\documentclass[11pt]{article}

\usepackage[T1]{fontenc}
\usepackage[utf8]{inputenc}
\usepackage[margin=1in]{geometry}
\usepackage{amsmath,amssymb,amsthm}
\usepackage{booktabs,tabularx}
\usepackage{microtype}
\usepackage{graphicx}
\usepackage{tikz}
\usepackage{float}
\usepackage{algorithm,algpseudocode}
\usepackage{needspace}
\usetikzlibrary{arrows.meta,positioning}
\usepackage[colorlinks=true,linkcolor=blue,citecolor=blue,urlcolor=blue]{hyperref}

\numberwithin{equation}{section}

\newtheorem{theorem}{Theorem}[section]
\newtheorem{proposition}[theorem]{Proposition}
\newtheorem{lemma}[theorem]{Lemma}
\newtheorem{corollary}[theorem]{Corollary}
\theoremstyle{definition}
\newtheorem{definition}[theorem]{Definition}
\newtheorem{example}[theorem]{Example}
\newtheorem{problem}{Problem}
\theoremstyle{remark}
\newtheorem{remark}[theorem]{Remark}
\theoremstyle{definition}

\newcommand{\RPAComparisonStatement}{%
For the family~\eqref{eq:rpa-review-family}, assume the input and coherent
access specified in Subsection~\ref{subsec:rpa-review}.
Both constructions use the full stable projector and the same
upper-row pseudoinverse recovery. For \(0<\varepsilon_T\le1\),
they return a block-encoding
with \(\|\widetilde T-T\|\le\varepsilon_T\).
For fixed \(V>0\) and \(0<\varepsilon_c\le(1-u)/(4V)\), the corresponding
estimates of \(e_c\) have additive error at most \(\varepsilon_c\),
with success probability at least \(2/3\).
The output normalizations and query bounds to \(U_{\mathcal H}\)
and its adjoint are those in Table~\ref{tab:rpa-review-comparison}.
The notation \(\widetilde O\) suppresses logarithms in \(u^{-1}\) and the inverse
target precisions; it does not include input preparation, coherent
coefficient loading, or gate-synthesis costs.}

\title{Near-Optimal Quantum Algorithm and Complexity Analysis\\
for Riccati Problems}
\author{Jingyao Wang$^{1,8}$, Yanqiao Wang$^{2,3,4}$, Bowen Li$^{5}$, Jin-Peng Liu$^{2,6,7,\dagger}$, Zhengfeng Ji$^{1,*}$
\vspace{2mm}\\
\footnotesize $^{1}$Department of Computer Science and Technology, Tsinghua University\\
\footnotesize $^{2}$ Yau Mathematical Sciences Center, Tsinghua University\\
\footnotesize $^{3}$ Qiuzhen College, Tsinghua University\\
\footnotesize $^{4}$ Institute for AI Industry Research (AIR), Tsinghua University\\
\footnotesize $^{5}$ Academy of Mathematics and Systems Science, Chinese Academy of Sciences\\
\footnotesize $^{6}$ Institute for Applied Mathematics, Tsinghua University\\
\footnotesize $^{7}$ Beijing Institute of Mathematical Sciences and Applications\\
\footnotesize $^{8}$ TraverseQuantum Co., Ltd.
}
\makeatletter
\renewcommand{\@thanks}{%
  \footnotetext[1]{Corresponding author: jizhengfeng@tsinghua.edu.cn}%
  \footnotetext[2]{Corresponding author: liujinpeng@tsinghua.edu.cn}%
}
\makeatother
\date{}

\begin{document}
\maketitle

\begin{abstract}
We develop quantum algorithms that construct block-encodings of solution
matrices for continuous- and discrete-time algebraic Riccati equations (CAREs and DAREs),
differential Riccati equations (DREs), and finite Riccati recursions. The Quantum Weighted Riesz Method, unifying four kinds of Riccati problems, constructs projectors onto the solution graphs using weighted Riesz branching
operators and recovers the solution matrices.
Under the stated access and normalization assumptions, our algorithms for CAREs and DAREs have near-optimal query complexity $\Theta\!\left(
 \mathcal R\,\alpha\right)$ up to logarithmic factors, with the generalized singularity factor
\(\mathcal R\) that reflects spectral separation, and the output normalization $\alpha$. For DREs and finite-horizon
problems, the query complexity has an additional initialization conditioning factor $\kappa_{\rm init}$.
We establish product query lower bounds for specified input-oracle
families, demonstrating necessary joint dependence on spectral
separation and initialization or output scale.
We also prove that selected-value promise problems are BQP-complete
on explicit circuit-generated families of single-control discrete-time
algebraic and zero-terminal differential Riccati equations.
Applications include classical feedback evaluation for linear--quadratic
control, a query-complexity comparison for stable random-phase
approximation Riccati equations in quantum chemistry, and heated boundary control with numerical validation.
\end{abstract}

\newpage

\tableofcontents

\newpage

\section{Introduction}
\label{sec:introduction}

\subsection{Background}
\label{subsec:riccati-operator-output}

Riccati equations are an important class of nonlinear matrix problems arising in optimal control and quantum chemistry. In control, four formulations arise depending on the time horizon and whether time is continuous or discrete. The continuous- and discrete-time algebraic Riccati equations (CARE and DARE) determine stationary value functions and feedback laws in linear--quadratic regulation (LQR); differential Riccati equations (DRE) and finite Riccati recursions (RR) govern finite-time control with prescribed boundary data~\cite{AndersonMoore2007OptimalControl}. In quantum chemistry, Riccati equations also determine random-phase-approximation(RPA) amplitudes and have further applications~\cite{RodenasRuizZhaoLee2026NonlinearMatrixEquations}.

A formulation underlies these nonlinear equations: the solution is closely connected to the solutions of a set of related linear matrix problems~\cite{Laub1979Schur,VanDooren1981Generalized}. 
The formulation admits a linear representation of the graph of its
solution, the subspace of pairs consisting of a state and its image under
the solution matrix.
For a CARE, the stabilizing solution is represented by an invariant graph
of a Hamiltonian matrix; for a DARE, it is represented by a deflating graph
of a symplectic pencil \cite{Laub1979Schur,VanDooren1981Generalized}.
A DRE propagates the initial or terminal graph through a Hamiltonian flow,
while a finite recursion propagates it through a linear-fractional lift
\cite{BittantiLaubWillems1991RiccatiEquation,LancasterRodman1995Riccati}.
These representations connect algebraic solutions to selected spectral
subspaces and finite-horizon solutions to the propagation of prescribed
graph data.

We study how block-encoding access to the input matrices can be used to
construct a block-encoding of the solution.
Here, a block-encoding is a unitary whose designated block approximates the
solution matrix divided by a known normalization
\cite{GilyenSuLowWiebe2019QSVT}.
This operator output allows the solution to be used in subsequent quantum
computation; a classical scalar feedback value is obtained by combining it
with state preparation and matrix-element estimation.
The algorithmic question is how spectral branch construction and graph
recovery contribute to the cost of producing this output.
We state the matrix access, spectral separation, invertibility, and
normalization assumptions for each construction, and quantify nonnormal
behavior through generalized singularity factors.

\subsection{From linear lifts to Riccati operators}
\label{subsec:lifts-to-operators}

The Quantum Weighted Riesz Method has three stages, shown in
Figure~\ref{fig:riccati-construction}: Embed, QET, and Recover.
Embed represents the Riccati problem by a linear matrix or a regular pencil
and its graph data. The quantum eigenvalue transformation (QET) stage
constructs the required weighted Riesz operators and uses them to form
the graph projector. Recover acts on its block rows to return a
block-encoding of the Riccati solution.

\begin{figure}[H]
\centering
\begin{tikzpicture}[
  stage/.style={draw=black!60,rounded corners=3pt,fill=blue!4,
    align=center,text width=12cm,minimum height=1.05cm,
    inner xsep=0.4cm,inner ysep=5pt,font=\small},
  branch/.style={draw=black!50,rounded corners=2pt,fill=white,
    align=center,text width=5.25cm,minimum height=1.7cm,
    inner sep=5pt,font=\small},
  flow/.style={-{Latex[length=2.2mm]},semithick,draw=black!70}
]
\node[stage] (embed) at (0,0)
  {\textbf{Embed}\\Linear lift and problem data};
\node[stage,minimum height=3.9cm] (qet) at (0,-3.1) {};
\node[align=center,font=\small] (qet-title) at (0,-1.8)
  {\textbf{QET}\\Weighted Riesz branch operators};
\node[branch] (algebraic) at (-3.05,-3.85)
  {\textbf{CARE / DARE}\\Weighted Riesz as\\graph projector};
\node[branch] (dynamic) at (3.05,-3.85)
  {\textbf{DRE / finite recursion}\\Combine initialization to form\\
   decaying graph projector};
\node[stage] (recover) at (0,-6.15)
  {\textbf{Recover}\\Lower block row $\times$ upper-block-row pseudoinverse};
\draw[flow] (embed.south) -- (qet.north);
\draw[semithick,draw=black!70] (qet-title.south) -- (0,-2.35);
\draw[flow] (0,-2.35) -| (dynamic.north);
\draw[flow] (0,-2.35) -| (algebraic.north);
\draw[flow] (qet.south) -- (recover.north);
\end{tikzpicture}
\caption{The Quantum Weighted Riesz Method.
The QET stage constructs weighted Riesz branch operators.
DREs and finite recursions combine these operators with initial data
to form a decaying graph projector (DGP); CARE and DARE directly
construct the limiting Riesz projector. Both routes use the same
block-row recovery to return a block-encoding of the Riccati solution,
with explicit output normalization and operator-norm error bounds
under the stated assumptions.}
\label{fig:riccati-construction}
\end{figure}
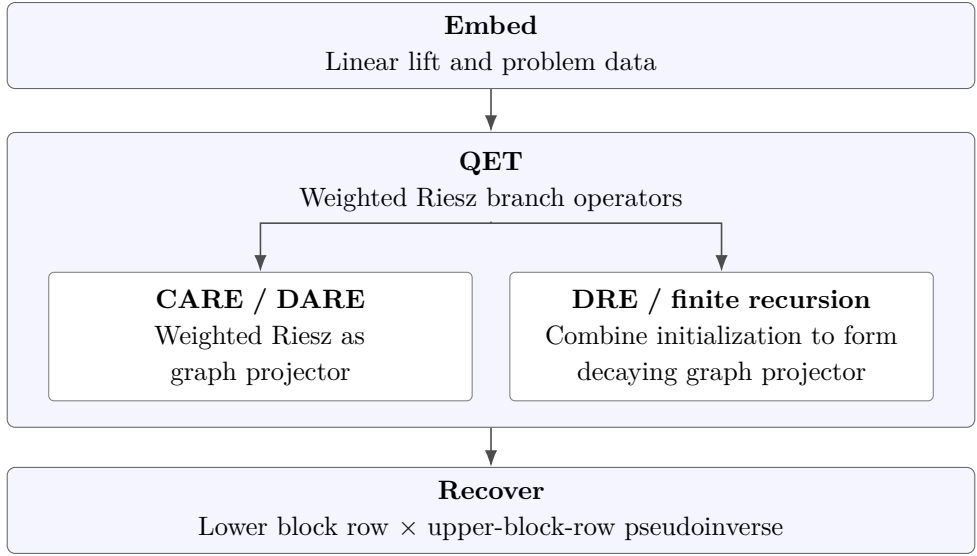

The graph of a Riccati solution matrix \(P\) is the \(n\)-dimensional
subspace \(\operatorname{graph}P=\{[v;Pv]:v\in\mathbb C^n\}\),
which represents the action of \(P\) on all state vectors \(v\).
In the real scalar case, this graph is a line whose slope is the solution \(p\).
A projector \(\Pi\) onto the solution graph preserves this relation
between the state and its image under \(P\).
With the coordinate embeddings \(E_1=[\mathsf I;0]\) and
\(E_2=[0;\mathsf I]\), its upper and lower block rows satisfy
\(E_2^*\Pi=P E_1^*\Pi\).
The upper block row has full row rank, so its Moore--Penrose
pseudoinverse~\cite{GilyenSuLowWiebe2019QSVT} is a right inverse, yielding
\begin{equation}
P=(E_2^*\Pi)(E_1^*\Pi)^+.
\label{eq:common-graph-recovery}
\end{equation}
For a prescribed state \(v\), the pseudoinverse supplies an input whose
projection has upper coordinate \(v\); applying the lower block row to
that input then gives \(Pv\).
This identity holds whether or not the projector is orthogonal.
For DREs and finite recursions, it applies to \(P(t)\) and \(P_k\)
using their decaying graph projectors.
See Subsection~\ref{subsec:dre-recovery} for the recovery construction
and Appendix~\ref{app:dre-recovery} for the general proof and
pseudoinverse norm bound.

Initialization also has a geometric meaning. For DREs and finite
recursions, \(R_0=[\mathsf I;P_0]\) spans the initial graph, and
\(\Pi\) denotes the spectral projector used to select its
initial branch component. If \(\|\Pi\|\) and \(\|P_0\|\)
are bounded by fixed constants, the absolute initialization nondegeneracy
factor satisfies
\[
\sigma_{\min}(\Pi R_0)
=\Theta\!\left(\sin\theta_{\min}
  (\operatorname{graph}P_0,\ker\Pi)\right).
\]
Here \(\theta_{\min}\) is the smallest principal angle~\cite[Sec.~2]{SimonciniSzyld2010Oblique}: the factor becomes
small as the initial graph approaches the discarded subspace
\(\ker\Pi\). Appendix~\ref{app:factor-geometry} gives the
general bounds and distinguishes this absolute factor from the relative
conditioning of the projected initial columns.

For algebraic equations, stability selects the graph to construct.
The top panels of Figure~\ref{fig:weighted-riesz-branch-schematics}
show the relevant regions: the left half-plane for the CARE Hamiltonian
and the finite spectrum inside the unit circle for the DARE pencil.
A Riesz operator with unit weight is a projector onto the selected branch~\cite[Chap.~I, Sec.~5.3]{Kato1995Perturbation}.
This complete projector supplies the block rows in
\eqref{eq:common-graph-recovery}.

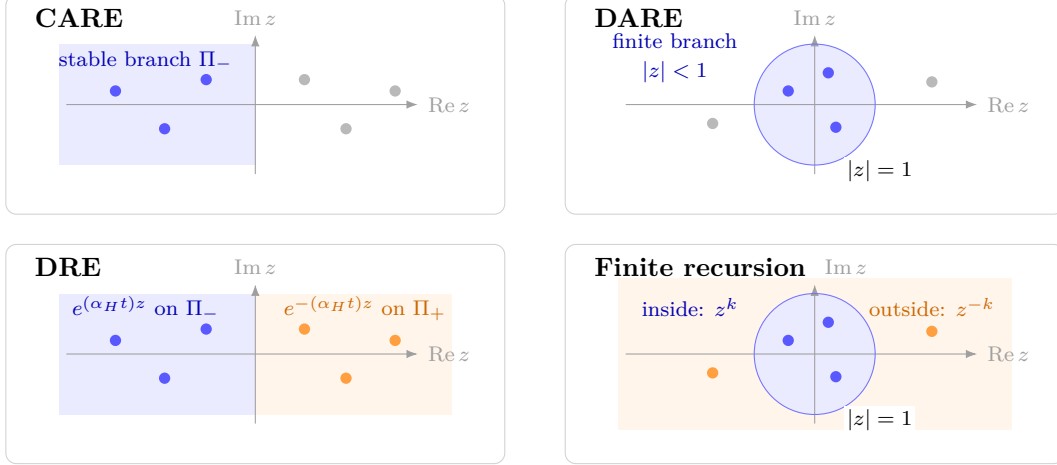
\begin{figure}[H]
\centering
\begin{tikzpicture}[
  branchpoint/.style={circle, fill=blue!65, inner sep=1.5pt},
  oppositepoint/.style={circle, fill=orange!75, inner sep=1.5pt},
  unselectedpoint/.style={circle, fill=gray!55, inner sep=1.5pt},
  spectralaxis/.style={-{Latex[length=1.5mm]}, gray!75, thin}
]
\begin{scope}[xshift=3.5cm,yshift=1.65cm]
  \draw[gray!35, rounded corners] (-3.3,-1.45) rectangle (3.3,1.45);
  \fill[blue!8] (-2.6,-0.8) rectangle (0,0.8);
  \draw[spectralaxis] (-2.5,0) -- (2.15,0)
    node[right,font=\scriptsize] {$\operatorname{Re}z$};
  \draw[spectralaxis] (0,-0.92) -- (0,0.92)
    node[above,font=\scriptsize] {$\operatorname{Im}z$};
  \node[font=\small\bfseries,anchor=west] at (-3.05,1.15) {CARE};
  \node[font=\scriptsize,blue!70!black] at (-1.45,0.58)
    {stable branch $\Pi_-$};
  \node[branchpoint] at (-1.85,0.18) {};
  \node[branchpoint] at (-1.2,-0.32) {};
  \node[branchpoint] at (-0.65,0.33) {};
  \node[unselectedpoint] at (1.85,0.18) {};
  \node[unselectedpoint] at (1.2,-0.32) {};
  \node[unselectedpoint] at (0.65,0.33) {};
\end{scope}

\begin{scope}[xshift=10.9cm,yshift=1.65cm]
  \draw[gray!35, rounded corners] (-3.3,-1.45) rectangle (3.3,1.45);
  \fill[blue!8] (0,0) circle (0.8);
  \draw[spectralaxis] (-2.5,0) -- (2.15,0)
    node[right,font=\scriptsize] {$\operatorname{Re}z$};
  \draw[spectralaxis] (0,-0.92) -- (0,0.92)
    node[above,font=\scriptsize] {$\operatorname{Im}z$};
  \draw[blue!55] (0,0) circle (0.8);
  \node[font=\scriptsize,fill=white,inner sep=1pt] at (0.86,-0.87) {$|z|=1$};
  \node[font=\small\bfseries,anchor=west] at (-3.05,1.15) {DARE};
  \node[font=\scriptsize,blue!70!black,align=center] at (-1.85,0.62)
    {finite branch\\$|z|<1$};
  \node[branchpoint] at (-0.35,0.18) {};
  \node[branchpoint] at (0.28,-0.30) {};
  \node[branchpoint] at (0.18,0.42) {};
  \node[unselectedpoint] at (-1.35,-0.25) {};
  \node[unselectedpoint] at (1.55,0.30) {};
\end{scope}

\begin{scope}[xshift=3.5cm,yshift=-1.65cm]
  \draw[gray!35, rounded corners] (-3.3,-1.45) rectangle (3.3,1.45);
  \fill[blue!8] (-2.6,-0.8) rectangle (0,0.8);
  \fill[orange!8] (0,-0.8) rectangle (2.6,0.8);
  \draw[spectralaxis] (-2.5,0) -- (2.15,0)
    node[right,font=\scriptsize] {$\operatorname{Re}z$};
  \draw[spectralaxis] (0,-0.92) -- (0,0.92)
    node[above,font=\scriptsize] {$\operatorname{Im}z$};
  \node[font=\small\bfseries,anchor=west] at (-3.05,1.15) {DRE};
  \node[font=\scriptsize,blue!70!black] at (-1.45,0.61)
    {$e^{(\alpha_Ht)z}$ on $\Pi_-$};
  \node[font=\scriptsize,orange!80!black] at (1.45,0.61)
    {$e^{-(\alpha_Ht)z}$ on $\Pi_+$};
  \node[branchpoint] at (-1.85,0.18) {};
  \node[branchpoint] at (-1.2,-0.32) {};
  \node[branchpoint] at (-0.65,0.33) {};
  \node[oppositepoint] at (1.85,0.18) {};
  \node[oppositepoint] at (1.2,-0.32) {};
  \node[oppositepoint] at (0.65,0.33) {};
\end{scope}

\begin{scope}[xshift=10.9cm,yshift=-1.65cm]
  \draw[gray!35, rounded corners] (-3.3,-1.45) rectangle (3.3,1.45);
  \fill[orange!8] (-2.6,-1.0) rectangle (2.6,1.0);
  \fill[blue!8] (0,0) circle (0.8);
  \draw[spectralaxis] (-2.5,0) -- (2.15,0)
    node[right,font=\scriptsize] {$\operatorname{Re}z$};
  \draw[spectralaxis] (0,-0.92) -- (0,0.92)
    node[above right,font=\scriptsize] {$\operatorname{Im}z$};
  \draw[blue!55] (0,0) circle (0.8);
  \node[font=\scriptsize,fill=white,inner sep=1pt] at (0.86,-0.87) {$|z|=1$};
  
  \node[font=\small\bfseries,anchor=west] at (-3.05,1.15)
    {Finite recursion};
  \node[font=\scriptsize,blue!70!black] at (-1.65,0.62)
    {inside: $z^k$};
  \node[font=\scriptsize,orange!80!black] at (1.55,0.62)
    {outside: $z^{-k}$};
  \node[branchpoint] at (-0.35,0.18) {};
  \node[branchpoint] at (0.28,-0.30) {};
  \node[branchpoint] at (0.18,0.42) {};
  \node[oppositepoint] at (-1.35,-0.25) {};
  \node[oppositepoint] at (1.55,0.30) {};
\end{scope}
\end{tikzpicture}
\caption{Spectral branches and weights in the QET stage.
The top panels select the blue stable branches for CARE and DARE;
gray points are unselected. The bottom panels include both blue and orange
branches, with forward weights on the left or inside and inverse weights
on the right or outside. Points are schematic spectral locations, and
each circle marks \(|z|=1\); the DRE weights use
\(\widehat{\mathcal H}_{\rm DRE}=\mathcal H_{\rm DRE}/\alpha_H\).}
\label{fig:weighted-riesz-branch-schematics}
\end{figure}

For a DRE or finite recursion, the initial graph generally has components
in both spectral branches. Its forward evolution can grow on the right
half-plane or outside the unit circle. Under the required branch-rank
conditions, changing the graph coordinates lets us represent the same
forward-evolved graph using inverse evolution on the growing branch and
forward evolution on the other. In the DRE panel, this gives the weights
\(e^{\alpha_Htz}\) on the left and \(e^{-\alpha_Htz}\) on the right,
where \(z\) belongs to the spectrum of the normalized Hamiltonian
\(\widehat{\mathcal H}_{\rm DRE}=\mathcal H_{\rm DRE}/\alpha_H\).
The recursion panel uses \(z^k\) inside the unit circle and \(z^{-k}\)
outside. These scalar weights decay on their respective branches;
the initial graph provides the relation between the two branches needed
to recover the full evolving graph.

For dynamic problems, the QET stage combines the weighted branch
operators with the initial data to form a decaying graph projector
onto the evolving graph. Its lower block row is the solution matrix times
its upper block row, whose full row rank permits recovery by a Moore--Penrose
pseudoinverse. The Recover stage therefore produces the matrix specified by
\eqref{eq:common-graph-recovery} directly from its graph projector.
The query analysis accounts separately for nonnormal amplification of the
branch operators and for the conditioning and normalization of recovery.

Contour--resolvent linear combinations of unitaries (LCU) implement
the branch operators by approximating contour integrals with weighted
sums of shifted inverses~\cite{HaleHighamTrefethen2008Contour,TakahiraOhashiSogabeUsuda2022Contour}. The contour selects the spectral region, while
the scalar weight specifies the action within that region.
This construction applies to ordinary matrices and regular pencils under
their respective contour and inverse assumptions.
For half-plane projectors of ordinary Hamiltonian matrices, we also develop
polynomial selectors implemented through the Weyl
linear-combination-of-Hermitian-matrices (LCHM) construction
\cite{WangLiangChenLiu2026LCHM}.
The weighted exponential and power operators are supplied by the contour
construction. Section~\ref{sec:quantum-tools} gives the common operators
and implementation assumptions used in the Riccati algorithms.

\subsection{Main results}
\label{subsec:main-results}

\paragraph{Quantum algorithms and feedback evaluation.}
We construct block-encodings of all four Riccati solution operators.
Under the stated coherent access, spectral separation, and
normalization assumptions, the matrix-oracle query bounds are
\begin{equation}
\begin{aligned}
Q_{\rm CARE}
 &\;=\widetilde O\!\left(
 \mathcal R_{\rm CARE}\,\alpha_X\right),\\
Q_{\rm DARE}
 &\;=\widetilde O\!\left(
 \mathcal R_{\rm DARE}\,\alpha_X\right),\\
Q_{\rm DRE}
 &\;=\widetilde O\!\left(
 \mathcal R_{\rm DRE}\,\kappa(\Pi_+R_0)\,\alpha_{P(t)}\right),\\
Q_{\rm RR}
 &\;=\widetilde O\!\left(
 \mathcal R_{\rm RR}\,\kappa(\Pi_>R_0)\,\alpha_{P_k}\right).
\end{aligned}
\label{eq:informal-query-upper}
\end{equation}
Each bound counts matrix queries for one call to a solution encoding with
operator-norm error \(\varepsilon\); \(\widetilde O\) hides logarithmic
dependence on accuracy and the displayed parameters.
These compact bounds assume bounded normalization-reduction overheads;
the dynamic bounds also require calibrated initial-column scales.
The actual output normalizations include the supplied solution-norm
bounds and the scales of the constructed projectors. General bounds
with all declared scales and re-encoding costs appear in the appendices.
The generalized singularity factor \(\mathcal R\) measures how close the
shifted matrices along a separating contour are to singularity, relative
to their declared scale. A large value means that some shifted matrix is
nearly singular. This can reflect narrow spectral separation, which eigenvalue distances alone do not
capture. For an ordinary Hamiltonian lift \(H\),
Proposition~\ref{prop:factor-condition-number-app} in
Appendix~\ref{app:complexity-factors} gives
\(\mathcal R=O(\kappa(H))\) for the specified rectangles
when the supplied normalization is comparable to \(\|H\|\) and the
minimum singular value over imaginary-axis shifts is comparable to
\(\sigma_{\min}(H)\). Here \(\kappa(H)=\|H\|\|H^{-1}\|\).
Thus \(\mathcal R\) extends the familiar conditioning of matrix inversion
to the separation of spectral branches.

The projector normalization records the geometry of the selected branch
and the scale of its quantum construction. The projectors
\(\Pi_-,\Pi_<\) select the stable CARE and finite unit-disk DARE
subspaces; \(\Pi_+,\Pi_>\) select the right-half-plane and
exterior-unit-disk branches used in the dynamic problems.
For CARE and DARE, the complete upper block row has full row rank,
and its pseudoinverse is controlled by the solution scale. This
dependence enters \(\alpha_X\), together with the actual projector
normalization. In particular, the complete DARE pencil projector includes
the normalization of its right factor \(L\).
For DRE and recursion, \(R_0=[\mathsf I;P_0]\) represents the initial
graph, and \(\kappa(\Pi R_0)\) measures how unevenly the projection
represents its independent directions. When \(\Pi R_0\) has full
column rank, this relative measure is
\(\kappa(\Pi R_0)=\|\Pi R_0\|/\sigma_{\min}(\Pi R_0)\), while
\(\sigma_{\min}(\Pi R_0)\) is the absolute initialization
nondegeneracy factor. The general implementation cost uses the actual
encoding normalization \(\alpha_{\Pi R_0}\) through the ratio
\(\alpha_{\Pi R_0}/\sigma_{\min}(\Pi R_0)\), which becomes a
condition number under the stated calibration. Appendix~\ref{app:factor-geometry}
develops these distinctions.

We also evaluate classical LQR feedback from the Riccati solution.
For scalar control, \(u(t)=-c^*P(t)x(t)\), where \(c=b/r\) for input
column \(b\) and control weight \(r>0\).
With vector-preparation costs accounted for separately, amplitude estimation
gives the informal query bound
\cite{BrassardHoyerMoscaTapp2002AmplitudeEstimation}
\begin{equation}
Q_u=\widetilde O\!\left(
 \frac{
 \mathcal R_{\rm DRE}\,\kappa(\Pi_+R_0)\,
 \alpha_{P(t)}^{\,2}\|c\|\|x(t)\|
 }{\varepsilon_u}
 \log\frac1\delta
\right).
\label{eq:informal-scalar-query}
\end{equation}
Here \(\alpha_{P(t)}\) is the output normalization, and
\(\varepsilon_u\) and \(\delta\) are the target error and failure probability.
One factor \(\alpha_{P(t)}\) enters the solution-circuit cost and the
other enters classical scalar readout.
We use one-dimensional control only to illustrate the procedure.
Higher-dimensional controls follow the same construction and componentwise
readout, with the total readout cost depending on the requested components.

\paragraph{Query lower bounds and computational hardness.}
For each Riccati problem, we construct a two-input family with independent
parameters \(0<\mu_1,\mu_2\le1/8\) and generalized singularity factor
\(\Theta(\mu_1^{-1})\). The algebraic families use a common actual
output normalization \(\alpha_X=\Theta(\mu_2^{-1})\); the dynamic
families have relative initialization condition number
\(\Theta(\mu_2^{-1})\). A query-hybrid argument for explicitly
specified smooth Julia input block-encodings gives
\begin{equation}
\begin{aligned}
q_{\rm CARE}
 &\;=\Omega\!\left(\mathcal R_{\rm CARE}\alpha_X\right),&
q_{\rm DARE}
 &\;=\Omega\!\left(\mathcal R_{\rm DARE}\alpha_X\right),\\
q_{\rm DRE}
 &\;=\Omega\!\left(\mathcal R_{\rm DRE}\kappa(\Pi_+R_0)\right),&
q_{\rm RR}
 &\;=\Omega\!\left(\mathcal R_{\rm RR}\kappa(\Pi_>R_0)\right).
\end{aligned}
\label{eq:informal-query-lower}
\end{equation}
Thus each family requires \(\Omega((\mu_1\mu_2)^{-1})\) queries for
coherent solution block-encoding with its declared public normalization
and common signal space. In the algebraic examples, \(\|X\|\) grows
without bound as \(\mu_2\) decreases, and the allowed error is at most
a fixed fraction of the smaller solution norm. Our direct-projector algorithm
attains an actual normalization of order \(\|X\|\) on these
families: explicit small contours keep the raw projector normalization
bounded, as proved in Proposition~\ref{prop:algebraic-product-normalization-app}.
This construction uses no normalization reduction.
The dynamic families have bounded solution norms, constant output
normalization and absolute error, and evaluation time or step
\(\Theta(\mu_1^{-1}\log(1/\mu_2))\).
For CARE and DARE, these lower bounds match
\eqref{eq:informal-query-upper} up to logarithmic factors, establishing
near-optimal worst-case query complexity at the actual output normalizations
achieved by our algorithms on these input-oracle families.
For DRE and finite recursion, the lower bounds show that the generalized
singularity factor and initialization condition number must enter
multiplicatively in the worst case.

We further prove that constructing DARE and DRE solution block-encodings
with constant normalization and accuracy is BQP-hard, even with a single
control input.
The associated selected-value decision problems are BQP-complete on
explicit circuit-generated families.

We illustrate the benefits of our algorithms through two applications: electronic correlation energy estimation and optimal feedback control.

\paragraph{Quantum chemistry.}
RPA correlation energies follow from a Riccati amplitude equation~\cite{RodenasRuizZhaoLee2026NonlinearMatrixEquations}. In our normalized family, $u>0$ denotes the smallest eigenvalue of the electronic stability matrix. It tends to zero during molecular bond stretching~\cite{TahirRen2019RPAChannels}, making the dependence on $u^{-1}$ central to the cost of RPA calculations. Our method reduces the number of queries needed for energy readout on this family from $\widetilde O(u^{-3}/\varepsilon_c)$ to $\widetilde O(u^{-2}/\varepsilon_c)$. This improves the dependence on $u^{-1}$ over the matched baseline in previous RPA work~\cite{RodenasRuizZhaoLee2026NonlinearMatrixEquations}.

\paragraph{Heated boundary control (HBC).}
For a path of $n$ thermal nodes with a heater at one endpoint, at fixed parameters and over a bounded time range, the spectral factor, initialization inverse, and output normalization remain $O(1)$ in $n$, allowing the optimal feedback value $u$ to be estimated to error $\varepsilon_u$ and failure probability $\delta$ using $O(\varepsilon_u^{-1}\log(1/\delta)\log^3(e+1/\varepsilon_u))$ matrix and state-oracle calls, given coherent access to the current state with known bounded normalization. We also conduct numerical experiments to demonstrate the efficiency of our algorithm.

\subsection{Related work}
\label{subsec:related-work}

Classical CARE and DARE algorithms use invariant-subspace, generalized-
eigenvalue, and matrix-sign methods, while Hamiltonian flows and
linear-fractional lifts describe differential and recursive Riccati problems
\cite{Laub1979Schur,Roberts1980Sign,VanDooren1981Generalized,
Byers1987MatrixSign,KenneyLaub1995MatrixSign,Mehrmann1991AutonomousLQ,
AbouKandilFreilingIonescuJank2003Riccati}.
For state dimension \(n\), dense Schur or QZ factorization costs
\(O(n^3)\), as does one dense matrix-sign iteration, and a \(K\)-step
dense recursion costs roughly \(O(Kn^3)\).
Large-scale and distributed-control methods instead exploit problem
structure or finite-section modeling
\cite{BennerLiPenzl2008LargeScaleRiccati,Simoncini2016LinearMatrixEquations,
CurtainZwart1995InfiniteDimensional,LasieckaTriggiani2000PDEControl}.

The quantum part of the paper builds on block-encoding, quantum linear-system
methods, QSP/QSVT, and qubitization
\cite{HarrowHassidimLloyd2009HHL,ChildsKothariSomma2017QLS,
LowChuang2017QSP,GilyenSuLowWiebe2019QSVT,LowChuang2019Qubitization,
MartynRossiTanChuang2021GrandUnification}.
Related quantum-control and matrix-equation work uses linear ODEs, Riesz
projectors, contour resolvents, and singular-value transformations
\cite{Krovi2024QuantumOptimalControl,RodenasRuizZhaoLee2026NonlinearMatrixEquations,
WangLiu2026SignEmbedding}.
We use these primitives to access the linear objects in the Riccati lifts.

The analysis also draws on matrix-function perturbation theory, resolvent
bounds for nonnormal matrices, contour quadrature, Faber approximation, Weyl
LCHM, and generalized quantum signal processing
\cite{Kato1995Perturbation,Higham2008Functions,
HaleHighamTrefethen2008Contour,TrefethenEmbree2005Spectra,
BeckermannReichel2009Faber,LowSu2024QEP,MotlaghWiebe2023GQSP,
JiangAn2026ContourQET,GutierrezLaneveSanz2026ArbitraryQET,
WangLiangChenLiu2026LCHM}.
Related matrix-geometric-mean and oscillator results provide comparison
points, but their lower bounds are not imported into the Riccati statements
here \cite{LiuWangWildeZhang2025GeometricMeans,
BabbushBerryKothariSommaWiebe2023Oscillators}.

The present paper combines these tools with the four Riccati linear
structures and graph recovery, and states the separation, resolvent, graph,
and readout conditions needed for each construction.
Finite-section PDE examples and physical control readouts involve additional
discretization, state preparation, trajectory generation, measurement, and
classical reconstruction costs, which are treated separately.

\begin{remark}[Relation to quantum oblique eigenprojection]
After we submitted this work to arXiv, Dalzell and
Su~\cite{DalzellSu2026QuantumObliqueEigenprojection} posted a preprint
introducing two-sided block preconditioning for oblique eigenprojections
under spectral-set assumptions. Under additional assumptions of
numerical-range separation or well-conditioned diagonalizability with a
known half-plane spectral gap for the invariant restrictions, which
strengthen the spectral hypotheses of our contour construction, their
primitive can improve our CARE quantum implementation by reducing
projector normalization. Our framework additionally covers weighted
Riccati evolution and regular matrix pencils with singular \(L\), beyond
the direct scope of their projection theorems. Furthermore, our
contour-dependent implementation can lower the query complexity relative
to their uniform half-plane bound when narrow gaps and strong nonnormality
occur in different spectral clusters, with input scales and output
normalization accounted for.
\end{remark}

\section{Preliminaries and quantum constructions}
\label{sec:quantum-tools}

We first specify the Riccati problems and the quantum output requested for
each of them. We then define the weighted Riesz operators used to access
their linear representations and give two constructions of these operators.
This order separates the control problem, the operator to be computed, and
the assumptions needed to implement it.

\subsection{From LQR to Riccati problems}
\label{subsec:lqr-riccati-problems}

Throughout the paper, \(\|\cdot\|\) denotes the Euclidean vector norm or
its induced matrix norm, \(^*\) denotes the adjoint, and \(\mathsf I\)
is the identity of the required dimension.
Consider a state \(x\in\mathbb C^n\), a control \(u\in\mathbb C^m\),
and time-independent matrices \(A\in\mathbb C^{n\times n}\) and
\(B\in\mathbb C^{n\times m}\).
The state, control, and terminal cost matrices satisfy
\(Q\succeq0\), \(R_c\succ0\), and \(P_T\succeq0\), respectively.
Continuous-time finite-horizon LQR minimizes
\begin{equation}
\begin{aligned}
J_T(u)&=\int_0^T\!\left(x(t)^*Qx(t)+u(t)^*R_cu(t)\right)\,\mathrm dt
       +x(T)^*P_Tx(T),\\
\dot x(t)&=Ax(t)+Bu(t),\qquad x(0)=x_0.
\end{aligned}
\label{eq:lqr-continuous-cost}
\end{equation}
Its discrete-time counterpart, with \(T\) control steps, minimizes
\begin{equation}
\begin{aligned}
J_T(u)&=\sum_{j=0}^{T-1}\left(x_j^*Qx_j+u_j^*R_cu_j\right)
          +x_T^*P_Tx_T,\\
x_{j+1}&=Ax_j+Bu_j,\qquad x_0\ \text{given}.
\end{aligned}
\label{eq:lqr-discrete-cost}
\end{equation}
The infinite-horizon problems use the corresponding infinite integral or
sum without a terminal term.
Under the usual existence conditions, the optimal value is a quadratic form
in the current state. Its matrix determines both the optimal cost and the
feedback law, leading to the four Riccati problems below
\cite{AndersonMoore2007OptimalControl,BittantiLaubWillems1991RiccatiEquation,
LancasterRodman1995Riccati}.
We write \(G=BR_c^{-1}B^*\) when referring to the control data.

For continuous-time infinite-horizon LQR, stabilizability of \((A,B)\)
and detectability of \((Q^{1/2},A)\) give
\(V_\infty(x)=x^*Xx\) and
\(u^*(t)=-R_c^{-1}B^*Xx(t)\) \cite[Secs.~3.1--3.2 and App.~B]{AndersonMoore2007OptimalControl}.
The matrix \(X\) is the stabilizing CARE solution \cite{Laub1979Schur,LancasterRodman1995Riccati}.

\begin{definition}[Continuous-time algebraic Riccati problem]
\label{def:problem-care}
Given \(A\in\mathbb C^{n\times n}\), \(B\in\mathbb C^{n\times m}\),
\(Q\succeq0\), and \(R_c\succ0\), set \(G=BR_c^{-1}B^*\).
Assume \((A,B)\) is stabilizable and \((Q^{1/2},A)\) is detectable
in continuous time. Find the unique positive semidefinite stabilizing
solution \(X\) of
\begin{equation}
A^*X+XA-XGX+Q=0,\qquad A-GX\ \text{is Hurwitz}.
\label{eq:problem-care}
\end{equation}
A matrix is Hurwitz if all its eigenvalues have strictly negative real part.
The stated LQR conditions guarantee existence and uniqueness of \(X\)
\cite{AndersonMoore2007OptimalControl,LancasterRodman1995Riccati}.
\end{definition}

The algebraic problems below use these standard LQR assumptions;
\(G\) and \(Q\) may be singular.
For discrete-time infinite-horizon LQR, the analogous stabilizability and
detectability conditions give \(V_\infty(x)=x^*Xx\) with feedback
\(u_j^*=-(R_c+B^*XB)^{-1}B^*XA\,x_j\) \cite[Sec.~3.3]{AndersonMoore2007OptimalControl}.

\begin{definition}[Discrete-time algebraic Riccati problem]
\label{def:problem-dare}
Given \(A\in\mathbb C^{n\times n}\), \(B\in\mathbb C^{n\times m}\),
\(Q\succeq0\), and \(R_c\succ0\), assume \((A,B)\) is stabilizable
and \((Q^{1/2},A)\) is detectable in discrete time. Find the unique
positive semidefinite stabilizing solution of
\begin{equation}
X=Q+A^*XA-A^*XB(R_c+B^*XB)^{-1}B^*XA.
\label{eq:problem-dare}
\end{equation}
With \(G=BR_c^{-1}B^*\), the equivalent form is \cite[Eq.~(24)]{Poloni2020RiccatiIterations}
\begin{equation}
X=Q+A^*X(\mathsf I+GX)^{-1}A,\qquad
F=(\mathsf I+GX)^{-1}A,\quad \rho(F)<1.
\label{eq:problem-dare-compact}
\end{equation}
The stated LQR conditions guarantee existence and uniqueness of the
stabilizing solution \cite{AndersonMoore2007OptimalControl,LancasterRodman1995Riccati}.
The displayed inverses exist because \(X\succeq0\) and \(R_c\succ0\).
Here \(\rho(F)\) denotes the spectral radius of the closed-loop matrix
\(F\); the condition \(\rho(F)<1\) means that \(F\) is Schur stable \cite[Sec.~3.3]{AndersonMoore2007OptimalControl}.
\end{definition}

On a finite continuous-time horizon, the value matrix depends on time:
\(V(t,x)=x^*P_{\rm LQR}(t)x\), with feedback
\(u^*(t)=-R_c^{-1}B^*P_{\rm LQR}(t)x(t)\).
Dynamic programming gives the terminal-value equation \cite[Secs.~2.2--2.3]{AndersonMoore2007OptimalControl}
\begin{equation}
-\dot P_{\rm LQR}
=Q+A^*P_{\rm LQR}+P_{\rm LQR}A-P_{\rm LQR}GP_{\rm LQR},
\qquad P_{\rm LQR}(T)=P_T.
\label{eq:lqr-terminal-dre}
\end{equation}
For the matrix algorithms, we fix an initial-value convention so that the
requested solution can be represented by a forward linear flow.

\begin{definition}[Differential Riccati problem]
\label{def:problem-dre}
Given time-independent \(A\), Hermitian \(G,Q,P_0\), and a requested time
\(t\in[0,T]\), find \(P(t)\) solving
\begin{equation}
\dot P(t)=-Q-A^*P(t)-P(t)A+P(t)GP(t),\qquad P(0)=P_0.
\label{eq:problem-dre}
\end{equation}
The solution is promised to exist throughout \([0,T]\).
\end{definition}

To express \eqref{eq:lqr-terminal-dre} in this convention, set
\(s=T-t\) and \(P(s)=P_{\rm LQR}(T-s)\).
The inputs to Definition~\ref{def:problem-dre} are then
\((-A,-G,-Q,P_0=P_T)\), where \(A,G,Q\) on the right are the original
control matrices. The physical control time \(t\) therefore corresponds to
the algorithmic request \(s=T-t\).
This sign change is one reason for allowing general Hermitian DRE data.

For the discrete finite-horizon problem, it is convenient to index the value
matrix by the number of remaining steps: \(V_j(x)=x^*P_jx\), with
\(P_0=P_T\). One Bellman update adds a control step \cite[Sec.~2.4]{AndersonMoore2007OptimalControl}.

\begin{definition}[Finite Riccati recursion]
\label{def:problem-recursion}
Given \(A\), positive semidefinite \(G,Q,P_0\), and an integer \(k\ge0\),
find \(P_k\) generated by
\begin{equation}
P_{j+1}=Q+A^*P_j(\mathsf I+GP_j)^{-1}A,
\qquad j=0,\ldots,k-1.
\label{eq:problem-recursion}
\end{equation}
The positive semidefinite assumptions ensure that each inverse exists and
that every iterate is positive semidefinite.
\end{definition}

Taking \(G=BR_c^{-1}B^*\) recovers the \(k\)-step LQR value.
For \(k\ge1\), the first optimal control uses the next-stage value matrix:
\begin{equation}
u_0^*=-(R_c+B^*P_{k-1}B)^{-1}B^*P_{k-1}A\,x_0.
\label{eq:recursion-first-control}
\end{equation}
Appendix~\ref{app:lqr-conventions} verifies the finite-horizon identities,
including the time reversal and positive semidefiniteness of the recursion.
These definitions fix the desired solution, branch, and initial or terminal
data. Spectral separation, graph invertibility, and matrix-access conditions
are additional algorithmic assumptions. In particular, invertibility of
\(A\) is needed by the recursion lift used later, but is not part of the
recursion problem itself.

\subsection{Inputs, outputs, and inverse primitives}
\label{subsec:block-encodings-inverse}

The preceding definitions specify matrices as mathematical outputs.
Our quantum algorithms represent these matrices as blocks of unitary
operators, with the normalization and approximation error stated explicitly.

\begin{definition}[Block-encoding and requested output]
\label{def:block-encoding}
A unitary \(U_T\) is an \((\alpha_T,a_T,\varepsilon_T)\) block-encoding
of \(T\) if
\begin{equation}
\left\|T-\alpha_T
(\langle0^{a_T}|\otimes\mathsf I)U_T
(|0^{a_T}\rangle\otimes\mathsf I)\right\|\le\varepsilon_T.
\label{eq:block-encoding}
\end{equation}
Here \(\alpha_T>0\) is the normalization and \(a_T\) is the number
of ancilla qubits.
The matrix obtained by multiplying the indicated block by \(\alpha_T\)
is the decoded matrix \(\widetilde T\).
For any problem in Definitions~\ref{def:problem-care}--\ref{def:problem-recursion},
the requested output is such a unitary for \(T=X\), \(P(t)\), or \(P_k\),
with a declared \(\alpha_T\) and
\(\|\widetilde T-T\|\le\varepsilon\).
Rectangular matrices use separate input and output signal spaces, with
dimension padding when needed.
\end{definition}

An exact input has \(\varepsilon_T=0\); an approximate input encodes a
nearby matrix. Whenever an exact-input construction is applied to approximate
data, we account for this perturbation in the final decoded error.
One matrix-oracle query is one call to a supplied block-encoding or its
adjoint. We count controlled calls at the same unit cost when controlled
access is supplied. Access may be given to the original coefficient matrices
or to an assembled lift; constructing the latter from the former has its own
query cost.

The common nonlinear operation on encoded matrices is inversion.
We use the following interface to a QSVT inverse construction
\cite{GilyenSuLowWiebe2019QSVT,ChildsKothariSomma2017QLS}.

\begin{proposition}[Inverse and pseudoinverse primitive]
\label{prop:local-inverse}
Let \(U_T\) be an exact block-encoding of \(T\) with normalization
\(\alpha_T\), and suppose the nonzero singular values of \(T\) lie in
\([\gamma,\alpha_T]\), where \(0<\gamma\le\alpha_T\) is supplied to the inverse
construction. Under the controlled-access and singular-value-transformation
assumptions of that primitive, an encoding of the Moore--Penrose
pseudoinverse \(T^+\) can be constructed
with normalization \(O(1/\gamma)\), decoded error
\(\varepsilon_{\rm inv}>0\) satisfying \(\alpha_T\varepsilon_{\rm inv}\le1\),
and query cost
\begin{equation}
O\!\left(
\frac{\alpha_T}{\gamma}
\log\!\left(e+\frac{1}{\gamma\varepsilon_{\rm inv}}\right)
\right).
\label{eq:scaled-inverse-query}
\end{equation}
For square invertible \(T\), the target is \(T^{-1}\).
For a rectangular matrix, \(T^+\) vanishes on the zero singular subspace.
\end{proposition}

Appendix~\ref{app:inverse-perturbation} proves the scaling in this proposition.
The singular-value threshold fixes both the inverse normalization and the
cost of resolving small nonzero singular values.
When \(\gamma\) estimates the smallest nonzero singular value within a
constant factor, these bounds can be written in terms of \(\|T^+\|\).
For approximate inputs, Lemma~\ref{lem:inverse-input-perturbation} in
Appendix~\ref{app:inverse-perturbation} quantifies how small input errors
preserve invertibility and affect the inverse.
Pseudoinverse perturbations additionally require control of rank and the
nonzero singular-value gap.

The constructions below also use coherent preparation of LCU coefficients
and coherent access to node-dependent inverse circuits or polynomial phases.
Their query bounds assume these operations can be performed without
measuring the node register. Coefficient computation, state preparation,
arithmetic, and rotation synthesis are recorded separately from matrix-oracle
queries. This convention makes the operator normalization and oracle cost
comparable across the two constructions.

\subsection{Weighted Riesz operators for the four Riccati problems}
\label{subsec:weighted-riesz-targets}

The common geometric object is the graph of the solution. For a matrix
\(P\in\mathbb C^{n\times n}\),
\[
\operatorname{graph}P=\{[u;Pu]:u\in\mathbb C^n\}.
\]
If the columns of \([U;V]\), with \(U,V\in\mathbb C^{n\times n}\),
form a basis of this graph, then \(V=PU\) and \(U\) is invertible.
Thus \(P=VU^{-1}\) is its matrix slope. In the scalar case, the graph
is simply the line \(V=pU\) with slope \(p\).
A projector \(\Pi\) onto this graph sends each vector to the graph
along \(\ker\Pi\): the projected component lies on the graph, and the
removed component lies in the kernel. This direction need not be
perpendicular to the graph. The complete projector encodes the same
matrix slope through \(E_2^*\Pi=P E_1^*\Pi\). Its upper block row has
full row rank, so the pseudoinverse recovery
in~\eqref{eq:common-graph-recovery} gives \(P\).

The block-encoding model supplies access to a linear matrix or pencil.
For algebraic equations, a Riesz projector selects the target graph from
the lift. Dynamic problems combine two spectral branches with the initial
data to construct the evolving graph projector. Weighted Riesz operators
select these branches and apply a scalar weight within each one.
The weight is constant for algebraic problems and carries time or step
dependence for dynamic problems.

Fix a bounded spectral domain \(D_\chi\) with positively oriented, piecewise smooth boundary
\(\Gamma_\chi\). The contour lies in the resolvent set, winds once around
the selected spectral component \(K_\chi\), and has winding number zero
around the other components. Let \(g\) be holomorphic on a neighborhood of the closure of \(D_\chi\).

\begin{definition}[Weighted Riesz operator]
\label{def:weighted-riesz-interface-main}
For an ordinary matrix \(M\), define~\cite[Chap.~I, Sec.~5.6]{Kato1995Perturbation}
\begin{equation}
\mathfrak R_\chi[g;M,\mathsf I]
=\frac{1}{2\pi\mathrm i}\oint_{\Gamma_\chi}
g(z)(z\mathsf I-M)^{-1}\,\mathrm dz.
\label{eq:ordinary-weighted-riesz}
\end{equation}
For a regular pencil \(M-zL\), apply the same domain and holomorphy
conditions to its finite generalized spectrum and set
\begin{equation}
\mathfrak R_\chi[g;M,L]
=\frac{1}{2\pi\mathrm i}\oint_{\Gamma_\chi}
g(z)(zL-M)^{-1}L\,\mathrm dz.
\label{eq:pencil-weighted-riesz}
\end{equation}
Regularity means that \(\det(M-zL)\) is not identically zero~\cite[Sec.~1]{VanDooren1981Generalized}; the contour
must satisfy \(\det(zL-M)\ne0\) at every point.
\end{definition}

For \(g=1\), these operators are the Riesz projector and the right
deflating projector onto the selected branch, respectively.
They need not be orthogonal. If \(g\) extends holomorphically to the full
spectrum of an ordinary matrix, the weighted operator equals
\(g(M)\Pi_\chi\).
The pencil integral selects finite modes; infinite modes do not contribute
to this finite-domain integral. A DARE application must separately ensure
that the selected finite deflating subspace represents its stabilizing graph.
When \(L\) is invertible,
\begin{equation}
\mathfrak R_\chi[g;M,L]
=\mathfrak R_\chi[g;L^{-1}M,\mathsf I],
\label{eq:pencil-ordinary-reduction}
\end{equation}
whereas the pencil form remains meaningful for singular \(L\).
Proposition~\ref{prop:riesz-calculus-app} in Appendix~\ref{app:riesz-calculus}
proves these identities using the finite and infinite canonical blocks.

Multiplying by a right input \(R\) gives the common target operator
\begin{equation}
\mathcal G_\chi=\mathfrak R_\chi[g;M,L]R.
\label{eq:common-riesz-input}
\end{equation}
Table~\ref{tab:weighted-riesz-routes} summarizes its use in the four
problems. Algebraic equations use \(R=\mathsf I\) to construct the
complete projector. We use \(E_1=[\mathsf I;0]\) and
\(E_2=[0;\mathsf I]\) as coordinate embeddings, and
\(R_0=[\mathsf I;P_0]\) as an initial graph.
The concrete lifts and the identities connecting these blocks to the
solution are introduced with the corresponding recovery constructions.

\begin{table}[!ht]
\centering
\small
\renewcommand{\arraystretch}{1.2}
\begin{tabularx}{\textwidth}{@{}l>{\raggedright\arraybackslash}X>{\raggedright\arraybackslash}Xl@{}}
\toprule
Problem & Lift and selected branch & Weights & Output\\
\midrule
CARE & Hamiltonian; left half-plane & \(1\) & \(X\)\\
DARE & Regular pencil; finite unit disk & \(1\) & \(X\)\\
DRE & Normalized Hamiltonian; left/right half-planes &
\(1,\ e^{\alpha_Htz},\ e^{-\alpha_Htz}\) & \(P(t)\)\\
Finite recursion & Linear-fractional lift; inside/outside unit circle &
\(1,\ z^k,\ z^{-k}\) & \(P_k\)\\
\bottomrule
\end{tabularx}
\caption{Weighted Riesz operators used before graph recovery.
The DRE weights act on \(\widehat{\mathcal H}_{\rm DRE}
=\mathcal H_{\rm DRE}/\alpha_H\); recursion powers refer to the unscaled
lift, with \(z^{-k}\) used on the outside branch.}
\label{tab:weighted-riesz-routes}
\end{table}

For CARE and DARE, stability selects a fixed target graph, and the
unit-weight operator is the complete stable projector used in
\eqref{eq:common-graph-recovery}. Dynamic solutions also depend on the
initial graph. Selecting one branch alone cannot preserve this dependence:
whenever \(\Pi R_0\) has full column rank, its columns span the same
fixed subspace \(\operatorname{ran}\Pi\).

The pseudoinverse preserves the connection to the initial data through the
column coordinates. Here \(\Pi=\Pi_+\) for DRE and \(\Pi=\Pi_>\)
for recursion. If \(\Pi R_0\) has full column rank, then for any
coefficient vector \(c\), \((\Pi R_0)^+(\Pi R_0c)=c\): it recovers the coefficients of the
initial vector from its selected component. Multiplying by \(R_0\)
recovers the full initial vector \(R_0c\), including its component in
the other branch. The weighted operators then use this relation to
assemble a projector onto the evolved graph. For DREs, forward weights
act on the left-half-plane branch and inverse weights on the right;
finite recursions use the analogous inside/outside branch construction.
The next two subsections construct the required operators; the initial-data
dependence and graph recovery are handled by the equation-specific analysis.

\subsection{Contour--resolvent construction}
\label{subsec:contour-resolvent-lcu}

We first construct \(\mathcal G_\chi=\mathfrak R_\chi[g;M,\mathsf I]R\)
by discretizing the Riesz integral for an ordinary matrix.
A quadrature rule replaces the integral by a finite sum of resolvents:
for \(m\) nodes \(z_j\in\Gamma_\chi\) and geometric weights \(\nu_j\), set
\begin{equation}
\omega_j=\frac{\nu_jg(z_j)}{2\pi\mathrm i},\qquad
S_\Gamma=\sum_{j=1}^{m}\omega_j(z_j\mathsf I-M)^{-1}R.
\label{eq:weighted-contour-coefficients}
\end{equation}
The coefficients \(\omega_j\) incorporate the contour orientation and the
scalar weight \(g\). Constructing the operator therefore reduces to
encoding the shifted inverses and combining them with these coefficients
by LCU
\cite{HaleHighamTrefethen2008Contour,TakahiraOhashiSogabeUsuda2022Contour}.

Each resolvent is obtained by inverting an encoded shifted matrix.
Suppose \(M\) has an exact block-encoding with normalization \(\alpha_M\).
At node \(z_j\), an LCU of \(M\) and the identity encodes
\((z_j\mathsf I-M)/s_j\), where \(s_j=|z_j|+\alpha_M\).
If the shifted matrix is invertible and a constant-factor upper bound on its
inverse norm is supplied, Proposition~\ref{prop:local-inverse} gives a
declared inverse normalization
\begin{equation}
\beta_j=\Theta(\|(z_j\mathsf I-M)^{-1}\|).
\label{eq:ordinary-node-normalization}
\end{equation}
For decoded inverse error \(\varepsilon_j\) with
\(s_j\varepsilon_j\le1\), its query cost is
\begin{equation}
O\!\left(
s_j\|(z_j\mathsf I-M)^{-1}\|
\log\!\left(e+\frac{\|(z_j\mathsf I-M)^{-1}\|}{\varepsilon_j}\right)
\right).
\label{eq:single-resolvent-query}
\end{equation}
The cost at a node is thus controlled by its generalized singularity factor,
with the shifted-input scale included.
For a nonnormal matrix, distance from the node to the spectrum alone does
not determine this norm~\cite{TrefethenEmbree2005Spectra}. We use the uniform contour bound
\begin{equation}
\mathcal R_\Gamma=\sup_{z\in\Gamma_\chi}
(|z|+\alpha_M)\|(z\mathsf I-M)^{-1}\|.
\label{eq:contour-resolvent-bound}
\end{equation}
We call this quantity the generalized singularity factor; the same term
also applies to the unnormalized inverse norm.
At each node, its reciprocal is the smallest singular value of the
normalized shifted matrix, and hence its distance to singularity.
Appendix~\ref{app:factor-spectral} explains the role of the contour and
when the factor can be controlled by the usual matrix condition number.

LCU combines the nodes by preparing their superposition, applying the
corresponding inverse, and unpreparing the index register.
The following proposition gives the resulting encoding of \(S_\Gamma\).
It assumes coherent coefficient and phase preparation, with node-dependent
inverse schedules padded to a common length to preserve the superposition.

\begin{proposition}[Ordinary contour construction]
\label{thm:local-contour-sum}
Let \(M,R\) have exact block-encodings with normalizations
\(\alpha_M,\alpha_R\), and let every node shift be invertible.
Assume the coherent access just described, including supplied uniform
constant-factor inverse-norm estimates. With \(\beta_j\) defined
in~\eqref{eq:ordinary-node-normalization}, set
\begin{equation}
\alpha_\Gamma=\alpha_R\sum_{j=1}^{m}|\omega_j|\beta_j.
\label{eq:ordinary-lcu-normalization}
\end{equation}
For \(0<\varepsilon_{\rm impl}<\alpha_\Gamma\), the construction returns
a block-encoding of \(S_\Gamma\) with normalization \(\alpha_\Gamma\),
decoded error at most \(\varepsilon_{\rm impl}\), and query cost
\begin{equation}
O\!\left(
\mathcal R_\Gamma
\log\!\left(e+\frac{\mathcal R_\Gamma\alpha_\Gamma}
{\varepsilon_{\rm impl}}\right)
\right)
\label{eq:ordinary-lcu-query}
\end{equation}
to the matrix oracle and its adjoint, together with one call to \(U_R\).
\end{proposition}

The normalization records the weighted sum of the node scales
\(\beta_j\alpha_R\), whereas the query bound is determined by the deepest
inverse circuit. This distinction follows from processing the nodes
coherently. Preparing the node and coefficient data has an additional cost,
separate from the matrix-oracle queries in
\eqref{eq:ordinary-lcu-query}.

The same construction applies when the spectral problem is specified by a
pencil \(M-zL\), as in DARE. In the defining Riesz integral,
\((z\mathsf I-M)^{-1}\) is replaced by \((zL-M)^{-1}L\).
Consequently, the right input becomes \(LR\), and the finite sum and
generalized singularity factor become
\begin{equation}
\begin{aligned}
S_{\Gamma,L}&=\sum_{j=1}^{m}\omega_j(z_jL-M)^{-1}LR,\\
\mathcal R_{\Gamma,L}
&=\sup_{z\in\Gamma_\chi}
(|z|\alpha_L+\alpha_M)\|(zL-M)^{-1}\|.
\end{aligned}
\label{eq:pencil-resolvent-bound}
\end{equation}

Encoding a shifted pencil uses both input matrices \(M\) and \(L\).
Combining these encodings gives the following counterpart of
Proposition~\ref{thm:local-contour-sum}.

\begin{proposition}[Affine-pencil contour construction]
\label{thm:affine-pencil-lcu-main}
Let \(M,L,R\) have exact block-encodings with normalizations
\(\alpha_M,\alpha_L,\alpha_R\), and suppose every \(z_jL-M\) is invertible.
Assume controlled access to \(M,L\) and their adjoints, coherent coefficient
and inverse schedules, and supplied constant-factor inverse-norm estimates.
For declared node normalizations
\(\beta_j=\Theta(\|(z_jL-M)^{-1}\|)\), set
\begin{equation}
\alpha_{\Gamma,L}=\alpha_L\alpha_R\sum_{j=1}^{m}|\omega_j|\beta_j.
\label{eq:pencil-lcu-normalization}
\end{equation}
For \(0<\varepsilon_{\rm impl}<\alpha_{\Gamma,L}\), an encoding of
\(S_{\Gamma,L}\) is obtained with normalization \(\alpha_{\Gamma,L}\),
decoded error at most \(\varepsilon_{\rm impl}\), and query cost
\begin{equation}
O\!\left(
\mathcal R_{\Gamma,L}
\log\!\left(e+\frac{\mathcal R_{\Gamma,L}\alpha_{\Gamma,L}}
{\varepsilon_{\rm impl}}\right)
\right),
\label{eq:pencil-lcu-query}
\end{equation}
plus one call to \(U_R\).
\end{proposition}

Appendix~\ref{app:contour-proofs} proves both construction propositions
and gives the nodewise error allocation.
This construction uses the pencil matrices directly and does not require
\(L^{-1}\). Its interpretation as a selected deflating block uses the
regularity and finite-branch conditions in
Definition~\ref{def:weighted-riesz-interface-main}.

The two propositions encode the finite quadrature sums.
To obtain an approximation to the Riesz operator itself, we also account
for the quadrature error \(\varepsilon_{\rm quad}\), which bounds the
difference between \(\mathcal G_\chi\) and its exact finite sum.
Lemma~\ref{lem:analytic-quadrature-app} gives an explicit bound when
the parametrized integrand has a bounded analytic continuation; other
quadrature rules require their own error bound.
If \(\varepsilon_{\rm LCU}\) accounts for coefficient preparation and
coherent composition, the ordinary construction satisfies
\begin{equation}
\varepsilon_{\rm branch}
\le\varepsilon_{\rm quad}
 +\alpha_R\sum_{j=1}^{m}|\omega_j|\varepsilon_j
 +\varepsilon_{\rm LCU}.
\label{eq:contour-branch-error}
\end{equation}
For the pencil construction, replace \(\alpha_R\) in the sum by
\(\alpha_L\alpha_R\).
Allocating the last two terms a combined budget
\(\varepsilon_{\rm impl}\) gives
\(\varepsilon_{\rm branch}\le\varepsilon_{\rm quad}
+\varepsilon_{\rm impl}\).
For approximate inputs, Proposition~\ref{prop:contour-input-errors-app}
adds the perturbation of the shifted inverses and right factors to this
budget, under explicit nodewise smallness conditions.
Algorithm~\ref{alg:weighted-riesz-be} uses the preceding spectral,
access, and quadrature assumptions. The declared output scale is
\(\alpha_T\); any normalization change has its implementation cost
included, and its propagated error contributes to the total target
\(\varepsilon_T\).
In the algorithms, BE denotes block-encoding, each \(U_T\) carries
its scale, and \(U_{\mathsf I}\) encodes the identity at scale one.
Displayed matrices are ideal targets; pseudoinverses act on decoded inputs.

\begin{algorithm}[H]
\caption{\(\operatorname{RieszBE}(U_M,U_L,\Gamma,g,U_R;\varepsilon_T)\):
weighted Riesz encoding for a pencil}
\label{alg:weighted-riesz-be}
\small
\begin{algorithmic}[1]
\Require Exact BEs \(U_M,U_L,U_R\) with scales
\(\alpha_M,\alpha_L,\alpha_R\); \(\Gamma,g,\varepsilon_T\).
\Ensure A BE \((U_T,\alpha_T,a_T)\) with
\(\|\widetilde T-T\|\le\varepsilon_T\), where
\[
T=\frac{1}{2\pi\mathrm i}\oint_\Gamma
g(z)(zL-M)^{-1}LR\,dz.
\]
\State Choose \(m\) quadrature nodes \((z_j,\nu_j)\) and set
\[
\omega_j=\frac{\nu_jg(z_j)}{2\pi\mathrm i},
\qquad s_j=|z_j|\alpha_L+\alpha_M.
\]
\State Coherently encode the shifted inverses:
\[
\frac{z_jL-M}{s_j}
\quad\longmapsto\quad
\frac{(z_jL-M)^{-1}}{\beta_j}.
\]
\State Set the quadrature target and its LCU scale:
\[
S_{\Gamma,L}=\sum_{j=1}^{m}\omega_j(z_jL-M)^{-1}LR,
\qquad
\alpha_{\Gamma,L}=\alpha_L\alpha_R\sum_{j=1}^{m}|\omega_j|\beta_j.
\]
\Statex If \(\alpha_{\Gamma,L}=0\), return a zero BE at scale \(\alpha_T\).
\State Combine the inverse BEs by LCU and compose \(U_L\), then
\(U_R\), on the right.
\State Implement the declared output scale \(\alpha_T\).
\State \Return \((U_T,\alpha_T,a_T)\).
\end{algorithmic}
\end{algorithm}

This completes the contour construction of the weighted Riesz block.
Its error subsequently enters graph recovery, where the Riccati operator
is obtained. For the unit-weight half-plane projector, the next subsection
gives a second construction based on polynomial approximation.

\subsection{Weyl--LCHM construction of half-plane projectors}
\label{subsec:weyl-lchm-riesz}

We now construct the unit-weight Riesz operator for an ordinary matrix
whose spectrum is separated by the imaginary axis. The target remains
\(\mathcal G_\chi=\Pi_\chi R\), where
\(\Pi_\chi=\mathfrak R_\chi[1;\widehat H,\mathsf I]\) selects the left
or right half-plane, \(\chi\in\{-,+\}\).
For CARE, \(R=\mathsf I\) gives the complete stable projector \(\Pi_-\).
A polynomial that approximates one on the selected branch and zero on the
other approximates the Riesz projector when applied to the matrix.
Weyl--LCHM provides a quantum implementation of this polynomial, giving
another way to construct the same branch block.
Let the normalized matrix satisfy
\begin{equation}
\|\widehat H\|\le1,\qquad
\operatorname{spec}(\widehat H)\subset
\{z:\operatorname{Re}z\le-\Delta_H\}
\cup\{z:\operatorname{Re}z\ge\Delta_H\},\qquad \Delta_H>0.
\label{eq:lchm-half-plane-gap}
\end{equation}
All geometric parameters below refer to the normalized variable.
If \(\widehat H=H/\alpha_H\), the spectral separation and matrix-input
error are divided by \(\alpha_H\) as well.

The scalar identity underlying the polynomial is
\(\operatorname{sign}(z)=z(z^2)^{-1/2}\), where the inverse square root
uses its principal branch. The sign is \(1\) in the right half-plane and
\(-1\) in the left half-plane, so
\((1\pm\operatorname{sign}(z))/2\) selects either branch~\cite[Chap.~5]{Higham2008Functions}.
To approximate the inverse square root uniformly around the
spectrum, choose \(0<r_d<\Delta_H\) and \(r_c>1\), and define
\begin{equation}
E_\pm=\{z\in\mathbb C:|z|\le r_c,\ \pm\operatorname{Re}z\ge r_d\},
\qquad
\Omega_{r_d,r_c}=\{z^2:z\in E_+\}.
\label{eq:lchm-squared-caps}
\end{equation}
The contours \(\Gamma_\pm=\partial E_\pm\) have positive orientation.
Each spectral branch lies strictly inside its cap, so both contours lie
in the resolvent set. Squaring maps either cap onto the same domain.
The squared domain avoids \((-\infty,0]\), so the same analytic branch of
the inverse square root applies throughout it.

Approximating this scalar function produces a polynomial selector.
For a nonnormal matrix, however, scalar approximation error can be
amplified when the polynomial is evaluated at the matrix. The next lemma
quantifies that amplification through the Riesz integral and gives the
error in the required right-multiplied block directly.

\begin{lemma}[Half-plane polynomial selector]
\label{lem:lchm-branch-selector}
Under \eqref{eq:lchm-half-plane-gap}--\eqref{eq:lchm-squared-caps},
let \(q_{r-1}\) be a polynomial of degree at most \(r-1\), \(r\ge1\), and set
\begin{equation}
\varepsilon_{-1/2,r}
=\sup_{w\in\Omega_{r_d,r_c}}|q_{r-1}(w)-w^{-1/2}|.
\label{eq:lchm-inverse-square-root-error}
\end{equation}
\begin{equation}
S_{\pm,r}(z)=\frac12\left(1\pm zq_{r-1}(z^2)\right),
\quad d_S=2r-1.
\label{eq:lchm-polynomial-selectors}
\end{equation}
Define the contour amplification factor
\begin{equation}
\mathcal A_H=\frac1{2\pi}\sum_{\xi\in\{-,+\}}
\int_{\Gamma_\xi}\|(z\mathsf I-\widehat H)^{-1}\|\,|\mathrm dz|.
\label{eq:lchm-nonnormal-amplification}
\end{equation}
Then \(\deg S_{\chi,r}\le d_S\) and
\begin{equation}
\|S_{\chi,r}(\widehat H)-\Pi_\chi\|
\le\frac{\mathcal A_Hr_c}{2}\varepsilon_{-1/2,r}.
\label{eq:lchm-selector-operator-bound}
\end{equation}
For any compatible right factor \(R\), the direct bound is
\begin{equation}
\begin{aligned}
\|(S_{\chi,r}(\widehat H)-\Pi_\chi)R\|
&\le\frac{\varepsilon_{-1/2,r}}{4\pi}
\sum_{\xi\in\{-,+\}}\int_{\Gamma_\xi}
|z|\,\|(z\mathsf I-\widehat H)^{-1}R\|\,|\mathrm dz|\\
&\le\frac{\|R\|\mathcal A_Hr_c}{2}\varepsilon_{-1/2,r}.
\end{aligned}
\label{eq:lchm-direct-branch-error}
\end{equation}
\end{lemma}

\begin{proof}
On the cap \(E_\xi\), the square-root identity gives
\[
|S_{\chi,r}(z)-\mathbf1_{\{\xi=\chi\}}|
\le\frac{|z|}{2}\varepsilon_{-1/2,r}.
\]
The Dunford--Riesz formulas for the polynomial and the projector~\cite[Chap.~I, Eqs.~(5.25) and~(5.47)]{Kato1995Perturbation} imply
\[
(S_{\chi,r}(\widehat H)-\Pi_\chi)R
=\frac1{2\pi\mathrm i}\sum_{\xi\in\{-,+\}}
\oint_{\Gamma_\xi}
\bigl(S_{\chi,r}(z)-\mathbf1_{\{\xi=\chi\}}\bigr)
(z\mathsf I-\widehat H)^{-1}R\,\mathrm dz.
\]
Taking norms proves the first inequality in
\eqref{eq:lchm-direct-branch-error}. The bounds
\(|z|\le r_c\) and
\(\|(z\mathsf I-\widehat H)^{-1}R\|
\le\|(z\mathsf I-\widehat H)^{-1}\|\|R\|\)
give the second. Setting \(R=\mathsf I\) proves
\eqref{eq:lchm-selector-operator-bound}.
\end{proof}

Thus the scalar approximation determines the projector accuracy once the
generalized singularity factor is known. The lemma allows arbitrary Jordan
structure~\cite[Sec.~1.2.1]{Higham2008Functions}; the resolvent accounts for nonnormality in the error analysis.
In particular, a sufficient condition for selector error
\(\varepsilon_{\rm sel}>0\) is
\begin{equation}
\frac{\alpha_R\mathcal A_Hr_c}{2}\varepsilon_{-1/2,r}
\le\varepsilon_{\rm sel},
\qquad \|R\|\le\alpha_R.
\label{eq:lchm-selector-budget}
\end{equation}
When bounds on the resolvent applied to \(R\) are available,
\eqref{eq:lchm-direct-branch-error} can reduce the required degree.

It remains to turn this polynomial approximation into a block-encoding.
The matrix \(\widehat H\) need not be Hermitian, so the construction uses
polynomials of a family of Hermitian matrices whose coherent average
recovers the desired polynomial of \(\widehat H\).
Specifically, the Weyl linear-combination-of-Hermitian-matrices construction
\cite{WangLiangChenLiu2026LCHM} uses the Hermitian contractions
\begin{equation}
X_\theta=\frac12\left(e^{-\mathrm i\theta}\widehat H
                         +e^{\mathrm i\theta}\widehat H^*\right).
\label{eq:lchm-hermitian-overview}
\end{equation}
A finite, coherent average over angles cancels the unwanted terms
containing \(\widehat H^*\) and yields \(S_{\chi,r}(\widehat H)\).
The angular identity is exact, so this averaging introduces no quadrature
error. Its proof and circuit realization are given in
Appendix~\ref{app:weyl-realization}.

For the quantum statements below, assume exact block-encodings
\(U_{\widehat H}\) and \(U_R\) with normalizations \(1\) and \(\alpha_R\),
controlled input and adjoint access, the required reflections,
and coherent polynomial implementation as in
Subsection~\ref{subsec:block-encodings-inverse}.
For each selector, choose a declared normalization
\begin{equation}
\alpha_{S_{\chi,r}}\ge
\max_{|z|=1}|2S_{\chi,r}(z)-\tfrac12|.
\label{eq:lchm-selector-normalization}
\end{equation}

The following proposition combines this implementation with the selector
error bound to obtain the Riesz block. Its query cost is linear in the
polynomial degree.

\begin{proposition}[Weyl--LCHM branch-block construction]
\label{cor:lchm-branch-block}
Under the input and access assumptions above, let
\(S_{\chi,r}\) be the selector in
\eqref{eq:lchm-polynomial-selectors} and satisfy
\(\|(S_{\chi,r}(\widehat H)-\Pi_\chi)R\|\le\varepsilon/2\),
where \(\varepsilon>0\).
Weyl--LCHM returns a block-encoding of \(\Pi_\chi R\) with normalization
\(\alpha_{S_{\chi,r}}\alpha_R\), decoded error at most \(\varepsilon\),
using
\begin{equation}
O(d_S)\quad\text{queries to the block-encodings of }\widehat H\text{ and }R.
\label{eq:lchm-total-query}
\end{equation}
\end{proposition}

Lemma~\ref{lem:lchm-branch-selector} supplies the mathematical
approximation, and Proposition~\ref{cor:lchm-branch-block} supplies its
quantum implementation. Assigning half the target error to each yields the
stated accuracy for \(\Pi_\chi R\).
Approximate-input and gate-synthesis errors are treated in
Appendix~\ref{app:weyl-errors}.

To express the query cost in input parameters, we choose a polynomial
that attains the required selector accuracy on a known spectral region.
The following theorem treats spectra whose oscillation frequencies are
bounded relative to their growth or decay rates. This sector condition
allows complex eigenvalues and arbitrarily small positive half-plane gaps.
For a diagonalizable lift, an eigenvector condition bound then controls
how scalar approximation errors propagate to the Riesz operator.

\begin{theorem}[Query bound for a sector-bounded spectrum]
\label{thm:lchm-explicit-query}
Assume the exact-input access of Proposition~\ref{cor:lchm-branch-block}
and \(\widehat H=V\Lambda V^{-1}\), where \(\Lambda\) is diagonal.
Let \(\Delta_H\) be a supplied lower bound in
\eqref{eq:lchm-half-plane-gap}, and let
\(\kappa_V\ge\|V\|\|V^{-1}\|\) be a supplied upper bound.
Suppose a supplied \(0\le\eta<1\) satisfies
\begin{equation}
|\operatorname{Im}\lambda|\le\eta|\operatorname{Re}\lambda|
\qquad(\lambda\in\operatorname{spec}(\widehat H)).
\label{eq:lchm-sector-condition}
\end{equation}
For any \(\varepsilon>0\), Weyl--LCHM constructs a block-encoding of
\(\Pi_\chi R\) with decoded error at most \(\varepsilon\)
and normalization \(\alpha_{S_{\chi,r}}\alpha_R\) as in
\eqref{eq:lchm-selector-normalization}, using
\begin{equation}
\widetilde O\!\left(\frac{1}{\Delta_H^2(1-\eta^2)^2}\right)
\label{eq:lchm-sector-query}
\end{equation}
queries to the block-encodings of \(\widehat H\) and \(R\).
Here \(\widetilde O\) suppresses logarithmic factors in
\(e+\alpha_R\kappa_V/\varepsilon\) and \((1-\eta^2)^{-1}\).
\end{theorem}

The parameter \(\eta\) bounds the ratio of oscillation frequency to
exponential growth or decay rate, while \(\Delta_H\) measures separation
from the imaginary axis. For fixed \(\eta<1\), the displayed query bound
scales as \(\Delta_H^{-2}\), up to logarithmic factors.
The eigenvector bound \(\kappa_V\) enters through those logarithms;
no eigenvector oracle is required. For a normal matrix,
\(\kappa_V=1\) is admissible.
Appendix~\ref{app:weyl-explicit-query} proves the theorem and gives a
sufficient output normalization, whose effect on subsequent graph recovery
must still be included.

The sector condition gives the explicit query bound above.
Appendix~\ref{app:weyl-degree} gives a Faber construction whose degree
is controlled by approximation geometry and generalized singularity factors;
this construction also permits arbitrary Jordan structure.
The weighted and pencil operators in Table~\ref{tab:weighted-riesz-routes}
continue to use the contour construction of
Subsection~\ref{subsec:contour-resolvent-lcu}.

\section{Riccati equations}
\label{sec:dre}

\subsection{Quantum DRE problem and total query complexity}
\label{subsec:dre-complexity}

We now construct the time-dependent solution in
Definition~\ref{def:problem-dre}. In its initial-value convention,
\begin{equation}
\dot P=-Q-A^*P-PA+PGP,\qquad P(0)=P_0.
\label{eq:dre-section-equation}
\end{equation}
The associated Hamiltonian \cite[Sec.~5.1]{Poloni2020RiccatiIterations} and initial graph column are
\begin{equation}
\mathcal H_{\rm DRE}
=\begin{bmatrix}A&-G\\-Q&-A^*\end{bmatrix},
\qquad
R_0=\begin{bmatrix}\mathsf I\\P_0\end{bmatrix}.
\label{eq:dre-lift-seed}
\end{equation}
The data \(G,Q,P_0\) remain Hermitian, and the solution is promised to
exist throughout the requested interval.
For the DRE spectral construction in this section and
Appendix~\ref{app:dre}, assume that
\(\mathcal H_{\rm DRE}\) has no imaginary-axis eigenvalues.
The LQR conditions of Definition~\ref{def:problem-care} guarantee this
property by \eqref{eq:care-triangularization-app}; it is unchanged by
the Hamiltonian sign reversal for finite-horizon LQR.

\begin{problem}[Quantum DRE]
\label{prob:quantum-dre}
For the data and existence interval in Definition~\ref{def:problem-dre},
let \(t\in[0,T]\) be the requested time.
Given exact block-encodings of \(\mathcal H_{\rm DRE}\) and \(R_0\)
in \eqref{eq:dre-lift-seed}, their adjoints and controlled versions,
with normalizations \(\alpha_H\) and \(\alpha_{R_0}\), and a target
\(0<\varepsilon\le1\), construct a block-encoding with decoded matrix
\(\widetilde P(t)\) satisfying
\(\|\widetilde P(t)-P(t)\|\le\varepsilon\).
\end{problem}

Write \(\widehat{\mathcal H}_{\rm DRE}=\mathcal H_{\rm DRE}/\alpha_H\)
and let \(\Pi_+\), \(\Pi_-\) denote its right- and left-half-plane
Riesz projectors. The query bound below uses the access, precision, and
normalization conditions in Appendix~\ref{app:dre-complexity}.

\begin{theorem}[DRE construction by direct projector-block recovery]
\label{thm:dre-construction-main}
For Problem~\ref{prob:quantum-dre}, suppose that
\(\Pi_+R_0\) has full column rank, and supply a bound
\(\sup_{0\le s\le T}\|P(s)\|\le M\). Choose positively oriented contours
\(\Gamma_\pm\) enclosing exactly the corresponding spectral branches
of \(\widehat{\mathcal H}_{\rm DRE}\), and define
\begin{equation}
\mathcal R_{\rm DRE}
=\max_{\sigma\in\{-,+\}}\sup_{z\in\Gamma_\sigma}
(1+|z|)\|(z\mathsf I-\widehat{\mathcal H}_{\rm DRE})^{-1}\|.
\label{eq:dre-generalized-factor}
\end{equation}
Under the implementation and normalization assumptions of
Appendix~\ref{app:dre-complexity}, the algorithm returns an
\((\alpha_{P(t)},a_{P(t)},\varepsilon)\) block-encoding of \(P(t)\), using
\begin{equation}
Q_{\rm DRE}
=\widetilde O\!\left(
\mathcal R_{\rm DRE}\,
\kappa(\Pi_+R_0)\,
\alpha_{P(t)}
\right)
\label{eq:dre-query-main}
\end{equation}
queries to the supplied matrix oracles.
Here \(\alpha_{P(t)}\) is the output normalization specified in that
appendix, and \(\widetilde O\) suppresses logarithmic dependence on
precision and the supplied scales.
\end{theorem}

\(\mathcal R_{\rm DRE}\) measures the relative proximity to singularity
of the shifted matrices used to construct the weighted spectral branches.
The full-rank condition means that projecting the initial graph onto the
right-half-plane branch loses no direction:
\(\operatorname{ran}R_0\cap\operatorname{ran}\Pi_-=\{0\}\).
Its smallest singular value is the initialization nondegeneracy factor,
and \(\kappa(\Pi_+R_0)=\|\Pi_+R_0\|/\sigma_{\min}(\Pi_+R_0)\) measures
the relative imbalance among the projected directions.
Rank loss makes these branch coordinates degenerate, even when the DRE
solution exists; Example~\ref{ex:dre-initialization-scope-app} illustrates
this distinction.
The output normalization \(\alpha_{P(t)}\) records the cost of direct
recovery from the upper projector block. It includes the supplied
solution bound \(M\), as specified in Appendix~\ref{app:dre-complexity}.
For the rectangular rule of
Proposition~\ref{prop:dre-weighted-quadrature-app}, with
\(M,\alpha_{R_0}=O(1)\) and a calibrated initial singular-value estimate,
Corollary~\ref{cor:dre-construction-normalization-app} gives
\begin{equation}
\alpha_{P(t)}=O\!\left(
\frac{\mathcal R_{\rm DRE}^{\,2}}{\sigma_{\min}(\Pi_+R_0)}\right).
\label{eq:dre-output-normalization-bound}
\end{equation}
The general query bound includes the actual output scale when \(M\)
varies with the problem.
Appendix~\ref{app:factor-geometry} explains the initialization geometry,
Appendix~\ref{app:factor-spectral} controls the spectral factor, and
Proposition~\ref{prop:factor-dre-initialization-app} gives coefficient
conditions for the initialization.
We next describe the graph construction underlying this bound.

\subsection{Hamiltonian flow and weighted spectral branches}
\label{subsec:dre-weighted-graph}

The linear Hamiltonian flow carries the solution graph \cite[Chap.~4]{BittantiLaubWillems1991RiccatiEquation}:
\begin{equation}
\mathcal W(t)=e^{t\mathcal H_{\rm DRE}}R_0,
\qquad
\operatorname{ran}\mathcal W(t)
=\operatorname{ran}\begin{bmatrix}\mathsf I\\P(t)\end{bmatrix}.
\label{eq:dre-flow-graph}
\end{equation}
Thus the initial column identifies the subspace to propagate.
Lemma~\ref{lem:dre-graph-flow-app} proves this relation on the promised
existence interval. This is the Embed step for the DRE; Hamiltonian
linearization also underlies earlier quantum approaches to
finite-horizon Riccati problems \cite{Krovi2024QuantumOptimalControl}.

For QET, Table~\ref{tab:dre-weighted-blocks} specifies four uses of the
weighted Riesz construction in Section~\ref{subsec:contour-resolvent-lcu}.
The normalized spectral variable requires the time parameter
\(\alpha_Ht\). Each exponential decays in its selected half-plane.
The pseudoinverse of \(\Pi_+R_0\) then combines the two branches while
preserving the initial data.

\begin{table}[htbp]
\centering
\renewcommand{\arraystretch}{1.2}
\begin{tabular}{llll}
\toprule
Target & Branch & Weight \(g(z)\) & Right input\\
\midrule
\(\Pi_+\) & Right & \(1\) & \(\mathsf I\)\\
\(\Pi_+R_0\) & Right & \(1\) & \(R_0\)\\
\(e^{t\mathcal H_{\rm DRE}}\Pi_-R_0\) & Left
& \(e^{\alpha_Htz}\) & \(R_0\)\\
\(e^{-t\mathcal H_{\rm DRE}}\Pi_+\) & Right
& \(e^{-\alpha_Htz}\) & \(\mathsf I\)\\
\bottomrule
\end{tabular}
\caption{Weighted Riesz blocks for the DRE, with matrix
\(\widehat{\mathcal H}_{\rm DRE}\).}
\label{tab:dre-weighted-blocks}
\end{table}

\begin{theorem}[Decaying graph projector]
\label{thm:dre-seeded-projector-main}
For Problem~\ref{prob:quantum-dre}, suppose that \(\Pi_+R_0\)
has full column rank. At the requested time \(t\), set
\begin{equation}
\mathcal E(t)=\Pi_+
+e^{t\mathcal H_{\rm DRE}}\Pi_-R_0
(\Pi_+R_0)^+e^{-t\mathcal H_{\rm DRE}}\Pi_+.
\label{eq:dre-seeded-projector}
\end{equation}
Then
\begin{equation}
\begin{aligned}
\mathcal E(t)^2&=\mathcal E(t),\\
\operatorname{ran}\mathcal E(t)
&=\operatorname{ran}\begin{bmatrix}\mathsf I\\P(t)\end{bmatrix},
&\ker\mathcal E(t)&=\operatorname{ran}\Pi_-.
\end{aligned}
\label{eq:dre-seeded-range}
\end{equation}
\end{theorem}

This construction is the key step in our quantum algorithm for the DRE
and one of the core contributions of this work. It preserves the range
of the Hamiltonian graph flow \(e^{t\mathcal H_{\rm DRE}}R_0\), while
expressing all time dependence through
\(e^{t\mathcal H_{\rm DRE}}\Pi_-\) and
\(e^{-t\mathcal H_{\rm DRE}}\Pi_+\).
Their scalar weights decay on the selected half-planes, so the construction
avoids directly encoding the exponentially amplified components of the
forward flow. Consequently,
\(\mathcal E(t)\) represents the finite-time solution determined by
\(P_0\), with the fixed kernel \(\operatorname{ran}\Pi_-\).

The projector property then connects the weighted Riesz blocks to the
requested Riccati operator. Since \(\mathcal E(t)\) fixes every vector
in the solution graph, its upper block row has full row rank, and its
lower block row is \(P(t)\) times the upper one. These identities enable
the direct recovery in the next subsection and control the recovery
pseudoinverse through the solution scale.
Appendix~\ref{app:dre-graph} proves the theorem, and
Appendix~\ref{app:dre-quadrature} gives the finite-contour construction
of all four weighted blocks.

\subsection{Direct recovery from the upper projector block}
\label{subsec:dre-recovery}

Let \(E_1=[\mathsf I;0]\) and \(E_2=[0;\mathsf I]\) select the state
and lower coordinates. The two block rows \(E_j^*\mathcal E(t)\)
inherit the encoding of the decaying graph projector. Its upper block has full
row rank and provides the coordinates needed to recover the solution.

\begin{lemma}[Direct projector-block recovery]
\label{lem:dre-projector-block-recovery}
Under the assumptions of Theorem~\ref{thm:dre-seeded-projector-main},
\begin{equation}
\begin{aligned}
(E_1^*\mathcal E(t))\begin{bmatrix}\mathsf I\\P(t)\end{bmatrix}
&=\mathsf I,\\
E_2^*\mathcal E(t)&=P(t)E_1^*\mathcal E(t),\\
P(t)&=(E_2^*\mathcal E(t))(E_1^*\mathcal E(t))^+.
\end{aligned}
\label{eq:dre-projector-direct-recovery}
\end{equation}
\end{lemma}

Appendix~\ref{app:dre-recovery} proves these identities and bounds
the upper-block pseudoinverse in terms of the solution norm.
Both block selections have normalization \(\alpha_{\mathcal E}\).
The upper-block pseudoinverse is encoded at normalization
\(\alpha_{P(t)}/\alpha_{\mathcal E}\), chosen from the supplied
solution bound as specified in \eqref{eq:dre-supplied-output-scale-app}.
Their product therefore returns the solution at the declared output scale.

Supply \(0<\gamma_+\le\sigma_{\min}(\Pi_+R_0)\), as in
\eqref{eq:factor-dre-initialization-app} when
Proposition~\ref{prop:factor-dre-initialization-app} applies, and fix actual
encoding scales before choosing \(\eta\) and pseudoinverse parameters
as in Appendix~\ref{app:dre-complexity}; Appendix~\ref{app:dre-quadrature}
supplies the contour data.
The input \(U_{\mathcal H_{\rm DRE}}\) also serves as
\(U_{\widehat{\mathcal H}_{\rm DRE}}\) at scale one.

\begin{algorithm}[H]
\caption{DRE: construct a block-encoding of \(P(t)\)}
\label{alg:dre-direct-projector-construction}
\small
\begin{algorithmic}[1]
\Require Exact BEs of \(\mathcal H_{\rm DRE}\) and
\(R_0=[\mathsf I_n;P_0]\), with their scales;
\(t\in[0,T]\), \(0<\varepsilon\le1\).
\Ensure A BE of \(P(t)\) with decoded error \(\le\varepsilon\)
and scale \(\alpha_{P(t)}\) defined in Appendix~\ref{app:dre-complexity}.

\State Encode \(\Pi_+R_0\) using Algorithm~\ref{alg:weighted-riesz-be}:
\[
U_{\Pi_+R_0}\gets\operatorname{RieszBE}
(U_{\widehat{\mathcal H}_{\rm DRE}},U_{\mathsf I},\Gamma_+,1,U_{R_0};\eta).
\]
\State Apply the pseudoinverse primitive to obtain \(U_{(\Pi_+R_0)^+}\).
\State Construct the other three weighted Riesz BEs:
\[
\begin{aligned}
U_{\Pi_+}&\gets\operatorname{RieszBE}
(U_{\widehat{\mathcal H}_{\rm DRE}},U_{\mathsf I},\Gamma_+,1,U_{\mathsf I};\eta),\\
U_{e^{t\mathcal H_{\rm DRE}}\Pi_-R_0}
&\gets\operatorname{RieszBE}
(U_{\widehat{\mathcal H}_{\rm DRE}},U_{\mathsf I},\Gamma_-,e^{\alpha_Htz},U_{R_0};\eta),\\
U_{e^{-t\mathcal H_{\rm DRE}}\Pi_+}
&\gets\operatorname{RieszBE}
(U_{\widehat{\mathcal H}_{\rm DRE}},U_{\mathsf I},\Gamma_+,e^{-\alpha_Htz},U_{\mathsf I};\eta).
\end{aligned}
\]
\State Assemble a BE of the DGP:
\[
\mathcal E(t)=\Pi_+
+e^{t\mathcal H_{\rm DRE}}\Pi_-R_0
(\Pi_+R_0)^+e^{-t\mathcal H_{\rm DRE}}\Pi_+.
\]
\State Use both complete rows of the same \(U_{\mathcal E}\)
to encode
\[
P(t)=(E_2^*\mathcal E(t))(E_1^*\mathcal E(t))^+.
\]
\Statex \Return \((U_{P(t)},\alpha_{P(t)},a_{P(t)})\).
\end{algorithmic}
\end{algorithm}

The Hamiltonian lift represents the nonlinear evolution by the motion
of a linear graph subspace. The weighted Riesz blocks separate its two
spectral branches and encode time dependence through decaying weights.
The first pseudoinverse recovers the initial graph coordinates from
their projection onto the right-half-plane branch. Combining the two
branches therefore selects exactly the evolving graph determined by
\(P_0\), with the left-half-plane subspace as its fixed complement.

The second pseudoinverse chooses an input whose upper projector block
equals a prescribed state \(x\). Applying the lower block to that input
then gives \(P(t)x\), because every projected vector lies in the solution
graph. Thus the recovery uses the defining relation between the two
graph coordinates. The resulting output normalization enters both the
construction bound and any subsequent feedback or measurement cost.
The same argument applies to the discrete evolution considered next.

\subsection{Finite Riccati recursions}
\label{sec:rr}
\label{subsec:rr-complexity}

The same construction applies to the finite recursion in
Definition~\ref{def:problem-recursion},
\begin{equation}
P_{j+1}=Q+A^*P_j(\mathsf I+GP_j)^{-1}A,
\qquad j=0,\ldots,k-1.
\label{eq:rr-section-equation}
\end{equation}
For \(G,Q,P_0\succeq0\), every iterate is well defined and positive
semidefinite. Assume that \(A\) is invertible and the forward lift
\begin{equation}
\mathcal S_F=\begin{bmatrix}
A^{-1}&A^{-1}G\\
QA^{-1}&A^*+QA^{-1}G
\end{bmatrix},\qquad
R_0=\begin{bmatrix}\mathsf I_n\\P_0\end{bmatrix}
\label{eq:rr-lift-seed}
\end{equation}
has no unit-circle eigenvalues. In this subsection and
Appendix~\ref{app:rr}, \(\Pi_<,\Pi_>\) are its interior and exterior
Riesz projectors. This local convention concerns the unscaled lift,
rather than the DARE pencil projector of
Section~\ref{subsec:dare-construction}.

\begin{problem}[Quantum RR]
\label{prob:quantum-rr}
For the data in Definition~\ref{def:problem-recursion} and a requested
integer \(k\ge0\), suppose that exact block-encodings of
\(\mathcal S_F\) and \(R_0\) in \eqref{eq:rr-lift-seed}, their
adjoints and controlled versions, are supplied with normalizations
\(\alpha_{\mathcal S_F}\) and \(\alpha_{R_0}\).
For \(0<\varepsilon\le1\), construct a block-encoding with decoded
matrix \(\widetilde P_k\) satisfying
\(\|\widetilde P_k-P_k\|\le\varepsilon\).
\end{problem}

\begin{theorem}[RR construction by direct projector-block recovery]
\label{thm:rr-construction-main}
For Problem~\ref{prob:quantum-rr}, suppose that \(\Pi_>R_0\) has full
column rank, and supply a bound \(\|P_j\|\le M\) for
\(0\le j\le k\). Choose an interior contour \(\Gamma_<\) in
\(|z|<1\) and an exterior contour system \(\Gamma_>\) in
\(|z|>1\), each enclosing exactly its spectral branch, with
\(\Gamma_>\) having winding number zero about the closed unit disk.
Define
\begin{equation}
\mathcal R_{\rm RR}
=\max_{\chi\in\{<,>\}}\sup_{z\in\Gamma_\chi}
(|z|+\alpha_{\mathcal S_F})
\|(z\mathsf I-\mathcal S_F)^{-1}\|.
\label{eq:rr-generalized-factor}
\end{equation}
Under the implementation and normalization assumptions of
Appendix~\ref{app:rr-complexity}, the algorithm returns an
\((\alpha_{P_k},a_{P_k},\varepsilon)\) block-encoding of \(P_k\), using
\begin{equation}
Q_{\rm RR}=\widetilde O\!\left(
\mathcal R_{\rm RR}\,\kappa(\Pi_>R_0)\,\alpha_{P_k}\right)
\label{eq:rr-query-main}
\end{equation}
queries to the supplied matrix oracles. Here \(\alpha_{P_k}\) is
the direct-product output normalization specified in that appendix;
\(\widetilde O\) suppresses logarithmic dependence on precision
and the supplied scales.
\end{theorem}

The three factors account for the shifted resolvents, the
initial-column pseudoinverse, and the final projector-block recovery.
Full column rank means that projection onto the exterior branch
loses no initial graph direction:
\(\operatorname{ran}R_0\cap\operatorname{ran}\Pi_<=\{0\}\).
The supplied solution bound enters through \(\alpha_{P_k}\);
the query formula does not require \(M=O(1)\).
For the circular rule of Appendix~\ref{app:rr-quadrature}, if
\(M=O(1)\) and \(\alpha_{R_0}=O(1)\),
Corollary~\ref{cor:rr-construction-normalization-app} gives
\begin{equation}
\alpha_{P_k}=O\!\left(
\frac{\mathcal R_{\rm RR}^{\,2}}{\sigma_{\min}(\Pi_>R_0)}\right).
\label{eq:rr-output-normalization-bound}
\end{equation}
These uniform bounds require corresponding supplied information;
finiteness of the iterates alone does not provide them.

The lift propagates the solution graph,
\begin{equation}
\operatorname{ran}(\mathcal S_F^kR_0)
=\operatorname{ran}\begin{bmatrix}\mathsf I_n\\P_k\end{bmatrix}.
\label{eq:rr-propagated-graph}
\end{equation}
This linear-fractional representation
\cite{Poloni2020RiccatiIterations,Mehrmann1991AutonomousLQ}
is verified in Lemma~\ref{lem:rr-graph-propagation-app}.
The weighted Riesz construction produces
\(\Pi_>\), \(\Pi_>R_0\), \(\mathcal S_F^k\Pi_<R_0\), and
\(\mathcal S_F^{-k}\Pi_>\), using the weights \(1,1,z^k,z^{-k}\),
respectively. The powers have modulus at most one on their selected
contours. As in the DRE, these blocks give the decaying graph projector
\begin{equation}
\mathcal E_k=\Pi_>
+\mathcal S_F^k\Pi_<R_0(\Pi_>R_0)^+
\mathcal S_F^{-k}\Pi_>.
\label{eq:rr-seeded-projector}
\end{equation}
Theorem~\ref{thm:rr-seeded-projector-main} proves that it projects
onto the graph of \(P_k\), with kernel \(\operatorname{ran}\Pi_<\).
Using the same coordinate selectors \(E_1,E_2\) as above, recover
\begin{equation}
P_k=(E_2^*\mathcal E_k)(E_1^*\mathcal E_k)^+.
\label{eq:rr-projector-direct-recovery}
\end{equation}
The upper block has full row rank, and its pseudoinverse is
controlled by the supplied solution bound; see
Corollary~\ref{cor:rr-projector-block-recovery-app}.

The full procedure is given in Algorithm~\ref{alg:rr-direct-projector-construction} in Appendix~\ref{app:rr-complexity}.

The lift replaces the nonlinear updates by graph propagation.
The two weighted branches and the initial-column pseudoinverse
identify the graph selected by \(P_0\); the last pseudoinverse
selects its vector with a prescribed state coordinate.
The discrete setting changes the spectral boundary and replaces
exponential weights by powers, while the projector-block recovery
is the same as for the DRE. Appendix~\ref{app:rr} gives the
propagation proof, complete quadrature, stability, and query bounds.

\subsection{Algebraic Riccati equations}
\label{sec:algebraic-riccati}

The algebraic problems use the limiting graph in the continuous and
discrete constructions above. Its projector is a unit-weight Riesz
operator, so it can be constructed directly and passed to the same
projector-block recovery. The resulting query bounds involve the
generalized singularity factor and the actual output normalization.
The limiting relations below specify the time direction and the
connection between the discrete lift and the pencil.

\subsubsection{CARE: the limiting graph projector}
\label{subsec:care-construction}

For the CARE in Definition~\ref{def:problem-care}, use the Hamiltonian
\begin{equation}
\mathcal H_{\rm CARE}=\begin{bmatrix}A&-G\\-Q&-A^*\end{bmatrix},
\qquad
\mathcal H_{\rm CARE}\begin{bmatrix}\mathsf I\\X\end{bmatrix}
=\begin{bmatrix}\mathsf I\\X\end{bmatrix}(A-GX).
\label{eq:care-graph-invariance}
\end{equation}
Its left-half-plane Riesz projector is denoted by \(\Pi_-\).
The stabilizing LQR solution makes its range the graph of \(X\)
\cite{Laub1979Schur,LancasterRodman1995Riccati}.

\begin{problem}[Quantum CARE]
\label{prob:quantum-care}
For the LQR data in Definition~\ref{def:problem-care}, let
\(X\succeq0\) be the stabilizing solution, with \(\|X\|<\infty\).
Given exact block-encodings of \(\mathcal H_{\rm CARE}\) in
\eqref{eq:care-graph-invariance}, its adjoint and controlled versions,
with normalization \(\alpha_H\), and a target
\(0<\varepsilon\le1\), construct a block-encoding with decoded
matrix \(\widetilde X\) satisfying \(\|\widetilde X-X\|\le\varepsilon\).
\end{problem}

\begin{theorem}[CARE construction by direct projector-block recovery]
\label{thm:care-construction-main}
For Problem~\ref{prob:quantum-care}, supply an upper bound on \(\|X\|\)
and a positively oriented contour \(\Gamma_-\) enclosing exactly the
stable spectrum of \(\mathcal H_{\rm CARE}\). Define
\begin{equation}
\mathcal R_{\rm CARE}
=\sup_{z\in\Gamma_-}(|z|+\alpha_H)
\|(z\mathsf I-\mathcal H_{\rm CARE})^{-1}\|.
\label{eq:care-construction-parameters}
\end{equation}
Under the implementation and normalization assumptions of
Appendix~\ref{app:algebraic-riccati}, the algorithm returns an
\((\alpha_X,a_X,\varepsilon)\) block-encoding of \(X\), using
\begin{equation}
Q_{\rm CARE}=\widetilde O\!\left(\mathcal R_{\rm CARE}\alpha_X\right)
\label{eq:care-query-bound}
\end{equation}
queries to the supplied Hamiltonian encoding and its adjoint and
controlled versions. Here \(\alpha_X\) is the actual output
normalization in \eqref{eq:algebraic-projector-output-app};
\(\widetilde O\) suppresses logarithmic dependence on precision and
the supplied scales.
\end{theorem}

The factor \(\mathcal R_{\rm CARE}\) controls the shifted inverses
used to encode \(\Pi_-\). The scale \(\alpha_X\) includes both
the actual projector normalization and the supplied solution bound.
For the rectangle rule of
Proposition~\ref{prop:care-rectangle-quadrature-app}, the projector
normalization is \(O(\mathcal R_{\rm CARE})\), uniformly in the
quadrature order. The output scale accounts for the cost of recovery;
any reduction of the projector normalization still has its declared
implementation cost.

The limiting projector arises from the LQR time-to-go evolution,
whose lift is \(-\mathcal H_{\rm CARE}\). With
\(R_0=[\mathsf I;P_0]\) and full-column-rank \(\Pi_-R_0\), its
decaying graph projector satisfies
\begin{equation}
\Pi_-+e^{-t\mathcal H_{\rm CARE}}\Pi_+R_0
(\Pi_-R_0)^+e^{t\mathcal H_{\rm CARE}}\Pi_-
\ \longrightarrow\ \Pi_-,\qquad t\to\infty.
\label{eq:care-limiting-projector}
\end{equation}
Thus the stable CARE graph uses the sign-reversed flow; the positive-time
flow in \eqref{eq:dre-seeded-projector} instead tends to \(\Pi_+\).
Lemma~\ref{lem:algebraic-projector-limits-app} proves the limit.
The algorithm constructs its value directly, using
\(\Pi_-=\mathfrak R_-[1;\mathcal H_{\rm CARE},\mathsf I]\),
without evolving an initial graph to a finite cutoff time.

\begin{lemma}[CARE graph projector and recovery]
\label{lem:care-graph-recovery}
Let \(X\) be the stabilizing solution in
Definition~\ref{def:problem-care}.
For \(E_1=[\mathsf I;0]\) and \(E_2=[0;\mathsf I]\), the stable
Riesz projector of \(\mathcal H_{\rm CARE}\) satisfies
\begin{equation}
\Pi_-^2=\Pi_-,
\qquad
\operatorname{ran}\Pi_-=\operatorname{ran}\begin{bmatrix}\mathsf I\\X\end{bmatrix}.
\label{eq:care-graph-projector}
\end{equation}
Its upper block row has full row rank, and
\begin{equation}
X=(E_2^*\Pi_-)(E_1^*\Pi_-)^+.
\label{eq:care-graph-recovery}
\end{equation}
\end{lemma}

Appendix~\ref{app:care-construction} proves the lemma and the query
bound. QET encodes the full projector,
and Recover selects its two block rows and multiplies the lower row by
the pseudoinverse of the upper row. The graph-projector bound in
Lemma~\ref{lem:dre-projector-block-recovery-app} controls this inverse
through the output scale. This removes the initialization pseudoinverse
from the algebraic algorithm.

Matrix-sign selection of the stable subspace is classical
\cite{Roberts1980Sign,Byers1987MatrixSign}; quantum CARE constructions
also use sign projectors or contour resolvents
\cite{WangLiu2026SignEmbedding,RodenasRuizZhaoLee2026NonlinearMatrixEquations}.
Here the full-projector recovery connects the algebraic case to the
dynamic construction. Weyl--LCHM can supply the same projector under
the assumptions of Corollary~\ref{cor:care-weyl-recovery-app}, which
includes its actual polynomial normalization.

\subsubsection{DARE: the finite stable graph projector}
\label{subsec:dare-construction}

For Definition~\ref{def:problem-dare}, use the symplectic pencil
\cite{VanDooren1981Generalized,LancasterRodman1995Riccati}
\begin{equation}
M=\begin{bmatrix}A&0\\-Q&\mathsf I\end{bmatrix},
\qquad L=\begin{bmatrix}\mathsf I&G\\0&A^*\end{bmatrix}.
\label{eq:dare-pencil}
\end{equation}
Its finite unit-disk projector is denoted here by \(\Pi_<\).
The pencil construction permits singular \(A\) and \(L\).

\begin{problem}[Quantum DARE]
\label{prob:quantum-dare}
For the LQR data in Definition~\ref{def:problem-dare}, let
\(X\succeq0\) be the stabilizing solution, with \(\|X\|<\infty\).
Given exact block-encodings of \(M,L\) in \eqref{eq:dare-pencil},
their adjoints and controlled versions, with normalizations
\(\alpha_M,\alpha_L\), and a target \(0<\varepsilon\le1\),
construct a block-encoding with decoded matrix \(\widetilde X\)
satisfying \(\|\widetilde X-X\|\le\varepsilon\).
\end{problem}

\begin{theorem}[DARE construction by direct projector-block recovery]
\label{thm:dare-construction-main}
For Problem~\ref{prob:quantum-dare}, supply an upper bound on \(\|X\|\)
and a positively oriented contour \(\Gamma_<\) enclosing exactly the
finite unit-disk spectrum of \(M-zL\). Define
\begin{equation}
\mathcal R_{\rm DARE}
=\sup_{z\in\Gamma_<}(|z|\alpha_L+\alpha_M)\|(zL-M)^{-1}\|.
\label{eq:dare-construction-parameters}
\end{equation}
Under the implementation and normalization assumptions of
Appendix~\ref{app:algebraic-riccati}, the algorithm returns an
\((\alpha_X,a_X,\varepsilon)\) block-encoding of \(X\), using
\begin{equation}
Q_{\rm DARE}=\widetilde O\!\left(\mathcal R_{\rm DARE}\alpha_X\right)
\label{eq:dare-query-bound}
\end{equation}
queries to the supplied pencil encodings and their adjoint and
controlled versions. Here \(\alpha_X\) is the actual output
normalization in \eqref{eq:algebraic-projector-output-app};
\(\widetilde O\) suppresses logarithmic dependence on precision and
the supplied scales.
\end{theorem}

The generalized singularity factor includes the shifted-input cost
\(|z|\alpha_L+\alpha_M\). Constructing the full projector also
requires the right factor \(L\):
\begin{equation}
\Pi_<=\mathfrak R_<[1;M,L]
=\frac{1}{2\pi\mathrm i}\oint_{\Gamma_<}(zL-M)^{-1}L\,\mathrm dz.
\label{eq:dare-full-riesz-projector}
\end{equation}
Its actual normalization, including this multiplication, enters
\(\alpha_X\). Appendix~\ref{app:dare-construction} gives the
finite unit-circle rule, its normalization, and the full query bound.
No inverse of \(L\) is needed.

The compact DARE identifies the selected graph through
\begin{equation}
M\begin{bmatrix}\mathsf I\\X\end{bmatrix}
=L\begin{bmatrix}\mathsf I\\X\end{bmatrix}F,
\qquad F=(\mathsf I+GX)^{-1}A,\qquad \rho(F)<1.
\label{eq:dare-deflating-graph}
\end{equation}
The following graph identity remains valid when the complementary
subspace contains infinite modes.

\begin{lemma}[DARE graph projector and recovery]
\label{lem:dare-graph-recovery}
Let \(X\) be the stabilizing solution in
Definition~\ref{def:problem-dare}.
For \(E_1=[\mathsf I;0]\) and \(E_2=[0;\mathsf I]\), the finite
unit-disk projector of \(M-zL\) satisfies
\begin{equation}
\Pi_<^2=\Pi_<,
\qquad
\operatorname{ran}\Pi_<=\operatorname{ran}\begin{bmatrix}\mathsf I\\X\end{bmatrix}.
\label{eq:dare-graph-projector}
\end{equation}
Its upper block row has full row rank, and
\begin{equation}
X=(E_2^*\Pi_<)(E_1^*\Pi_<)^+.
\label{eq:dare-graph-recovery}
\end{equation}
\end{lemma}

The limiting projector arises from the RR evolution. When \(A\) is
invertible, its unscaled lift is \(\mathcal S_F=M^{-1}L\).
With \(R_0=[\mathsf I;P_0]\) and full-column-rank \(\Pi_>R_0\),
the decaying graph projector satisfies
\begin{equation}
\Pi_>+\mathcal S_F^k(\mathsf I-\Pi_>)R_0
(\Pi_>R_0)^+\mathcal S_F^{-k}\Pi_>
\longrightarrow\Pi_>,\qquad k\to\infty.
\label{eq:dare-limiting-projector}
\end{equation}
Here \(\Pi_>\) is the exterior projector of \(\mathcal S_F\), and
\(\Pi_>=\Pi_<=\mathfrak R_<[1;M,L]\) for the DARE pencil, as proved
in Lemma~\ref{lem:algebraic-projector-limits-app}.
Although \(A\) need not be invertible and \(\mathcal S_F\) may then
be unavailable, the pencil Riesz operator defines the same stable
graph projector. The algorithm constructs it directly from \(M,L\)
and applies the upper-row pseudoinverse, without evolving an initial
graph to a finite iteration count. Appendix~\ref{app:dare-construction}
proves the graph identity and construction bound.

\section{Control outputs and Riccati applications}
\label{sec:control-applications}

The preceding sections construct Riccati operators. We first follow one
control evaluation from the original coefficient encodings to a classical
feedback value. A stable RPA family then quantifies how local inverse
scales reduce query costs near a stability boundary. Finally, a thermal
network with one control input and an arbitrary number of state variables
makes explicit which physical and input-access conditions allow the gate
complexity to grow only polylogarithmically with the state dimension.

\subsection{Finite-horizon continuous LQR with classical feedback}
\label{subsec:lqr-classical-output}

Consider the finite-horizon continuous LQR problem
\begin{equation}
\begin{aligned}
\dot x(t)&=Ax(t)+bu(t),\\
J(u)&=\int_0^T
\bigl(x(t)^*Qx(t)+r|u(t)|^2\bigr)\,\mathrm dt
+x(T)^*P_Tx(T),
\end{aligned}
\label{eq:lqr-continuous-demonstration}
\end{equation}
with prescribed initial state, where \(b\in\mathbb C^{n\times1}\),
\(Q,P_T\succeq0\), and \(r>0\). The objective is to minimize \(J(u)\).
The single control input illustrates the procedure; the state dimension
\(n\) is arbitrary. At a requested physical time \(t\), write
\(s=T-t\), \(x=x(t)\), and
\(P(s)=P_{\rm LQR}(T-s)\), so the optimal feedback is
\(u(t)=-r^{-1}b^*P(s)x\)~\cite[Section~2.3]{AndersonMoore2007OptimalControl}.
The inputs are block-encodings of \(A,b,Q,P_T,x\), with their actual
normalizations, adjoints and controlled versions, together with the
classical values \(r,T,t\). The state-column encoding includes its
physical scale \(\alpha_x\); normalized state preparation alone does
not supply that scale.

\begin{proposition}[Finite-horizon LQR feedback and total query complexity]
\label{prop:lqr-be-classical-main}
For \eqref{eq:lqr-continuous-demonstration}, suppose the input encodings
are exact and the reverse-time DRE satisfies the structural, weighted-block
access and general inverse assumptions of Appendix~\ref{app:dre-complexity}.
Use direct sums and products with the raw normalizations, without
normalization reduction. Supply a bound on the solution over the
requested interval and use the output normalization
\(\alpha_{P(s)}\) specified in that appendix.
For \(0<\delta\le1/3\) and
\(0<\varepsilon_u<\alpha_b\alpha_{P(s)}\alpha_x/r\), the algorithm
returns a classical estimate satisfying
\begin{equation}
|\widehat u(t)+r^{-1}b^*P(s)x|\le\varepsilon_u
\label{eq:lqr-continuous-output}
\end{equation}
with probability at least \(1-\delta\), using a total of
\begin{equation}
Q_u=O\!\left(
\frac{\alpha_b\alpha_x\alpha_{P(s)}}{r\varepsilon_u}\,
\log\frac1\delta
\right)Q_{\rm DRE}
\label{eq:lqr-continuous-total}
\end{equation}
queries to the original coefficient and state block-encodings, their
adjoints and controlled versions. Here \(Q_{\rm DRE}\) is the
full query cost of one DRE solution circuit at the accuracy specified
in \eqref{eq:lqr-continuous-outer-accuracy-app}, including its actual
normalizations and inverse thresholds in
Proposition~\ref{prop:dre-general-resources-app}.
The compact bound \eqref{eq:dre-query-main} applies under its additional
normalization assumptions. For a tolerance at least the displayed output bound, or a
known zero factor, returning zero suffices.
\end{proposition}

The factor outside \(Q_{\rm DRE}\) is the cost of estimating the
encoded scalar at its physical scale. In particular,
\(\varepsilon_u^{-1}\log(1/\delta)\) comes from classical amplitude
estimation. The construction and normalization costs of the Riccati
encoding are already included in \(Q_{\rm DRE}\).
Appendix~\ref{app:lqr-classical-output} proves the result by composing
this readout with the existing DRE construction.

\begin{algorithm}[htbp]
\caption{Finite-horizon continuous LQR feedback from block-encodings}
\label{alg:lqr-be-classical}
\begin{algorithmic}[1]
\Require Encodings of \(A,b,Q,P_T,x(t)\), their normalizations and controlled
adjoints; \(r,T,t,\varepsilon_u,\delta\) and the supplied DRE bounds.
\State Set \(s=T-t\) and form \(G=bb^*/r\) with scale \(\alpha_b^2/r\).
\State Assemble \(\mathcal H_{\rm DRE}=[-A,G;Q,A^*]\) and
\(R_0=[\mathsf I;P_T]\).
\State Choose the solution precision from
\eqref{eq:lqr-continuous-outer-accuracy-app} and apply
Algorithm~\ref{alg:dre-direct-projector-construction} at time \(s\), keeping
the raw normalizations.
\State Compose \(U_b^*U_{P(s)}U_x\) and estimate its signal amplitude
to the prescribed accuracy and failure probability.
\State Multiply the estimate by
\(-\alpha_b\alpha_{P(s)}\alpha_x/r\) and return \(\widehat u(t)\).
\end{algorithmic}
\end{algorithm}

Products use separate signal ancillas, as in the block-encoding
multiplication rule. The signal amplitude in the last two steps is
\begin{equation}
\frac{b^*\widetilde P(s)x}{\alpha_b\alpha_{P(s)}\alpha_x}.
\label{eq:lqr-feedback-amplitude}
\end{equation}
A controlled Hadamard test followed by amplitude estimation reads its
real and imaginary parts
\cite{BrassardHoyerMoscaTapp2002AmplitudeEstimation}.
The sign or phase of the control vector is preserved in \(U_b\): supplying
\(G\) alone would not distinguish \(b\) from \(-b\).
No preparation of \(P(s)x/\|P(s)x\|\) is required, so a small
\(\|P(s)x\|\) does not introduce an inverse success-amplitude factor.

For \(B\in\mathbb C^{n\times m}\), the same Riccati construction and
matrix-element estimation evaluate components of \(-R_c^{-1}B^*P(s)x\).
Proposition~\ref{prop:local-inverse} supplies the control-space inverse.
Its conditioning, normalization and implementation cost must be included,
and requesting all \(m\) classical components requires the corresponding
readouts with a shared total failure probability.

\subsection{Random phase approximation in quantum chemistry}
\label{subsec:rpa-review}

The real random-phase approximation (RPA) amplitude equation provides
an algebraic Riccati application in quantum chemistry
\cite{RodenasRuizZhaoLee2026NonlinearMatrixEquations}.
For real symmetric \(A,B\), consider
\begin{equation}
B+AT+TA+TBT=0,\qquad
K=\begin{bmatrix}A&B\\B&A\end{bmatrix}\succ0,\qquad
\mathcal H=\begin{bmatrix}-A&-B\\B&A\end{bmatrix}.
\label{eq:rpa-review-problem}
\end{equation}
The target is the symmetric solution \(T\) for which \(-A-BT\) is
Hurwitz. To quantify the cost as stability deteriorates, we use
\begin{equation}
A=\mathsf I_2,\qquad B=\operatorname{diag}(1-u,0),\qquad
0<u\le\tfrac12,\qquad \lambda_{\min}(K)=u.
\label{eq:rpa-review-family}
\end{equation}
All matrices in this subsection depend on \(u\).
The input is an exact block-encoding \(U_{\mathcal H}\) with
normalization \(\alpha_H=\|\mathcal H\|=2-u\), together with its
controlled and adjoint versions. We seek either a solution encoding
with decoded operator-norm error \(\varepsilon_T\), or a classical
estimate of the correlation-energy quantity
\begin{equation}
e_c=\frac{\operatorname{Tr}(BT)}{4V}
   =\frac{1-u}{4V}\,e_1^*Te_1,
\label{eq:rpa-review-energy}
\end{equation}
where \(V>0\) is fixed and the energy convention is specified by this
formula.

Table~\ref{tab:rpa-review-comparison} compares uniform and local
inverse scales for the same stable-projector algorithm. Both
constructions use the finite rectangle and positive panel quadrature
in Appendix~\ref{app:rpa-review} and the same recovery.
Let \(m\) be the number of quadrature nodes.
The uniform node inverse threshold is taken from the proof of
\cite[Theorem~C.4]{RodenasRuizZhaoLee2026NonlinearMatrixEquations};
the contour and recovery are the common choices specified here.
Thus the comparison isolates the effect of inverse normalization.
As in Subsection~\ref{subsec:contour-resolvent-lcu}, node coefficients
and inverse circuits are coherently accessible. The local
implementation uses supplied bounds
\(\beta_j=\Theta(\|(z_j\mathsf I-\mathcal H)^{-1}\|)\), uniformly
in the nodes and in \(u\); their preparation cost is separate.

\begin{table}[!t]
\centering
\small
\renewcommand{\arraystretch}{1.18}
\begin{tabularx}{\textwidth}{@{}>{\raggedright\arraybackslash}Xcc@{}}
\toprule
Quantity & Previous work & Our work\\
\midrule
Generalized singularity factor & \(\Theta(u^{-1})\) & \(\Theta(u^{-1})\)\\
Projector normalization & \(\Theta(u^{-1})\) & \(\Theta(u^{-1/2})\)\\
Upper-row pseudoinverse norm & \(1\) & \(1\)\\
Solution normalization \(\alpha_T\) & \(\Theta(u^{-1})\) & \(\Theta(u^{-1/2})\)\\
Queries for a solution encoding & \(\widetilde O(u^{-2})\) & \(\widetilde O(u^{-3/2})\)\\
Queries for energy readout & \(\widetilde O(u^{-3}/\varepsilon_c)\)
 & \(\widetilde O(u^{-2}/\varepsilon_c)\)\\
\bottomrule
\end{tabularx}
\caption{Query bounds as \(u\downarrow0\) for the two implementations on
the same input family. Solution encodings have the same decoded error;
energy estimates have the same absolute error. The bounds include the different output
normalizations. Both constructions recover from the
full stable projector using the same upper-row pseudoinverse.}
\label{tab:rpa-review-comparison}
\end{table}

\begin{proposition}[RPA solution encoding and energy readout]
\label{prop:rpa-review-comparison}
\RPAComparisonStatement
\end{proposition}
\Needspace{6\baselineskip}
\begin{proof}[Proof outline]
The explicit stable graph and the finite quadrature in
Appendix~\ref{app:rpa-review} give the first three rows of
Table~\ref{tab:rpa-review-comparison}. The common projector recovery
then gives the solution bounds. Matrix-element estimation introduces
one further factor of the actual \(\alpha_T\). The appendix verifies
these scales and applies the general stability bound.
\end{proof}

Both constructions encode the full stable projector and recover
\[
T=(E_2^*\Pi_-)(E_1^*\Pi_-)^+.
\]
The query bounds improve by factors \(u^{-1/2}\) for solution encoding
and \(u^{-1}\) for energy readout. The source of this improvement is the
difference between a peak inverse norm and its weighted contribution
along the contour:
\begin{equation}
\begin{aligned}
\mathcal R_{\rm CARE}
&=\sup_{z\in\Gamma}(|z|+\alpha_H)
  \|(z\mathsf I-\mathcal H)^{-1}\|=\Theta(u^{-1}),\\
\alpha_{\Pi_-}=\alpha_{\rm CARE}
&=\sum_{j=1}^{m}|\omega_j|\beta_j=\Theta(u^{-1/2}).
\end{aligned}
\label{eq:rpa-review-local-scales}
\end{equation}
The largest inverse scale still determines the depth of the coherent
inverse stage. Only a short part of the contour has this scale, so
weighting each node by its own bound reduces the projector
normalization, the cost of recovery, and the scale paid in readout.
The upper-row pseudoinverse norm equals one for both implementations.
For this nonnormal lift, the singular-value gap on the imaginary axis
is \(u\), whereas the smallest eigenvalue magnitude is
\(\sqrt{2u-u^2}\); these gaps cannot be interchanged in the cost bound.

Nodewise rebalancing for CARE is already present in
\cite[Definitions~5.8 and~9.10, Theorems~5.9 and~9.12]
{WangLiu2026SignEmbedding}.
The comparison quantifies local versus uniform normalization for the
specified constructions. It does not establish superiority over that
rebalanced algorithm or over optimized contour choices.
The explicitly solvable two-mode family illustrates dependence on the
stability parameter, rather than a quantum advantage for molecular
instances. The thermal-network application that follows addresses a
different regime: increasing the state dimension while keeping the
physical scales bounded.

\subsection{Heated boundary control and finite-horizon feedback}
\label{subsec:thermal-model}

We consider heated boundary control (HBC), in which boundary heat input
regulates the temperature distribution. Such problems have important
engineering applications in thermal processing, including the control
of cooling profiles in glass manufacturing~\cite{PinnauSchulzeGlassCooling}.
The thermal network below provides a finite-dimensional example with
a single heater at one endpoint.

Consider a path of \(n\ge2\) thermal nodes, with one local heater,
nearest-neighbor conductance and uniform heat loss to the environment.
The continuous model is
\begin{equation}
\begin{aligned}
\dot x(t)&=A_{\rm c}x(t)+b_{\rm c}u(t),\\
A_{\rm c}&=-\nu\mathsf I-\kappa\mathcal L,
&b_{\rm c}&=e_{1},\\
q&>0,
&r&>0,\qquad P_{T}=0,
\end{aligned}
\label{eq:thermal-continuous-model}
\end{equation}
where \(\nu,\kappa>0\) and \(e_{1}\) is the first coordinate
column. The unweighted path Laplacian is defined by
\[
\mathcal L_{ij}=
\begin{cases}
1, & i=j\in\{1,n\},\\
2, & 1<i=j<n,\\
-1, & |i-j|=1,\\
0, & \text{otherwise},
\end{cases}
\qquad 1\le i,j\le n.
\]
The finite-horizon objective is
\[
\int_0^{T}
\bigl(q\|x(t)\|^2+r|u(t)|^2\bigr)\,\mathrm dt.
\]
The physical coefficients are efficiently computable and remain fixed
as nodes are added.

The control matrix is \(G_{\rm c}=b_{\rm c}(b_{\rm c})^*/r\).
Splitting the path edges into odd and even matchings implements the
continuous coefficient encodings with \(\operatorname{polylog}n\)
gates per call. Appendix~\ref{app:thermal-inputs} gives these circuits
and their actual scales, including coefficient synthesis.
Assume a supplied current-state column
encoding with \(\alpha_x=O(1)\) and
\(\operatorname{polylog}n\) gates per call, with polylogarithmic
precision overhead. A localized excitation or an efficient
preparation with known bounded physical norm meets this condition.
The state preparation and its physical scale are part of the input;
an arbitrary classical temperature array is not supplied for free.

For the continuous network, set \(s=T-t\). The value matrix at
physical time \(t\) solves the reverse-time DRE
\begin{equation}
\begin{aligned}
P'&=(A_{\rm c})^*P+P A_{\rm c}
-PG_{\rm c}P+q\mathsf I,&P(0)&=0,\\
\mathcal H_{{\rm DRE}}
&=\begin{bmatrix} -A_{\rm c}&G_{\rm c}\\
q\mathsf I&(A_{\rm c})^*\end{bmatrix},
&R_{0}&=E_{1}.
\end{aligned}
\label{eq:thermal-dre-lift}
\end{equation}
Here \(E_{1}=[\mathsf I;0]\); let \(\Pi_{+}\) be the
right-half-plane projector of this Hamiltonian.

\begin{corollary}[Finite-horizon feedback for the thermal network]
\label{cor:thermal-dre-main}
For \eqref{eq:thermal-continuous-model}, fix
\(\nu,\kappa,q,r>0\) and a bounded, finitely represented
time range \(0\le s\le T\), independently of \(n\).
Use the input implementation and state access of
Subsection~\ref{subsec:thermal-model}. For \(0<\varepsilon\le1\),
the direct projector-block construction encodes \(P(s)\) with decoded
error at most \(\varepsilon\), output normalization
\(\alpha_P=O(1)\) as specified in \eqref{eq:thermal-dre-output-scale}, and
\(O(\log^3(e+1/\varepsilon))\) basic continuous-matrix queries.
For \(0<\delta\le1/3\) and a nontrivial control tolerance
\(0<\varepsilon_{u}\le1\), it returns
\begin{equation}
\widehat u(t)\approx
-r^{-1}(b_{\rm c})^*P(T-t)x(t)
\label{eq:thermal-dre-feedback}
\end{equation}
with absolute error at most \(\varepsilon_{u}\) and failure
probability at most \(\delta\), using
\begin{equation}
O\!\left(\frac{\log(1/\delta)}{\varepsilon_{u}}
\log^3(e+1/\varepsilon_{u})\right)
\label{eq:thermal-dre-total}
\end{equation}
basic matrix and state-oracle calls. The total gate cost is
\(\varepsilon_{u}^{-1}\log(1/\delta)
\operatorname{polylog}(n)\operatorname{polylog}(e+1/\varepsilon_{u})\).
The constants depend on the fixed physical parameters, time range and
supplied state normalization, but not on \(n\).
\end{corollary}

The factor \(\varepsilon_{u}^{-1}\) outside the logarithms
comes from classical amplitude estimation. The solution encoding
uses the three precision logarithms of
Proposition~\ref{prop:dre-general-resources-app}; reading the control
value adds the factor in Proposition~\ref{prop:lqr-be-classical-main}.

For fixed \(\nu,\kappa,q,r>0\) and a bounded time range, the three
complexity parameters satisfy \(\mathcal R_{\rm DRE}=O(1)\),
\(\sigma_{\min}(\Pi_+R_0)^{-1}=O(1)\), and \(\alpha_P=O(1)\).
Their bounds are independent of the matrix dimension \(n\) but depend
on the original problem parameters \(\nu,\kappa,q,r\);
Appendix~\ref{app:thermal-dre} gives this dependence explicitly.

This result concerns a single classical feedback evaluation as the
network grows with fixed local physics and bounded state normalization.
Increasing the resolution of a heat equation on a fixed interval would
instead change the diffusion scale and hence the same complexity
parameters. State loading and time-weight arithmetic remain part of
the gate count.

We evaluate the numerical construction and its resources with the unscaled
path Laplacian, \(\kappa=1/4\), \(q=2\), and \(T=1\).
The baseline uses \(\nu=5\), \(r=0.5\), and target spectral error
\(\varepsilon=10^{-8}\). The right contour is the physical rectangle
\(\operatorname{Re}z\in[1.7,9.624555320336759]\),
\(|\operatorname{Im}z|\le3.3\); the left contour is its image under
\(z\mapsto-z\). The baseline uses 104 Gauss--Legendre nodes per contour,
with 28 on each horizontal edge and 24 on each vertical edge.

The algorithm and reference trajectories in
Figure~\ref{fig:hbc-control} overlap at the displayed scale.
Both costs are approximately \(0.19173020837115\);
the maximum state and control differences are
\(1.00\times10^{-11}\) and \(3.81\times10^{-10}\), respectively.
The trajectory uses repeated classical feedback evaluations.

\begin{figure}[H]
\centering
\includegraphics[width=\textwidth]{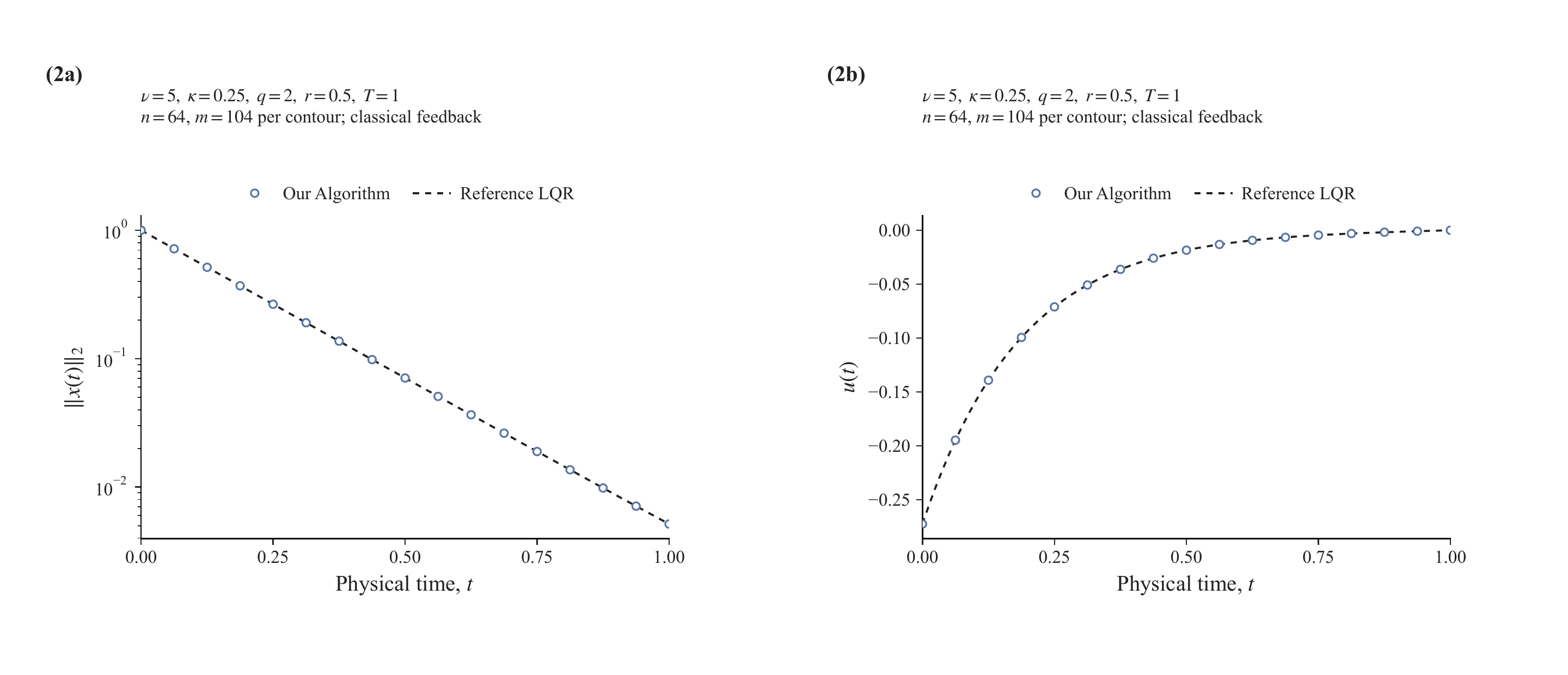}
\small\caption{%
Classical closed-loop response with \(n=64\) and 104 nodes per contour
under the common settings of Appendix~\ref{app:hbc-numerics}.
(a) State norm. (b) Scalar control.
Open circles without connecting lines show Our Algorithm; dashed
curves show Reference LQR. One circle is displayed every 25 samples.
The algorithm uses cubic-spline interpolation of 161 Riccati evaluations;
the reference uses the independent closed-form solution.
Each trajectory stores all 401 sampled values.
Amplitude estimation is not simulated.%
}
\label{fig:hbc-control}
\end{figure}

Figure~\ref{fig:hbc-dimension-cost} compares dimension dependence for
the same fixed physical parameters and target accuracy.
The executed RK4 arithmetic count grows from \(4.95\times10^7\) to
\(1.276\times10^{10}\) over \(n=64,128,256,512,1024\), while the
complete quantum construction uses \(82\,516\,028\) logical
Hamiltonian-oracle queries at every dimension.
The maximum measured spectral errors are \(5.18\times10^{-11}\)
for RK4 and \(1.88\times10^{-10}\) for the quantum success-block
simulation, both below \(10^{-8}\).
The distinct units and output representations are specified in the caption;
the counting and numerical-validation details are given in
Appendix~\ref{app:hbc-numerics}.

\begin{figure}[H]
\centering
\includegraphics[width=0.95\textwidth]{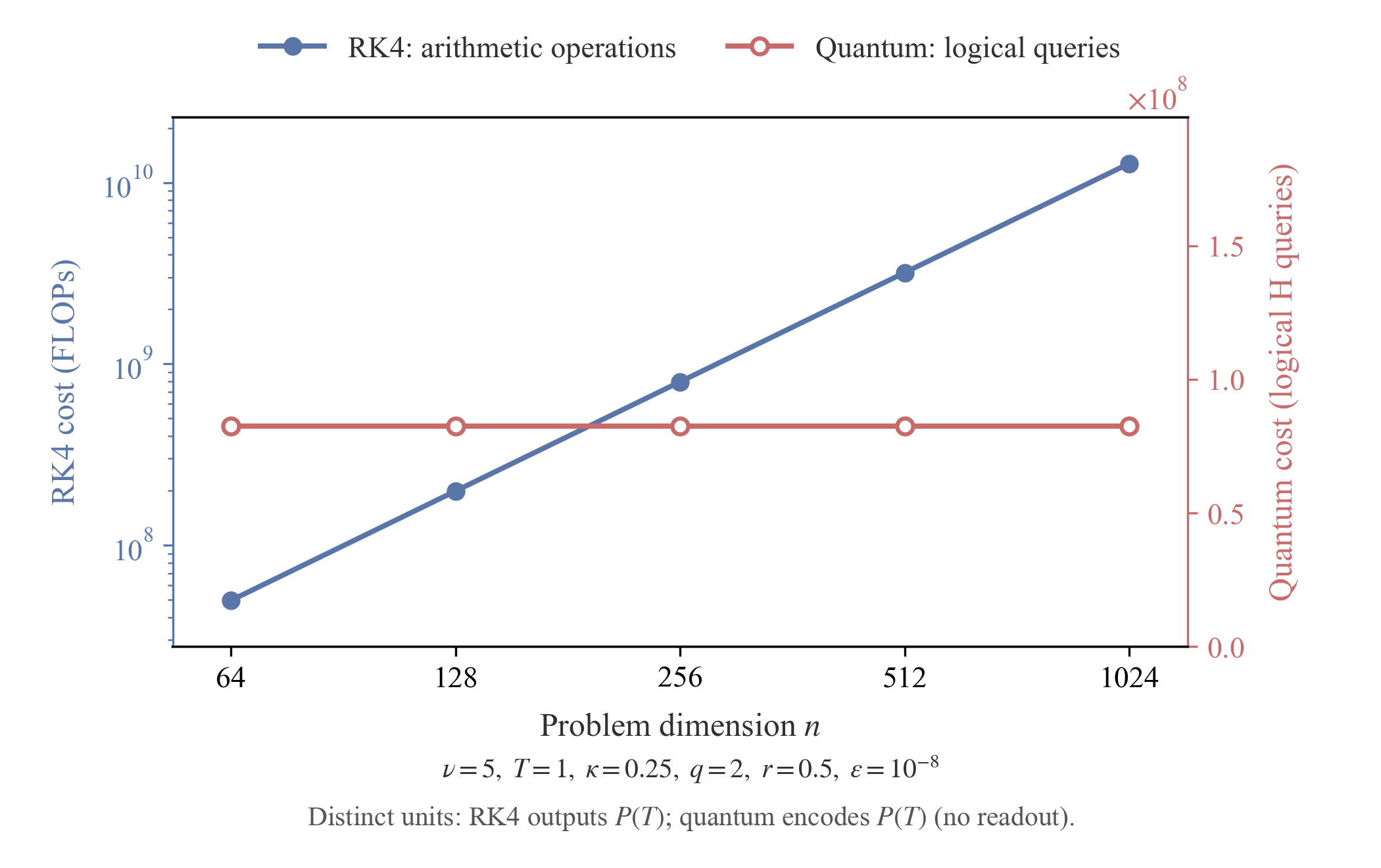}
\small\caption{%
Dimension dependence at \(\nu=5\), \(T=1\), \(\kappa=1/4\), \(q=2\),
\(r=0.5\), and \(\varepsilon=10^{-8}\), with the unscaled path Laplacian.
The blue curve uses the logarithmic left axis and counts RK4 additions,
subtractions, multiplications, and divisions, including all five refinement
runs and error-indicator calculations. The red curve uses the linear right
axis and counts logical Hamiltonian-oracle queries for one complete
\(P(1)\) block-encoding, with 104 nodes per contour and frozen QSP
degrees \((757,29,57)\). It includes the transient term and both recovery
inverses; no normalization amplification is used.
RK4 outputs an explicit matrix, whereas the quantum output is a block-encoding
without readout. The axes have different units, so neither a curve crossing
nor a count ratio represents a runtime speedup.%
}
\label{fig:hbc-dimension-cost}
\end{figure}

\clearpage
\section{Query lower bounds and BQP completeness}
\label{sec:lower-bounds}

We now establish two limitations on computing Riccati solution operators.
For each of the four problems, one pair of oracle inputs gives a lower
bound for the product of the generalized singularity factor and the
prescribed algebraic output normalization or dynamic initialization
condition number. We then embed
quantum circuits into single-control DARE and DRE instances, proving
hardness for solution encodings with constant normalization and accuracy.
These two conclusions use different input families and access models.

\subsection{Product query lower bounds}
\label{subsec:product-lower-bounds}

The lower-bound idea is to hide a visible change in the solution behind
a very small change in the available input oracles.
Fix public parameters \(0<\mu_1,\mu_2\le1/8\) and a hidden bit
\(c\in\{0,1\}\). Each family below has
\(\mathcal R=\Theta(\mu_1^{-1})\). The algebraic families have
solution norms and prescribed output normalizations of order
\(\mu_2^{-1}\); the dynamic families have initialization condition
numbers of this order. These estimates hold for both inputs with
uniform constants. Thus \(\mu_1\) controls spectral separation,
whereas \(\mu_2\) controls the algebraic solution scale or a small
initial component along the selected dynamic branch.
Access is through the complete Julia input unitaries specified in
Appendix~\ref{app:product-lower-bounds}, including their adjoints and
controlled versions. All other gates, numerical data, auxiliary
registers, and signal spaces are common to the two inputs. In
particular, the declared output normalization and the requested time or
step are public and independent of \(c\). We count the worst-case
queries of one coherent circuit that constructs a correct solution
encoding for either input.

The input normalizations and contours are fixed as in
Appendix~\ref{app:product-lower-bounds}. For CARE and DARE, let
\(X_c\) be the stabilizing solution and use the common actual output
normalization \(\alpha_X=\Theta(\|X_c\|)=\Theta(1/\mu_2)\)
constructed in Proposition~\ref{prop:algebraic-product-normalization-app}.
Require \(\varepsilon\le\|X_1-X_0\|/8\) in the accuracy range
of the corresponding quantum problem. For these families, the solution
difference equals the smaller solution norm. For DRE and RR, use
\(\alpha_{P(t)}=\alpha_{P_k}=3\) and \(\varepsilon\le1/16\),
respectively at the public time and step
\[
t=\frac{\log(1/\mu_2)}{2\mu_1},\qquad
k=\left\lceil\frac{\log(1/\mu_2)}{2\log(1+\mu_1)}\right\rceil.
\]
These prescribed normalizations are part of the lower-bound tasks.
After division by the common output normalization, the two solutions
remain a constant distance apart. Solving the problem must therefore
reveal which input was supplied, whereas each query changes the two
computations by only \(O(\mu_1\mu_2)\). Accumulating a constant
separation then requires \(\Omega((\mu_1\mu_2)^{-1})\) queries.

\begin{theorem}[DARE product query lower bound]
\label{thm:dare-product-lower-main}
Under the above access and accuracy conventions, there exists a family
of scalar instances of Problem~\ref{prob:quantum-dare} requiring
\[
q_{\rm DARE}=\Omega\!\left(\mathcal R_{\rm DARE}\alpha_X\right)
\]
quantum queries in the worst case. The two factors can be made
independently arbitrarily large.
\end{theorem}

\begin{proof}[Informal proof]
We construct two scalar DARE inputs whose finite pencil eigenvalues are
\(r=1-\mu_1\) and \(r^{-1}\). Their stable graph projectors and
solutions are
\begin{equation}
\Pi_<=\frac1{1+c+\mu_2}
\begin{bmatrix}\mu_2&\mu_2\\1+c&1+c\end{bmatrix},\qquad
X_c=\frac{1+c}{\mu_2}.
\label{eq:dare-product-idea-main}
\end{equation}
Appendix~\ref{app:product-lower-bounds} gives positive scalar
coefficients realizing these data and proves
\(\mathcal R_{\rm DARE}=\Theta(\mu_1^{-1})\).
The parameter \(\mu_1\) fixes spectral separation, while
\(\mu_2\) fixes the solution and output scales.

The two complete pencil input unitaries differ by only
\(O(\mu_1\mu_2)\). Since
\(\|X_1-X_0\|=1/\mu_2\) and \(\alpha_X=\Theta(1/\mu_2)\),
the normalized solution blocks remain a constant distance apart under
the stated error bound. Lemma~\ref{lem:query-hybrid-app} therefore requires
\(\Omega((\mu_1\mu_2)^{-1})\) queries to produce their constant
separation, which is the claimed product. The complete coefficient,
contour, and oracle calculations appear in
Appendix~\ref{app:product-lower-bounds}.
\end{proof}

\begin{theorem}[CARE product query lower bound]
\label{thm:care-product-lower-main}
Under the above access and accuracy conventions, there exists a family
of two-state instances of Problem~\ref{prob:quantum-care} requiring
\[
q_{\rm CARE}=\Omega\!\left(\mathcal R_{\rm CARE}\alpha_X\right)
\]
quantum queries in the worst case. The two factors can be made
independently arbitrarily large.
\end{theorem}

\begin{theorem}[DRE product query lower bound]
\label{thm:dre-product-lower-main}
Under the above access and accuracy conventions, there exists a family
of two-state instances of Problem~\ref{prob:quantum-dre} requiring
\[
q_{\rm DRE}=\Omega\!\left(\mathcal R_{\rm DRE}\kappa(\Pi_+R_0)\right)
\]
quantum queries in the worst case. The two factors can be made
independently arbitrarily large.
\end{theorem}

\begin{theorem}[RR product query lower bound]
\label{thm:rr-product-lower-main}
Under the above access and accuracy conventions, there exists a family
of two-state instances of Problem~\ref{prob:quantum-rr} requiring
\[
q_{\rm RR}=\Omega\!\left(\mathcal R_{\rm RR}\kappa(\Pi_>R_0)\right)
\]
quantum queries in the worst case. The two factors can be made
independently arbitrarily large.
\end{theorem}

The algebraic solution norms can grow arbitrarily large as
\(\mu_2\to0\). For these same families,
Proposition~\ref{prop:algebraic-product-normalization-app} constructs
the full projectors at constant raw normalization and applies projector-block
recovery at an actual public output normalization
\(\alpha_X=\Theta(\|X_c\|)=\Theta(1/\mu_2)\), with no
projector re-encoding cost. The lower bounds above use this actual
normalization, so the product captures both spectral separation and the algebraic output
scale, even for unbounded solutions. An arbitrary enlargement of the
declared normalization would reduce the normalized output separation
and is outside these fixed-scale lower-bound statements.
The dynamic examples have bounded solutions and constant absolute
error at the displayed time or step; their second factor instead
measures the difficulty of initialization. Optimal dependence on
additional normalization and re-encoding costs for general inputs
remains open.

\begin{remark}[Scope of optimality]
\label{rem:product-optimality-scope}
The matching bounds establish worst-case optimal parameter dependence,
up to logarithmic factors, for the specified oracle families and
output conventions. They do not imply that our algorithms are efficient
or optimal for every Riccati instance. For example, some problems
with a poorly conditioned Hamiltonian lift \(H\) may still incur large
query costs. These instances could effect the generalized singularity factor
\(\mathcal R\) and the stated initialization and normalization factors
(see Proposition~\ref{prop:factor-condition-number-app}).
\end{remark}

\subsection{BQP hardness of DARE and DRE solution encodings}
\label{subsec:bqp-hardness}

For the computational hardness results, the input is a finite uniform
circuit description of a specified Riccati family, including efficient
matrix-access circuits. BQP denotes bounded-error quantum polynomial
time \cite{BernsteinVazirani1997QuantumComplexity,Watrous2018Theory}.
We prove completeness for selected-value decision problems on the
images of two explicit circuit reductions. A solution encoding with
constant normalization and sufficiently small constant decoded error
would decide the same promises by a constant number of matrix-element
tests, as formalized in Lemma~\ref{lem:selected-value-hardness-app}.

All solution encodings below include controlled access.

\begin{theorem}[Single-control DARE hardness]
\label{thm:dare-bqp-main}
There exists a uniformly circuit-generated family of single-control
instances of Problem~\ref{prob:quantum-dare} for which constructing a
solution block-encoding with \(\alpha_X=2\) and
\(\varepsilon\le1/100\) is \(\mathsf{BQP}\)-hard.
On the construction family in Appendix~\ref{app:dare-bqp},
distinguishing \(x_0^*Xx_0\ge2/9\) from \(x_0^*Xx_0\le1/18\),
for a specified basis vector \(x_0\), is \(\mathsf{BQP}\)-complete.
\end{theorem}

\begin{proof}[Informal proof]
Let \(V\) be a real source circuit on \(|0\rangle_{\rm work}\), with
acceptance probability \(p\). Apply \(V\), copy its accepting output
bit to a fresh flag qubit, and undo \(V\). The amplitude of
\(|0\rangle_{\rm work}|1\rangle_{\rm flag}\) is then
\begin{equation}
{}_{\rm work}\langle0|V^*
\bigl(|1\rangle\langle1|_{\rm out}\otimes\mathsf I\bigr)
V|0\rangle_{\rm work}=p.
\label{eq:dare-bqp-amplitude-main}
\end{equation}
Thus the accepting and rejecting source circuits give different
amplitudes at one fixed, known basis state.

Write the gates of this augmented circuit as \(U_1,\ldots,U_N\).
Introduce clock levels \(0,\ldots,N\) and set
\begin{equation}
\begin{aligned}
A&=\sum_{j=0}^{N-1}|j+1\rangle\langle j|\otimes U_{j+1},\\
x_0&=|0\rangle_{\rm clk}|0\rangle_{\rm work}|0\rangle_{\rm flag},&
b&=|N\rangle_{\rm clk}|0\rangle_{\rm work}|1\rangle_{\rm flag}.
\end{aligned}
\label{eq:dare-bqp-clock-main}
\end{equation}
Each multiplication by \(A\) advances one clock level and applies
one circuit gate, so \(b^*A^Nx_0=p\). The terminal level has no
outgoing transition, giving \(A^{N+1}=0\). Choose the single control
and the state cost at that terminal basis state:
\(B=b\), \(R_c=1\), and \(Q=G=bb^*\).

The control can partially cancel the arriving terminal amplitude.
For a real state \(z\) on clock level \(N-1\), its Bellman update
minimizes \(u^2+|b^*Az+u|^2\), giving \(|b^*Az|^2/2\).
Propagating this value back through the preceding clock levels yields
\begin{equation}
X=Q+\frac12\sum_{j=1}^{N}(A^*)^jQA^j.
\label{eq:dare-bqp-solution-main}
\end{equation}
The summands act on different clock levels, so \(\|X\|=1\).
The closed loop \(F=(\mathsf I-Q/2)A\) remains nilpotent, making
this the stabilizing Riccati solution.

Since \(x_0\) lies on clock level zero, only the term \(j=N\)
contributes to its value. Hence
\begin{equation}
x_0^*Xx_0=\frac12|b^*A^Nx_0|^2=\frac{p^2}{2}.
\label{eq:dare-bqp-value-main}
\end{equation}
Source acceptance probabilities at least \(2/3\) and at most
\(1/3\) therefore give optimal costs at least \(2/9\) and at most
\(1/18\), respectively. A solution encoding with normalization
\(\alpha_X=2\) and the stated accuracy reveals this gap through a
constant number of matrix-element tests. Conversely, running the
given source circuit distinguishes the two cases, proving membership
on this construction family. The clock and coefficient-access
circuits have polynomial-size descriptions; their precise
implementation and the full Riccati verification are given in
Appendix~\ref{app:dare-bqp}.
\end{proof}

\begin{theorem}[Single-control DRE hardness]
\label{thm:dre-bqp-main}
There exists a uniformly circuit-generated family of single-control,
zero-terminal LQR instances with polynomial horizon \(T\) whose
time-reversed DRE solution \(P(T)\) in
Problem~\ref{prob:quantum-dre} is \(\mathsf{BQP}\)-hard to
block-encode with \(\alpha_{P(T)}=4\) and \(\varepsilon\le1/100\).
On the construction family in Appendix~\ref{app:dre-bqp},
distinguishing \(x_0^*P(T)x_0>1/20\) from \(x_0^*P(T)x_0<1/200\),
for a specified basis vector \(x_0\), is \(\mathsf{BQP}\)-complete.
\end{theorem}

Both construction families satisfy the corresponding problem assumptions
and admit the algorithms in Appendices~\ref{app:dre-complexity}
and~\ref{app:dare-construction}, as verified in their proofs.

\begin{remark}[Computational difficulty of LQR values]
\label{rem:lqr-computational-hardness}
Theorems~\ref{thm:dare-bqp-main} and~\ref{thm:dre-bqp-main} show that
LQR optimal values can encode general quantum computations even with
one control input and positive semidefinite quadratic costs. A single
known initial basis state carries the acceptance decision through its
optimal cost, with a constant gap. This identifies an intrinsic
computational difficulty in extracting value information from the
dynamics. Here input size is the length of the finite circuit
description, which can be much smaller than the state dimension.
A randomized classical algorithm that estimates these optimal values
to the stated constant accuracy in time polynomial in that description
length would decide every problem in BQP with bounded error.
\end{remark}


\section{Conclusion}
\label{sec:conclusion}

We develop a unified quantum approach to continuous- and discrete-time
algebraic Riccati equations, differential Riccati equations, and finite
Riccati recursions. The Quantum Weighted Riesz Method combines spectral
branch selection with the initial data to construct decaying graph
projectors and recovers algebraic solutions from their limiting
projectors. Complete block-row recovery then produces the solution
block-encoding, avoiding an additional square-block invertibility
assumption and controlling the recovery pseudoinverse through the
solution norm. This common construction preserves the initial information
needed for dynamics and the stability selection needed for algebraic
solutions.

We establish matching upper and lower query bounds, up to logarithmic
factors, on explicit algebraic input-oracle families at the actual
normalizations achieved by our algorithm. The dynamic lower bounds show
that spectral separation and initialization costs must multiply in the
worst case. Our single-control DARE and DRE reductions also establish
BQP-hardness of constant-normalization, constant-accuracy solution
encodings and BQP-completeness of selected-value promise problems on
the circuit-generated families. Applications to RPA in quantum chemistry and heated boundary
control demonstrate reduced query bounds from local inverse scales and
efficient feedback evaluation in the specified thermal-network regime.

Future work will focus on extending quantum matrix-equation methods
beyond Riccati problems, including coupled and higher-degree nonlinear
equations. For matrix differential equations, an important direction is
to combine suitable linear representations with quantum linear-system
algorithms (QLSAs), and to study linearization and time discretization
when such representations are unavailable. The resulting methods should
account for both approximation error and the cost of recovering the
requested matrix solution or observable. Efficient input block-encoding
preparation is another practical ingredient that merits further study.

\section*{Acknowledgement}
JPL acknowledges support from the Quantum Science and Technology National Science and Technology Major Project under Grant No.~2024ZD0300500, the Excellent Young Scientists Fund Program, start-up funding from Tsinghua University, and the Beijing Institute of Mathematical Sciences and Applications.

\bibliographystyle{unsrt}
\bibliography{Riccati}

\clearpage
\appendix
\section{Foundations and contour construction}
\label{app:common-foundations}

The constructions in Section~\ref{sec:quantum-tools} use a common sequence
of operations: identify a weighted spectral branch, encode its shifted
inverses, and combine them into a Riesz block. This appendix supplies the
identities and error bounds behind that sequence. We first verify the LQR
conventions and the inverse primitive, then establish the spectral meaning
of the ordinary and pencil integrals. The contour proofs construct their
finite sums; quadrature and input perturbations complete the approximation
to the original Riesz operator.

\subsection{LQR identities and time conventions}
\label{app:lqr-conventions}

The finite-horizon formulations in Subsection~\ref{subsec:lqr-riccati-problems}
use two time conventions: physical time for continuous control and the
number of remaining steps for discrete control. The following identities
verify the signs, matrix products, and feedback indices used there.
They also explain why the positive semidefinite recursion remains defined
when its iterates are singular.

\begin{proposition}[Finite-horizon LQR identities]
\label{prop:lqr-conventions-app}
Use the control data of \eqref{eq:lqr-continuous-cost}
and~\eqref{eq:lqr-discrete-cost}, with \(G=BR_c^{-1}B^*\).
If a differentiable Hermitian \(P_{\rm LQR}\) satisfies
\eqref{eq:lqr-terminal-dre} on \([0,T]\), then every admissible control
and its trajectory satisfy \cite[Sec.~2.3, Problem~2.3-2]{AndersonMoore2007OptimalControl}
\begin{equation}
J_T(u)=x_0^*P_{\rm LQR}(0)x_0+
\int_0^T v(t)^*R_cv(t)\,\mathrm dt,
\qquad v(t)=u(t)+R_c^{-1}B^*P_{\rm LQR}(t)x(t).
\label{eq:lqr-square-continuous-app}
\end{equation}
Thus the feedback in Subsection~\ref{subsec:lqr-riccati-problems} is optimal.
The reversed matrix \(P(s)=P_{\rm LQR}(T-s)\) solves
Definition~\ref{def:problem-dre} with data \((-A,-G,-Q,P_T)\).

For any \(P\succeq0\), the one-step discrete minimization satisfies \cite[Sec.~2.4]{AndersonMoore2007OptimalControl}
\begin{equation}
\begin{aligned}
\min_u\{x^*Qx+u^*R_cu+(Ax+Bu)^*P(Ax+Bu)\}
 &=x^*\bigl[Q+A^*P(\mathsf I+GP)^{-1}A\bigr]x,\\
u^*&=-(R_c+B^*PB)^{-1}B^*PAx.
\end{aligned}
\label{eq:lqr-bellman-app}
\end{equation}
More generally, for any \(G,P\succeq0\), \(\mathsf I+GP\) is
invertible and \(P(\mathsf I+GP)^{-1}\succeq0\).
Consequently, Definition~\ref{def:problem-recursion} preserves positive
semidefiniteness, and its \(k\)-step control uses \(P_{k-1}\) as in
\eqref{eq:recursion-first-control}.
\end{proposition}

\begin{proof}
Along \(\dot x=Ax+Bu\), substitute the terminal Riccati equation into
the derivative of \(x^*P_{\rm LQR}x\). Completing the square gives
\[
x^*Qx+u^*R_cu+\frac{\mathrm d}{\mathrm dt}(x^*P_{\rm LQR}x)
=(u+R_c^{-1}B^*P_{\rm LQR}x)^*R_c
 (u+R_c^{-1}B^*P_{\rm LQR}x).
\]
Integration and \(P_{\rm LQR}(T)=P_T\) prove
\eqref{eq:lqr-square-continuous-app}. Its integral is nonnegative
and vanishes under the displayed feedback. The change of variable gives
\(P'(s)=Q+A^*P+PA-PGP\), which is precisely the stated input
substitution in \eqref{eq:problem-dre}.

For the discrete step, \(R_c+B^*PB\succ0\). Completing its quadratic
form in \(u\) gives the optimizer in \eqref{eq:lqr-bellman-app} and
the value matrix
\[
Q+A^*\bigl[P-PB(R_c+B^*PB)^{-1}B^*P\bigr]A.
\]
The identity
\[
(\mathsf I+BR_c^{-1}B^*P)^{-1}
=\mathsf I-B(R_c+B^*PB)^{-1}B^*P
\]
follows by multiplication, and proves the compact formula \cite[Eq.~(24)]{Poloni2020RiccatiIterations}.

For arbitrary positive semidefinite \(G,P\), Sylvester's determinant
identity \cite[Problem~1.49]{Higham2008Functions} yields
\(\det(\mathsf I+GP)=\det(\mathsf I+P^{1/2}GP^{1/2})>0\).
Multiplying the following expression on the right by \(\mathsf I+GP\)
verifies that
\begin{equation}
P(\mathsf I+GP)^{-1}
=P^{1/2}(\mathsf I+P^{1/2}GP^{1/2})^{-1}P^{1/2}\succeq0.
\label{eq:lqr-psd-identity-app}
\end{equation}
This proves existence of each inverse and positive semidefiniteness by
induction. With \(j+1\) steps remaining, the next-stage value is \(P_j\);
setting \(j=k-1\) proves the first-control formula.
\end{proof}

These identities identify the matrix outputs of the finite-horizon
problems. The next step is to specify how inverses of encoded matrices
are represented, including their normalization and decoded error.

\subsection{Inverse primitives and input perturbations}
\label{app:inverse-perturbation}

Proposition~\ref{prop:local-inverse} applies the QSVT pseudoinverse
construction to a matrix with arbitrary input normalization. The proof
below makes that rescaling explicit, so that the same primitive can be
used at each contour node.

\newtheorem*{inverserestatement}{Proposition~\ref{prop:local-inverse}}
\begin{inverserestatement}
Let \(U_T\) be an exact block-encoding of \(T\) with normalization
\(\alpha_T\), and suppose the nonzero singular values of \(T\) lie in
\([\gamma,\alpha_T]\), where \(0<\gamma\le\alpha_T\) is supplied to the inverse
construction. Under the controlled-access and singular-value-transformation
assumptions of that primitive, an encoding of the Moore--Penrose
pseudoinverse \(T^+\) can be constructed
with normalization \(O(1/\gamma)\), decoded error
\(\varepsilon_{\rm inv}>0\) satisfying \(\alpha_T\varepsilon_{\rm inv}\le1\),
and query cost
\begin{equation*}
O\!\left(
\frac{\alpha_T}{\gamma}
\log\!\left(e+\frac{1}{\gamma\varepsilon_{\rm inv}}\right)
\right).
\end{equation*}
For square invertible \(T\), the target is \(T^{-1}\).
For a rectangular matrix, \(T^+\) vanishes on the zero singular subspace.
\end{inverserestatement}

\begin{proof}
Set \(A=T/\alpha_T\),
\(\delta=\gamma/(2\alpha_T)\), and
\(\zeta=\gamma\varepsilon_{\rm inv}/4\).
Then \(0<\zeta\le\delta/2\le1/4\), and every singular value of
\(A\) is either zero or at least \(2\delta\).
The QSVT pseudoinverse construction
\cite{GilyenSuLowWiebe2019QSVT} produces a block
approximating \((\delta/2)A^+\) to error \(\zeta\), using
\(O(\delta^{-1}\log(1/\zeta))\) oracle calls.
Since
\[
\frac{\delta}{2}A^+=\frac{\gamma}{4}T^+,
\]
the output normalization is \(4/\gamma\), and its decoded error
is at most \((4/\gamma)\zeta=\varepsilon_{\rm inv}\).
Substitution gives \eqref{eq:scaled-inverse-query}.

To check the rectangular orientation, write \(T=W\Sigma V^*\).
Applying the odd inverse polynomial to the singular-value transformation
of \(A^*\) gives the orientation \(V p(\Sigma/\alpha_T)W^*\),
which approximates \((\gamma/4)T^+\).
Odd parity gives \(p(0)=0\), so the ideal polynomial block vanishes
on the zero singular subspace. The promise excludes nonzero singular
values below \(\gamma\); no approximation on that excluded interval
is needed. For square invertible \(T\), \(T^+=T^{-1}\).
\end{proof}

The primitive above encodes the inverse of its actual matrix input.
When that input is a decoded approximation to the intended matrix, a
separate perturbation estimate is needed. For the square shifted matrices
used by the contour construction, the following bound suffices.

\begin{lemma}[Inverse perturbation bound]
\label{lem:inverse-input-perturbation}
Let \(T\in\mathbb C^{d\times d}\) be invertible, and let
\(\Delta T\in\mathbb C^{d\times d}\) satisfy
\(\|\Delta T\|\le\eta\), where \(\eta\ge0\) and
\(\eta\|T^{-1}\|<1\). Then \(T+\Delta T\) is invertible and
\begin{equation}
\begin{aligned}
\|(T+\Delta T)^{-1}\|
&\le\frac{\|T^{-1}\|}{1-\eta\|T^{-1}\|},\\
\|(T+\Delta T)^{-1}-T^{-1}\|
&\le\frac{\eta\|T^{-1}\|^2}{1-\eta\|T^{-1}\|}.
\end{aligned}
\label{eq:inverse-input-perturbation}
\end{equation}
\end{lemma}

\begin{proof}
Set \(E=T^{-1}\Delta T\). Since
\(\|E\|\le\eta\|T^{-1}\|<1\), the Neumann series~\cite[Chap.~I, Sec.~4.4]{Kato1995Perturbation} for
\((\mathsf I+E)^{-1}\) converges in norm and gives
\[
\|(\mathsf I+E)^{-1}\|
\le\frac{1}{1-\|E\|}
\le\frac{1}{1-\eta\|T^{-1}\|}.
\]
The factorization \(T+\Delta T=T(\mathsf I+E)\) proves invertibility
and the first bound. The resolvent identity
\[
(T+\Delta T)^{-1}-T^{-1}
=-(T+\Delta T)^{-1}\Delta T\,T^{-1}
\]
then yields the second bound by submultiplicativity of the norm.
\end{proof}

The smallness condition supplies both invertibility of the perturbed
shift and a bound on its inverse norm. We will apply it node by node
after establishing the exact contour construction.

\subsection{Spectral meaning of weighted Riesz operators}
\label{app:riesz-calculus}

Inversion provides access to resolvents, but the branch interpretation
comes from their contour integral. The next proposition justifies the
ordinary and pencil targets in
Definition~\ref{def:weighted-riesz-interface-main}. It allows arbitrary
Jordan structure and a singular pencil coefficient \(L\).

\begin{proposition}[Weighted spectral selection]
\label{prop:riesz-calculus-app}
Under the domain, contour, and holomorphy conditions of
Definition~\ref{def:weighted-riesz-interface-main}, the ordinary Riesz
operator acts on a Jordan block \(J_\lambda=\lambda\mathsf I+N_\lambda\)
of size \(s_\lambda\) as~\cite[Sec.~1.2.1]{Higham2008Functions}
\begin{equation}
\mathfrak R_\chi[g;J_\lambda,\mathsf I]
=\begin{cases}
\displaystyle\sum_{r=0}^{s_\lambda-1}
 \frac{g^{(r)}(\lambda)}{r!}N_\lambda^r,&\lambda\in K_\chi,\\[2mm]
0,&\lambda\notin K_\chi.
\end{cases}
\label{eq:riesz-jordan-app}
\end{equation}
In particular, \(g=1\) gives the spectral projector \(\Pi_\chi\).
If \(g\) extends holomorphically to a neighborhood of the full spectrum,
then \(\mathfrak R_\chi[g;M,\mathsf I]=g(M)\Pi_\chi\).

For a regular pencil, choose nonsingular matrices \(T_\ell,T_r\) in a
Weierstrass form \cite[Sec.~2, Eq.~(2.1)]{BennerStykel2014ProjectedRiccati}
\[
T_\ell MT_r=\operatorname{diag}(J_f,\mathsf I),\qquad
T_\ell LT_r=\operatorname{diag}(\mathsf I,N),
\]
where \(J_f\) contains the finite Jordan blocks and \(N\) is nilpotent.
Then
\begin{equation}
\mathfrak R_\chi[g;M,L]
=T_r\begin{bmatrix}
\mathfrak R_\chi[g;J_f,\mathsf I]&0\\0&0
\end{bmatrix}T_r^{-1}.
\label{eq:riesz-weierstrass-app}
\end{equation}
Thus \(g=1\) gives the right deflating projector onto the selected
finite modes. If \(L\) is invertible, the pencil operator equals
\(\mathfrak R_\chi[g;L^{-1}M,\mathsf I]\).
\end{proposition}

\begin{proof}
On a Jordan block,
\[
(z\mathsf I-J_\lambda)^{-1}
=\sum_{r=0}^{s_\lambda-1}\frac{N_\lambda^r}{(z-\lambda)^{r+1}}.
\]
Cauchy's differentiation formula~\cite[Chap.~I, Sec.~5.6]{Kato1995Perturbation} gives \eqref{eq:riesz-jordan-app}
on selected blocks; on every other block the integrand is holomorphic
inside the domain and its integral vanishes. Similarity then gives the
ordinary formula. For \(g=1\), its Jordan blocks are identities or
zeros, proving idempotence. For a full-spectrum extension of \(g\),
the same block formula is the holomorphic matrix function \(g(M)\)
restricted by \(\Pi_\chi\).

For the pencil, its Weierstrass form yields
\[
(zL-M)^{-1}L
=T_r\operatorname{diag}\bigl((z\mathsf I-J_f)^{-1},
                         (zN-\mathsf I)^{-1}N\bigr)T_r^{-1}.
\]
If \(N^\nu=0\), the second block is the polynomial
\[
(zN-\mathsf I)^{-1}N=-\sum_{r=0}^{\nu-1}z^rN^{r+1}.
\]
Its product with \(g\) is holomorphic throughout the selected domain,
so it integrates to zero. Applying the ordinary result to the finite
block proves \eqref{eq:riesz-weierstrass-app}. Finally, if \(L\)
is invertible, the factorization
\(zL-M=L(z\mathsf I-L^{-1}M)\) gives
\((zL-M)^{-1}L=(z\mathsf I-L^{-1}M)^{-1}\), proving the reduction.
\end{proof}

The similarity matrices in this argument need not be unitary, so the
resulting projectors need not be orthogonal. The argument requires
holomorphy only on the selected domain; it therefore also permits a
bounded annular domain for an exterior branch weighted by \(z^{-k}\).
Such a domain excludes zero and its inner boundary has clockwise
orientation. The integral identifies the selected spectral subspace;
its relation to a particular Riccati graph uses the assumptions of
the corresponding equation-specific construction.

\subsection{Proofs of the contour constructions}
\label{app:contour-proofs}

The preceding proposition identifies the desired integral. We now
implement its finite quadrature sum. The ordinary construction contains
the common argument: construct each inverse, choose the LCU amplitudes
from its declared normalization, and allocate the implementation error.
The pencil construction then changes the shifted matrix and right factor.

\newtheorem*{ordinarycontourrestatement}{Proposition~\ref{thm:local-contour-sum}}
\begin{ordinarycontourrestatement}
Let \(M,R\) have exact block-encodings with normalizations
\(\alpha_M,\alpha_R\), and let every node shift be invertible.
Assume the coherent access just described, including supplied uniform
constant-factor inverse-norm estimates. With \(\beta_j\) defined
in~\eqref{eq:ordinary-node-normalization}, set
\begin{equation*}
\alpha_\Gamma=\alpha_R\sum_{j=1}^{m}|\omega_j|\beta_j.
\end{equation*}
For \(0<\varepsilon_{\rm impl}<\alpha_\Gamma\), the construction returns
a block-encoding of \(S_\Gamma\) with normalization \(\alpha_\Gamma\),
decoded error at most \(\varepsilon_{\rm impl}\), and query cost
\begin{equation*}
O\!\left(
\mathcal R_\Gamma
\log\!\left(e+\frac{\mathcal R_\Gamma\alpha_\Gamma}
{\varepsilon_{\rm impl}}\right)
\right)
\end{equation*}
to the matrix oracle and its adjoint, together with one call to \(U_R\).
\end{ordinarycontourrestatement}

\begin{proof}
Discard nodes with zero weight and set
\(A_j=z_j\mathsf I-M\), \(s_j=|z_j|+\alpha_M\).
A two-term LCU of the identity and the matrix oracle encodes
\(A_j/s_j\) exactly using a constant number of matrix calls.
Choose the supplied inverse normalizations so that
\[
4\|A_j^{-1}\|\le\beta_j\le C\|A_j^{-1}\|,
\]
where \(C\ge4\) is fixed independently of the node.
Proposition~\ref{prop:local-inverse}, applied to \(A_j\) with input
normalization \(s_j\) and threshold \(4/\beta_j\), then gives inverse
normalization \(\beta_j\). These declared quantities are known upper
estimates; the construction does not require exact inverse norms.

Assign node errors
\begin{equation}
\varepsilon_j=\min\!\left\{\frac1{s_j},
                  \frac{\varepsilon_{\rm impl}\beta_j}{2\alpha_\Gamma}\right\}.
\label{eq:contour-node-budget-app}
\end{equation}
If the decoded inverse at node \(j\) is \(\widetilde A_j^{-1}\),
then \(\|\widetilde A_j^{-1}-A_j^{-1}\|\le\varepsilon_j\).
Prepare an index state with probabilities
\(p_j=\alpha_R|\omega_j|\beta_j/\alpha_\Gamma\).
The controlled inverse blocks, with phases \(e^{\mathrm i\arg\omega_j}\),
and subsequent unpreparation produce the signal block
\[
\sum_{j=1}^{m} p_j e^{\mathrm i\arg\omega_j}
        \frac{\widetilde A_j^{-1}}{\beta_j}
=\frac{\alpha_R}{\alpha_\Gamma}\sum_{j=1}^{m}\omega_j\widetilde A_j^{-1}.
\]
Multiplying this block by the encoding of \(R/\alpha_R\), with
separate signal ancillas, gives a block for
\(\sum_{j=1}^{m}\omega_j\widetilde A_j^{-1}R\) with normalization \(\alpha_\Gamma\).
This product uses \(U_R\) once.

The node approximation contributes at most
\[
\alpha_R\sum_{j=1}^{m}|\omega_j|\varepsilon_j
\le\frac{\varepsilon_{\rm impl}}{2\alpha_\Gamma}
       \alpha_R\sum_{j=1}^{m}|\omega_j|\beta_j
=\frac{\varepsilon_{\rm impl}}2.
\]
Choose coefficient preparation and coherent composition to contribute
at most \(\varepsilon_{\rm impl}/2\) in decoded norm. This proves
the error claim at the declared normalization.

It remains to bound the longest node circuit. The inverse query bound is
\(O(s_j\beta_j\log(e+\beta_j/\varepsilon_j))\), and
\begin{equation}
\frac{\beta_j}{\varepsilon_j}
=\max\!\left\{s_j\beta_j,
                     \frac{2\alpha_\Gamma}{\varepsilon_{\rm impl}}\right\}.
\label{eq:contour-log-budget-app}
\end{equation}
Uniform constant-factor estimates imply
\(s_j\beta_j=O(\mathcal R_\Gamma)\).
Also \(\mathcal R_\Gamma\ge1\), since
\(s_j\ge\|A_j\|\), and \(\alpha_\Gamma/\varepsilon_{\rm impl}>1\).
The assumed coherent access to node-dependent phases and inverse
schedules therefore implements all nodes with query depth
\[
O\!\left(\max_j s_j\beta_j
 \log(e+\beta_j/\varepsilon_j)\right)
=O\!\left(\mathcal R_\Gamma
 \log\!\left(e+\frac{\mathcal R_\Gamma\alpha_\Gamma}
                         {\varepsilon_{\rm impl}}\right)\right).
\]
This is \eqref{eq:ordinary-lcu-query}.
\end{proof}

The sum in the normalization and the maximum in the query count have
different origins: LCU adds the weighted inverse scales, while coherent
node processing shares the matrix calls across nodes. Padding alone does
not provide the required access to node-dependent phases and activation
schedules; that access is part of the proposition's assumptions.

\newtheorem*{pencilcontourrestatement}{Proposition~\ref{thm:affine-pencil-lcu-main}}
\begin{pencilcontourrestatement}
Let \(M,L,R\) have exact block-encodings with normalizations
\(\alpha_M,\alpha_L,\alpha_R\), and suppose every \(z_jL-M\) is invertible.
Assume controlled access to \(M,L\) and their adjoints, coherent coefficient
and inverse schedules, and supplied constant-factor inverse-norm estimates.
For declared node normalizations
\(\beta_j=\Theta(\|(z_jL-M)^{-1}\|)\), set
\begin{equation*}
\alpha_{\Gamma,L}=\alpha_L\alpha_R\sum_{j=1}^{m}|\omega_j|\beta_j.
\end{equation*}
For \(0<\varepsilon_{\rm impl}<\alpha_{\Gamma,L}\), an encoding of
\(S_{\Gamma,L}\) is obtained with normalization \(\alpha_{\Gamma,L}\),
decoded error at most \(\varepsilon_{\rm impl}\), and query cost
\begin{equation*}
O\!\left(
\mathcal R_{\Gamma,L}
\log\!\left(e+\frac{\mathcal R_{\Gamma,L}\alpha_{\Gamma,L}}
{\varepsilon_{\rm impl}}\right)
\right),
\end{equation*}
plus one call to \(U_R\).
\end{pencilcontourrestatement}

\begin{proof}
Apply the preceding inverse-and-LCU argument with
\(A_j=z_jL-M\), \(s_j=|z_j|\alpha_L+\alpha_M\), and right
factor \(LR\) encoded at normalization \(\alpha_L\alpha_R\).
The two-term shift encoding uses \(U_L,U_M\). The same supplied
inverse bounds and thresholds \(4/\beta_j\) give normalization
\(\beta_j\) at each node. In \eqref{eq:contour-node-budget-app}
and the LCU probabilities, replace \(\alpha_\Gamma\) by
\(\alpha_{\Gamma,L}\); also replace \(\alpha_R\) by
\(\alpha_L\alpha_R\) in those probabilities. The decoded target is then
\(\sum_{j=1}^{m}\omega_j\widetilde A_j^{-1}LR\), and the same error
allocation gives total implementation error at most
\(\varepsilon_{\rm impl}\).
Since \(s_j\beta_j=O(\mathcal R_{\Gamma,L})\), the maximum-depth
argument following \eqref{eq:contour-log-budget-app} proves
\eqref{eq:pencil-lcu-query}. The right product uses one additional
call to \(U_L\), absorbed in that bound, and one call to \(U_R\).
At no point is \(L\) inverted.
\end{proof}

The ordinary proposition addresses
\(0<\varepsilon_{\rm impl}<\alpha_\Gamma\).
If all weights vanish, an exact zero block represents the finite sum.
If \(\varepsilon_{\rm impl}\ge\alpha_\Gamma>0\), then
\(\|S_\Gamma\|\le\alpha_\Gamma\), so a zero block at normalization
\(\alpha_\Gamma\) suffices. The same conclusions hold for the pencil
sum with \(\alpha_{\Gamma,L}\) in place of \(\alpha_\Gamma\).
These shortcuts concern the finite sums; the quadrature error must
still be included when approximating the Riesz operator.

\subsection{Quadrature error}
\label{app:contour-quadrature}

The implementation bounds above take the nodes and weights as supplied
data. To connect them to the continuous Riesz target, one needs a
quadrature error bound. Analytic parametrizations give a useful sufficient
condition for logarithmic dependence of the node count on inverse
precision, as in the classical trapezoidal-rule analysis
\cite{TrefethenWeideman2014Trap}.

\begin{lemma}[Analytic periodic quadrature]
\label{lem:analytic-quadrature-app}
Let \(F\) be a matrix-valued, \(2\pi\)-periodic function, holomorphic
on a neighborhood of the closed strip \(|\operatorname{Im}t|\le a\),
where \(a>0\).
For an integer \(m\ge1\), set
\[
\mathcal G=\frac1{2\pi}\int_0^{2\pi}F(t)\,\mathrm dt,
\qquad S_m=\frac1m\sum_{j=0}^{m-1}F(2\pi j/m).
\]
Then
\begin{equation}
\|S_m-\mathcal G\|\le\frac{2\sup_{|\operatorname{Im}t|\le a}\|F(t)\|}{e^{am}-1}.
\label{eq:analytic-quadrature-error-app}
\end{equation}
In particular, error at most \(\varepsilon_{\rm quad}>0\) is obtained by
\begin{equation}
m\ge\max\!\left\{1,\left\lceil\frac1a
       \log\!\left(1+\frac{2\sup_{|\operatorname{Im}t|\le a}\|F(t)\|}{\varepsilon_{\rm quad}}\right)
                   \right\rceil\right\}.
\label{eq:analytic-quadrature-count-app}
\end{equation}
\end{lemma}

\begin{proof}
Let \(F_k=(2\pi)^{-1}\int_0^{2\pi}F(t)e^{-\mathrm ikt}\,\mathrm dt\).
For \(k>0\), shift the integration path to \(\operatorname{Im}t=-a\);
for \(k<0\), shift it to \(\operatorname{Im}t=a\).
Holomorphy and periodicity cancel the vertical sides and give
\(\|F_k\|\le \sup_{|\operatorname{Im}t|\le a}\|F(t)\|e^{-a|k|}\).
The Fourier series is therefore absolutely convergent on the real axis.
Averaging it at the \(m\) equispaced nodes leaves precisely the terms with
indices divisible by \(m\), so
\[
\begin{aligned}
S_m-\mathcal G&=\sum_{\ell\ne0}F_{\ell m},\\
\|S_m-\mathcal G\|
&\le2\left(\sup_{|\operatorname{Im}t|\le a}\|F(t)\|\right)
\sum_{\ell=1}^{\infty}e^{-a\ell m}
=\frac{2\sup_{|\operatorname{Im}t|\le a}\|F(t)\|}{e^{am}-1}.
\end{aligned}
\]
Solving this inequality for \(m\) proves the stated node count.
A supplied upper bound on the strip supremum gives a conservative
choice of \(m\).
\end{proof}

For a parametrized ordinary contour \(z(t)\), apply the lemma to
\begin{equation}
F(t)=\frac{g(z(t))z'(t)}{\mathrm i}
            (z(t)\mathsf I-M)^{-1}R.
\label{eq:quadrature-integrand-app}
\end{equation}
Then \(S_m\) has the form \eqref{eq:weighted-contour-coefficients}
with \(\nu_j=2\pi z'(2\pi j/m)/m\).
For a pencil, replace the resolvent-right-factor product by
\((z(t)L-M)^{-1}LR\). With several boundary components, apply the
lemma to each oriented parametrization and add their errors.

The assumption concerns the full integrand: its analytic strip must
avoid resolvent poles and singularities of the scalar weight.
The quantities \(a\) and \(\sup_{|\operatorname{Im}t|\le a}\|F(t)\|\) can depend on spectral
geometry, nonnormality, time, or the recursion index. Logarithmic
dependence on inverse precision holds with those data fixed.
For a general piecewise smooth contour, the implementation propositions
still apply to any supplied quadrature rule, but its error requires a
separate bound; piecewise smoothness alone does not imply
\eqref{eq:analytic-quadrature-error-app}.

\subsection{Approximate inputs and total error}
\label{app:contour-input-errors}

Quadrature compares an exact finite sum with the desired integral.
An approximate input oracle changes that finite sum before the inverse
circuits are applied. The next proposition bounds this change at the
quadrature nodes and then adds the implementation error. It does not
require the perturbed matrices to define the same spectral branch on
the entire contour.

\begin{proposition}[Contour construction with perturbed inputs]
\label{prop:contour-input-errors-app}
Let \(\widetilde M,\widetilde R\) be the decoded matrices of input
block-encodings with normalizations \(\alpha_M,\alpha_R\), and suppose
\[
\|\widetilde M-M\|\le\varepsilon_M,\qquad
\|\widetilde R-R\|\le\varepsilon_R.
\]
For the ordinary construction, set
\[
\begin{aligned}
A_j&=z_j\mathsf I-M,&
\widetilde A_j&=z_j\mathsf I-\widetilde M,&
\eta_j&=\varepsilon_M,\\
S&=\sum_{j=1}^{m}\omega_j A_j^{-1}R,&
\widetilde S&=\sum_{j=1}^{m}\omega_j\widetilde A_j^{-1}\widetilde R.
\end{aligned}
\]
Assume every \(A_j\) is invertible and
\(\eta_j\|A_j^{-1}\|<1\).
If \(\|S-\mathcal G_\chi\|\le\varepsilon_{\rm quad}\) and a decoded
output satisfies
\(\|S_{\rm out}-\widetilde S\|\le\varepsilon_{\rm impl}\), then
\begin{equation}
\|S_{\rm out}-\mathcal G_\chi\|
\le\varepsilon_{\rm quad}+\varepsilon_{\rm input}
+\varepsilon_{\rm impl},
\label{eq:contour-total-error-app}
\end{equation}
where the ordinary input error is bounded by
\[
\varepsilon_{\rm input}
=\sum_{j=1}^{m}|\omega_j|
\left[
\frac{\alpha_R\varepsilon_M\|A_j^{-1}\|^2}
{1-\varepsilon_M\|A_j^{-1}\|}
+\|A_j^{-1}\|\varepsilon_R
\right].
\]

For the pencil construction, let \(\widetilde L\) be the decoded
matrix of an input encoding with normalization \(\alpha_L\), with
\(\|\widetilde L-L\|\le\varepsilon_L\).
Use
\[
\begin{aligned}
A_j&=z_jL-M,&
\widetilde A_j&=z_j\widetilde L-\widetilde M,\\
S&=\sum_{j=1}^{m}\omega_jA_j^{-1}LR,&
\widetilde S&=\sum_{j=1}^{m}\omega_j\widetilde A_j^{-1}
                          \widetilde L\widetilde R.
\end{aligned}
\]
Under the same invertibility, smallness, quadrature, and implementation
conditions, \eqref{eq:contour-total-error-app} holds with
\begin{equation}
\begin{aligned}
\eta_j&=|z_j|\varepsilon_L+\varepsilon_M,\\
\varepsilon_{\rm input}
&=\sum_{j=1}^{m}|\omega_j|
\Biggl[
\frac{\alpha_L\alpha_R\eta_j\|A_j^{-1}\|^2}
{1-\eta_j\|A_j^{-1}\|}\\
&\hspace{29mm}
+\|A_j^{-1}\|
 \bigl(\alpha_L\varepsilon_R+\alpha_R\varepsilon_L
                         +\varepsilon_L\varepsilon_R\bigr)
\Biggr].
\end{aligned}
\label{eq:pencil-input-data-app}
\end{equation}

For any of these constructions, suppose the supplied node
normalizations satisfy
\[
4\|A_j^{-1}\|\le\beta_j\le C\|A_j^{-1}\|,
\]
with a fixed \(C\ge4\). The sufficient condition
\(\eta_j\beta_j\le2\) guarantees that inverse normalization
\(2\beta_j\) is admissible for \(\widetilde A_j\) and remains
within a constant factor of its inverse norm. Thus the normalization
and matrix-query estimates for exact inputs change by at most
constant factors, with \(\varepsilon_{\rm input}\) counted separately.
\end{proposition}

\begin{proof}
For ordinary shifts,
\(\|\widetilde A_j-A_j\|\le\varepsilon_M=\eta_j\).
For pencil shifts, the triangle inequality gives
\[
\|\widetilde A_j-A_j\|
=\|z_j(\widetilde L-L)-(\widetilde M-M)\|
\le|z_j|\varepsilon_L+\varepsilon_M=\eta_j.
\]
Lemma~\ref{lem:inverse-input-perturbation} therefore gives invertibility
of every \(\widetilde A_j\) and
\[
\|\widetilde A_j^{-1}-A_j^{-1}\|
\le
\frac{\eta_j\|A_j^{-1}\|^2}{1-\eta_j\|A_j^{-1}\|}.
\]
For the ordinary right factor, write
\[
\widetilde A_j^{-1}\widetilde R-A_j^{-1}R
=(\widetilde A_j^{-1}-A_j^{-1})\widetilde R
 +A_j^{-1}(\widetilde R-R).
\]
Since a decoded unitary block satisfies
\(\|\widetilde R\|\le\alpha_R\), taking norms and summing with
\(|\omega_j|\) gives the stated ordinary input error.

For the pencil right factor, the analogous identity is
\[
\begin{aligned}
\widetilde A_j^{-1}\widetilde L\widetilde R-A_j^{-1}LR
={}&(\widetilde A_j^{-1}-A_j^{-1})\widetilde L\widetilde R\\
&+A_j^{-1}(\widetilde L\widetilde R-LR).
\end{aligned}
\]
Decoded unitary blocks give
\(\|\widetilde L\widetilde R\|\le\alpha_L\alpha_R\).
With \(\Delta L=\widetilde L-L\) and \(\Delta R=\widetilde R-R\),
\[
\widetilde L\widetilde R-LR
=\widetilde L\Delta R+\Delta L\widetilde R-\Delta L\Delta R,
\]
so
\[
\|\widetilde L\widetilde R-LR\|
\le\alpha_L\varepsilon_R+\alpha_R\varepsilon_L
                         +\varepsilon_L\varepsilon_R.
\]
This proves \eqref{eq:pencil-input-data-app}, without assuming
that the true input matrices are bounded by the decoded normalizations.

In each case, \(\|\widetilde S-S\|\le\varepsilon_{\rm input}\);
adding the quadrature and implementation errors proves
\eqref{eq:contour-total-error-app}.

Finally, \(\eta_j\beta_j\le2\) implies
\(\eta_j\|A_j^{-1}\|\le1/2\), and therefore
\[
\|\widetilde A_j^{-1}\|
\le2\|A_j^{-1}\|\le\frac{\beta_j}{2}.
\]
The inverse primitive with threshold \(2/\beta_j\) has output
normalization \(2\beta_j\). The reverse resolvent identity also gives
\[
\begin{aligned}
\|A_j^{-1}\|
&\le
\bigl(1+\eta_j\|A_j^{-1}\|\bigr)
\|\widetilde A_j^{-1}\|\\
&\le\tfrac32\|\widetilde A_j^{-1}\|.
\end{aligned}
\]
Consequently \(2\beta_j\le3C\|\widetilde A_j^{-1}\|\), proving
the uniform constant-factor claim. Applying the exact-input
construction to the decoded matrices then preserves the stated
asymptotic normalization and query bounds.
\end{proof}

The smallness conditions and input precision can be chosen
conservatively using the supplied inverse-norm upper estimates.
For a prescribed branch tolerance, allocate a third each to
quadrature, input perturbations, and implementation, then use
\eqref{eq:contour-node-budget-app} within the last budget, or its
pencil counterpart with \(\alpha_{\Gamma,L}\).
Coefficient preparation and rotation errors are measured on the
normalized circuit: if its unitary factors have errors \(\delta_\ell\),
telescoping bounds the total by \(\sum_\ell\delta_\ell\).
Multiplication by the declared output normalization converts this to
decoded error. A fixed oracle's decoded matrix error is already counted
in \(\varepsilon_{\rm input}\); errors in physically executing its
individual calls belong to the circuit budget.
This completes the approximation of the weighted Riesz block. Any
subsequent graph recovery propagates this branch error through the
equation-specific inverse or projector-block pseudoinverse construction.

\clearpage
\section{Complexity factors, spectral separation, and coefficient bounds}
\label{app:complexity-factors}

The query bounds include the cost of constructing a graph projector,
the normalization of the recovered Riccati matrix, and, for a dynamic
problem, the conditioning of its initialization. We first explain these
quantities, then relate spectral separation to contours and ordinary
condition numbers. Coefficient bounds give sufficient conditions for
estimating them before solving the Riccati equation.

\subsection{How the factors enter the query complexity}
\label{app:factor-geometry}

In the contour construction, each node requires an inverse of a shifted
matrix. The generalized singularity factor measures the largest inverse
scale that these calls resolve. Both algebraic and dynamic recovery then
multiply the lower block row of the complete graph projector by the
pseudoinverse of its upper block row. For a graph with slope \(X\),
Lemma~\ref{lem:dre-projector-block-recovery-app} gives
\[
X=(E_2^*\Pi)(E_1^*\Pi)^+,
\qquad \|(E_1^*\Pi)^+\|\le\sqrt{1+\|X\|^2}.
\]
The supplied solution-norm bound therefore enters the actual output
normalization and the required projector accuracy. A dynamic projector
also depends on the pseudoinverse of its initial branch column.

Under the calibration and access assumptions of
\eqref{eq:informal-query-upper}, the principal factors are
\begin{equation}
\widetilde O\!\left(r_\Pi\mathcal R\alpha_X\right),\qquad
\widetilde O\!\left(\mathcal R\,\alpha_P
\frac{\|\Pi_{\rm sel}R_0\|}{\sigma_{\min}(\Pi_{\rm sel}R_0)}\right).
\label{eq:factor-cost-interpretation-app}
\end{equation}
The first expression describes algebraic recovery and includes the
projector re-encoding overhead \(r_\Pi\) defined in
Appendix~\ref{app:direct-graph-recovery}. The compact main bound assumes
\(r_\Pi=O(1)\). The second describes the dynamic construction, with
\(\alpha_P=\alpha_{P(t)}\) for DRE and \(\alpha_P=\alpha_{P_k}\) for RR.
It still depends on the stated initial-column calibration, normalization changes,
inverse-error conditions, and supplied solution-norm bound. The expanded
algebraic precision dependence appears in
\eqref{eq:care-projector-query-expanded-app}
and~\eqref{eq:dare-query-expanded-app}.

An LCU normalization records the magnitudes of its summands before
cancellation and can exceed the norm of their sum. For the specific
positive-weight rectangle and unit-circle rules in
Lemmas~\ref{lem:care-normalization-factor-app}
and~\ref{lem:dare-normalization-factor-app}, the raw full-projector scales
satisfy
\[
\alpha_{\Pi_-}^{\rm raw}=\alpha_{\rm CARE}
=O(\mathcal R_{\rm CARE}),\qquad
\alpha_{\Pi_<}^{\rm raw}=\alpha_L\alpha_{\rm DARE}
=O(\mathcal R_{\rm DARE}),
\]
uniformly in dimension and quadrature accuracy. The pencil scale includes
the right factor \(L\). These estimates concern the specified quadrature
rules; the general query bounds include the actual supplied scale and the
cost of any normalization reduction. For direct LCU recovery, a supplied
bound \(\|X\|\le u\) gives the output scale
\(\alpha_X=8\alpha_\Pi\sqrt{1+u^2}\), as in
\eqref{eq:algebraic-projector-output-app}. Conservatism in either supplied
bound remains in the implementation cost. The corresponding full-projector
errors and node counts are given in
Propositions~\ref{prop:care-rectangle-quadrature-app}
and~\ref{prop:dare-circle-quadrature-app}.

We use the name generalized singularity factor both for an unnormalized
inverse norm and for the quantity that includes its shifted-input scale.
For an ordinary shift these are
\(\|(z\mathsf I-M)^{-1}\|\) and
\((|z|+\alpha_M)\|(z\mathsf I-M)^{-1}\|\), respectively;
the corresponding pencil scale is \(|z|\alpha_L+\alpha_M\).
The formula specifies which quantity is meant. Their geometric meaning
follows from the standard distance characterization of the smallest
singular value \cite{TrefethenEmbree2005Spectra}.

\begin{lemma}[Distance to rank deficiency]
\label{lem:factor-rank-distance-app}
For \(T\in\mathbb C^{m\times n}\), \(m\ge n\),
\begin{equation}
\sigma_{\min}(T)=\inf_{\operatorname{rank}(T+\Delta)<n}\|\Delta\|.
\label{eq:factor-rank-distance-app}
\end{equation}
Consequently, the reciprocal of the ordinary contour factor is
\begin{equation}
\mathcal R_\Gamma^{-1}
=\inf_{z\in\Gamma}
\sigma_{\min}\!\left(\frac{z\mathsf I-M}{|z|+\alpha_M}\right).
\label{eq:factor-normalized-distance-app}
\end{equation}
For a pencil, replace the normalized shift by
\((zL-M)/(|z|\alpha_L+\alpha_M)\).
\end{lemma}
\begin{proof}
A perturbation of norm less than \(\sigma_{\min}(T)\) preserves
full column rank by the singular-value perturbation inequality.
For unit singular vectors with \(Tv=\sigma_{\min}(T)u\), the
perturbation \(-\sigma_{\min}(T)uv^*\) attains rank deficiency.
If \(T\) is already rank deficient, use the zero perturbation.
Apply this identity to each normalized shift and use
\(\|T^{-1}\|=1/\sigma_{\min}(T)\) for an invertible square matrix.
\end{proof}

Thus a large generalized singularity factor means that some shift is
close to singularity at the stated scale. This is a statement about
the shifted matrices; it is not a condition number for structured
perturbations of the original Riccati data. That interpretation would
require a specified perturbation model for the coefficients.

The dynamic construction additionally depends on its initial branch
column. Its absolute and relative nondegeneracy factors play different
roles in the inverse implementation.

\begin{definition}[Initialization nondegeneracy factors]
\label{def:factor-nondegeneracy-app}
For a selected dynamic projector, the absolute initialization
nondegeneracy factor is \(\sigma_{\min}(\Pi_{\rm sel}R_0)\), where
\(R_0=[\mathsf I;P_0]\).
When this column has full rank, its relative initialization
nondegeneracy factor is
\begin{equation}
\frac{\sigma_{\min}(\Pi_{\rm sel}R_0)}{\|\Pi_{\rm sel}R_0\|}
=\kappa(\Pi_{\rm sel}R_0)^{-1}.
\label{eq:factor-relative-initialization-app}
\end{equation}
The selected projectors for DRE and RR are respectively \(\Pi_+\)
and \(\Pi_>\), with the branch conventions in
Table~\ref{tab:weighted-riesz-routes}.
\end{definition}

An absolute lower bound on \(\sigma_{\min}(\Pi_{\rm sel}R_0)\) does not remove the
column norm from its condition number. The initialization pseudoinverse
depends on its actual encoding normalization divided by
\(\sigma_{\min}(\Pi_{\rm sel}R_0)\); this ratio equals the condition number only under
the corresponding norm calibration.
At rank deficiency the Moore--Penrose inverse can still have finite norm;
these reciprocal identities require full column rank.

The full-projector recovery bound above depends on the graph slope, while
the actual projector normalization can also reflect the angle between its
range and kernel.
For a nontrivial projector,
\(\|\Pi\|=1/\sin\theta_{\min}(\operatorname{ran}\Pi,\ker\Pi)\)
\cite{SimonciniSzyld2010Oblique}. This exact norm is a lower bound on an
exact encoding scale; it does not give a free normalization reduction.
The following geometry translates coefficient bounds on the discarded
graph into initialization bounds.

\begin{lemma}[Initial-column geometry]
\label{lem:factor-graph-distance-app}
Let \(\Pi\) be a rank-\(n\) projector on \(\mathbb C^{2n}\)
and let \(R\in\mathbb C^{2n\times n}\) have full column rank.
Then
\begin{equation}
\begin{aligned}
\sin\theta_{\min}(\operatorname{ran}R,\ker\Pi)
&\le\sigma_{\min}\!\left(\Pi R(R^*R)^{-1/2}\right),\\
\sigma_{\min}\!\left(\Pi R(R^*R)^{-1/2}\right)
&\le\|\Pi\|\sin\theta_{\min}(\operatorname{ran}R,\ker\Pi).
\end{aligned}
\label{eq:factor-initial-angle-app}
\end{equation}
The column \(\Pi R\) has full rank exactly when
\(\operatorname{ran}R\cap\ker\Pi=\{0\}\).
For Hermitian \(X,Y\) and \(R_0=[\mathsf I;P_0]\),
\begin{align}
\operatorname{dist}(R_0u,\operatorname{ran}[\mathsf I;X])
&=\|(\mathsf I+X^2)^{-1/2}(P_0-X)u\|,\notag\\
\operatorname{dist}(R_0u,\operatorname{ran}[-Y;\mathsf I])
&=\|(\mathsf I+Y^2)^{-1/2}(\mathsf I+YP_0)u\|.
\label{eq:factor-graph-distances-app}
\end{align}
\end{lemma}
\begin{proof}
Let \(P_{(\ker\Pi)^\perp}\) denote the orthogonal projector onto
\((\ker\Pi)^\perp\). Since \(x-\Pi x\in\ker\Pi\) and
\(\Pi P_{(\ker\Pi)^\perp}=\Pi\),
\[
\|P_{(\ker\Pi)^\perp}x\|\le\|\Pi x\|
\le\|\Pi\|\|P_{(\ker\Pi)^\perp}x\|.
\]
Minimizing on the unit sphere of \(\operatorname{ran}R\) proves the
angle bounds, since \(R(R^*R)^{-1/2}\) is an isometry onto that range.
The kernel gives the rank criterion. Orthonormal frames
for the two graph complements are
\([-X;\mathsf I](\mathsf I+X^2)^{-1/2}\) and
\([\mathsf I;Y](\mathsf I+Y^2)^{-1/2}\).
Multiplication by their adjoints proves the distance identities.
\end{proof}

The spectral factor can also bound the norm of the projector used in
initialization and recovery. For an ordinary Riesz projector,
\begin{equation}
\|\Pi\|\le\frac{\mathcal R_\Gamma}{2\pi}
\int_\Gamma\frac{|\mathrm dz|}{|z|+\alpha_M}.
\label{eq:factor-projector-integral-app}
\end{equation}
For a pencil the analogous bound has numerator
\(\mathcal R_{\Gamma,L}\|L\|\) and denominator
\(|z|\alpha_L+\alpha_M\) inside the integral.
These follow directly by taking norms in the Riesz integral.
If \(\|P_0\|\le p\), then
\(\kappa(\Pi_{\rm sel}R_0)
\le\|\Pi_{\rm sel}\|\sqrt{1+p^2}/\sigma_{\min}(\Pi_{\rm sel}R_0)\).
Using \eqref{eq:factor-projector-integral-app} in this estimate can
introduce another generalized singularity factor; it does not justify
discarding the column norm from \eqref{eq:factor-cost-interpretation-app}.

A supplied positive lower bound on an initialization singular value can
be used conservatively in its inverse threshold and error budget. Its
reciprocal remains in the cost unless it approximates the actual value
within a constant factor. The threshold also divides the supplied bound
by the encoding normalization, as in
Proposition~\ref{prop:local-inverse}. For algebraic recovery, conservatism
in a supplied solution-norm bound instead enters
\eqref{eq:algebraic-projector-output-app} and its precision budget.
With uniform bounds on the spectral factors, actual normalizations,
supplied solution norms, and re-encoding overheads, the two algebraic
theorems give \(O(\log^2(e+1/\varepsilon))\) matrix queries.
Dimension can enter these parameters and each oracle's implementation
cost; it is not an additional multiplier in the count of supplied-oracle
calls.

\subsection{Spectral separation, contours, and condition numbers}
\label{app:factor-spectral}

The usual condition number of an invertible \(H\) examines its singular
values at the zero shift. Selecting a half-plane branch instead requires
separation along the imaginary axis. The following bounds turn that
separation into finite contours and keep the supplied encoding scale
explicit.

\begin{theorem}[Half-plane factors on finite rectangles]
\label{thm:factor-axis-app}
Let \(H\) have no imaginary-axis spectrum and let
\(\alpha_H\ge\|H\|\). Set
\begin{equation}
\Delta_{\rm ax}=\inf_{\omega\in\mathbb R}
\sigma_{\min}(\mathrm i\omega\mathsf I-H)>0.
\label{eq:factor-axis-gaps-app}
\end{equation}
For \(0<\tau<1\), the positively oriented boundaries of
\begin{equation}
\{x+\mathrm iy:\tau\Delta_{\rm ax}\le x\le2\alpha_H,
\ |y|\le2\alpha_H\},\qquad
\{x+\mathrm iy:-2\alpha_H\le x\le-\tau\Delta_{\rm ax},
\ |y|\le2\alpha_H\}
\label{eq:factor-finite-rectangles-app}
\end{equation}
enclose exactly the corresponding half-plane branches.
On both contours the inverse norm is at most
\(1/((1-\tau)\Delta_{\rm ax})\), and their maximum generalized
singularity factor satisfies
\begin{equation}
\frac{\alpha_H}{(1+\tau)\Delta_{\rm ax}}
\le\mathcal R_{H,\tau}\le
\frac{1+2\sqrt2}{1-\tau}\frac{\alpha_H}{\Delta_{\rm ax}}.
\label{eq:factor-rectangle-bound-app}
\end{equation}
Here \(\mathcal R_{H,\tau}\) is the supremum of
\((|z|+\alpha_H)\|(z\mathsf I-H)^{-1}\|\) over both contours.
\end{theorem}
\begin{proof}
The smallest singular value is continuous and is bounded below by
\(|\omega|-\|H\|\) on the imaginary axis. Its infimum is attained
and positive, and evaluation at zero gives
\(\Delta_{\rm ax}\le\sigma_{\min}(H)\le\alpha_H\).
Every eigenvalue \(\lambda\) has modulus at most \(\alpha_H\) and
\(|\operatorname{Re}\lambda|\ge\Delta_{\rm ax}\), by testing a
unit eigenvector at the shift \(\mathrm i\operatorname{Im}\lambda\).
Thus the rectangles enclose exactly their branches.
On the inner edges, singular-value perturbation from the imaginary axis
gives a gap of at least \((1-\tau)\Delta_{\rm ax}\).
On the remaining edges, \(|z|\ge2\alpha_H\), so the gap is at
least \(\alpha_H\). These estimates give the inverse bound and,
since \(|z|+\alpha_H\le(1+2\sqrt2)\alpha_H\), the upper bound
in \eqref{eq:factor-rectangle-bound-app}.

An imaginary-axis minimizer \(\omega_*\) satisfies
\(|\omega_*|\le\|H\|+\Delta_{\rm ax}\le2\alpha_H\).
The points \(\pm\tau\Delta_{\rm ax}+\mathrm i\omega_*\)
therefore lie on the inner edges, where singular-value perturbation
bounds the gap above by \((1+\tau)\Delta_{\rm ax}\).
The shifted encoding scale is at least \(\alpha_H\), proving
the lower bound.
\end{proof}

A supplied positive lower bound on \(\Delta_{\rm ax}\) may replace
it in the contour locations and upper bounds. The matching lower bound
uses the true gap. The rectangles provide a decay margin for time weights;
their corners require a piecewise quadrature rule rather than the global
analytic parametrization of Lemma~\ref{lem:analytic-quadrature-app}.

\begin{proposition}[When a condition number controls the rectangle factor]
\label{prop:factor-condition-number-app}
In the setting of Theorem~\ref{thm:factor-axis-app}, fix \(0<\tau<1\).
Then
\begin{equation}
\mathcal R_{H,\tau}
=\Theta\!\left(
\frac{\alpha_H}{\|H\|}
\frac{\sigma_{\min}(H)}{\Delta_{\rm ax}}\,\kappa(H)\right),
\label{eq:factor-condition-number-app}
\end{equation}
where the constants depend only on \(\tau\).
Over a family, \(\mathcal R_{H,\tau}=O(\kappa(H))\) if and only if
\(\alpha_H=O(\|H\|)\) and
\(\Delta_{\rm ax}\ge c\sigma_{\min}(H)\) for a uniform \(c>0\).
The gap comparison follows from either of the following sufficient conditions:
\begin{enumerate}
\item Suppose \(\sigma_{\min}((H+H^*)/2)>0\) and
\[
\left\|\frac{H-H^*}{2}\right\|
\le\rho\,\sigma_{\min}\!\left(\frac{H+H^*}{2}\right),
\qquad 0\le\rho<1,
\]
where \(\rho\) is uniformly bounded away from one. Then
\(\Delta_{\rm ax}\ge(1-\rho)\sigma_{\min}(H)/(1+\rho)\).
\item If \(H=V\Lambda V^{-1}\) and
\(|\operatorname{Re}\lambda|\ge\beta|\lambda|\) for every eigenvalue,
then
\[
\Delta_{\rm ax}\ge
\frac{\beta}{\|V\|\|V^{-1}\|}\,\sigma_{\min}(H).
\]
Uniform comparison requires \(\|V\|\|V^{-1}\|/\beta\) to be bounded.
\end{enumerate}
In particular, an invertible Hermitian \(H\) has
\(\Delta_{\rm ax}=\sigma_{\min}(H)\).
\end{proposition}
\begin{proof}
Substitute \(\kappa(H)=\|H\|/\sigma_{\min}(H)\) in
\eqref{eq:factor-rectangle-bound-app}.
Both ratios in \eqref{eq:factor-condition-number-app} are at least
one, so their product is uniformly bounded exactly when each is.
For the first sufficient condition, the spectral theorem applied to
\((H+H^*)/2\) and perturbation by \((H-H^*)/2\) give
\[
\begin{aligned}
\Delta_{\rm ax}
&\ge(1-\rho)\sigma_{\min}\!\left(\frac{H+H^*}{2}\right),\\
\sigma_{\min}(H)
&\le(1+\rho)\sigma_{\min}\!\left(\frac{H+H^*}{2}\right).
\end{aligned}
\]
For the second, diagonalization bounds the imaginary-axis inverse norm
by \(\|V\|\|V^{-1}\|/\min_\lambda|\operatorname{Re}\lambda|\), and
\(\min_\lambda|\lambda|\ge\sigma_{\min}(H)\).
For Hermitian \(H\), equality in the gap occurs at zero.
\end{proof}

Without a tight encoding scale the comparison still includes
\(\alpha_H/\sigma_{\min}(H)\), which need not be the usual condition
number. The proposition concerns the stated rectangles: an arbitrary
contour may approach an eigenvalue and have a much larger factor.
A conservative supplied gap likewise remains in the upper bound.

For an ordinary discrete lift the separating set is the unit circle.
The exterior branch requires a contour system that excludes the
interior spectrum as well as the origin, where the inverse power weight
is singular.

\begin{theorem}[Unit-circle separation and annular contours]
\label{thm:factor-circle-app}
Let \(S\) have no unit-circle spectrum and have nonempty interior
and exterior branches. Let \(\alpha_S\ge\|S\|\), and set
\begin{equation}
\eta_{\mathbb T}=\min_{|z|=1}\sigma_{\min}(z\mathsf I-S).
\label{eq:factor-circle-gap-app}
\end{equation}
Then \(0<\eta_{\mathbb T}\le1\), \(\alpha_S>1\), and for
\(0<\tau<1\) the positively oriented contours
\begin{equation}
|z|=1-\tau\eta_{\mathbb T},\qquad
\partial\{1+\tau\eta_{\mathbb T}<|z|<3\alpha_S\}
\label{eq:factor-annular-contours-app}
\end{equation}
enclose exactly the interior and exterior branches, respectively.
The annular inner boundary is oriented clockwise, so the exterior
system has winding number zero about the closed unit disk.
On the two circles adjacent to the unit circle, the smallest singular
value is at least \((1-\tau)\eta_{\mathbb T}\).
The maximum generalized singularity factor over both contour systems
is at most
\begin{equation}
\frac{1+\alpha_S+\tau\eta_{\mathbb T}}
{(1-\tau)\eta_{\mathbb T}}.
\label{eq:factor-annular-bound-app}
\end{equation}
\end{theorem}
\begin{proof}
Compactness gives a positive gap. A unit eigenvector for an interior
eigenvalue, tested at a nearest point on the unit circle, gives
\(\eta_{\mathbb T}\le1\); an exterior eigenvalue gives \(\alpha_S>1\).
The same argument gives
\(\bigl||\lambda|-1\bigr|\ge\eta_{\mathbb T}\) for every eigenvalue,
so the proposed contours enclose exactly their branches.
Singular-value perturbation from the unit circle proves the stated
gap on its two adjacent circles, where the shifted scale is at most
\(1+\alpha_S+\tau\eta_{\mathbb T}\).
On \(|z|=3\alpha_S\), the weighted inverse norm is at most
\(4\alpha_S/(3\alpha_S-\|S\|)\le2\), which is bounded by
\eqref{eq:factor-annular-bound-app}.
\end{proof}

For \(S=\mathcal S_F\), these contours give
\(|z^k|\le(1-\tau\eta_{\mathbb T})^k\) on the interior circle and
\(|z^{-k}|\le(1+\tau\eta_{\mathbb T})^{-k}\) on the exterior system.
Scaling the matrix moves the separating circle by the same scale.
The circle gap is separate from the singular value at zero, so the
usual \(\kappa(S)\) alone does not supply it.
For the DARE pencil, a fixed unit-circle contour gives directly
\begin{equation}
\mathcal R_{\mathbb T,L}
=\frac{\alpha_M+\alpha_L}
{\min_{|z|=1}\sigma_{\min}(zL-M)}.
\label{eq:factor-pencil-circle-app}
\end{equation}
This identity holds whenever the pencil is regular with no spectrum
on that circle, including singular \(L\).

\begin{example}[Rotational LQR]
\label{ex:care-rotation-dimension-app}
For even \(n\ge2\), take
\begin{equation}
J=\begin{bmatrix}0&1\\-1&0\end{bmatrix},\qquad
A_n=\operatorname{diag}(J,\ldots,J),\qquad
B_n=Q_n=R_{c,n}=\mathsf I_n.
\label{eq:care-rotation-matrix-app}
\end{equation}
Then \(G_n=\mathsf I_n\), and the stabilizing CARE solution and
feedback are
\begin{equation}
X_n=\mathsf I_n,\qquad u^*(t)=-x(t),\qquad
A_n-G_nX_n=A_n-\mathsf I_n.
\label{eq:care-rotation-solution-app}
\end{equation}
The Hamiltonian \(\mathcal H_n\) has
\(\|\mathcal H_n\|=\sqrt2\), \(\kappa(\mathcal H_n)=1\), and
\(\Delta_{\rm ax}=1\). Use the left rectangle of
\eqref{eq:factor-finite-rectangles-app} with
\(\alpha_H=\sqrt2\) and \(\tau=1/2\). This same contour in every
dimension gives
\begin{equation}
\mathcal R_{\rm CARE}=\Theta(1),\qquad n=2,4,\ldots.
\label{eq:care-rotation-resolvent-app}
\end{equation}
The complete stable projector satisfies
\begin{equation}
\Pi_-=\frac12\begin{bmatrix}\mathsf I_n&\mathsf I_n\\
\mathsf I_n&\mathsf I_n\end{bmatrix},\qquad
\|\Pi_-\|=1,\qquad
\|(E_1^*\Pi_-)^+\|=\sqrt2.
\label{eq:care-rotation-recovery-app}
\end{equation}
The finite-rectangle rule of
Proposition~\ref{prop:care-rectangle-quadrature-app}, with supplied
gap one, has \(\alpha_{\Pi_-}=\Theta(1)\) uniformly in its accuracy
and in \(n\). With \(\|X_n\|=1\), direct recovery therefore attains
\(\alpha_X=8\sqrt2\alpha_{\Pi_-}=\Theta(1)\) without normalization
reduction.

If instead \(n=2\) and \(A=TJ\), \(T\ge1\), the corresponding
Hamiltonian \(\mathcal H_T\), at scale
\(\alpha_H=\|\mathcal H_T\|=\sqrt{1+T^2}\), satisfies
\begin{equation}
\kappa(\mathcal H_T)=1,\qquad \Delta_{\rm ax}=1,\qquad
\mathcal R_{\rm CARE}=\Theta(\sqrt{1+T^2})
\label{eq:factor-rotation-frequency-app}
\end{equation}
on the left rectangle with \(\tau=1/2\).
\end{example}
\begin{proof}
Skew-Hermiticity of \(A_n\) verifies the CARE at \(X_n=\mathsf I_n\)
and makes \(A_n-\mathsf I_n\) Hurwitz. The Hamiltonian
\[
\mathcal H_n=
\begin{bmatrix}A_n&-\mathsf I_n\\-\mathsf I_n&A_n\end{bmatrix}
\]
is normal and acts on \([v;v]\) and \([v;-v]\) as
\(A_n-\mathsf I_n\) and \(A_n+\mathsf I_n\), respectively.
Its spectrum \(\{\pm1\pm\mathrm i\}\) gives the stated norm,
condition number, axis gap and projector. The upper projector row has
all singular values \(1/\sqrt2\).
Theorem~\ref{thm:factor-axis-app} bounds the rectangle factor above
by a constant; every normalized inverse factor is at least one.
Lemma~\ref{lem:care-normalization-factor-app} bounds the actual
positive-weight LCU scale independently of the quadrature order.
Padding to a fixed constant scale and using
\eqref{eq:algebraic-projector-output-app} proves the encoding claim.

For frequency \(T\), the normal spectrum is
\(\{\pm1\pm\mathrm iT\}\), so the matrix norm and smallest singular
value are both \(\sqrt{1+T^2}\), while the axis gap stays one.
The rectangle upper bound is \(O(\sqrt{1+T^2})\).
At its inner-edge point \(z=-1/2+\mathrm iT\), the inverse norm
is two and \(|z|+\alpha_H\ge\alpha_H\), giving the matching lower
bound in \eqref{eq:factor-rotation-frequency-app}.
\end{proof}

Thus the construction factor and actual output normalization can stay
bounded as the dimension grows. Normality and positive definite
coefficients alone do not make the factor on the stated rectangles
comparable to the usual condition number: varying the rotational
frequency changes the relevant gap ratio.

\subsection{Coefficient bounds for CARE and DARE}
\label{app:factor-algebraic-coefficients}

The preceding geometric identities identify the matrices that control
spectral selection and recovery. We now bound them using supplied
coefficient information. For the coefficient estimates below, assume
\begin{equation}
\|A\|\le a,\qquad
g_0\mathsf I\preceq G\preceq\bar g\mathsf I,\qquad
q_0\mathsf I\preceq Q\preceq\bar q\mathsf I,\qquad g_0,q_0>0.
\label{eq:factor-coefficient-data-app}
\end{equation}
These are sufficient conditions for the bounds below, rather than
additional assumptions on all problems in Section~\ref{sec:quantum-tools}.
In particular, singular \(G\) arising from a control input of deficient
row rank is outside these coefficient estimates. The scalar bounds are supplied;
the work needed to obtain them is separate from their use here.

Positive stabilizing solutions and their complementary graphs are
classical tools in algebraic Riccati theory
\cite{LancasterRodman1995Riccati,Mehrmann1991AutonomousLQ}.
The following identity also applies to semidefinite solutions. It
identifies the complete projector used in the spectral and initialization
estimates below.

\begin{lemma}[Complementary semidefinite graphs]
\label{lem:factor-positive-graphs-app}
Let \(X,Y\) be Hermitian positive semidefinite matrices.
The matrix \(\mathsf I+YX\) is invertible, the graphs
\(\operatorname{ran}[\mathsf I;X]\) and
\(\operatorname{ran}[-Y;\mathsf I]\) are complementary, and the
projector onto the first along the second is
\begin{equation}
\Pi=[\mathsf I;X](\mathsf I+YX)^{-1}[\mathsf I,Y].
\label{eq:factor-positive-projector-app}
\end{equation}
\end{lemma}
\begin{proof}
Sylvester's determinant identity gives
\[
\det(\mathsf I+YX)
=\det(\mathsf I+X^{1/2}YX^{1/2})>0,
\]
because \(X^{1/2}YX^{1/2}\succeq0\). Thus \(\mathsf I+YX\) is invertible. In the decomposition
\([f;h]=[\mathsf I;X]s+[-Y;\mathsf I]t\), elimination of \(t\) gives
\((\mathsf I+YX)s=f+Yh\). This proves complementarity and
\eqref{eq:factor-positive-projector-app}.
\end{proof}

For CARE, the same positivity bounds the solution norms and controls
the generalized singularity factor through the imaginary-axis gap.

\begin{proposition}[CARE coefficient bounds]
\label{prop:factor-care-coefficients-app}
Under \eqref{eq:factor-coefficient-data-app}, the equations
\begin{equation}
A^*X+XA-XGX+Q=0,\qquad
AY+YA^*-YQY+G=0
\label{eq:factor-care-dual-app}
\end{equation}
have positive definite stabilizing solutions. They satisfy
\begin{equation}
\begin{aligned}
\|X\|&\le u_{\rm CARE}
 :=\frac{a+\sqrt{a^2+g_0\bar q}}{g_0},\\
\|Y\|&\le v_{\rm CARE}
 :=\frac{a+\sqrt{a^2+q_0\bar g}}{q_0}.
\end{aligned}
\label{eq:factor-care-coefficient-bounds-app}
\end{equation}
For \(H=[A,-G;-Q,-A^*]\), its stable projector is
\eqref{eq:factor-positive-projector-app}, and
\begin{equation}
\Delta_{\rm ax}\ge\min\{g_0,q_0\}.
\label{eq:factor-care-two-bounds-app}
\end{equation}
\end{proposition}
\begin{proof}
Set \(J_e=[0,\mathsf I;\mathsf I,0]\). For every real \(\Omega\)
and every \(w=[x;y]\),
\[
\operatorname{Re}\bigl(w^*J_e(\mathrm i\Omega\mathsf I-H)w\bigr)
=x^*Qx+y^*Gy\ge\min\{g_0,q_0\}\|w\|^2.
\]
Since \(J_e\) is unitary, Cauchy--Schwarz proves the gap bound and
excludes imaginary-axis eigenvalues. With
\(J=[0,\mathsf I;-\mathsf I,0]\), the identity \(H^*J+JH=0\)
pairs eigenvalues as \(\lambda,-\overline\lambda\), including
algebraic multiplicities \cite[Lemma~4]{Poloni2020RiccatiIterations}. The stable subspace therefore has dimension
\(n\).

Along a stable trajectory \(\dot w=Hw\), direct differentiation gives
\[
\frac{\mathrm d}{\mathrm dt}(x^*y)=-x^*Qx-y^*Gy.
\]
Integration to infinity shows that \(x(0)^*y(0)>0\) for every nonzero
stable vector. In particular, no such vector is vertical. The stable
subspace is consequently \(\operatorname{ran}[\mathsf I;X]\);
the displayed positive quadratic forms imply \(X=X^*\succ0\).
Graph invariance gives the CARE and the Hurwitz matrix \(A-GX\).
Applying the same argument to \((A^*,Q,G)\) gives \(Y\succ0\) and
the Hurwitz matrix \(A^*-QY\).

Evaluate the CARE at a unit eigenvector for the largest eigenvalue
of \(X\). This gives \(g_0\|X\|^2\le2a\|X\|+\bar q\), whose
positive root is \(u_{\rm CARE}\). The dual equation yields
\(v_{\rm CARE}\) in the same way. Finally,
\[
H[\mathsf I;X]=[\mathsf I;X](A-GX),\qquad
H[-Y;\mathsf I]=[-Y;\mathsf I](QY-A^*).
\]
The second reduced matrix has spectrum in the right half-plane.
These are the two complete spectral subspaces, so
Lemma~\ref{lem:factor-positive-graphs-app} gives the projector.
\end{proof}

If a bound \((A+A^*)/2\preceq h\mathsf I\) is supplied, \(a\) in
the two quadratic-root bounds can be replaced by \(h\), with
its sign unchanged even when \(h<0\). The gap in
\eqref{eq:factor-care-two-bounds-app} supplies the finite separating
contours of Theorem~\ref{thm:factor-axis-app}. A further comparison
with \(\kappa(H)\) still requires the gap ratio and normalization
conditions of Proposition~\ref{prop:factor-condition-number-app};
positive \(G,Q\) alone do not give that comparison.
The supplied bound \(u_{\rm CARE}\) can be inserted directly in
\eqref{eq:algebraic-projector-output-app} and its precision budget;
the actual projector normalization and any re-encoding cost remain.

For DARE, a dual graph identifies the complete complementary branch
of the pencil. The next proof keeps this branch explicit when \(A\)
is singular, as required by the generalized-eigenvalue formulation
\cite{VanDooren1981Generalized,LancasterRodman1995Riccati}.

\begin{proposition}[DARE coefficient bounds]
\label{prop:factor-dare-coefficients-app}
Under \eqref{eq:factor-coefficient-data-app}, the primal and dual
DAREs have positive definite stabilizing solutions satisfying
\begin{equation}
\begin{aligned}
Q\preceq X&\preceq Q+A^*G^{-1}A,&
\|X\|&\le u_{\rm DARE}:=\bar q+\frac{a^2}{g_0},\\
G\preceq Y&\preceq G+AQ^{-1}A^*,&
\|Y\|&\le v_{\rm DARE}:=\bar g+\frac{a^2}{q_0},
\end{aligned}
\label{eq:factor-dare-coefficient-bounds-app}
\end{equation}
where \(Y=G+AY(\mathsf I+QY)^{-1}A^*\).
The pencil \(M-zL\) in \eqref{eq:dare-pencil} is regular and has
no finite spectrum on the unit circle. Its finite interior branch
has dimension \(n\), and its right deflating projector is
\eqref{eq:factor-positive-projector-app}.
No invertibility assumption on \(A\) or \(L\) is required.
\end{proposition}
\begin{proof}
For \(P\succeq0\),
\[
P(\mathsf I+GP)^{-1}
=G^{-1/2}\!\left[\mathsf I-
(\mathsf I+G^{1/2}PG^{1/2})^{-1}\right]G^{-1/2}
\]
is monotone in \(P\) and lies between \(0\) and \(G^{-1}\).
The DARE iteration from zero \cite[Sec.~4.2]{Poloni2020RiccatiIterations} is therefore monotone and bounded
above by \(Q+A^*G^{-1}A\). Its limit \(X\) solves the DARE and
satisfies the first line of \eqref{eq:factor-dare-coefficient-bounds-app}.
For \(F=(\mathsf I+GX)^{-1}A\), substitution gives
\begin{equation}
X-F^*XF=Q+F^*XGXF\succ0.
\label{eq:factor-dare-stein-app}
\end{equation}
Evaluation at an eigenvector of \(F\) proves \(\rho(F)<1\).
The dual iteration gives the second line of the coefficient bounds
and \(K=(\mathsf I+QY)^{-1}A^*\), with
\(Y-K^*YK=G+K^*YQYK\succ0\). Hence \(\rho(K)<1\).

Write \(V_s=[\mathsf I;X]\), \(V_u=[-Y;\mathsf I]\), and
\begin{equation}
T=[V_s,V_u],\qquad
S_p=[LV_s,MV_u]
=\begin{bmatrix}\mathsf I+GX&-AY\\A^*X&\mathsf I+QY\end{bmatrix}.
\label{eq:factor-dare-equivalence-matrices-app}
\end{equation}
Lemma~\ref{lem:factor-positive-graphs-app} makes \(T\) invertible.
For \(D=\operatorname{diag}(X,Y)\),
\[
\frac{DS_p+S_p^*D}{2}
=\operatorname{diag}(X+XGX,Y+YQY)\succ0,
\]
which also makes \(S_p\) invertible. The primal and dual equations
give \(MV_s=LV_sF\) and \(LV_u=MV_uK\), so
\begin{equation}
S_p^{-1}MT=\operatorname{diag}(F,\mathsf I),\qquad
S_p^{-1}LT=\operatorname{diag}(\mathsf I,K).
\label{eq:factor-dare-strict-equivalence-app}
\end{equation}
This proves regularity. The second block has only exterior finite
eigenvalues, namely reciprocals of the nonzero eigenvalues of \(K\),
and possible infinite modes. In particular,
\[
(zL-M)^{-1}L
=T\operatorname{diag}\bigl((z\mathsf I-F)^{-1},
                         (zK-\mathsf I)^{-1}K\bigr)T^{-1}.
\]
The second block is analytic on a neighborhood of the closed unit
disk. Integration gives \(\Pi_<=T\operatorname{diag}(\mathsf I,0)T^{-1}\),
and Lemma~\ref{lem:factor-positive-graphs-app} gives the stated
complete projector.
\end{proof}

For example, \(A=0\) gives \(X=Q\) and \(Y=G\), with an
entirely infinite complementary branch. The bound \(u_{\rm DARE}\)
supplies the solution-norm estimate for
\eqref{eq:algebraic-projector-output-app} and its precision budget,
including in this singular case.

The Stein inequalities \cite[Sec.~2]{Poloni2020RiccatiIterations} also control the unit-circle generalized
singularity factor. We give the bound explicitly and reuse the same
contraction estimates for the forward recursion.

\begin{proposition}[DARE unit-circle generalized singularity factor]
\label{prop:factor-dare-resolvent-app}
Under \eqref{eq:factor-coefficient-data-app},
\begin{equation}
\begin{aligned}
\sup_{|z|=1}\|(zL-M)^{-1}\|
&\le\bigl(1+\max\{u_{\rm DARE},v_{\rm DARE}\}\bigr)
\sqrt{\frac{\max\{u_{\rm DARE},v_{\rm DARE}\}}{\min\{q_0,g_0\}}}\\
&\quad\times\max\left\{
\frac{\sqrt{u_{\rm DARE}/q_0}}{1-\sqrt{1-q_0/u_{\rm DARE}}},
\frac{\sqrt{v_{\rm DARE}/g_0}}{1-\sqrt{1-g_0/v_{\rm DARE}}}
\right\}.
\end{aligned}
\label{eq:factor-dare-resolvent-app}
\end{equation}
Multiplying the right-hand side by \(\alpha_M+\alpha_L\) bounds
\(\mathcal R_{\mathbb T,L}\). All denominators are positive because
\(0<q_0/u_{\rm DARE}\le1\) and \(0<g_0/v_{\rm DARE}\le1\).
The estimate includes zero contraction factors and singular \(A\).
\end{proposition}
\begin{proof}
Since \(Q\succeq(q_0/u_{\rm DARE})X\), the Stein identity gives
\(F^*XF\preceq(1-q_0/u_{\rm DARE})X\).
Together with its dual, this yields
\begin{equation}
\|X^{1/2}FX^{-1/2}\|\le\sqrt{1-q_0/u_{\rm DARE}},\qquad
\|Y^{1/2}KY^{-1/2}\|\le\sqrt{1-g_0/v_{\rm DARE}}.
\label{eq:factor-discrete-contractions-app}
\end{equation}
Neumann series therefore give, for \(|z|=1\),
\[
\begin{aligned}
\|(z\mathsf I-F)^{-1}\|
&\le\frac{\sqrt{u_{\rm DARE}/q_0}}{1-\sqrt{1-q_0/u_{\rm DARE}}},\\
\|(zK-\mathsf I)^{-1}\|
&\le\frac{\sqrt{v_{\rm DARE}/g_0}}{1-\sqrt{1-g_0/v_{\rm DARE}}}.
\end{aligned}
\]
The calculation in the preceding proof gives
\((DS_p+S_p^*D)/2\succeq D\). Thus the Hermitian part of
\(D^{1/2}S_pD^{-1/2}\) is at least \(\mathsf I\), and its inverse
has norm at most one. Consequently
\[
\|S_p^{-1}\|\le
\sqrt{\frac{\max\{u_{\rm DARE},v_{\rm DARE}\}}{\min\{q_0,g_0\}}},
\qquad
\|T\|\le1+\max\{u_{\rm DARE},v_{\rm DARE}\}.
\]
Taking norms in the inverse of
\[
zL-M=S_p\operatorname{diag}(z\mathsf I-F,zK-\mathsf I)T^{-1}
\]
proves \eqref{eq:factor-dare-resolvent-app}.
On the unit circle, the shifted encoding scale is the constant
\(\alpha_M+\alpha_L\), giving the stated bound on
\(\mathcal R_{\mathbb T,L}\).
\end{proof}

Equation~\eqref{eq:factor-pencil-circle-app} identifies the weighted
factor on this fixed contour, including when \(L\) is singular.

\subsection{Coefficient bounds for DRE and RR}
\label{app:factor-dynamic-coefficients}

A dynamic construction additionally needs the coordinates of the
initial graph in its selected branch. For DRE, \(\Pi_+\) annihilates
the stable positive graph. For the forward recursion, \(\Pi_>\)
annihilates the complementary negative graph. We first bound the
corresponding nondegeneracy factors, then connect them to the
generalized singularity factors and the condition numbers used in
the dynamic query bounds.

\begin{proposition}[DRE initialization from coefficient bounds]
\label{prop:factor-dre-initialization-app}
Use the DRE convention of Definition~\ref{def:problem-dre}, with
\(H=[A,-G;-Q,-A^*]\). Suppose
\[
Q\succeq q_0\mathsf I,\qquad
0\prec G\preceq\bar g\mathsf I,\qquad
(A+A^*)/2\succeq\omega\mathsf I,\qquad
P_0=P_0^*,\quad\|P_0\|\le p,
\]
where \(q_0,\bar g>0\), \(p\ge0\), and \(\omega\) is real.
If \(q_0+2\omega p-\bar g p^2>0\), then \(H\) has no imaginary-axis
spectrum and
\begin{equation}
\sigma_{\min}(\Pi_+R_0)
\ge\frac{\omega+\sqrt{\omega^2+\bar gq_0}-\bar g p}
{\sqrt{\bar g^2+(\omega+\sqrt{\omega^2+\bar gq_0})^2}}>0.
\label{eq:factor-dre-initialization-app}
\end{equation}
The initial graph also satisfies
\begin{equation}
\sin\theta_{\min}(\operatorname{ran}R_0,\ker\Pi_+)
\ge\frac{\omega+\sqrt{\omega^2+\bar gq_0}-\bar g p}
{\sqrt{\bar g^2+(\omega+\sqrt{\omega^2+\bar gq_0})^2}
 \sqrt{1+p^2}}.
\label{eq:factor-dre-initial-angle-app}
\end{equation}
\end{proposition}
\begin{proof}
In finite dimensions \(G\succ0\) has a positive smallest eigenvalue.
The energy argument in Proposition~\ref{prop:factor-care-coefficients-app}
therefore proves hyperbolicity and
\(\ker\Pi_+=\operatorname{ran}[\mathsf I;X]\) with \(X\succ0\).
Evaluating its CARE at a unit eigenvector for the smallest eigenvalue
of \(X\) gives
\[
0\ge q_0+2\omega\lambda_{\min}(X)-\bar g\lambda_{\min}(X)^2.
\]
The quadratic has precisely one positive root, so
\begin{equation}
\lambda_{\min}(X)
\ge\frac{\omega+\sqrt{\omega^2+\bar gq_0}}{\bar g}>p.
\label{eq:factor-dre-lower-slope-app}
\end{equation}
The strict inequality follows from the coefficient condition and
\(p\ge0\).

For every unit vector \(w\), Lemma~\ref{lem:factor-graph-distance-app} gives
\[
\begin{aligned}
\|\Pi_+R_0w\|
&\ge\|(\mathsf I+X^2)^{-1/2}(P_0-X)w\|\\
&\ge\sigma_{\min}\bigl(X(\mathsf I+X^2)^{-1/2}\bigr)
   -p\|(\mathsf I+X^2)^{-1/2}\|\\
&=\frac{\lambda_{\min}(X)-p}{\sqrt{1+\lambda_{\min}(X)^2}}.
\end{aligned}
\]
The equality follows by functional calculus for \(X\succ0\).
The function \((s-p)/\sqrt{1+s^2}\) increases for \(s>0\) when
\(p\ge0\). Applying \eqref{eq:factor-dre-lower-slope-app} and
multiplying numerator and denominator by \(\bar g\) proves
\eqref{eq:factor-dre-initialization-app}; no commutation of \(X\)
and \(P_0\) is needed.
Apply the same distance estimate to the orthonormal frame
\(R_0(\mathsf I+P_0^2)^{-1/2}\), whose right factor has smallest
singular value at least \(1/\sqrt{1+p^2}\), to obtain
\eqref{eq:factor-dre-initial-angle-app}.
\end{proof}

This result bounds an initial branch projection; it does not establish
existence of the DRE solution up to the requested time.
Its positive \(G,Q\) assumptions also do not automatically apply after
the sign changes used to represent terminal LQR in
Definition~\ref{def:problem-dre}. If a uniform lower bound
\(G\succeq g_0\mathsf I\) is additionally supplied, the energy estimate
gives \(\Delta_{\rm ax}\ge\min\{g_0,q_0\}\) and the separating-contour
bounds of Theorem~\ref{thm:factor-axis-app}.
Strict positivity without such a uniform lower bound suffices for the
initialization argument but does not supply that uniform spectral bound.

For RR, the branch assignment follows from the forward update itself.
The linear-fractional representation and its relation to DARE
iterations are standard \cite{Poloni2020RiccatiIterations,Mehrmann1991AutonomousLQ};
the following calculation fixes their orientation in our convention.

\begin{lemma}[The forward recursion and its selected graph]
\label{lem:factor-rr-branches-app}
Assume \eqref{eq:factor-coefficient-data-app} and let \(A\) be
invertible. For the pencil matrices in \eqref{eq:dare-pencil}, set
\begin{equation}
\mathcal S_F=M^{-1}L
=\begin{bmatrix}
A^{-1}&A^{-1}G\\
QA^{-1}&A^*+QA^{-1}G
\end{bmatrix}.
\label{eq:factor-rr-lift-app}
\end{equation}
Whenever \(\mathsf I+GP\) is invertible,
\begin{equation}
\mathcal S_F[\mathsf I;P]
=[\mathsf I;Q+A^*P(\mathsf I+GP)^{-1}A]\,
A^{-1}(\mathsf I+GP).
\label{eq:factor-rr-graph-update-app}
\end{equation}
For the DARE solutions \(X,Y\) and matrices \(F,K\) above,
\[
\mathcal S_F[\mathsf I;X]=[\mathsf I;X]F^{-1},
\qquad
\mathcal S_F[-Y;\mathsf I]=[-Y;\mathsf I]K.
\]
Thus the positive graph is the exterior branch and the negative
graph is the interior branch, each of dimension \(n\). In particular,
\begin{equation}
\Pi_>=\Pi_<,\qquad
\Pi_>R_0=[\mathsf I;X](\mathsf I+YX)^{-1}(\mathsf I+YP_0),
\label{eq:factor-rr-initial-graph-app}
\end{equation}
where \(\Pi_<\) on the right is the DARE pencil projector.
\end{lemma}
\begin{proof}
Block multiplication proves \eqref{eq:factor-rr-graph-update-app}.
Multiplying the relations \(MV_s=LV_sF\) and \(LV_u=MV_uK\)
by \(M^{-1}\) gives the two invariant graphs.
Here \(F\) is invertible because \(A\) is invertible.
Since \(F,K\) are Schur stable, \(F^{-1}\) has exterior spectrum
and \(K\) has interior spectrum. The canonical projectors therefore
have the same range and kernel as in
Proposition~\ref{prop:factor-dare-coefficients-app}.
Its projector formula gives \eqref{eq:factor-rr-initial-graph-app}.
\end{proof}

The interior DARE branch and exterior RR branch therefore represent
the same positive graph in two different spectral descriptions.
The first initialization bound uses the positive semidefinite data of
Definition~\ref{def:problem-recursion}.

\begin{proposition}[Positive semidefinite RR initialization]
\label{prop:factor-rr-initialization-app}
Under the assumptions of Lemma~\ref{lem:factor-rr-branches-app},
every \(P_0\succeq0\) satisfies
\begin{equation}
\sigma_{\min}(\Pi_>R_0)
\ge\frac{\sqrt{1+q_0^2}}{1+u_{\rm DARE}v_{\rm DARE}}
\sqrt{\frac{g_0}{v_{\rm DARE}}}>0.
\label{eq:factor-rr-psd-initialization-app}
\end{equation}
No upper restriction on \(\|P_0\|\) is needed for this absolute
nondegeneracy bound.
\end{proposition}
\begin{proof}
The factorization
\[
\mathsf I+YP_0
=Y^{1/2}(\mathsf I+Y^{1/2}P_0Y^{1/2})Y^{-1/2}
\]
has a middle factor at least \(\mathsf I\). Since
\(g_0\mathsf I\preceq Y\preceq v_{\rm DARE}\mathsf I\),
the smallest singular value of \(\mathsf I+YP_0\) is at least
\(\sqrt{g_0/v_{\rm DARE}}\).
Also,
\[
\sigma_{\min}([\mathsf I;X])\ge\sqrt{1+q_0^2},\qquad
\sigma_{\min}((\mathsf I+YX)^{-1})
\ge(1+u_{\rm DARE}v_{\rm DARE})^{-1}.
\]
Multiplying these three bounds in
\eqref{eq:factor-rr-initial-graph-app} proves the result.
\end{proof}

The same discrete contraction estimates now bound the generalized
singularity factor used to select these two branches.

\begin{proposition}[RR unit-circle generalized singularity factor]
\label{prop:factor-rr-resolvent-app}
Under the assumptions of Lemma~\ref{lem:factor-rr-branches-app},
\begin{equation}
\begin{aligned}
\sup_{|z|=1}\|(z\mathsf I-\mathcal S_F)^{-1}\|
&\le\bigl(1+\max\{u_{\rm DARE},v_{\rm DARE}\}\bigr)
\sqrt{\frac{\max\{u_{\rm DARE},v_{\rm DARE}\}}{\min\{q_0,g_0\}}}\\
&\quad\times\max\left\{
\frac{\sqrt{u_{\rm DARE}/q_0}\sqrt{1-q_0/u_{\rm DARE}}}
 {1-\sqrt{1-q_0/u_{\rm DARE}}},
\frac{\sqrt{v_{\rm DARE}/g_0}}{1-\sqrt{1-g_0/v_{\rm DARE}}}
\right\}.
\end{aligned}
\label{eq:factor-rr-resolvent-app}
\end{equation}
The reciprocal of the right-hand side of
\eqref{eq:factor-rr-resolvent-app} is a lower bound on
\(\eta_{\mathbb T}\). Multiplying that right-hand side by
\(1+\alpha_{\mathcal S_F}\) bounds the unit-circle generalized
singularity factor at the declared scale
\(\alpha_{\mathcal S_F}\ge\|\mathcal S_F\|\).
The projector satisfies
\begin{equation}
\|\Pi_>\|\le
\bigl(1+\max\{u_{\rm DARE},v_{\rm DARE}\}\bigr)
\sqrt{\frac{\max\{u_{\rm DARE},v_{\rm DARE}\}}{\min\{q_0,g_0\}}}.
\label{eq:factor-rr-projector-bound-app}
\end{equation}
\end{proposition}
\begin{proof}
The invariant graph relations give
\(\mathcal S_F=T\operatorname{diag}(F^{-1},K)T^{-1}\).
For \(D=\operatorname{diag}(X,Y)\),
\((DT+T^*D)/2=D\), so the accretivity argument in
Proposition~\ref{prop:factor-dare-resolvent-app} gives
\[
\|T^{-1}\|\le
\sqrt{\frac{\max\{u_{\rm DARE},v_{\rm DARE}\}}{\min\{q_0,g_0\}}}.
\]
Together with the bound on \(\|T\|\) in that proof, this proves
\eqref{eq:factor-rr-projector-bound-app}.
In the \(X\) metric, the identity
\[
(z\mathsf I-F^{-1})^{-1}=-(\mathsf I-zF)^{-1}F
\]
and the contraction bounds in \eqref{eq:factor-discrete-contractions-app}
give, for \(|z|=1\),
\[
\begin{aligned}
\|(z\mathsf I-F^{-1})^{-1}\|
&\le\frac{\sqrt{u_{\rm DARE}/q_0}\sqrt{1-q_0/u_{\rm DARE}}}
 {1-\sqrt{1-q_0/u_{\rm DARE}}},\\
\|(z\mathsf I-K)^{-1}\|
&\le\frac{\sqrt{v_{\rm DARE}/g_0}}{1-\sqrt{1-g_0/v_{\rm DARE}}}.
\end{aligned}
\]
The second estimate follows from the \(Y\)-metric Neumann series.
Similarity by \(T\) proves \eqref{eq:factor-rr-resolvent-app}.
Taking reciprocal singular values gives the lower bound on
\(\eta_{\mathbb T}\), and the shifted scale on the unit circle
is \(1+\alpha_{\mathcal S_F}\).
\end{proof}

A fixed fraction of the reciprocal of the right-hand side of
\eqref{eq:factor-rr-resolvent-app} supplies a radial margin for
the interior circle and exterior annulus of
Theorem~\ref{thm:factor-circle-app}. The exterior domain must exclude
zero for the weight \(z^{-k}\); scaling \(\mathcal S_F\) moves its
separating circle and must also rescale the contour.
Its encoding scale and construction cost include those of \(A^{-1}\),
which cannot be bounded from \(\|A\|\) above alone.

Finally, these lower bounds concern absolute initialization
nondegeneracy. If \(\|P_0\|\le p\), then for either selected projector
\[
\kappa(\Pi_{\rm sel}R_0)
=\frac{\|\Pi_{\rm sel}R_0\|}
       {\sigma_{\min}(\Pi_{\rm sel}R_0)}
\le\frac{\|\Pi_{\rm sel}\|\sqrt{1+p^2}}
         {\sigma_{\min}(\Pi_{\rm sel}R_0)}.
\]
For RR, \eqref{eq:factor-rr-projector-bound-app} bounds the projector
norm; the appropriate initialization result bounds the denominator. For DRE, a separate projector-norm bound is still
needed. If it is obtained from the integral estimate in
Appendix~\ref{app:factor-geometry}, its generalized singularity factor
and contour dependence remain in this condition-number bound.
This distinction preserves the \(\kappa\) dependence in the dynamic
complexity formulas while keeping the required coefficient,
initialization, and time or step assumptions explicit.

\clearpage

\section{Weyl--LCHM construction of half-plane projectors}
\label{app:weyl-lchm}

Subsection~\ref{subsec:weyl-lchm-riesz} reduces the construction of
\(\Pi_\chi R\) to approximating a scalar selector and implementing the
resulting polynomial. We first prove the Weyl implementation in
Proposition~\ref{cor:lchm-branch-block}, then derive the polynomial degree
required by Theorem~\ref{thm:lchm-explicit-query}. A Faber construction
then combines scalar approximation with contour resolvent control,
allowing matrices with arbitrary Jordan structure.
The final subsection applies these exact-input results to approximate
block-encodings by separating selector, matrix-input, and circuit errors.
We use the normalized variables and selector notation of
\eqref{eq:lchm-half-plane-gap}--\eqref{eq:lchm-polynomial-selectors}.

\subsection{Weyl realization and proof of the branch-block construction}
\label{app:weyl-realization}

Proposition~\ref{cor:lchm-branch-block} turns a polynomial selector into
an encoding of the Riesz block. We prove that implementation first;
the degree estimates in the following subsections can then be converted
directly into query bounds. We use the exact-input and coherent-access
assumptions of Subsection~\ref{subsec:weyl-lchm-riesz}.

The angular identity itself holds for any contraction. Let
\(\widetilde H\) denote the signal block of the supplied unit-normalized
encoding, equal to \(\widehat H\) for exact inputs. Write the Weyl lift as
\begin{equation}
2S_{\chi,r}(z)-\tfrac12=\sum_{m=0}^{d_S}c_mz^m,
\qquad c_0=\tfrac12.
\label{eq:lchm-weyl-lift}
\end{equation}
Use the normalization defined in \eqref{eq:lchm-selector-normalization}.
The subtraction of \(1/2\) preserves the constant coefficient once;
the other coefficients are doubled.

\begin{lemma}[Exact angular realization]
\label{lem:lchm-angular-realization}
Let \(N>d_S\), \(\theta_j=\pi j/N\), and
\[
X_j=\tfrac12(e^{-\mathrm i\theta_j}\widetilde H+
                     e^{\mathrm i\theta_j}\widetilde H^*).
\]
For the Chebyshev polynomials \(T_m\) of the first kind, define
\(q_{\theta_j}(x)=c_0+\sum_{m=1}^{d_S}
c_me^{\mathrm im\theta_j}T_m(x)\).
Then
\begin{equation}
S_{\chi,r}(\widetilde H)
=\frac1N\sum_{j=0}^{N-1}q_{\theta_j}(X_j),
\qquad
\max_{x\in[-1,1]}|q_{\theta_j}(x)|
\le\alpha_{S_{\chi,r}}.
\label{eq:lchm-angular-identity}
\end{equation}
\end{lemma}

\begin{proof}
For \(m\ge1\), \(T_m(x)\) has leading term \(2^{m-1}x^m\), and
all its other powers have degrees \(m-2,m-4,\ldots\).
In a term of degree \(m-2\ell\), a word containing \(k\) factors of
\(\widetilde H^*\) acquires the phase \(e^{2\mathrm i(\ell+k)\theta_j}\)
after multiplication by \(e^{\mathrm im\theta_j}\).
The integer \(\ell+k\) lies between zero and \(m<N\).
Its discrete average vanishes unless \(\ell=k=0\), in which case the
word is \(\widetilde H^m\) and its coefficient is \(1/2\).
Thus
\[
\frac1N\sum_{j=0}^{N-1}
e^{\mathrm im\theta_j}T_m(X_j)=\frac12\widetilde H^m
\quad(m\ge1).
\]
The constant term averages to \(c_0\mathsf I\), proving the identity.
For \(x=\cos\phi\), the scalar angular polynomial satisfies
\[
q_\theta(\cos\phi)
=S_{\chi,r}(e^{\mathrm i(\theta+\phi)})
+S_{\chi,r}(e^{\mathrm i(\theta-\phi)})-\tfrac12.
\]
This is the average of the Weyl lift at two points on the unit circle,
so its absolute value is at most \(\alpha_{S_{\chi,r}}\).
\end{proof}

The identity expresses the selector as an average of Hermitian
polynomials. To realize this average with \(O(d_S)\) queries, the angles
must be processed coherently. The following signal construction lets a
single polynomial signal-processing sequence act on all angles
\cite{WangLiangChenLiu2026LCHM}.
On one extra qubit define the Hermitian unitary
\[
\mathcal V_j=
\begin{bmatrix}
0&e^{-\mathrm i\theta_j}U_{\widehat H}\\
e^{\mathrm i\theta_j}U_{\widehat H}^*&0
\end{bmatrix}.
\]
Let \(J|\psi\rangle=|+\rangle|0^{a_H}\rangle|\psi\rangle\) be the
signal isometry. Then \(J^*\mathcal V_jJ=X_j\), and the qubitized walk
\(W_j=(2JJ^*-\mathsf I)\mathcal V_j\) satisfies
\(J^*W_j^mJ=T_m(X_j)\).
The angle-controlled unitary
\[
\mathsf Z=\sum_{j=0}^{N-1}|j\rangle\langle j|
\otimes e^{\mathrm i\theta_j}W_j
\]
therefore has the required Chebyshev powers in its signal block.
Generalized quantum signal processing implements
\(\sum_{m=0}^{d_S}c_m\mathsf Z^m/\alpha_{S_{\chi,r}}\), since
the normalized polynomial is bounded by one on the unit circle
\cite{MotlaghWiebe2023GQSP}.
Any supplied bound in \eqref{eq:lchm-selector-normalization} is therefore
admissible. Preparing and unpreparing
\(N^{-1/2}\sum_j|j\rangle\) gives the average in
Lemma~\ref{lem:lchm-angular-realization}.

Averaging therefore implements the selector at the declared normalization.
Composing it with the encoding of \(R\) and accounting for the two error
sources proves the main-text construction result.

\newtheorem*{lchmbranchrestatement}{Proposition~\ref{cor:lchm-branch-block}}
\begin{lchmbranchrestatement}
Under the input and access assumptions above, let
\(S_{\chi,r}\) be the selector in
\eqref{eq:lchm-polynomial-selectors} and satisfy
\(\|(S_{\chi,r}(\widehat H)-\Pi_\chi)R\|\le\varepsilon/2\),
where \(\varepsilon>0\).
Weyl--LCHM returns a block-encoding of \(\Pi_\chi R\) with normalization
\(\alpha_{S_{\chi,r}}\alpha_R\), decoded error at most \(\varepsilon\),
using
\[
O(d_S)\quad\text{queries to the block-encodings of }\widehat H\text{ and }R.
\]
\end{lchmbranchrestatement}

\begin{proof}
Choose \(N=2^b\), where \(b=\lceil\log_2(d_S+1)\rceil\).
The construction above gives the exact signal block
\(S_{\chi,r}(\widehat H)/\alpha_{S_{\chi,r}}\) for ideal circuits.
Composing with the encoding of \(R\) gives
\(S_{\chi,r}(\widehat H)R\) at normalization
\(\alpha_{S_{\chi,r}}\alpha_R\).
Implement the polynomial circuit to normalized block error at most
\(\varepsilon/(2\alpha_{S_{\chi,r}}\alpha_R)\).
Adding this implementation error to the selector error gives decoded error at most
\(\varepsilon\).
The signal-processing sequence has \(O(d_S)\) signal calls,
each using a constant number of controlled input or adjoint calls.
Multiplication by \(R\) uses one further call.
\end{proof}

Coherent angular processing introduces no additional factor \(N\) in the
query count. With \(b=\lceil\log_2(d_S+1)\rceil\), the angle register
uses \(b\) qubits and the \(O(d_S)\) signal calls use
\(O(d_S\log(d_S+1))\) angle rotations in total.
A direct composition uses at most \(a_H+a_R+b+2\) logical ancillas;
the two extra qubits implement the Hermitian unitary and signal processing.
Workspace internal to the oracles, reflections, and synthesis is additional.
The cost now depends on which selector is chosen. We next determine its
degree under the spectral assumptions of the main theorem.

\subsection{Proof of the sector query bound}
\label{app:weyl-explicit-query}

The implementation above applies to any selector with the required
accuracy. Theorem~\ref{thm:lchm-explicit-query} obtains an explicit cost by
using the sector condition to choose that selector. After squaring the
eigenvalues, a scaled binomial series converges uniformly on the spectral
region. Its tail determines the degree, while the eigenvector condition
bound transfers scalar accuracy to the Riesz block.

\newtheorem*{lchmqueryrestatement}{Theorem~\ref{thm:lchm-explicit-query}}
\begin{lchmqueryrestatement}
Assume the exact-input access of Proposition~\ref{cor:lchm-branch-block}
and \(\widehat H=V\Lambda V^{-1}\), where \(\Lambda\) is diagonal.
Let \(\Delta_H\) be a supplied lower bound in
\eqref{eq:lchm-half-plane-gap}, and let
\(\kappa_V\ge\|V\|\|V^{-1}\|\) be a supplied upper bound.
Suppose a supplied \(0\le\eta<1\) satisfies
\[
|\operatorname{Im}\lambda|\le\eta|\operatorname{Re}\lambda|
\qquad(\lambda\in\operatorname{spec}(\widehat H)).
\]
For any \(\varepsilon>0\), Weyl--LCHM constructs a block-encoding of
\(\Pi_\chi R\) with decoded error at most \(\varepsilon\)
and normalization \(\alpha_{S_{\chi,r}}\alpha_R\) as in
\eqref{eq:lchm-selector-normalization}, using
\[
\widetilde O\!\left(\frac{1}{\Delta_H^2(1-\eta^2)^2}\right)
\]
queries to the block-encodings of \(\widehat H\) and \(R\).
Here \(\widetilde O\) suppresses logarithmic factors in
\(e+\alpha_R\kappa_V/\varepsilon\) and \((1-\eta^2)^{-1}\).
\end{lchmqueryrestatement}

\begin{proof}
We first construct the inverse-square-root approximation on the squared
spectrum. Set \(b_k=4^{-k}\binom{2k}{k}\) and define
\begin{equation}
q_{r-1}(w)
=\sqrt{\frac{1-\eta^2}{1+\eta^2}}
\sum_{k=0}^{r-1}b_k
\left(1-\frac{1-\eta^2}{1+\eta^2}w\right)^k.
\label{eq:lchm-spectral-binomial}
\end{equation}
Use the selectors in \eqref{eq:lchm-polynomial-selectors}.
Then \(\deg S_{\pm,r}\le d_S=2r-1\) and
\(S_{\pm,r}(0)=1/2\), as required by the Weyl construction.
For a spectral point \(\lambda\), the sector assumption gives
\[
\operatorname{Re}(\lambda^2)
\ge\frac{1-\eta^2}{1+\eta^2}|\lambda|^2.
\]
Writing
\(t_\lambda=|1-\frac{1-\eta^2}{1+\eta^2}\lambda^2|\)
and using \(|\lambda|\le1\), we obtain
\[
\begin{aligned}
t_\lambda^2
&\le1-2\left(\frac{1-\eta^2}{1+\eta^2}\right)^2|\lambda|^2
+\left(\frac{1-\eta^2}{1+\eta^2}\right)^2|\lambda|^4\\
&\le1-\left(\frac{1-\eta^2}{1+\eta^2}\right)^2|\lambda|^2<1,
\end{aligned}
\]
and therefore
\[
1-t_\lambda
\ge\frac12\left(\frac{1-\eta^2}{1+\eta^2}\right)^2|\lambda|^2.
\]
Consequently the binomial series converges to the principal inverse
square root of \(\lambda^2\). The coefficients \(b_k\) decrease with
\(k\), so for \(0\le t<1\),
\[
\sum_{k=r}^{\infty}b_kt^k
=t^r\sum_{j=0}^{\infty}b_{j+r}t^j
\le\frac{t^r}{\sqrt{1-t}}.
\]
Let \(f_\pm\) denote the corresponding half-plane indicator.
Combining this tail estimate with
\(\lambda(\lambda^2)^{-1/2}=\operatorname{sign}(\operatorname{Re}\lambda)\)
gives
\[
\begin{aligned}
|S_{\pm,r}(\lambda)-f_\pm(\lambda)|
&\le\frac{|\lambda|}{2}
\sqrt{\frac{1-\eta^2}{1+\eta^2}}
\frac{t_\lambda^r}{\sqrt{1-t_\lambda}}\\
&\le\sqrt{\frac{1+\eta^2}{2(1-\eta^2)}}
\exp\!\left(
-\frac{r(1-\eta^2)^2|\lambda|^2}{2(1+\eta^2)^2}
\right)\\
&\le\frac{1}{\sqrt{1-\eta^2}}
\exp\!\left(-\frac{r(1-\eta^2)^2\Delta_H^2}{8}\right).
\end{aligned}
\]
The last inequality uses \(1+\eta^2\le2\) and
\(|\lambda|\ge|\operatorname{Re}\lambda|\ge\Delta_H\).
It remains valid when \(t_\lambda=0\), since the tail then vanishes
for every \(r\ge1\).
This scalar estimate controls the Riesz block through diagonalization.
Using \(\|R\|\le\alpha_R\), we obtain
\[
\|(S_{\chi,r}(\widehat H)-\Pi_\chi)R\|
\le\frac{\alpha_R\kappa_V}{\sqrt{1-\eta^2}}
\exp\!\left(-\frac{r(1-\eta^2)^2\Delta_H^2}{8}\right).
\]
With \(\log_+x=\max\{0,\log x\}\), choose
\begin{equation}
r=\max\!\left\{1,\left\lceil
\frac{8}{(1-\eta^2)^2\Delta_H^2}
\log_+\!\left(
\frac{2\alpha_R\kappa_V}{\varepsilon\sqrt{1-\eta^2}}
\right)
\right\rceil\right\}.
\label{eq:lchm-sector-degree}
\end{equation}
The selector error is then at most \(\varepsilon/2\).

The degree choice controls the selector error on the spectrum. To invoke
Proposition~\ref{cor:lchm-branch-block}, we also need a normalization on
the unit circle. For \(|z|=1\), the same polynomial satisfies
\[
\begin{aligned}
|2S_{\chi,r}(z)-1/2|
&\le\frac12+\sqrt{\frac{1-\eta^2}{1+\eta^2}}
\sum_{k=0}^{r-1}b_k
\left(1+\frac{1-\eta^2}{1+\eta^2}\right)^k\\
&\le\frac12+\sum_{k=0}^{r-1}2^k<2^r.
\end{aligned}
\]
Hence an admissible normalization is
\begin{equation}
\alpha_{S_{\chi,r}}=2^r.
\label{eq:lchm-explicit-normalization}
\end{equation}
Implement the normalized polynomial block to error at most
\(\varepsilon/(2^{r+1}\alpha_R)\).
At output scale \(2^r\alpha_R\), the resulting implementation error
is at most \(\varepsilon/2\), giving total decoded error at most
\(\varepsilon\).
The Weyl construction uses \(O(d_S)\) queries to the block-encodings
of \(\widehat H\) and \(R\), including one query to the latter.
Substituting the chosen \(r\) yields
\begin{equation}
O\!\left(
\frac{1}{(1-\eta^2)^2\Delta_H^2}
\log\!\left(e+
\frac{\alpha_R\kappa_V}{\varepsilon\sqrt{1-\eta^2}}
\right)\right)
\label{eq:lchm-sector-query-expanded}
\end{equation}
queries. Suppressing the indicated logarithms gives
\eqref{eq:lchm-sector-query}.
\end{proof}

The scale in \eqref{eq:lchm-explicit-normalization} is sufficient for
this construction. Its effect on input precision is accounted for in
Appendix~\ref{app:weyl-errors}; subsequent amplification and graph
recovery must also include the chosen scale.

\subsection{Faber approximation with resolvent control}
\label{app:weyl-degree}

Faber approximation removes the sector restriction. Combined with
Lemma~\ref{lem:lchm-branch-selector}, it controls selector error for
arbitrary Jordan structure through the squared-cap geometry and
generalized singularity factor.

\begin{proposition}[Faber approximation on a squared cap]
\label{cor:lchm-general-gap-faber}
Let \(\Omega\) be the image under \(z\mapsto z^2\) of a closed circular
cap with nonempty interior contained in the open right half-plane,
and let \(F_k\) be its
Faber polynomials. Let \(\Lambda_\varrho\) be the image of
\(|\zeta|=\varrho>1\) under the conformal map from the exterior unit disk
to the exterior of \(\Omega\), normalized at infinity. Assume that its
closed interior excludes zero.
Let \(h\) be the holomorphic inverse square root on a neighborhood of
that interior, agreeing with the principal branch on \(\Omega\).
Given \(B_\Omega\ge\sup_{k\ge0}\|F_k\|_{\infty,\Omega}\), set
\begin{equation}
M_\varrho=\max_{w\in\Lambda_\varrho}|w|^{-1/2},\qquad
q_{r-1}(w)=\sum_{k=0}^{r-1}a_kF_k(w),\qquad r\ge1,
\label{eq:lchm-faber-polynomial}
\end{equation}
where \(a_k\) are the Faber coefficients of \(h\). Then
\begin{equation}
\|q_{r-1}-h\|_{\infty,\Omega}
\le\frac{B_\Omega M_\varrho}{\varrho-1}\varrho^{-(r-1)}.
\label{eq:lchm-faber-error}
\end{equation}
\end{proposition}

\begin{proof}
The square map is injective near the cap, so \(\Omega\) is a Jordan
compact set with piecewise analytic boundary. The assumed holomorphic
extension gives a convergent Faber expansion on \(\Omega\), with
\(|a_k|\le M_\varrho\varrho^{-k}\) by Cauchy's estimate in the exterior
conformal coordinate \cite{BeckermannReichel2009Faber}. Thus
\[
\left\|h-\sum_{k=0}^{r-1}a_kF_k\right\|_{\infty,\Omega}
\le B_\Omega M_\varrho\sum_{k=r}^{\infty}\varrho^{-k}
=\frac{B_\Omega M_\varrho}{\varrho-1}\varrho^{-(r-1)}.
\]
\end{proof}

For \(\Omega=\Omega_{r_d,r_c}\) in \eqref{eq:lchm-squared-caps},
Lemma~\ref{lem:lchm-branch-selector} converts this error to a sufficient
degree. With \(\log_+x=\max\{0,\log x\}\), the choice
\begin{equation}
r=1+\left\lceil
\frac{\log_+\!\left(
\dfrac{\alpha_R\mathcal A_Hr_cB_\Omega M_\varrho}
{2(\varrho-1)\varepsilon_{\rm sel}}\right)}
{\log\varrho}
\right\rceil,\qquad d_S=2r-1,
\label{eq:lchm-faber-degree}
\end{equation}
satisfies \eqref{eq:lchm-selector-budget} for arbitrary Jordan structure.
The level curve and \(B_\Omega,M_\varrho\) are supplied classical
approximation data. Finite total boundary rotation ensures a finite
uniform Faber bound \cite{BeckermannReichel2009Faber}; its value for
the chosen geometry must still be supplied. The degree thus depends on
both geometric data and \(\mathcal A_H\).
Proposition~\ref{cor:lchm-branch-block} implements the selector at its
declared normalization; the next subsection controls input and circuit
errors through that scale and the polynomial sensitivity.

\subsection{Input errors and precision}
\label{app:weyl-errors}

The preceding constructions choose a selector for the exact lift and
then implement it coherently. With approximate inputs, we keep that
polynomial fixed and compare its values at the exact and decoded matrices.
This avoids requiring the perturbed spectrum to satisfy the half-plane
separation condition. Both matrices are contractions, so coefficient telescoping
gives the sensitivity needed for this comparison.

\begin{lemma}[Polynomial input sensitivity]
\label{lem:lchm-input-sensitivity}
With the coefficients in \eqref{eq:lchm-weyl-lift}, define
\[
\mathcal L_{\chi,r}=\frac12\sum_{m=1}^{d_S}m|c_m|.
\]
If \(\|\widetilde H-\widehat H\|\le\varepsilon_{\widehat H}\)
and both matrices are contractions, then
\begin{equation}
\begin{aligned}
\|S_{\chi,r}(\widetilde H)-S_{\chi,r}(\widehat H)\|
&\le\mathcal L_{\chi,r}\varepsilon_{\widehat H},\\
\mathcal L_{\chi,r}
&\le\frac{\alpha_{S_{\chi,r}}}{2}
\sqrt{\frac{d_S(d_S+1)(2d_S+1)}6}
\le\frac{\alpha_{S_{\chi,r}}d_S^{3/2}}2.
\end{aligned}
\label{eq:lchm-input-sensitivity}
\end{equation}
\end{lemma}

\begin{proof}
Telescoping gives
\(\widetilde H^m-\widehat H^m
=\sum_{\ell=0}^{m-1}\widetilde H^\ell
(\widetilde H-\widehat H)\widehat H^{m-1-\ell}\),
whose norm is at most \(m\varepsilon_{\widehat H}\).
The nonconstant coefficients of \(S_{\chi,r}\) are \(c_m/2\).
Parseval's identity and \eqref{eq:lchm-selector-normalization} give
\(\sum_m|c_m|^2\le\alpha_{S_{\chi,r}}^2\).
Cauchy--Schwarz, with
\(\sum_{m=1}^{d_S}m^2=d_S(d_S+1)(2d_S+1)/6\le d_S^3\),
proves the remaining estimates.
\end{proof}

The sensitivity bound isolates the matrix-input error. To combine it
with the other errors, let \(\varepsilon_{\rm sel}\) bound the error of
the chosen selector at the exact matrix. The following result applies
regardless of which preceding approximation supplies that bound.

\begin{proposition}[Branch-block error with approximate inputs]
\label{prop:lchm-approximate-inputs}
Let \(\widehat H\) satisfy \eqref{eq:lchm-half-plane-gap}, and let the
selector in \eqref{eq:lchm-polynomial-selectors} satisfy
\(\|(S_{\chi,r}(\widehat H)-\Pi_\chi)R\|\le\varepsilon_{\rm sel}\).
Suppose \(U_{\widehat H}\) is a
\((1,a_H,\varepsilon_{\widehat H})\) block-encoding and
\(U_R\) is an \((\alpha_R,a_R,\varepsilon_R)\) block-encoding
with \(\|R\|\le\alpha_R\).
Assume the controlled access and coherent realization of
Proposition~\ref{cor:lchm-branch-block}.
Let \(\varepsilon_{\rm circ}\) bound the normalized block error of
the polynomial circuit relative to
\(S_{\chi,r}(\widetilde H)/\alpha_{S_{\chi,r}}\).
Then the composed block-encoding has normalization
\(\alpha_{S_{\chi,r}}\alpha_R\) and decoded error
\begin{equation}
\begin{aligned}
\varepsilon_{\rm branch}\le{}&
\varepsilon_{\rm sel}
+\alpha_R\mathcal L_{\chi,r}\varepsilon_{\widehat H}\\
&+\alpha_R\alpha_{S_{\chi,r}}\varepsilon_{\rm circ}
+\alpha_{S_{\chi,r}}\varepsilon_R.
\end{aligned}
\label{eq:lchm-branch-block-error}
\end{equation}
The factors \(\alpha_R\) in the second and third terms may be
replaced by \(\|R\|\).
The matrix-query count is unchanged.
\end{proposition}

\begin{proof}
Let \(B_\chi\) be the implemented normalized polynomial block and
\(\widetilde R\) the decoded block of \(U_R\).
The decoded output is
\(\alpha_{S_{\chi,r}}B_\chi\widetilde R\).
Subtracting the target gives
\[
\begin{aligned}
\alpha_{S_{\chi,r}}B_\chi\widetilde R-\Pi_\chi R
={}&\alpha_{S_{\chi,r}}B_\chi(\widetilde R-R)\\
&+\bigl(\alpha_{S_{\chi,r}}B_\chi
                  -S_{\chi,r}(\widetilde H)\bigr)R\\
&+\bigl(S_{\chi,r}(\widetilde H)
                  -S_{\chi,r}(\widehat H)\bigr)R\\
&+\bigl(S_{\chi,r}(\widehat H)-\Pi_\chi\bigr)R.
\end{aligned}
\]
Since \(B_\chi\) is a block of a unitary, \(\|B_\chi\|\le1\).
The circuit-error assumption, Lemma~\ref{lem:lchm-input-sensitivity},
and the assumed selector bound control the other three terms.
\end{proof}

For a contour-based approximation, take \(\varepsilon_{\rm sel}\)
from \eqref{eq:lchm-direct-branch-error} or
\eqref{eq:lchm-selector-budget}, with the Faber degree choice in
\eqref{eq:lchm-faber-degree} when applicable. Under the sector assumption,
use the selector estimate in the proof of
Theorem~\ref{thm:lchm-explicit-query}.
The same input-error analysis therefore applies to each choice of polynomial.

For a requested branch error \(\varepsilon>0\), one sufficient allocation is
\begin{equation}
\begin{aligned}
\varepsilon_{\rm sel}
&\le\frac{\varepsilon}{4},
&\varepsilon_{\widehat H}
&\le\frac{\varepsilon}{4\alpha_R\mathcal L_{\chi,r}},\\
\varepsilon_{\rm circ}
&\le\frac{\varepsilon}
{4\alpha_R\alpha_{S_{\chi,r}}},
&\varepsilon_R
&\le\frac{\varepsilon}{4\alpha_{S_{\chi,r}}}.
\end{aligned}
\label{eq:lchm-precision-allocation}
\end{equation}
A condition with zero multiplying factor is omitted.
Choose the degree to meet the selector budget, then evaluate the
normalization and input sensitivity before fixing the remaining precisions.
For the exact-input degree choice in
\eqref{eq:lchm-sector-degree}, replace its tolerance
\(\varepsilon\) by \(\varepsilon/2\) to obtain
selector error at most \(\varepsilon/4\).
For fixed input accuracy, increasing the degree need not reduce the
total error, because these two factors can grow.

Apply the circuit-error convention of
Appendix~\ref{app:contour-input-errors}, keeping the fixed decoded input
error \(\varepsilon_{\widehat H}\) separate from errors in executing
its unitary. The realization in Appendix~\ref{app:weyl-realization}
uses \(n_{\rm sig}=O(d_S)\) signal calls and
\(n_{\rm rot}=O(d_S\log(d_S+1))\) angle rotations.
For per-operation errors \(\delta_{\rm sig},\delta_{\rm rot}\),
the sum \(n_{\rm sig}\delta_{\rm sig}+n_{\rm rot}\delta_{\rm rot}\),
together with coefficient and phase-data errors, must fit within
\(\varepsilon_{\rm circ}\) in \eqref{eq:lchm-precision-allocation}.
Classical coefficient and phase computation, oracle and reflection costs,
and subsequent graph recovery remain separate from these resource counts.

\clearpage
\section{Differential Riccati constructions and direct recovery}
\label{app:dre}

We prove the graph identities used in Section~\ref{sec:dre}, construct
the four weighted blocks by finite-contour quadrature, and propagate
their errors through the recovery algorithm.
The final resource analysis separates the general implementation
scales from the normalization assumptions used in
Theorem~\ref{thm:dre-construction-main}.

\subsection{Hamiltonian flow and the decaying graph projector}
\label{app:dre-graph}

The graph-flow identity \cite[Chap.~4]{BittantiLaubWillems1991RiccatiEquation} below requires only existence of the solution;
it does not require spectral separation.

\begin{lemma}[Hamiltonian graph flow]
\label{lem:dre-graph-flow-app}
For the data in Definition~\ref{def:problem-dre}, write
\(\mathcal W(t)=e^{t\mathcal H_{\rm DRE}}R_0
=[\mathcal W_1(t);\mathcal W_2(t)]\).
On the promised existence interval,
\begin{equation}
\mathcal W(t)=\begin{bmatrix}\mathsf I\\P(t)\end{bmatrix}Y(t),
\qquad
\dot Y(t)=(A-GP(t))Y(t),\quad Y(0)=\mathsf I,
\label{eq:dre-flow-factorization-app}
\end{equation}
and \(Y(t)\) is invertible.
Consequently \(\mathcal W_1(t)\) is invertible and
\(P(t)=\mathcal W_2(t)\mathcal W_1(t)^{-1}\).
Conversely, on any interval containing zero where \(\mathcal W_1(t)\)
is invertible, this formula solves the DRE with initial value \(P_0\).
\end{lemma}
\begin{proof}
The coefficient \(A-GP(t)\) is continuous on the existence interval,
so its fundamental matrix \(Y(t)\) is invertible.
Using the DRE to differentiate the product gives
\[
\frac{d}{dt}(P(t)Y(t))
=(-Q-A^*P(t))Y(t).
\]
Thus \([\mathsf I;P(t)]Y(t)\) solves the same linear system and has
the same initial value as \(e^{t\mathcal H_{\rm DRE}}R_0\).
Uniqueness proves \eqref{eq:dre-flow-factorization-app}.
For the converse, the two blocks of the linear flow satisfy
\[
\dot{\mathcal W}_1=A\mathcal W_1-G\mathcal W_2,\qquad
\dot{\mathcal W}_2=-Q\mathcal W_1-A^*\mathcal W_2.
\]
Differentiating \(\mathcal W_2\mathcal W_1^{-1}\) gives
\(-Q-A^*P-PA+PGP\), and its initial value is \(P_0\).
\end{proof}

Existence supplies this qualitative invertibility, but does not imply
spectral separation or full rank of \(\Pi_+R_0\).
The next proof uses these two additional conditions separately.

\newtheorem*{dreseededrestatement}{Theorem~\ref{thm:dre-seeded-projector-main}}
\begin{dreseededrestatement}
For Problem~\ref{prob:quantum-dre}, suppose that \(\Pi_+R_0\)
has full column rank. At the requested time \(t\), set
\begin{equation*}
\mathcal E(t)=\Pi_+
+e^{t\mathcal H_{\rm DRE}}\Pi_-R_0
(\Pi_+R_0)^+e^{-t\mathcal H_{\rm DRE}}\Pi_+.
\end{equation*}
Then
\begin{equation*}
\begin{aligned}
\mathcal E(t)^2&=\mathcal E(t),\\
\operatorname{ran}\mathcal E(t)
&=\operatorname{ran}\begin{bmatrix}\mathsf I\\P(t)\end{bmatrix},
&\ker\mathcal E(t)&=\operatorname{ran}\Pi_-.
\end{aligned}
\end{equation*}
\end{dreseededrestatement}

\begin{proof}
For \(J=\begin{bmatrix}0&\mathsf I\\-\mathsf I&0\end{bmatrix}\),
Hermitian \(G,Q\) give
\(\mathcal H_{\rm DRE}^*J+J\mathcal H_{\rm DRE}=0\).
Hence \(\mathcal H_{\rm DRE}\) is similar to
\(-\mathcal H_{\rm DRE}^*\), and its spectrum is paired by
\(\lambda\mapsto-\overline\lambda\), with algebraic multiplicities.
In the absence of imaginary-axis eigenvalues both Riesz projectors
therefore have rank \(n\).
They commute with the flow, sum to the identity, and annihilate
one another.

Set
\[
L(t)=(\Pi_+R_0)^+e^{-t\mathcal H_{\rm DRE}}\Pi_+.
\]
The full-column-rank matrix \(\Pi_+R_0\) spans
\(\operatorname{ran}\Pi_+\). Its orthogonal projector
\((\Pi_+R_0)(\Pi_+R_0)^+\) therefore fixes the columns of
\(e^{-t\mathcal H_{\rm DRE}}\Pi_+\). Consequently
\begin{equation}
\begin{aligned}
\mathcal W(t)L(t)&=\mathcal E(t),\\
L(t)\mathcal W(t)&=(\Pi_+R_0)^+(\Pi_+R_0)=\mathsf I.
\end{aligned}
\label{eq:dre-seeded-factorization-app}
\end{equation}
For the first identity, split \(R_0\) into its two spectral branches:
the right-branch term is
\[
e^{t\mathcal H_{\rm DRE}}(\Pi_+R_0)(\Pi_+R_0)^+
e^{-t\mathcal H_{\rm DRE}}\Pi_+=\Pi_+.
\]
The remaining term is the second term of
\eqref{eq:dre-seeded-projector}.
The second identity follows by commuting \(\Pi_+\) through the flow.
Equation~\eqref{eq:dre-seeded-factorization-app} proves idempotence
and equality of the ranges of \(\mathcal E(t)\) and \(\mathcal W(t)\).
Furthermore \(\mathcal E(t)\Pi_-=0\) and
\(\Pi_+\mathcal E(t)=\Pi_+\), which imply
\(\ker\mathcal E(t)=\ker\Pi_+=\operatorname{ran}\Pi_-\).
Lemma~\ref{lem:dre-graph-flow-app} identifies the range with the
solution graph on the promised interval.
\end{proof}

The factorization proof is valid for arbitrary Jordan structure.
Its projector identities hold for every real \(t\); identification
with a finite Riccati matrix uses the existence interval.
In particular, \((\Pi_+R_0)(\Pi_+R_0)^+\) need not equal the spectral projector \(\Pi_+\),
and the final \(\Pi_+\) in \eqref{eq:dre-seeded-projector} is kept.

\begin{example}[Existence without full-rank initialization]
\label{ex:dre-initialization-scope-app}
Take \(A=0\), \(G=Q=1\), and \(P_0=1\).
Then \(P(t)=1\) solves \(\dot P=-1+P^2\) for every \(t\), and
\(\mathcal H_{\rm DRE}=\begin{bmatrix}0&-1\\-1&0\end{bmatrix}\)
has eigenvalues \(1,-1\).
However, \(R_0=[1;1]\) belongs to the negative eigenspace, so
\(\Pi_+R_0=0\).
Thus existence and spectral separation do not imply the
initialization condition of Theorem~\ref{thm:dre-construction-main}.
\end{example}

Lemma~\ref{lem:factor-graph-distance-app} identifies full rank of
\(\Pi_+R_0\) with
\(\operatorname{ran}R_0\cap\operatorname{ran}\Pi_-=\{0\}\).
Under the coefficient assumptions of
Proposition~\ref{prop:factor-dre-initialization-app},
\eqref{eq:factor-dre-initialization-app} gives an explicit positive
lower bound on \(\sigma_{\min}(\Pi_+R_0)\).
That proposition supplies an initialization bound, not a time-existence
result. Its conversion to a condition-number bound includes the factor
\(\|\Pi_+R_0\|\); it is not automatically applicable to the
sign-reversed terminal-LQR data in Appendix~\ref{app:lqr-conventions}.

\subsection{Recovery from a graph projector}
\label{app:dre-recovery}

The recovery identity depends only on the range of the projector.
We state it for a general graph so that the discrete construction can
use the same argument.

\begin{lemma}[Upper-block inverse bound]
\label{lem:dre-projector-block-recovery-app}
Let \(\mathcal E\) be an idempotent matrix with range
\(\operatorname{ran}[\mathsf I;P]\), and set
\(E_1=[\mathsf I;0]\), \(E_2=[0;\mathsf I]\). Then
\begin{equation}
\begin{aligned}
(E_1^*\mathcal E)\begin{bmatrix}\mathsf I\\P\end{bmatrix}
&=\mathsf I,& E_2^*\mathcal E&=P E_1^*\mathcal E,\\
P&=(E_2^*\mathcal E)(E_1^*\mathcal E)^+,&
\|(E_1^*\mathcal E)^+\|&\le\sqrt{1+\|P\|^2}.
\end{aligned}
\label{eq:dre-projector-inverse-bound-app}
\end{equation}
\end{lemma}
\begin{proof}
Write \(V=[\mathsf I;P]\). The range identity gives
\(\mathcal E=V E_1^*\mathcal E\), while idempotence gives
\(\mathcal EV=V\). Taking upper blocks yields
\((E_1^*\mathcal E)V=\mathsf I\). The upper block therefore has full
row rank, and multiplying the lower-block identity by its pseudoinverse
gives the formula for \(P\). Finally,
\[
(E_1^*\mathcal E)^+
=(E_1^*\mathcal E)^+(E_1^*\mathcal E)V.
\]
The first two factors form an orthogonal projector, so the norm of the
product is at most \(\|V\|=\sqrt{1+\|P\|^2}\).
\end{proof}

\newtheorem*{drerecoveryrestatement}{Lemma~\ref{lem:dre-projector-block-recovery}}
\begin{drerecoveryrestatement}
Under the assumptions of Theorem~\ref{thm:dre-seeded-projector-main},
\begin{equation*}
\begin{aligned}
(E_1^*\mathcal E(t))\begin{bmatrix}\mathsf I\\P(t)\end{bmatrix}
&=\mathsf I,\\
E_2^*\mathcal E(t)&=P(t)E_1^*\mathcal E(t),\\
P(t)&=(E_2^*\mathcal E(t))(E_1^*\mathcal E(t))^+.
\end{aligned}
\end{equation*}
\end{drerecoveryrestatement}

\begin{proof}
Theorem~\ref{thm:dre-seeded-projector-main} supplies an idempotent
matrix with the stated graph range. Apply
Lemma~\ref{lem:dre-projector-block-recovery-app} with \(P=P(t)\).
\end{proof}

\subsection{Weighted finite-contour quadrature and normalization}
\label{app:dre-quadrature}

The four operators in Table~\ref{tab:dre-weighted-blocks} are ordinary
weighted Riesz blocks for \(\widehat{\mathcal H}_{\rm DRE}\):
\begin{equation}
\begin{aligned}
\mathfrak R_+[1;\widehat{\mathcal H}_{\rm DRE},\mathsf I]
&=\Pi_+,\\
\mathfrak R_+[1;\widehat{\mathcal H}_{\rm DRE},\mathsf I]R_0
&=\Pi_+R_0,\\
\mathfrak R_-[e^{\alpha_Htz};\widehat{\mathcal H}_{\rm DRE},\mathsf I]R_0
&=e^{t\mathcal H_{\rm DRE}}\Pi_-R_0,\\
\mathfrak R_+[e^{-\alpha_Htz};\widehat{\mathcal H}_{\rm DRE},\mathsf I]
&=e^{-t\mathcal H_{\rm DRE}}\Pi_+.
\end{aligned}
\label{eq:dre-weighted-riesz-blocks-app}
\end{equation}
These identities follow from Proposition~\ref{prop:riesz-calculus-app}
and hold without a diagonalization. For each corresponding choice of
\((\sigma,g,R)\), write
\[
\mathcal T_{\sigma,g,R}
=\frac1{2\pi\mathrm i}\int_{\Gamma_\sigma}
g(z)(z\mathsf I-\widehat{\mathcal H}_{\rm DRE})^{-1}R\,\mathrm dz.
\]
Here \(R=\mathsf I\) has normalization \(\alpha_R=1\), and
\(R=R_0\) has supplied normalization \(\alpha_R=\alpha_{R_0}\).
The supplied Hamiltonian encoding also encodes
\(\widehat{\mathcal H}_{\rm DRE}\) at normalization one.

For geometric weights \(\nu_j\), use the convention of
\eqref{eq:weighted-contour-coefficients}:
\begin{equation}
\begin{aligned}
\omega_j&=\frac{\nu_jg(z_j)}{2\pi\mathrm i},\qquad
\alpha_{g,R}=\alpha_R\sum_j|\omega_j|\beta_j,\\
S_{\sigma,g,R}&=\sum_j\omega_j
 (z_j\mathsf I-\widehat{\mathcal H}_{\rm DRE})^{-1}R,\\
\beta_j&=\Theta\!\left(
 \|(z_j\mathsf I-\widehat{\mathcal H}_{\rm DRE})^{-1}\|\right).
\end{aligned}
\label{eq:dre-weighted-contour-sum-app}
\end{equation}
Thus \(\omega_j\) already includes the scalar weight.
Under the exact-input and coherent-access assumptions of
Proposition~\ref{thm:local-contour-sum}, including uniform supplied
constant-factor node estimates, this sum has a block-encoding with
normalization \(\alpha_{g,R}\) and decoded implementation error
\(0<\varepsilon_{\rm impl}<\alpha_{g,R}\), using
\begin{equation}
O\!\left(\mathcal R_{\rm DRE}
\log\!\left(e+\frac{\mathcal R_{\rm DRE}\alpha_{g,R}}
{\varepsilon_{\rm impl}}\right)\right)
\label{eq:dre-weighted-block-query-app}
\end{equation}
Hamiltonian queries, together with one call to the encoding of \(R\).
For \(R=\mathsf I\), that last call is trivial; the two blocks with
\(R=R_0\) each use its encoding once. If the quadrature error is
\(\varepsilon_{\rm quad}\), the complete decoded target error is at
most \(\varepsilon_{\rm quad}+\varepsilon_{\rm impl}\).

The following rule gives an absolute error bound uniform in forward
time by using the decaying weight on each branch.

\begin{proposition}[Paired-rectangle weighted Gauss--Legendre construction]
\label{prop:dre-weighted-quadrature-app}
Let the exact input encodings and coherent access be as above. Supply
\begin{equation}
0<\underline\Delta_{\rm ax}\le
\inf_{\omega\in\mathbb R}
\sigma_{\min}(\mathrm i\omega\mathsf I-\mathcal H_{\rm DRE}),
\qquad \delta=\frac{\underline\Delta_{\rm ax}}{\alpha_H}\in(0,1].
\label{eq:dre-quadrature-gap-app}
\end{equation}
For \(\widehat{\mathcal H}_{\rm DRE}\), choose the positively
oriented rectangle boundaries
\begin{equation}
\begin{aligned}
\Gamma_+&=\partial\{x+\mathrm iy:\delta/2\le x\le2,
\ |y|\le2\},\\
\Gamma_-&=\partial\{x+\mathrm iy:-2\le x\le-\delta/2,
\ |y|\le2\},\qquad
\ell(\Gamma_\sigma)=12-\delta.
\end{aligned}
\label{eq:dre-quadrature-rectangles-app}
\end{equation}
Divide each side of length \(L\) into \(\lceil4L/\delta\rceil\)
equal panels. Write an oriented panel as
\(z_\ell(s)=c_\ell+h_\ell s\), \(-1\le s\le1\), and use the
\(p\)-point Gauss--Legendre nodes \(s_j,w_j\) on \([-1,1]\).
Its nodes and weights are
\begin{equation}
z_{\ell j}=c_\ell+h_\ell s_j,\qquad
\nu_{\ell j}=h_\ell w_j,\qquad
\omega_{\ell j}=\frac{h_\ell w_jg(z_{\ell j})}{2\pi\mathrm i}.
\label{eq:dre-quadrature-nodes-app}
\end{equation}
For every target in \eqref{eq:dre-weighted-riesz-blocks-app}, its
finite sum in \eqref{eq:dre-weighted-contour-sum-app} satisfies
\begin{equation}
\|S_{\sigma,g,R}-\mathcal T_{\sigma,g,R}\|
\le\frac{8\ell(\Gamma_\sigma)\alpha_R}{\pi\delta}\,4^{-p},
\qquad t\ge0.
\label{eq:dre-weighted-quadrature-error-app}
\end{equation}
For a target tolerance \(\varepsilon_T>0\), choose
\begin{equation}
p=\left\lceil
\frac{\log\!\left(e+16\ell(\Gamma_\sigma)\alpha_R/
(\pi\delta\varepsilon_T)\right)}{\log4}
\right\rceil.
\label{eq:dre-weighted-quadrature-order-app}
\end{equation}
Then the quadrature error is at most \(\varepsilon_T/2\).
The exact number of panels on either rectangle and the node count
for one target are
\begin{equation}
\begin{aligned}
N_{\rm pan}
&=2\left\lceil\frac{16}{\delta}\right\rceil
 +2\left\lceil\frac{4(2-\delta/2)}{\delta}\right\rceil
 \le\frac{4\ell(\Gamma_\sigma)}{\delta}+4,\\
m_T&=N_{\rm pan}p
=O\!\left(\delta^{-1}
\log\!\left(e+\frac{\alpha_R}{\delta\varepsilon_T}\right)\right).
\end{aligned}
\label{eq:dre-weighted-quadrature-count-app}
\end{equation}
For these contours and rules,
\begin{equation}
\begin{aligned}
\mathcal R_{\rm DRE}
&\le\frac{2(1+2\sqrt2)}{\delta},\\
\alpha_{g,R}
&\le\frac{6c_\beta}{\pi}\alpha_R\mathcal R_{\rm DRE},&
\|\mathcal T_{\sigma,g,R}\|
&\le\frac6\pi\alpha_R\mathcal R_{\rm DRE},
\end{aligned}
\label{eq:dre-weighted-quadrature-scales-app}
\end{equation}
where \(c_\beta\) is a uniform constant satisfying
\(\beta_{\ell j}\le c_\beta
\|(z_{\ell j}\mathsf I-\widehat{\mathcal H}_{\rm DRE})^{-1}\|\).
In particular, the four declared normalizations can be chosen with
\begin{equation}
\begin{aligned}
\alpha_{\Pi_+},\quad
\alpha_{e^{-t\mathcal H_{\rm DRE}}\Pi_+}
&=O(\mathcal R_{\rm DRE}),\\
\alpha_{\Pi_+R_0},\quad
\alpha_{e^{t\mathcal H_{\rm DRE}}\Pi_-R_0}
&=O(\alpha_{R_0}\mathcal R_{\rm DRE}).
\end{aligned}
\label{eq:dre-four-block-normalizations-app}
\end{equation}
These bounds and the absolute quadrature error are uniform in
\(t\ge0\) and the number of nodes.
\end{proposition}
\begin{proof}
Theorem~\ref{thm:factor-axis-app}, applied to
\(\widehat{\mathcal H}_{\rm DRE}\) with supplied axis bound
\(\delta\) and inner-edge distance \(\delta/2\), gives the
required spectral enclosure, contour inverse bound \(2/\delta\),
and first inequality in \eqref{eq:dre-weighted-quadrature-scales-app}.
Also \(1+|z|\ge\|z\mathsf I-\widehat{\mathcal H}_{\rm DRE}\|\)
gives \(\mathcal R_{\rm DRE}\ge1\).

Fix a panel midpoint \(c_\ell\). Lemma~\ref{lem:inverse-input-perturbation},
with \(T=c_\ell\mathsf I-\widehat{\mathcal H}_{\rm DRE}\)
and \(\Delta T=(z-c_\ell)\mathsf I\), gives
\[
\|(z\mathsf I-\widehat{\mathcal H}_{\rm DRE})^{-1}\|
\le4/\delta,\qquad |z-c_\ell|\le\delta/4,
\]
since \(\|T^{-1}\|\le2/\delta\).
This closed disk remains within the corresponding open half-plane:
its real part is at least \(\delta/4\) on the right and at most
\(-\delta/4\) on the left. Consequently \(|g(z)|\le1\) throughout
the disk for every permitted weight and every \(t\ge0\).
The full integrand
\(g(z)(z\mathsf I-\widehat{\mathcal H}_{\rm DRE})^{-1}R\)
is analytic on a neighborhood of the disk and has norm at most
\(4\alpha_R/\delta\).

Expand this integrand at \(c_\ell\) through degree \(2p-1\).
Cauchy's coefficient bound~\cite[Chap.~I, Sec.~1.7]{Kato1995Perturbation}, applied to unit-vector matrix elements,
also holds in operator norm. The panel half-length satisfies
\(|h_\ell|\le\delta/8\), so the ratio to the disk radius is at
most \(1/2\). Hence the Taylor remainder on the panel has norm at most
\[
\frac{4\alpha_R}{\delta}
\sum_{k=2p}^{\infty}2^{-k}
=\frac{8\alpha_R}{\delta}\,4^{-p}.
\]
The positive Gauss--Legendre weights satisfy \(\sum_jw_j=2\)
and integrate every polynomial of degree at most \(2p-1\) exactly
\cite{NISTDLMFQuadrature}. Only the remainder contributes to the integral-minus-sum error. Its contribution
on this oriented panel, including the factor \(1/(2\pi\mathrm i)\),
is at most
\[
\frac{|h_\ell|}{2\pi}
\left(2+\sum_jw_j\right)
\frac{8\alpha_R}{\delta}\,4^{-p}
=\frac{16|h_\ell|\alpha_R}{\pi\delta}\,4^{-p}.
\]
Summing with \(\sum_\ell2|h_\ell|=\ell(\Gamma_\sigma)\)
proves \eqref{eq:dre-weighted-quadrature-error-app} for the entire
matrix target. The displayed choice of \(p\) yields its half-budget
bound. The two vertical sides have length four and the two horizontal
sides have length \(2-\delta/2\), which proves the exact panel count.
Summing their ceiling bounds and using \(11\le\ell(\Gamma_\sigma)<12\)
gives \eqref{eq:dre-weighted-quadrature-count-app}.

For the normalization, at every node,
\(\beta_{\ell j}\le c_\beta\mathcal R_{\rm DRE}/(1+|z_{\ell j}|)\),
and the geometric coefficients satisfy
\[
\sum_{\ell,j}\frac{|h_\ell|w_j}{2\pi}
=\frac{\ell(\Gamma_\sigma)}{2\pi}\le\frac6\pi.
\]
Using \(|g(z_{\ell j})|\le1\) in
\eqref{eq:dre-weighted-contour-sum-app} proves the bound on
\(\alpha_{g,R}\). Bounding the ideal integral by the same
length estimate proves the target-norm bound. Taking
\(\alpha_R=1\) or \(\alpha_{R_0}\) gives
\eqref{eq:dre-four-block-normalizations-app}. All estimates of the
weight used only \(t\ge0\), so the stated uniformity follows.
\end{proof}

With exact inputs, allocate the remaining \(\varepsilon_T/2\) to
the finite-sum implementation in
\eqref{eq:dre-weighted-block-query-app}. If
\(\alpha_{g,R}\le\varepsilon_T/2\), the zero-block shortcut in
Appendix~\ref{app:contour-proofs} already approximates the finite sum
within that budget. For approximate inputs,
Proposition~\ref{prop:contour-input-errors-app} supplies an additional
error on the fixed ideal contours; choose quadrature, input, and
implementation budgets whose sum is at most \(\varepsilon_T\).
The half-budget choice above then changes according to
\eqref{eq:dre-weighted-quadrature-error-app}.

The uniform absolute error does not remove the cost of preparing
time-dependent coefficients or performing arithmetic with \(\alpha_Ht\).
Node and weight preparation and the coherent reflection circuits
have separate gate costs. The node count in
\eqref{eq:dre-weighted-quadrature-count-app} does not multiply the
Hamiltonian query bound under the stated coherent access model.
The supplied imaginary-axis lower bound also need not provide the
required constant-factor inverse-norm estimate at every node.
Finally, \eqref{eq:dre-four-block-normalizations-app} controls these
four input blocks for the specified rules. The initial-column
pseudoinverse, products, sums, and recompression still determine the
DGP normalization \(\alpha_{\mathcal E}\) and the
output normalization \(\alpha_{P(t)}\), as accounted for in
Appendix~\ref{app:dre-complexity}.

\subsection{Stability and query complexity}
\label{app:dre-complexity}

There are two pseudoinverses in the construction. The first recovers
the initial branch coordinates, and the second recovers the solution
from its graph projector. We first bound the second operation, then
substitute the cost and accuracy of constructing the projector.

At the requested time, the graph projector is
\begin{equation}
\mathcal E=\Pi_+
+e^{t\mathcal H_{\rm DRE}}\Pi_-R_0(\Pi_+R_0)^+
e^{-t\mathcal H_{\rm DRE}}\Pi_+.
\label{eq:dre-local-blocks-app}
\end{equation}
Supply \(0<\gamma_+\le\sigma_{\min}(\Pi_+R_0)\) and a solution bound
\(\sup_{0\le s\le T}\|P(s)\|\le M\). If the encoding of
\(\mathcal E\) has normalization \(\alpha_{\mathcal E}\), choose
\begin{equation}
\alpha_{P(t)}\ge8\alpha_{\mathcal E}\sqrt{1+M^2}.
\label{eq:dre-supplied-output-scale-app}
\end{equation}
Equality is the default choice; a larger supplied output scale is
also valid. Choose a supplied normalization
\(\alpha_{\mathcal E}\ge\|\mathcal E\|\).
Since \(\mathcal E\) is a nonzero projector, \(\alpha_{\mathcal E}\ge1\).

\begin{proposition}[Stability and cost of direct recovery]
\label{prop:dre-direct-projector-stability-app}
Suppose an encoding of \(\mathcal E(t)\) with normalization
\(\alpha_{\mathcal E}\) and decoded error at most \(\delta_E\)
costs \(Q_{\mathcal E}(\delta_E)\) matrix queries per call.
For \(0<\varepsilon\le1\), it suffices to use
\begin{equation}
\delta_E\le\frac{\varepsilon}{8(1+M^2)},\qquad
\varepsilon_{\rm inv}=\frac{\varepsilon}{4(1+\alpha_{\mathcal E})},
\label{eq:dre-direct-projector-budget-app}
\end{equation}
with decoded product error at most \(\varepsilon/4\).
The upper-block pseudoinverse threshold
\(4\alpha_{\mathcal E}/\alpha_{P(t)}\) then returns an
\((\alpha_{P(t)},a_{P(t)},\varepsilon)\) block-encoding of \(P(t)\), using
\begin{equation}
O\!\left(Q_{\mathcal E}(\delta_E)\alpha_{P(t)}
\log\!\left(e+\frac{\alpha_{P(t)}}{\varepsilon}\right)\right)
\label{eq:dre-direct-projector-local-query-app}
\end{equation}
matrix queries. The same statement holds for any graph projector
satisfying Lemma~\ref{lem:dre-projector-block-recovery-app}.
\end{proposition}
\begin{proof}
Let \(\widetilde{\mathcal E}=\mathcal E+\Delta\mathcal E\) be the
decoded matrix. The upper block satisfies
\[
\|(E_1^*\Delta\mathcal E)(E_1^*\mathcal E)^+\|
\le\delta_E\sqrt{1+M^2}\le\tfrac18.
\]
Hence \((E_1^*\widetilde{\mathcal E})(E_1^*\mathcal E)^+\)
is invertible. It provides a right inverse of the perturbed upper block
after multiplication by its inverse, and the minimum-norm right inverse
therefore obeys
\begin{equation}
\|(E_1^*\widetilde{\mathcal E})^+\|
\le\frac{\sqrt{1+M^2}}{1-\delta_E\sqrt{1+M^2}}
\le2\sqrt{1+M^2}.
\label{eq:dre-direct-perturbed-inverse-app}
\end{equation}
In particular, the upper block remains full row rank. Using its right
inverse and the exact lower-block identity gives
\begin{equation}
(E_2^*\widetilde{\mathcal E})(E_1^*\widetilde{\mathcal E})^+-P(t)
=\begin{bmatrix}-P(t)&\mathsf I\end{bmatrix}
\Delta\mathcal E\,(E_1^*\widetilde{\mathcal E})^+.
\label{eq:dre-direct-projector-error-identity-app}
\end{equation}
The error is at most \(2(1+M^2)\delta_E\le\varepsilon/4\).
This argument uses the shared matrix perturbation of the two blocks;
the approximate matrix need not be idempotent.

The inverse implementation error contributes at most
\(\|E_2^*\widetilde{\mathcal E}\|\varepsilon_{\rm inv}
\le\alpha_{\mathcal E}\varepsilon_{\rm inv}\le\varepsilon/4\).
Adding the product error proves the required accuracy.
Equation~\eqref{eq:dre-direct-perturbed-inverse-app} permits the
threshold \(4\alpha_{\mathcal E}/\alpha_{P(t)}\), whose inverse
normalization under Proposition~\ref{prop:local-inverse} is
\(\alpha_{P(t)}/\alpha_{\mathcal E}\). That proposition uses
\(O(\alpha_{P(t)}\log(e+\alpha_{P(t)}/\varepsilon))\) calls
to the upper-block encoding. Multiplication adds one lower-block call,
which proves \eqref{eq:dre-direct-projector-local-query-app}.
\end{proof}

For the decaying graph projector, use the exact inputs in
Problem~\ref{prob:quantum-dre} and the coherent access assumptions of
Section~\ref{sec:quantum-tools}. Each weighted target
\[
T\in\bigl\{\Pi_+,\ \Pi_+R_0,
e^{t\mathcal H_{\rm DRE}}\Pi_-R_0,
e^{-t\mathcal H_{\rm DRE}}\Pi_+\bigr\}
\]
has quadrature normalization \(\alpha_T^{\rm quad}\) and actual
normalization \(\alpha_T\).
Any reduction between them has a declared query overhead \(r_T\ge1\);
define \(r_{\mathcal E}\) similarly for the assembled projector.
Use overhead one when no reduction is made. Fix accuracy-independent
normalization bounds before choosing the internal precision, padding
the encodings to these declared scales when needed. The scales and
overheads are supplied uniformly over the required accuracies. The compact
theorem additionally assumes
\begin{equation}
\begin{gathered}
\gamma_+=\Theta(\sigma_{\min}(\Pi_+R_0)),\qquad
\alpha_{\Pi_+R_0}=\Theta(\|\Pi_+R_0\|),\\
r_{\mathcal E},\ \max_T r_T=O(1).
\end{gathered}
\label{eq:dre-direct-calibration-app}
\end{equation}
These are separate from the quadrature bounds; \(M\) need not be a
constant because it remains in the output normalization.

\begin{proposition}[DRE stability and general query bound]
\label{prop:dre-general-resources-app}
Under the preceding access and scale assumptions, the common internal
accuracy
\begin{equation}
\begin{aligned}
\eta={}&\frac{c\varepsilon}
{(1+M^2)(1+\alpha_{e^{t\mathcal H_{\rm DRE}}\Pi_-R_0})
(1+\alpha_{e^{-t\mathcal H_{\rm DRE}}\Pi_+})}\\
&\times\frac{1}{(1+\gamma_+^{-1})^2
(1+\alpha_{\Pi_+R_0}+\alpha_{\mathcal E})}
\end{aligned}
\label{eq:dre-outer-budget-app}
\end{equation}
is sufficient for a small absolute constant \(c>0\). Use this decoded
accuracy for the four weighted blocks, both pseudoinverse implementations,
and the combined multiplication and re-encoding error at each stage.
The initial pseudoinverse threshold is \(\gamma_+/2\), and direct
sums and products give
\begin{equation}
\alpha_{\mathcal E}^{\rm raw}
=\alpha_{\Pi_+}+\frac{8\alpha_{e^{t\mathcal H_{\rm DRE}}\Pi_-R_0}\alpha_{e^{-t\mathcal H_{\rm DRE}}\Pi_+}}{\gamma_+}.
\label{eq:dre-raw-normalizations-app}
\end{equation}
The final encoding has decoded error at most \(\varepsilon\).
Including the declared normalization reductions, its matrix-query cost is
\begin{equation}
\begin{aligned}
Q_{P(t)}=O\!\Bigg(&\mathcal R_{\rm DRE}
\frac{\alpha_{\Pi_+R_0}}{\gamma_+}\alpha_{P(t)}
\,r_{\mathcal E}\max_T r_T\\
&\times\log^3\!\left(e+
\frac{\mathcal R_{\rm DRE}(1+\alpha_{P(t)})(1+\gamma_+^{-1})
\bigl(1+\max_T\alpha_T^{\rm quad}\bigr)}{\eta}\right)\Bigg).
\end{aligned}
\label{eq:dre-general-query-expanded-app}
\end{equation}
For the rectangular rule in
Proposition~\ref{prop:dre-weighted-quadrature-app}, with normalized
gap \(\delta\), a common sufficient node count for each target is
\begin{equation}
O\!\left(\delta^{-1}\log\!\left(e+
\frac{\alpha_{R_0}}{\delta\eta}\right)\right).
\label{eq:dre-final-node-counts-app}
\end{equation}
The count does not multiply the coherent matrix-query bound.
\end{proposition}
\begin{proof}
Let \(\widetilde{\Pi_+R_0}\) be the decoded approximation to
\(\Pi_+R_0\), with error at most \(\eta\).
The chosen accuracy gives \(\eta\le\gamma_+/2\), so this approximation
remains full column rank. The identity
\[
\begin{aligned}
(\widetilde{\Pi_+R_0})^+-(\Pi_+R_0)^+
={}&-(\widetilde{\Pi_+R_0})^+
(\widetilde{\Pi_+R_0}-\Pi_+R_0)(\Pi_+R_0)^+\\
&+\bigl((\widetilde{\Pi_+R_0})^*
\widetilde{\Pi_+R_0}\bigr)^{-1}
(\widetilde{\Pi_+R_0}-\Pi_+R_0)^*\\
&\qquad\times\bigl(\mathsf I-(\Pi_+R_0)(\Pi_+R_0)^+\bigr)
\end{aligned}
\]
implies
\[
\|(\widetilde{\Pi_+R_0})^+-(\Pi_+R_0)^+\|
\le6\eta/\gamma_+^2.
\]
Expanding the product defining \(\mathcal E\), including the inverse
implementation and assembly errors, therefore gives
\begin{equation}
\begin{aligned}
\|\widetilde{\mathcal E}-\mathcal E\|
&=O\!\left((1+\alpha_{e^{t\mathcal H_{\rm DRE}}\Pi_-R_0})
(1+\alpha_{e^{-t\mathcal H_{\rm DRE}}\Pi_+})
(1+\gamma_+^{-1})^2\eta\right)\\
&\le\frac{\varepsilon}{8(1+M^2)}.
\end{aligned}
\label{eq:dre-projector-product-error-app}
\end{equation}
The same choice ensures
\(\eta\le\varepsilon/[4(1+\alpha_{\mathcal E})]\),
\(\alpha_{\Pi_+R_0}\eta\le1\), and \(\alpha_{\mathcal E}\eta\le1\).
Thus the two pseudoinverse implementations meet the precision conditions
of Proposition~\ref{prop:local-inverse}, and
Proposition~\ref{prop:dre-direct-projector-stability-app} gives the
final error. No successive list of smaller error budgets is needed.

A weighted block costs
\(O(\mathcal R_{\rm DRE}
\log(e+\mathcal R_{\rm DRE}\alpha_T^{\rm quad}/\eta))\)
before re-encoding. The initial pseudoinverse uses
\(O((\alpha_{\Pi_+R_0}/\gamma_+)
\log(e+\gamma_+^{-1}/\eta))\) calls to its input, and the final
pseudoinverse uses
\(O(\alpha_{P(t)}\log(e+\alpha_{P(t)}/\eta))\) calls to the
DGP encoding. Their product, with the stated overheads,
is bounded by \eqref{eq:dre-general-query-expanded-app}.
The node inverses and the two pseudoinverses supply three logarithms.
Each is bounded by the displayed common logarithm, so their product
is bounded by its third power. The rectangle node count follows by substituting
\(\eta\) in \eqref{eq:dre-weighted-quadrature-count-app} and using
\(\alpha_{R_0}\ge\|R_0\|\ge1\).
\end{proof}

\begin{corollary}[Output normalization for the rectangular rule]
\label{cor:dre-construction-normalization-app}
For the rule of Proposition~\ref{prop:dre-weighted-quadrature-app},
with \(\gamma_+=\Theta(\sigma_{\min}(\Pi_+R_0))\), direct assembly
and the default output scale allow
\begin{equation}
\alpha_{P(t)}
=O\!\left(\sqrt{1+M^2}\left(
\mathcal R_{\rm DRE}
+\frac{\alpha_{R_0}\mathcal R_{\rm DRE}^{\,2}}
{\sigma_{\min}(\Pi_+R_0)}\right)\right).
\label{eq:dre-direct-output-general-bound-app}
\end{equation}
If \(M,\alpha_{R_0}=O(1)\), this reduces to
\eqref{eq:dre-output-normalization-bound}.
\end{corollary}
\begin{proof}
The specified quadrature gives
\[
\begin{aligned}
\alpha_{\Pi_+},\quad\alpha_{e^{-t\mathcal H_{\rm DRE}}\Pi_+}
&=O(\mathcal R_{\rm DRE}),\\
\alpha_{e^{t\mathcal H_{\rm DRE}}\Pi_-R_0}
&=O(\alpha_{R_0}\mathcal R_{\rm DRE}).
\end{aligned}
\]
Substitute in \eqref{eq:dre-raw-normalizations-app} and then in
\eqref{eq:dre-supplied-output-scale-app}, taking equality.
When \(M,\alpha_{R_0}=O(1)\), the same rule gives
\(\sigma_{\min}(\Pi_+R_0)\le\|\Pi_+R_0\|
=O(\mathcal R_{\rm DRE})\), so the second term absorbs the first.
The constants are uniform in time, quadrature order, and precision
under the supplied estimates. These conclusions concern this specified
rule; any later normalization reduction has its own query cost.
\end{proof}

\newtheorem*{dreproblemrestatement}{Problem~\ref{prob:quantum-dre}}
\begin{dreproblemrestatement}
For the data and existence interval in Definition~\ref{def:problem-dre},
let \(t\in[0,T]\) be the requested time.
Given exact block-encodings of \(\mathcal H_{\rm DRE}\) and \(R_0\)
in \eqref{eq:dre-lift-seed}, their adjoints and controlled versions,
with normalizations \(\alpha_H\) and \(\alpha_{R_0}\), and a target
\(0<\varepsilon\le1\), construct a block-encoding with decoded matrix
\(\widetilde P(t)\) satisfying
\(\|\widetilde P(t)-P(t)\|\le\varepsilon\).
\end{dreproblemrestatement}

\newtheorem*{dremainrestatement}{Theorem~\ref{thm:dre-construction-main}}
\begin{dremainrestatement}
For Problem~\ref{prob:quantum-dre}, suppose that
\(\Pi_+R_0\) has full column rank, and supply a bound
\(\sup_{0\le s\le T}\|P(s)\|\le M\). Choose positively oriented contours
\(\Gamma_\pm\) enclosing exactly the corresponding spectral branches
of \(\widehat{\mathcal H}_{\rm DRE}\), and define
\begin{equation*}
\mathcal R_{\rm DRE}
=\max_{\sigma\in\{-,+\}}\sup_{z\in\Gamma_\sigma}
(1+|z|)\|(z\mathsf I-\widehat{\mathcal H}_{\rm DRE})^{-1}\|.
\end{equation*}
Under the implementation and normalization assumptions of
Appendix~\ref{app:dre-complexity}, the algorithm returns an
\((\alpha_{P(t)},a_{P(t)},\varepsilon)\) block-encoding of \(P(t)\), using
\begin{equation*}
Q_{\rm DRE}
=\widetilde O\!\left(
\mathcal R_{\rm DRE}\,
\kappa(\Pi_+R_0)\,
\alpha_{P(t)}
\right)
\end{equation*}
queries to the supplied matrix oracles.
Here \(\alpha_{P(t)}\) is the output normalization specified in that
appendix, and \(\widetilde O\) suppresses logarithmic dependence on
precision and the supplied scales.
\end{dremainrestatement}

\begin{proof}
Theorem~\ref{thm:dre-seeded-projector-main} and
Lemma~\ref{lem:dre-projector-block-recovery} identify the final product
with \(P(t)\). Proposition~\ref{prop:dre-general-resources-app} supplies
its encoding and accuracy. Under \eqref{eq:dre-direct-calibration-app},
\(\alpha_{\Pi_+R_0}/\gamma_+=O(\kappa(\Pi_+R_0))\) and the re-encoding
overheads are constant, so \eqref{eq:dre-general-query-expanded-app}
gives the claimed bound. Coefficient preparation, node arithmetic, and
gate synthesis have separate costs.
\end{proof}

\begin{proposition}[DRE construction from approximate inputs]
\label{prop:dre-approximate-inputs-app}
Under the fixed ideal data of Proposition~\ref{prop:dre-general-resources-app}
and the nodewise perturbation conditions of
Proposition~\ref{prop:contour-input-errors-app}, choose \(\eta\) by
\eqref{eq:dre-outer-budget-app}. If each weighted block has total decoded
error at most \(\eta\), including input, quadrature, and implementation
errors, the same construction returns \(P(t)\) to error \(\varepsilon\)
with the actual scales and re-encoding costs.
\end{proposition}
\begin{proof}
Proposition~\ref{prop:contour-input-errors-app} supplies the four input
errors required by Proposition~\ref{prop:dre-general-resources-app}.
Apply the latter with \(\varepsilon_{\widehat H}=\varepsilon_H/\alpha_H\),
including weight preparation in the implementation error as in
Appendix~\ref{app:contour-input-errors}.
\end{proof}

\clearpage
\section{Finite Riccati recursions and direct projector-block recovery}
\label{app:rr}

We use the data and forward lift of Subsection~\ref{sec:rr}.
Thus \(G,Q,P_0\succeq0\), \(A\) is invertible, and
\(\mathcal S_F\) has no unit-circle eigenvalues.
The projectors \(\Pi_<,\Pi_>\) refer to this lift.
We prove the discrete graph identities and power-weighted quadrature,
then apply the projector-block recovery of
Appendix~\ref{app:dre-recovery}.

\subsection{Linear-fractional propagation and the decaying graph projector}
\label{app:rr-graph}

\begin{lemma}[Forward propagation of the Riccati graph]
\label{lem:rr-graph-propagation-app}
For the data above, every iterate in \eqref{eq:rr-section-equation}
is well defined and positive semidefinite. The lift \(\mathcal S_F\)
is invertible, its interior and exterior spectral subspaces each
have dimension \(n\), and for every integer \(j\ge0\),
\begin{equation}
\begin{aligned}
\mathcal S_F^jR_0
&=\begin{bmatrix}\mathsf I_n\\P_j\end{bmatrix}
 E_1^*\mathcal S_F^jR_0,\\
P_j&=(E_2^*\mathcal S_F^jR_0)
 (E_1^*\mathcal S_F^jR_0)^{-1}.
\end{aligned}
\label{eq:rr-graph-blocks-app}
\end{equation}
In particular, the leading block in this formula is invertible.
The solution obeys
\begin{equation}
\|P_k\|\le\|A\|^{2k}\|P_0\|
+\|Q\|\sum_{\ell=0}^{k-1}\|A\|^{2\ell},
\label{eq:rr-solution-norm-app}
\end{equation}
with the empty sum equal to zero when \(k=0\).
\end{lemma}
\begin{proof}
Proposition~\ref{prop:lqr-conventions-app} gives existence and
positive semidefiniteness of every iterate. Its identity
\eqref{eq:lqr-psd-identity-app} also gives
\(0\preceq P(\mathsf I+GP)^{-1}\preceq P\) for \(P,G\succeq0\).
Use the algebraic identity \eqref{eq:factor-rr-graph-update-app}
with \(P=P_j\); it requires only the inverse just established,
not the coefficient bounds used elsewhere in that lemma.
Starting with \(E_1^*R_0=\mathsf I_n\), induction gives
\[
E_1^*\mathcal S_F^{j+1}R_0
=A^{-1}(\mathsf I_n+GP_j)E_1^*\mathcal S_F^jR_0,
\]
whose factors are invertible, and proves
\eqref{eq:rr-graph-blocks-app}.

To verify the spectral split, factor the lift as
\[
\mathcal S_F
=\begin{bmatrix}\mathsf I&0\\Q&\mathsf I\end{bmatrix}
 \begin{bmatrix}A^{-1}&0\\0&A^*\end{bmatrix}
 \begin{bmatrix}\mathsf I&G\\0&\mathsf I\end{bmatrix}.
\]
Every factor is invertible and preserves the matrix
\(\begin{bmatrix}0&\mathsf I\\-\mathsf I&0\end{bmatrix}\)
under congruence by its adjoint. Consequently
\begin{equation}
\mathcal S_F^*
\begin{bmatrix}0&\mathsf I\\-\mathsf I&0\end{bmatrix}
\mathcal S_F
=\begin{bmatrix}0&\mathsf I\\-\mathsf I&0\end{bmatrix}.
\label{eq:rr-symplectic-identity-app}
\end{equation}
Hence \(\mathcal S_F^{-1}\) is similar to \(\mathcal S_F^*\).
Eigenvalues occur in reciprocal-conjugate pairs, with algebraic
multiplicity \cite[Lemma~2]{Poloni2020RiccatiIterations}. Absence of unit-circle eigenvalues leaves exactly
\(n\) eigenvalues on each side, counted with multiplicity, and
therefore rank \(n\) for each Riesz projector.
This argument does not require diagonalizability.

Finally, the positive semidefinite inequality above gives
\(\|P_{j+1}\|\le\|Q\|+\|A\|^2\|P_j\|\).
Iterating this scalar inequality proves
\eqref{eq:rr-solution-norm-app}.
\end{proof}

Supplied upper bounds on \(\|A\|,\|Q\|,\|P_0\|\) may be used in
\eqref{eq:rr-solution-norm-app}. The estimate supplies an output
scale but does not assert that this scale is uniform in \(k\).
We next identify the graph through its spectral coordinates.

\begin{theorem}[Decaying graph projector]
\label{thm:rr-seeded-projector-main}
For Problem~\ref{prob:quantum-rr}, suppose that \(\Pi_>R_0\) has
full column rank. At the requested step \(k\), set
\begin{equation*}
\mathcal E_k=\Pi_>
+\mathcal S_F^k\Pi_<R_0(\Pi_>R_0)^+
\mathcal S_F^{-k}\Pi_>.
\end{equation*}
Then
\begin{equation*}
\begin{aligned}
\mathcal E_k^2&=\mathcal E_k,\\
\operatorname{ran}\mathcal E_k
&=\operatorname{ran}\begin{bmatrix}\mathsf I_n\\P_k\end{bmatrix},
&\ker\mathcal E_k&=\operatorname{ran}\Pi_<.
\end{aligned}
\end{equation*}
\end{theorem}

\begin{proof}
Full column rank and Lemma~\ref{lem:rr-graph-propagation-app} give
\(\operatorname{ran}(\Pi_>R_0)=\operatorname{ran}\Pi_>\).
Thus \((\Pi_>R_0)(\Pi_>R_0)^+\) acts as the identity on the
exterior subspace, even though it need not equal the spectral
projector \(\Pi_>\) on the whole space.
Commutation of the spectral projectors with \(\mathcal S_F\) yields
\[
(\mathcal S_F^{-k}\Pi_>)(\mathcal S_F^kR_0)=\Pi_>R_0.
\]
Splitting the propagated column into its two branches in
\eqref{eq:rr-seeded-projector} therefore gives
\begin{equation}
\mathcal E_k
=\mathcal S_F^kR_0(\Pi_>R_0)^+\mathcal S_F^{-k}\Pi_>,
\qquad
\mathcal E_k\mathcal S_F^kR_0=\mathcal S_F^kR_0.
\label{eq:rr-projector-factorization-app}
\end{equation}
The first identity bounds the rank of \(\mathcal E_k\) by \(n\);
the second proves that its range contains the rank-\(n\) column
\(\mathcal S_F^kR_0\) and that it is the identity there.
It is therefore idempotent, with the range stated in the theorem.
It annihilates \(\operatorname{ran}\Pi_<\), which has dimension
\(n\), so this subspace is exactly its kernel.
\end{proof}

The rank condition is equivalent to
\(\operatorname{ran}R_0\cap\operatorname{ran}\Pi_<=\{0\}\).
It need not follow from positive semidefiniteness alone: for
\(A=2\) and \(G=Q=P_0=0\), the lift is
\(\operatorname{diag}(1/2,2)\), the recursion is well defined,
but \(\Pi_>R_0=0\).
Under the stronger coefficient hypotheses of
Proposition~\ref{prop:factor-rr-initialization-app}, its lower
bound does ensure full column rank for every \(P_0\succeq0\).
The branch correspondence in Lemma~\ref{lem:factor-rr-branches-app}
then identifies the exterior RR projector with the positive-graph
DARE pencil projector; it does not identify the two interior
projectors under the local conventions of the respective sections.

\subsection{Direct recovery from the projector blocks}
\label{app:rr-recovery}

The decaying graph projector fixes every vector of the solution graph.
Its upper block therefore has a right inverse supplied by the graph
itself, which controls the final pseudoinverse without an additional
leading-block condition.

\begin{corollary}[Direct recovery of a finite iterate]
\label{cor:rr-projector-block-recovery-app}
Under Theorem~\ref{thm:rr-seeded-projector-main},
\begin{equation}
\begin{aligned}
(E_1^*\mathcal E_k)\begin{bmatrix}\mathsf I_n\\P_k\end{bmatrix}
&=\mathsf I_n,\qquad
E_2^*\mathcal E_k=P_k(E_1^*\mathcal E_k),\\
P_k&=(E_2^*\mathcal E_k)(E_1^*\mathcal E_k)^+,\\
\|(E_1^*\mathcal E_k)^+\|&\le\sqrt{1+\|P_k\|^2}.
\end{aligned}
\label{eq:rr-projector-recovery-scales-app}
\end{equation}
In particular, \(E_1^*\mathcal E_k\) has full row rank.
\end{corollary}
\begin{proof}
Theorem~\ref{thm:rr-seeded-projector-main} makes \(\mathcal E_k\)
an idempotent with range \(\operatorname{ran}[\mathsf I_n;P_k]\).
Apply Lemma~\ref{lem:dre-projector-block-recovery-app} with
\(\mathcal E=\mathcal E_k\) and \(P=P_k\).
\end{proof}

Both selected blocks inherit \(\alpha_{\mathcal E_k}\), the
normalization of the decaying graph projector, chosen to satisfy
\(\alpha_{\mathcal E_k}\ge\|\mathcal E_k\|\ge1\). A supplied bound
\(\max_{0\le j\le k}\|P_j\|\le M\) therefore permits the output scale
\begin{equation}
\alpha_{P_k}=8\alpha_{\mathcal E_k}\sqrt{1+M^2}.
\label{eq:rr-direct-output-scale}
\end{equation}
The upper-block pseudoinverse uses normalization
\(8\sqrt{1+M^2}\); its product with the lower block has exactly
the scale in \eqref{eq:rr-direct-output-scale}.
Equation~\eqref{eq:rr-solution-norm-app}, maximized over the requested
steps and evaluated using supplied coefficient bounds, provides one
admissible \(M\). It may grow with \(k\):
\(A=2,G=0,Q=1,P_0=0\) gives \(P_k=(4^k-1)/3\).

\subsection{Power-weighted annular quadrature and normalization}
\label{app:rr-quadrature}

The unit-circle separation in
Theorem~\ref{thm:factor-circle-app}, with \(S=\mathcal S_F\) and
\(\tau=1/2\), gives the contours
\begin{equation}
\Gamma_<:\ |z|=1-\eta_{\mathbb T}/2,\qquad
\Gamma_>=\partial\{1+\eta_{\mathbb T}/2<|z|<3\alpha_{\mathcal S_F}\}.
\label{eq:rr-quadrature-contours-app}
\end{equation}
The interior circle and the outer boundary of the annulus are
counterclockwise; the annular inner boundary is clockwise.
In particular, \(\Gamma_>\) excludes the origin, where the inverse
power weight is singular. The factor \(\mathcal R_{\rm RR}\) below
is evaluated on these complete contours.
The gap satisfies \(0<\eta_{\mathbb T}\le1\), and
\(\alpha_{\mathcal S_F}>1\). To select the contours algorithmically,
one may replace \(\eta_{\mathbb T}\) in their radii, the strip widths,
and the node counts by any supplied positive lower bound. All estimates
then use that lower bound and the factor on the resulting contours.

Each required matrix is a weighted Riesz target
\begin{equation}
\begin{gathered}
T=\frac1{2\pi\mathrm i}\int_{\Gamma_\chi}
g(z)(z\mathsf I-\mathcal S_F)^{-1}R\,\mathrm dz,\\[1mm]
\begin{array}{c|ccc}
T&\chi&g(z)&R\\ \hline
\Pi_>&>&1&\mathsf I\\
\Pi_>R_0&>&1&R_0\\
\mathcal S_F^k\Pi_<R_0&<&z^k&R_0\\
\mathcal S_F^{-k}\Pi_>&>&z^{-k}&\mathsf I
\end{array}
\end{gathered}
\label{eq:rr-weighted-riesz-blocks-app}
\end{equation}
Here \(\alpha_R=1\) for \(R=\mathsf I\), and
\(\alpha_R=\alpha_{R_0}\) for \(R=R_0\).
On a circle of radius \(\rho\) and orientation \(s\in\{1,-1\}\),
the \(m\)-point trapezoidal rule uses
\begin{equation}
z_j=\rho e^{2\pi\mathrm i j/m},\qquad
\omega_j=\frac{s z_jg(z_j)}m,\qquad 0\le j<m.
\label{eq:rr-quadrature-nodes-app}
\end{equation}
The complete finite sum \(S_\Gamma\) is the sum of
\(\omega_j(z_j\mathsf I-\mathcal S_F)^{-1}R\) over all circles
in \(\Gamma_\chi\). Thus the LCU coefficient already contains the
power weight, which is applied only once.

\begin{proposition}[Power-weighted circular construction]
\label{prop:rr-weighted-quadrature-app}
Let \(k\ge0\) be an integer. Assume exact encodings of
\(\mathcal S_F\) and \(R_0\), and the coherent access and supplied
uniform constant-factor node inverse-norm estimates of
Proposition~\ref{thm:local-contour-sum}.
For each boundary circle in \eqref{eq:rr-quadrature-contours-app},
use the following orientation and angular strip width:
\begin{equation}
\begin{array}{c|c|c}
\rho&s&a\\ \hline
1-\eta_{\mathbb T}/2&+1&
\displaystyle\log\frac{1-\eta_{\mathbb T}/4}{1-\eta_{\mathbb T}/2}\\[2mm]
1+\eta_{\mathbb T}/2&-1&
\displaystyle\log\frac{1+3\eta_{\mathbb T}/4}{1+\eta_{\mathbb T}/2}\\[2mm]
3\alpha_{\mathcal S_F}&+1&\log(4/3)
\end{array}
\label{eq:rr-quadrature-strips-app}
\end{equation}
For every target in \eqref{eq:rr-weighted-riesz-blocks-app},
the complete quadrature error is bounded by
\begin{equation}
\|S_\Gamma-T\|
\le\sum_{\text{circles of }\Gamma_\chi}
\frac{8\alpha_R\mathcal R_{\rm RR}}{e^{am}-1}.
\label{eq:rr-weighted-quadrature-error-app}
\end{equation}
Given a target tolerance \(\varepsilon_T>0\), choose the node
count separately on each circle as
\begin{equation}
\begin{aligned}
m&=\left\lceil\frac1a\log\!\left(1+
\frac{16\alpha_R\mathcal R_{\rm RR}}{\varepsilon_T}\right)
\right\rceil,&&\chi=<,\\[1mm]
m&=\left\lceil\frac1a\log\!\left(1+
\frac{32\alpha_R\mathcal R_{\rm RR}}{\varepsilon_T}\right)
\right\rceil,&&\chi=>.
\end{aligned}
\label{eq:rr-weighted-quadrature-order-app}
\end{equation}
Then \(\|S_\Gamma-T\|\le\varepsilon_T/2\), with total node count
\begin{equation}
m_T=\sum_{\text{circles of }\Gamma_\chi}m
=O\!\left(\eta_{\mathbb T}^{-1}
\log\!\left(e+\frac{\alpha_R\mathcal R_{\rm RR}}
{\varepsilon_T}\right)\right).
\label{eq:rr-weighted-quadrature-count-app}
\end{equation}
The outer circle requires only the logarithmic factor in this bound.

Choose the supplied inverse normalizations with
\(4\|(z_j\mathsf I-\mathcal S_F)^{-1}\|\le\beta_j
\le c_\beta\|(z_j\mathsf I-\mathcal S_F)^{-1}\|\),
where \(c_\beta\) is fixed independently of the node.
The raw normalization of the finite-sum encoding satisfies
\begin{equation}
\begin{aligned}
\alpha_T^{\rm quad}
&=\alpha_R\sum_j|\omega_j|\beta_j\\
&\le c_\beta\alpha_R\mathcal R_{\rm RR}
\sum_{\text{circles of }\Gamma_\chi}
\frac{\rho}{\rho+\alpha_{\mathcal S_F}}
\le2c_\beta\alpha_R\mathcal R_{\rm RR}.
\end{aligned}
\label{eq:rr-weighted-quadrature-scales-app}
\end{equation}
In particular, this rule gives
\begin{equation}
\begin{aligned}
\alpha_{\Pi_>}^{\rm quad},\quad
\alpha_{\mathcal S_F^{-k}\Pi_>}^{\rm quad}
&=O(\mathcal R_{\rm RR}),\\
\alpha_{\Pi_>R_0}^{\rm quad},\quad
\alpha_{\mathcal S_F^k\Pi_<R_0}^{\rm quad}
&=O(\alpha_{R_0}\mathcal R_{\rm RR}).
\end{aligned}
\label{eq:rr-four-block-normalizations-app}
\end{equation}
After assigning the remaining \(\varepsilon_T/2\) to inverse
and LCU implementation, a decoded approximation to \(T\) has error
at most \(\varepsilon_T\). It uses
\begin{equation}
O\!\left(\mathcal R_{\rm RR}
\log\!\left(e+\frac{\mathcal R_{\rm RR}\alpha_T^{\rm quad}}
{\varepsilon_T}\right)\right)
\label{eq:rr-weighted-block-query-app}
\end{equation}
queries to the supplied \(\mathcal S_F\) encoding and its adjoint,
and one call to the encoding of \(R\).
The absolute quadrature and normalization bounds are uniform in
\(k\) and the node counts, and allow arbitrary Jordan structure.
\end{proposition}
\begin{proof}
Theorem~\ref{thm:factor-circle-app} ensures that these contours
enclose the indicated complete branches. The Riesz functional calculus
in Appendix~\ref{app:riesz-calculus} identifies the four integrals
in \eqref{eq:rr-weighted-riesz-blocks-app}, including inverse powers
on the annulus, without assuming diagonalizability.

Write \(z(t)=\rho e^{\mathrm it}\) for one circle, with
\(|\operatorname{Im}t|\le a\).
For the interior circle, its largest strip radius is
\(1-\eta_{\mathbb T}/4\); for the inner annular circle, it is
\(1+3\eta_{\mathbb T}/4\).
In both cases the outward radial displacement from the real-parameter
circle is exactly \(\eta_{\mathbb T}/4\), and the inward
displacement is smaller, since \(1-e^{-a}\le e^a-1\).
The real-parameter circles have smallest singular values at least
\(\eta_{\mathbb T}/2\), by perturbation from the unit circle.
With \(z_0=\rho e^{\mathrm i\operatorname{Re}t}\), apply
Lemma~\ref{lem:inverse-input-perturbation} to
\(T=z_0\mathsf I-\mathcal S_F\) and
\(\Delta T=(z(t)-z_0)\mathsf I\). Their norm product is at most
\((\eta_{\mathbb T}/4)(2/\eta_{\mathbb T})=1/2\), so
\begin{equation}
\|(z(t)\mathsf I-\mathcal S_F)^{-1}\|
\le2\|(z_0\mathsf I-\mathcal S_F)^{-1}\|
\le\frac{2\mathcal R_{\rm RR}}
{\rho+\alpha_{\mathcal S_F}}.
\label{eq:rr-strip-resolvent-app}
\end{equation}
Both closed strips therefore avoid every resolvent pole.
The interior strip lies strictly inside the unit disk, while the
annular inner strip lies strictly outside it, because its smallest
radius is at least \(1+\eta_{\mathbb T}/4\).
Thus \(|z(t)^k|\le1\) on the interior strip and
\(|z(t)^{-k}|\le1\) on the annular inner strip for every \(k\ge0\).

For the outer circle, the strip radii range from
\(9\alpha_{\mathcal S_F}/4\) to \(4\alpha_{\mathcal S_F}\).
They are outside the spectrum and the unit disk.
Lemma~\ref{lem:inverse-input-perturbation}, with
\(T=z(t)\mathsf I\) and \(\Delta T=-\mathcal S_F\), gives
\[
|z(t)|\|(z(t)\mathsf I-\mathcal S_F)^{-1}\|
\le\frac{|z(t)|}{|z(t)|-\alpha_{\mathcal S_F}}
\le\frac95,
\]
and again \(|z(t)^{-k}|\le1\).
Since \(\mathcal R_{\rm RR}\ge1\) follows directly from
\(|z|+\alpha_{\mathcal S_F}\ge\|z\mathsf I-\mathcal S_F\|\),
all three strips satisfy the common full-integrand bound
\begin{equation}
\sup_{|\operatorname{Im}t|\le a}
\|s z(t)g(z(t))(z(t)\mathsf I-\mathcal S_F)^{-1}R\|
\le4\alpha_R\mathcal R_{\rm RR}.
\label{eq:rr-strip-integrand-app}
\end{equation}
For the near-unit circles, this uses
\(|z(t)|\le\rho+\eta_{\mathbb T}/4
<\rho+\alpha_{\mathcal S_F}\) in
\eqref{eq:rr-strip-resolvent-app}. Every bound has a strictly
positive margin, so the full integrand is holomorphic on a
neighborhood of the closed strip.

Lemma~\ref{lem:analytic-quadrature-app}, applied to the periodic
angular integrand with bound \eqref{eq:rr-strip-integrand-app},
gives circle error at most
\(8\alpha_R\mathcal R_{\rm RR}/(e^{am}-1)\).
Adding the oriented components proves
\eqref{eq:rr-weighted-quadrature-error-app}.
The choices in \eqref{eq:rr-weighted-quadrature-order-app} give
error at most \(\varepsilon_T/2\) on the sole interior circle,
or \(\varepsilon_T/4\) on each of the two exterior circles.
The first two strip widths are bounded below by constant multiples
of \(\eta_{\mathbb T}\), whereas the third is constant, proving
\eqref{eq:rr-weighted-quadrature-count-app}.
This argument bounds absolute error and imposes no condition
\(m>k\); it already includes all Fourier indices created by the power.

At the real-parameter nodes, \(|g(z_j)|\le1\) and
\(\beta_j\le c_\beta\mathcal R_{\rm RR}/
(|z_j|+\alpha_{\mathcal S_F})\).
The sum of the absolute geometric weights on a circle is \(\rho\).
Substitution in the LCU normalization proves
\eqref{eq:rr-weighted-quadrature-scales-app} and
\eqref{eq:rr-four-block-normalizations-app}.
The same integral estimate gives
\(\|T\|\le2\alpha_R\mathcal R_{\rm RR}\).
These are upper bounds for the specified rule, not lower bounds
on its raw normalization.

Finally, apply Proposition~\ref{thm:local-contour-sum} to the
finite sum with implementation tolerance \(\varepsilon_T/2\).
Its inverse and LCU error allocation, proved in
Appendix~\ref{app:contour-proofs}, gives
\eqref{eq:rr-weighted-block-query-app}; the right-input encoding
is used once. If the finite-sum normalization is below this tolerance,
its zero approximation suffices instead. Adding the quadrature error
proves the complete decoded-error claim.
\end{proof}

The node count controls the node register, coefficient preparation,
and arithmetic; it does not multiply the matrix-query count in
\eqref{eq:rr-weighted-block-query-app}.
Computing the power weights still depends on the representation of
\(k\) and the required working precision. Uniform absolute
quadrature error therefore does not make all algorithmic resources
independent of \(k\). Likewise, access to the supplied lift encoding
does not make the construction of \(A^{-1}\) from the original
data free. Approximate input encodings are handled separately by
Proposition~\ref{prop:contour-input-errors-app} on these fixed
ideal contours, with an input-error budget added to quadrature
and implementation errors.

For the specified contours,
Theorem~\ref{thm:factor-circle-app} gives
\begin{equation}
\mathcal R_{\rm RR}
\le\frac{2(1+\alpha_{\mathcal S_F})+\eta_{\mathbb T}}
{\eta_{\mathbb T}}
=O\!\left(\frac{1+\alpha_{\mathcal S_F}}{\eta_{\mathbb T}}\right).
\label{eq:rr-quadrature-spectral-bound-app}
\end{equation}
An additional comparison
\(\eta_{\mathbb T}\ge c\,\sigma_{\min}(\mathcal S_F)\),
with fixed \(c>0\), changes this to
\(O(\alpha_{\mathcal S_F}/\sigma_{\min}(\mathcal S_F))\).
This becomes \(O(\kappa(\mathcal S_F))\) only when
\(\alpha_{\mathcal S_F}=O(\|\mathcal S_F\|)\).
Alternatively, under the coefficient hypotheses of
Proposition~\ref{prop:factor-rr-resolvent-app}, use the reciprocal of
the right-hand side of \eqref{eq:factor-rr-resolvent-app} as a
supplied lower bound on \(\eta_{\mathbb T}\).
Substitution in \eqref{eq:rr-quadrature-spectral-bound-app} bounds
\(\mathcal R_{\rm RR}\) by \(O(1+\alpha_{\mathcal S_F})\) times
that same right-hand side. These substitutions apply to the stated
contour construction.

\subsection{Stability and query complexity}
\label{app:rr-complexity}

Use the coherent inverse and contour access of
Section~\ref{sec:quantum-tools}, and construct the four targets in
\eqref{eq:rr-weighted-riesz-blocks-app} by
Proposition~\ref{prop:rr-weighted-quadrature-app}, or a supplied
rule with the same error and query bounds. Write \(\alpha_T^{\rm quad}\)
for a target's raw quadrature normalization, \(\alpha_T\) for its
actual normalization, and \(r_T\ge1\) for the query overhead of any
normalization reduction. Define \(r_{\mathcal E_k}\) similarly;
an unused reduction has overhead one. These scales and bounds are
supplied uniformly over the required accuracies.

Assume that \(\Pi_>R_0\) has full column rank and that the solution
bound \(M\) is supplied. We write its smallest initial singular
value explicitly; a supplied positive lower bound may replace it
in every threshold, precision, normalization, and query bound below.
For the compact bound in Theorem~\ref{thm:rr-construction-main},
this lower bound is within a constant factor of the true singular
value, and we assume
\begin{equation}
\alpha_{\Pi_>R_0}=\Theta(\|\Pi_>R_0\|),\qquad
r_{\mathcal E_k},\ \max_T r_T=O(1).
\label{eq:rr-direct-calibration}
\end{equation}
Fix accuracy-independent normalization bounds before choosing the
internal precision, padding the encodings to the declared scales when
needed. The output scale remains \eqref{eq:rr-direct-output-scale}; the
solution bound need not be constant.

\begin{proposition}[Stability and cost of direct recovery]
\label{prop:rr-direct-projector-stability}
Suppose an encoding of \(\mathcal E_k\), with known normalization
\(\alpha_{\mathcal E_k}\ge1\) and decoded error at most \(\delta\),
costs \(Q_{\mathcal E_k}(\delta)\) queries per call. Set
\(\alpha_{P_k}\) by \eqref{eq:rr-direct-output-scale}. For
\(0<\varepsilon\le1\), it suffices to use
\begin{equation}
\delta=\frac{\varepsilon}{8(1+M^2)},\qquad
\varepsilon_{\rm inv}=\frac{\varepsilon}{4(1+\alpha_{\mathcal E_k})}
\label{eq:rr-direct-projector-budget}
\end{equation}
for the projector error and the upper-block inverse implementation
error, respectively. With inverse threshold
\(4\alpha_{\mathcal E_k}/\alpha_{P_k}\) and decoded product error
at most \(\varepsilon/4\), the result has error at most
\(\varepsilon\) and query cost
\begin{equation}
O\!\left(Q_{\mathcal E_k}(\delta)\,\alpha_{P_k}
\log\!\left(e+\frac{\alpha_{P_k}}{\varepsilon}\right)\right).
\label{eq:rr-direct-projector-local-query}
\end{equation}
\end{proposition}
\begin{proof}
Apply Proposition~\ref{prop:dre-direct-projector-stability-app}
with \(P(t)=P_k\), \(\mathcal E(t)=\mathcal E_k\),
\(\alpha_{\mathcal E}=\alpha_{\mathcal E_k}\),
\(\alpha_{P(t)}=\alpha_{P_k}\), and \(\delta_E=\delta\).
Corollary~\ref{cor:rr-projector-block-recovery-app} supplies its
exact graph identities, and \eqref{eq:rr-direct-output-scale}
is its permitted output normalization. Its thresholds, implementation
budget and query bound become
\eqref{eq:rr-direct-projector-budget} and
\eqref{eq:rr-direct-projector-local-query}.
The approximate projector need not be idempotent.
\end{proof}

The remaining task is to obtain the required projector accuracy.
The first pseudoinverse preserves the coordinates selected by the
initial data; its singular value, rather than the solution bound,
controls this step.

\begin{proposition}[Seed construction and total query bound]
\label{prop:rr-general-resources-app}
Under the assumptions above, there is an absolute \(c>0\) such that
the common decoded accuracy
\begin{equation}
\begin{aligned}
\eta={}&\frac{c\varepsilon}
{(1+M^2)(1+\alpha_{\Pi_>R_0}+\alpha_{\mathcal E_k})}\\
&\times\frac{1}{(1+\alpha_{\mathcal S_F^k\Pi_<R_0})
(1+\alpha_{\mathcal S_F^{-k}\Pi_>})
(1+\sigma_{\min}(\Pi_>R_0)^{-1})^2}
\end{aligned}
\label{eq:rr-outer-budget-app}
\end{equation}
suffices for the four weighted targets, the initial pseudoinverse,
and each fixed sum, product, and normalization reduction used to
assemble \(\mathcal E_k\). With threshold
\(\sigma_{\min}(\Pi_>R_0)/2\), direct composition gives
\begin{equation}
\alpha_{\mathcal E_k}^{\rm raw}
=\alpha_{\Pi_>}+
\frac{8\alpha_{\mathcal S_F^k\Pi_<R_0}
\alpha_{\mathcal S_F^{-k}\Pi_>}}{\sigma_{\min}(\Pi_>R_0)}.
\label{eq:rr-raw-normalizations-app}
\end{equation}
Combining this construction with direct recovery gives
\begin{equation}
\begin{aligned}
Q_{P_k}=O\!\Bigg(&\mathcal R_{\rm RR}
\frac{\alpha_{\Pi_>R_0}}{\sigma_{\min}(\Pi_>R_0)}
\alpha_{P_k}\,r_{\mathcal E_k}\max_T r_T\\
&\times\log^3\!\left(e+
\frac{\mathcal R_{\rm RR}(1+\alpha_{P_k})
(1+\sigma_{\min}(\Pi_>R_0)^{-1})
(1+\max_T\alpha_T^{\rm quad})}{\eta}\right)\Bigg).
\end{aligned}
\label{eq:rr-general-query-expanded-app}
\end{equation}
A common sufficient number of quadrature nodes for each target is
\begin{equation}
m_T=O\!\left(\eta_{\mathbb T}^{-1}
\log\!\left(e+\frac{\alpha_{R_0}\mathcal R_{\rm RR}}{\eta}\right)\right).
\label{eq:rr-final-node-counts-app}
\end{equation}
The coherent query bound does not multiply by this node count.
\end{proposition}
\begin{proof}
Use the stability and cost argument of
Proposition~\ref{prop:dre-general-resources-app} with the substitutions
\[
\begin{aligned}
\Pi_+&\mapsto\Pi_>,& \Pi_+R_0&\mapsto\Pi_>R_0,\\
e^{t\mathcal H_{\rm DRE}}\Pi_-R_0
&\mapsto\mathcal S_F^k\Pi_<R_0,&
e^{-t\mathcal H_{\rm DRE}}\Pi_+
&\mapsto\mathcal S_F^{-k}\Pi_>,\\
\gamma_+&\mapsto\sigma_{\min}(\Pi_>R_0),&
\mathcal E&\mapsto\mathcal E_k,\\
P(t)&\mapsto P_k,&
\mathcal R_{\rm DRE}&\mapsto\mathcal R_{\rm RR}.
\end{aligned}
\]
The corresponding actual normalizations, raw quadrature scales and
re-encoding overheads are replaced as well.
Theorem~\ref{thm:rr-seeded-projector-main} supplies the same
projector composition from four targets and full-column-rank
initialization;
\eqref{eq:rr-weighted-block-query-app} supplies the required weighted
block cost. Thus that argument gives the common precision
\eqref{eq:rr-outer-budget-app}, raw normalization
\eqref{eq:rr-raw-normalizations-app}, and three-logarithm query bound
\eqref{eq:rr-general-query-expanded-app}, with the absolute initial
singular value and every declared overhead included.
Substituting \(\eta\) in
\eqref{eq:rr-weighted-quadrature-count-app}, and using
\(\alpha_{R_0}\ge\|R_0\|\ge1\), gives
\eqref{eq:rr-final-node-counts-app}.
\end{proof}

Algorithm~\ref{alg:rr-direct-projector-construction} uses the supplied
bounds above and the working accuracy in \eqref{eq:rr-outer-budget-app}.

\begin{algorithm}[H]
\caption{RR: construct a block-encoding of \(P_k\)}
\label{alg:rr-direct-projector-construction}
\small
\begin{algorithmic}[1]
\Require Exact BEs of \(\mathcal S_F\) and
\(R_0=[\mathsf I_n;P_0]\), with their scales;
integer \(k\ge0\), \(0<\varepsilon\le1\).
\Ensure A BE of \(P_k\) with decoded error \(\le\varepsilon\)
and scale \(\alpha_{P_k}\) defined in Appendix~\ref{app:rr-recovery}.

\State Encode \(\Pi_>R_0\) using Algorithm~\ref{alg:weighted-riesz-be}:
\[
U_{\Pi_>R_0}\gets\operatorname{RieszBE}
(U_{\mathcal S_F},U_{\mathsf I},\Gamma_>,1,U_{R_0};\eta).
\]
\State Apply the pseudoinverse primitive to obtain \(U_{(\Pi_>R_0)^+}\).
\State Construct the other three weighted Riesz BEs:
\[
\begin{aligned}
U_{\Pi_>}&\gets\operatorname{RieszBE}
(U_{\mathcal S_F},U_{\mathsf I},\Gamma_>,1,U_{\mathsf I};\eta),\\
U_{\mathcal S_F^k\Pi_<R_0}
&\gets\operatorname{RieszBE}
(U_{\mathcal S_F},U_{\mathsf I},\Gamma_<,z^k,U_{R_0};\eta),\\
U_{\mathcal S_F^{-k}\Pi_>}
&\gets\operatorname{RieszBE}
(U_{\mathcal S_F},U_{\mathsf I},\Gamma_>,z^{-k},U_{\mathsf I};\eta).
\end{aligned}
\]
\State Assemble a BE of the DGP:
\[
\mathcal E_k=\Pi_>
+\mathcal S_F^k\Pi_<R_0(\Pi_>R_0)^+
\mathcal S_F^{-k}\Pi_>.
\]
\State Use both complete rows of the same \(U_{\mathcal E_k}\)
to encode
\[
P_k=(E_2^*\mathcal E_k)(E_1^*\mathcal E_k)^+.
\]
\Statex \Return \((U_{P_k},\alpha_{P_k},a_{P_k})\).
\end{algorithmic}
\end{algorithm}

Calls to the initial-column oracle obey the same bound. Coefficient
preparation and arithmetic for \(z^{\pm k}\) have separate gate
costs. In particular, the uniform quadrature estimate does not
assert that all resources are independent of \(k\). The explicit
logarithm also records the supplied encoding scales suppressed
by \(\widetilde O\).

\begin{corollary}[RR output normalization]
\label{cor:rr-construction-normalization-app}
For the specific contours and quadrature of
Proposition~\ref{prop:rr-weighted-quadrature-app}, a constant-factor
initial singular-value estimate, and direct sums and products,
one may choose
\begin{equation}
\alpha_{P_k}=O\!\left(\sqrt{1+M^2}\left(
\mathcal R_{\rm RR}
+\frac{\alpha_{R_0}\mathcal R_{\rm RR}^{\,2}}
{\sigma_{\min}(\Pi_>R_0)}\right)\right).
\label{eq:rr-direct-output-general-bound}
\end{equation}
If \(M=O(1)\) and \(\alpha_{R_0}=O(1)\), this gives
\eqref{eq:rr-output-normalization-bound}, uniformly in \(k\)
when the supplied bounds are uniform.
\end{corollary}
\begin{proof}
The four normalization bounds in
\eqref{eq:rr-four-block-normalizations-app}, substituted into
\eqref{eq:rr-raw-normalizations-app} and
\eqref{eq:rr-direct-output-scale}, give the first assertion.
When \(\alpha_{R_0}=O(1)\), the same contour estimates give
\(\sigma_{\min}(\Pi_>R_0)\le\|\Pi_>R_0\|=O(\mathcal R_{\rm RR})\).
Consequently the second term absorbs the first. This bound concerns
the specified quadrature and its direct assembly; reducing the
declared output normalization has its own query cost.
\end{proof}

\newtheorem*{rrproblemrestatement}{Problem~\ref{prob:quantum-rr}}
\begin{rrproblemrestatement}
For the data in Definition~\ref{def:problem-recursion} and a requested
integer \(k\ge0\), suppose that exact block-encodings of
\(\mathcal S_F\) and \(R_0\) in \eqref{eq:rr-lift-seed}, their
adjoints and controlled versions, are supplied with normalizations
\(\alpha_{\mathcal S_F}\) and \(\alpha_{R_0}\).
For \(0<\varepsilon\le1\), construct a block-encoding with decoded
matrix \(\widetilde P_k\) satisfying
\(\|\widetilde P_k-P_k\|\le\varepsilon\).
\end{rrproblemrestatement}

\newtheorem*{rrmainrestatement}{Theorem~\ref{thm:rr-construction-main}}
\begin{rrmainrestatement}
For Problem~\ref{prob:quantum-rr}, suppose that \(\Pi_>R_0\) has full
column rank, and supply a bound \(\|P_j\|\le M\) for
\(0\le j\le k\). Choose an interior contour \(\Gamma_<\) in
\(|z|<1\) and an exterior contour system \(\Gamma_>\) in
\(|z|>1\), each enclosing exactly its spectral branch, with
\(\Gamma_>\) having winding number zero about the closed unit disk.
Define
\begin{equation*}
\mathcal R_{\rm RR}
=\max_{\chi\in\{<,>\}}\sup_{z\in\Gamma_\chi}
(|z|+\alpha_{\mathcal S_F})
\|(z\mathsf I-\mathcal S_F)^{-1}\|.
\end{equation*}
Under the implementation and normalization assumptions of
Appendix~\ref{app:rr-complexity}, the algorithm returns an
\((\alpha_{P_k},a_{P_k},\varepsilon)\) block-encoding of \(P_k\), using
\begin{equation*}
Q_{\rm RR}=\widetilde O\!\left(
\mathcal R_{\rm RR}\,\kappa(\Pi_>R_0)\,\alpha_{P_k}\right)
\end{equation*}
queries to the supplied matrix oracles. Here \(\alpha_{P_k}\) is
the direct-product output normalization specified in that appendix;
\(\widetilde O\) suppresses logarithmic dependence on precision
and the supplied scales.
\end{rrmainrestatement}

\begin{proof}
Lemma~\ref{lem:rr-graph-propagation-app} and
Theorem~\ref{thm:rr-seeded-projector-main} identify the range of
\(\mathcal E_k\) with the requested graph.
Corollary~\ref{cor:rr-projector-block-recovery-app} then proves
that the final product equals \(P_k\), and
Propositions~\ref{prop:rr-direct-projector-stability}
and~\ref{prop:rr-general-resources-app} give the required decoded
accuracy and total query count. Under
\eqref{eq:rr-direct-calibration}, the initial inverse factor is
\(O(\kappa(\Pi_>R_0))\) and the normalization-reduction overheads
are constant. Equation~\eqref{eq:rr-general-query-expanded-app}
therefore gives \eqref{eq:rr-query-main}. The solution bound remains
in \(\alpha_{P_k}\), without requiring \(M=O(1)\).
\end{proof}

\subsection{Approximate inputs and construction of the lift}
\label{app:rr-inputs}

The matrix-query bound takes the lift and initial-column encodings
as inputs. Their approximation errors must satisfy the same weighted
target accuracy as the quadrature. If the lift is instead assembled
from \(A,G,Q\), its inverse of \(A\) also has a query cost.

\begin{proposition}[RR construction from approximate inputs]
\label{prop:rr-approximate-inputs-app}
The same recovery result holds for approximate lift and initial-column
encodings when the total error in each of the four weighted targets,
including input, coefficient, quadrature, and implementation errors,
meets \eqref{eq:rr-outer-budget-app}. Sufficient input conditions
are those of Proposition~\ref{prop:contour-input-errors-app},
applied to the fixed ideal contours in
\eqref{eq:rr-quadrature-contours-app}. The query bounds use the actual
perturbed normalizations and normalization-reduction overheads.

For exact encodings of \(A,G,Q\), an approximate inverse of \(A\)
gives the lift accuracy \(\varepsilon_{\mathcal S_F}\) whenever
\begin{equation}
\sqrt{1+\|Q\|^2}\sqrt{1+\|G\|^2}\,
\varepsilon_{A^{-1}}+\varepsilon_{\mathcal S_F,\rm impl}
\le\varepsilon_{\mathcal S_F}.
\label{eq:rr-inverse-a-error-app}
\end{equation}
Direct composition has scale
\begin{equation}
\alpha_{\mathcal S_F}^{\rm raw}
=\alpha_A+(1+\alpha_Q)\alpha_{A^{-1}}(1+\alpha_G).
\label{eq:rr-lift-raw-scale-app}
\end{equation}
With \(\alpha_{A^{-1}}=4/\sigma_{\min}(A)\) and
\(\alpha_A\varepsilon_{A^{-1}}\le1\), each lift call uses
\begin{equation}
O\!\left(\frac{\alpha_A}{\sigma_{\min}(A)}
\log\!\left(e+\frac1{\sigma_{\min}(A)\varepsilon_{A^{-1}}}\right)\right)
\label{eq:rr-inverse-a-query-app}
\end{equation}
queries to the \(A\) oracle, in addition to constant calls to the
coefficient oracles. A supplied lower bound can replace
\(\sigma_{\min}(A)\) throughout these inverse formulas.
\end{proposition}
\begin{proof}
Proposition~\ref{prop:contour-input-errors-app} bounds the errors
relative to the four ideal weighted targets.
Proposition~\ref{prop:rr-general-resources-app} then supplies the
required projector accuracy, and
Proposition~\ref{prop:rr-direct-projector-stability} gives the final
error. For construction from the coefficients, write
\[
\mathcal S_F=\begin{bmatrix}0&0\\0&A^*\end{bmatrix}
+[\mathsf I;Q]A^{-1}[\mathsf I,G].
\]
The rectangular factors have norms \(\sqrt{1+\|Q\|^2}\) and
\(\sqrt{1+\|G\|^2}\), which prove
\eqref{eq:rr-inverse-a-error-app}. Sums and products give
\eqref{eq:rr-lift-raw-scale-app}, and
Proposition~\ref{prop:local-inverse} gives
\eqref{eq:rr-inverse-a-query-app}.
\end{proof}

Any reduction of the lift normalization contributes its actual
error and query overhead. The factor \(\mathcal R_{\rm RR}\) is
evaluated using the resulting declared normalization of
\(\mathcal S_F\).

\clearpage
\section{Algebraic Riccati constructions and proofs}
\label{app:algebraic-riccati}

The algebraic algorithms encode the limiting Riesz projectors directly.
We apply the graph recovery and stability results of
Appendix~\ref{app:dre} and specialize the quadrature arguments of
Appendices~\ref{app:dre-quadrature} and~\ref{app:rr-quadrature}.
The Cayley alternative has additional inverse costs.

Throughout, assume the exact input encodings, coherent node access,
and uniform constant-factor inverse-norm estimates of
Propositions~\ref{thm:local-contour-sum}
and~\ref{thm:affine-pencil-lcu-main}. Contours are positively oriented
and lie in the relevant resolvent set. Supply a bound on \(\|X\|\)
before choosing the output scale and precision. Node preparation,
arithmetic, and gate synthesis have separate costs.

\subsection{Limiting projectors and shared recovery}
\label{app:direct-graph-recovery}

\begin{lemma}[Limits of the decaying graph projectors]
\label{lem:algebraic-projector-limits-app}
For the CARE Hamiltonian, let \(\Pi_-,\Pi_+\) be its two half-plane
projectors. If \(\Pi_-R_0\) has full column rank, the decaying graph projector
for \(e^{-t\mathcal H_{\rm CARE}}R_0\) is
\begin{equation}
\Pi_-+e^{-t\mathcal H_{\rm CARE}}\Pi_+R_0
(\Pi_-R_0)^+e^{t\mathcal H_{\rm CARE}}\Pi_-.
\label{eq:care-reversed-seeded-projector-app}
\end{equation}
It converges in norm to \(\Pi_-\) as \(t\to\infty\).
For the unscaled RR lift and the full-rank initialization of
Theorem~\ref{thm:rr-seeded-projector-main},
\(\mathcal E_k\to\Pi_>\) as \(k\to\infty\).
When \(A\) is invertible in the DARE pencil,
\(\mathcal S_F=M^{-1}L=(L^{-1}M)^{-1}\), and this exterior
projector equals the finite unit-disk pencil projector
\(\mathfrak R_<[1;M,L]\).
\end{lemma}
\begin{proof}
Apply the factorization in
Theorem~\ref{thm:dre-seeded-projector-main} to
\(-\mathcal H_{\rm CARE}\), interchanging the two branches.
Its correction term has norm at most
\[
\|e^{-t\mathcal H_{\rm CARE}}\Pi_+\|\,
\|R_0\|\,\|(\Pi_-R_0)^+\|\,
\|e^{t\mathcal H_{\rm CARE}}\Pi_-\|.
\]
Both restricted exponentials tend to zero. This follows from their
strict half-plane separation, including the polynomial factors of any
Jordan blocks. In \eqref{eq:rr-seeded-projector}, the analogous
factors \(\mathcal S_F^k\Pi_<\) and
\(\mathcal S_F^{-k}\Pi_>\) both tend to zero, proving the
discrete limit. If \(A\) is invertible, so are \(M,L\); inversion
maps the inside spectrum of \(L^{-1}M\) to the outside spectrum of
\(\mathcal S_F\), with the same spectral subspaces and projector.
Equation~\eqref{eq:pencil-ordinary-reduction} identifies it with the
pencil Riesz operator.
\end{proof}

These are limits of projectors. Their graph property on any
interval where the Riccati solution is finite follows from the existing
graph-flow and recursion results. The algebraic algorithms encode the
limit directly, so their query costs do not include an initialization
pseudoinverse or a convergence time. For singular \(A\), the DARE
graph is identified directly from its regular pencil in
Appendix~\ref{app:dare-construction}.

For either limiting projector, Lemma~\ref{lem:dre-projector-block-recovery-app}
and Proposition~\ref{prop:dre-direct-projector-stability-app} apply with
\(P(t)=X\) and \(\mathcal E(t)=\Pi\). In particular,
\begin{equation}
X=(E_2^*\Pi)(E_1^*\Pi)^+,
\qquad \|(E_1^*\Pi)^+\|\le\sqrt{1+\|X\|^2}.
\label{eq:algebraic-projector-recovery-app}
\end{equation}
For actual normalization \(\alpha_\Pi\ge\|\Pi\|\), use
\begin{equation}
\alpha_X\ge8\alpha_\Pi\sqrt{1+\|X\|^2},
\label{eq:algebraic-projector-output-app}
\end{equation}
and, for \(0<\varepsilon\le1\),
\begin{equation}
\delta=\frac{\varepsilon}{8(1+\|X\|^2)},\qquad
\varepsilon_{\rm inv}=\frac{\varepsilon}{4(1+\alpha_\Pi)}.
\label{eq:algebraic-projector-budget-app}
\end{equation}
Replace every occurrence of \(\|X\|\) in these choices by the same
supplied upper bound; equality in the output scale is the default.
With product error at most \(\varepsilon/4\), the cited proposition
gives decoded error at most \(\varepsilon\) using
\begin{equation}
O\!\left(Q_\Pi(\delta)\alpha_X
\log\!\left(e+\frac{\alpha_X}{\varepsilon}\right)\right)
\label{eq:algebraic-projector-query-app}
\end{equation}
matrix queries, where \(Q_\Pi(\delta)\) is the cost of one projector
call at error \(\delta\). For the contour implementations, split this
error equally between quadrature and circuit implementation.

For the two contour constructions, keep the existing quadrature scales
\begin{equation}
\begin{aligned}
\alpha_{\rm CARE}&=\sum_j|\omega_j|\beta_j,
&\beta_j&=\Theta\!\left(\|(z_j\mathsf I-\mathcal H_{\rm CARE})^{-1}\|\right),\\
\alpha_{\rm DARE}&=\sum_j|\omega_j|\beta_j,
&\beta_j&=\Theta\!\left(\|(z_jL-M)^{-1}\|\right).
\end{aligned}
\label{eq:algebraic-quadrature-normalizations-app}
\end{equation}
With unit weight, the raw normalizations of the complete projectors are
\begin{equation}
\alpha_{\Pi_-}^{\rm raw}=\alpha_{\rm CARE},\qquad
\alpha_{\Pi_<}^{\rm raw}=\alpha_L\alpha_{\rm DARE}.
\label{eq:algebraic-raw-projector-scales-app}
\end{equation}
The factor \(L\) remains in the pencil integral. Choose
accuracy-independent supplied scales, padding when necessary; the
actual projector normalization \(\alpha_\Pi\) in
\eqref{eq:algebraic-projector-output-app} may differ from the raw
one only through an implemented normalization change. Denote its
query overhead by \(r_\Pi\ge1\), equal to one without a reduction.
The two compact main theorems assume \(r_\Pi=O(1)\), uniformly
over the required accuracies. The expanded bounds below include it.

\subsection{CARE stable projector, finite contours, and construction bound}
\label{app:care-construction}

The classical Hamiltonian graph relation
\cite{Laub1979Schur,LancasterRodman1995Riccati} identifies the target
of both the sign and contour constructions. We record the full
projector property needed by the recovery.

\newtheorem*{caregraphrestatement}{Lemma~\ref{lem:care-graph-recovery}}
\begin{caregraphrestatement}
Let \(X\) be the stabilizing solution in
Definition~\ref{def:problem-care}.
For \(E_1=[\mathsf I;0]\) and \(E_2=[0;\mathsf I]\), the stable
Riesz projector of \(\mathcal H_{\rm CARE}\) satisfies
\begin{equation*}
\Pi_-^2=\Pi_-,
\qquad
\operatorname{ran}\Pi_-=\operatorname{ran}\begin{bmatrix}\mathsf I\\X\end{bmatrix}.
\end{equation*}
Its upper block row has full row rank, and
\begin{equation*}
X=(E_2^*\Pi_-)(E_1^*\Pi_-)^+.
\end{equation*}
\end{caregraphrestatement}

\begin{proof}
Let \(F=A-GX\), \(V=[\mathsf I;X]\), and
\(T=\begin{bmatrix}\mathsf I&0\\X&\mathsf I\end{bmatrix}\).
The CARE and the identities \(X=X^*\), \(G=G^*\) give
\begin{equation}
T^{-1}\mathcal H_{\rm CARE}T
=\begin{bmatrix}F&-G\\0&-F^*\end{bmatrix}.
\label{eq:care-triangularization-app}
\end{equation}
Since \(F\) is Hurwitz, the Hamiltonian has no imaginary-axis spectrum
and its stable subspace has dimension \(n\).
The relation \(\mathcal H_{\rm CARE}V=VF\) shows that
\(\operatorname{range}(V)\) is precisely this subspace.
Idempotence follows from the Riesz calculus. Applying
Lemma~\ref{lem:dre-projector-block-recovery-app} to this graph range
proves full row rank and the recovery formula.
\end{proof}

\newtheorem*{careproblemrestatement}{Problem~\ref{prob:quantum-care}}
\begin{careproblemrestatement}
For the LQR data in Definition~\ref{def:problem-care}, let
\(X\succeq0\) be the stabilizing solution, with \(\|X\|<\infty\).
Given exact block-encodings of \(\mathcal H_{\rm CARE}\) in
\eqref{eq:care-graph-invariance}, its adjoint and controlled versions,
with normalization \(\alpha_H\), and a target
\(0<\varepsilon\le1\), construct a block-encoding with decoded
matrix \(\widetilde X\) satisfying \(\|\widetilde X-X\|\le\varepsilon\).
\end{careproblemrestatement}

\newtheorem*{caremainrestatement}{Theorem~\ref{thm:care-construction-main}}
\begin{caremainrestatement}
For Problem~\ref{prob:quantum-care}, supply an upper bound on \(\|X\|\)
and a positively oriented contour \(\Gamma_-\) enclosing exactly the
stable spectrum of \(\mathcal H_{\rm CARE}\). Define
\begin{equation*}
\mathcal R_{\rm CARE}
=\sup_{z\in\Gamma_-}(|z|+\alpha_H)
\|(z\mathsf I-\mathcal H_{\rm CARE})^{-1}\|.
\end{equation*}
Under the implementation and normalization assumptions of
Appendix~\ref{app:algebraic-riccati}, the algorithm returns an
\((\alpha_X,a_X,\varepsilon)\) block-encoding of \(X\), using
\begin{equation*}
Q_{\rm CARE}=\widetilde O\!\left(\mathcal R_{\rm CARE}\alpha_X\right)
\end{equation*}
queries to the supplied Hamiltonian encoding and its adjoint and
controlled versions. Here \(\alpha_X\) is the actual output
normalization in \eqref{eq:algebraic-projector-output-app};
\(\widetilde O\) suppresses logarithmic dependence on precision and
the supplied scales.
\end{caremainrestatement}

\begin{proof}
Lemma~\ref{lem:care-graph-recovery} identifies \(\Pi_-\) as the
solution graph projector. Choose \(\delta\) by
\eqref{eq:algebraic-projector-budget-app}. Proposition~\ref{thm:local-contour-sum}
with \(M=\mathcal H_{\rm CARE}\), \(g=1\), and \(R=\mathsf I\)
encodes the full projector to decoded error \(\delta\), after combining
the quadrature and implementation errors. Its raw normalization is
\(\alpha_{\rm CARE}\), so one call after any declared normalization
change costs
\[
Q_{\Pi_-}(\delta)=O\!\left(r_{\Pi_-}\mathcal R_{\rm CARE}
\log\!\left(e+\frac{\mathcal R_{\rm CARE}\alpha_{\rm CARE}}{\delta}\right)\right).
\]
Proposition~\ref{prop:dre-direct-projector-stability-app} gives the
output normalization and error, with the expanded bound
\begin{equation}
\begin{aligned}
Q_{\rm CARE}=O\!\Bigg(&r_{\Pi_-}\mathcal R_{\rm CARE}\alpha_X
\log\!\left(e+
\frac{\mathcal R_{\rm CARE}\alpha_{\rm CARE}(1+\|X\|^2)}{\varepsilon}\right)\\
&\times\log\!\left(e+\frac{\alpha_X}{\varepsilon}\right)\Bigg).
\end{aligned}
\label{eq:care-projector-query-expanded-app}
\end{equation}
The supplied solution-norm bound is used throughout. The logarithms
come from the node inverses and the upper-row pseudoinverse.
Under \(r_{\Pi_-}=O(1)\), this proves the main query bound.
\end{proof}

Algorithm~\ref{alg:care-direct-projector-construction} uses the
projector and inverse accuracies in \eqref{eq:algebraic-projector-budget-app},
with the product error budget specified there.

\begin{algorithm}[H]
\caption{CARE: construct a block-encoding of \(X\)}
\label{alg:care-direct-projector-construction}
\small
\begin{algorithmic}[1]
\Require An exact BE of \(\mathcal H_{\rm CARE}\) with scale
\(\alpha_H\); \(0<\varepsilon\le1\).
\Ensure A BE of \(X\) with decoded error \(\le\varepsilon\)
and scale \(\alpha_X\) defined in
Appendix~\ref{app:direct-graph-recovery}.

\State Construct the complete stable projector BE:
\[
U_{\Pi_-}\gets\operatorname{RieszBE}
(U_{\mathcal H_{\rm CARE}},U_{\mathsf I},\Gamma_-,1,U_{\mathsf I};\delta).
\]
\State Encode \((E_1^*\Pi_-)^+\) from the complete upper row of
\(U_{\Pi_-}\).
\State Compose the complete lower-row BE from the same
\(U_{\Pi_-}\) with this pseudoinverse BE to encode
\[
X=(E_2^*\Pi_-)(E_1^*\Pi_-)^+.
\]
\Statex \Return \((U_X,\alpha_X,a_X)\).
\end{algorithmic}
\end{algorithm}

A finite rectangle makes the quadrature and normalization explicit.

\begin{lemma}[CARE normalization on a finite rectangle]
\label{lem:care-normalization-factor-app}
For the parameters in \eqref{eq:care-construction-parameters}
and~\eqref{eq:algebraic-quadrature-normalizations-app}, let
\(c_\beta\) be a uniform constant such that
\(\beta_j\le c_\beta\|(z_j\mathsf I-\mathcal H_{\rm CARE})^{-1}\|\).
Then
\begin{equation}
\alpha_{\rm CARE}\le c_\beta\mathcal R_{\rm CARE}
\sum_j\frac{|\omega_j|}{|z_j|+\alpha_H}.
\label{eq:care-normalization-weight-bound-app}
\end{equation}
In particular, take \(\Gamma_-\) to be the left rectangle in
\eqref{eq:factor-finite-rectangles-app}, with
\(H=\mathcal H_{\rm CARE}\) and \(0<\tau<1\).
On each oriented straight panel, use a nonnegative quadrature rule
in a linear parameter that is exact for constants. Then
\begin{equation}
\alpha_{\Pi_-}=\alpha_{\rm CARE}
\le\frac{6c_\beta}{\pi}\mathcal R_{\rm CARE}
=O(\mathcal R_{\rm CARE}),
\label{eq:care-normalization-factor-app}
\end{equation}
uniformly in the number of panels and quadrature nodes, using the
direct LCU normalization without reduction.
\end{lemma}
\begin{proof}
At every node, \(\beta_j\le c_\beta\mathcal R_{\rm CARE}/(|z_j|+\alpha_H)\),
which gives \eqref{eq:care-normalization-weight-bound-app}.
The positive-weight argument in Proposition~\ref{prop:dre-weighted-quadrature-app},
rescaled by \(\alpha_H\), gives
\(\sum_j|\omega_j|=\ell(\Gamma_-)/(2\pi)\le6\alpha_H/\pi\).
Since \(|z_j|+\alpha_H\ge\alpha_H\), this proves
\eqref{eq:care-normalization-factor-app}.
\end{proof}

Here \(c_\beta\) is uniform in the instance, node count, and accuracy.
The same panel argument gives the full-projector error below.

\begin{proposition}[CARE on a finite rectangle]
\label{prop:care-rectangle-quadrature-app}
Use the exact-input, coherent-access, node inverse-norm, and recovery
assumptions stated at the start of Appendix~\ref{app:algebraic-riccati}.
Suppose a lower bound
\(0<\underline\Delta_{\rm ax}\le\Delta_{\rm ax}\) is supplied, and set
\(d=\underline\Delta_{\rm ax}/2\). Choose the positively oriented contour
\begin{equation}
\Gamma_-=\partial\{x+\mathrm iy:-2\alpha_H\le x\le-d,
\ |y|\le2\alpha_H\},\qquad
\ell(\Gamma_-)=12\alpha_H-\underline\Delta_{\rm ax}.
\label{eq:care-rectangle-contour-app}
\end{equation}
Divide each side into equal panels, taking the ceiling of its length
divided by \(d\) as the number of panels.
Write each oriented panel as \(z_\ell(t)=c_\ell+h_\ell t\),
\(-1\le t\le1\), so \(2|h_\ell|\le d\), and let
\(t_j,w_j\), \(1\le j\le p\), be the \(p\)-point Gauss--Legendre
nodes and weights on \([-1,1]\). With \(N_{\rm pan}\) panels, set
\begin{equation}
\begin{aligned}
z_{\ell j}&=c_\ell+h_\ell t_j,&
\omega_{\ell j}&=\frac{h_\ell w_j}{2\pi\mathrm i},&
m&=N_{\rm pan}p,\\
S_m&=\sum_{\ell,j}\omega_{\ell j}
 (z_{\ell j}\mathsf I-\mathcal H_{\rm CARE})^{-1}.
\end{aligned}
\label{eq:care-rectangle-rule-app}
\end{equation}
Then \(\Gamma_-\) encloses exactly the stable branch and
\begin{equation}
\begin{aligned}
N_{\rm pan}&\le\frac{\ell(\Gamma_-)}d+4,&
\|S_m-\Pi_-\|
&\le\frac{2\ell(\Gamma_-)}{\pi d}\,4^{-p},\\
\mathcal R_{\rm CARE}&\le\frac{(1+2\sqrt2)\alpha_H}{d},&
\|\Pi_-\|&\le\frac{\ell(\Gamma_-)}{2\pi d}.
\end{aligned}
\label{eq:care-rectangle-bounds-app}
\end{equation}
Lemma~\ref{lem:care-normalization-factor-app} applies to this rule and gives
\begin{equation}
\alpha_{\Pi_-}=\alpha_{\rm CARE}
\le\frac{6c_\beta}{\pi}\mathcal R_{\rm CARE}
=O\!\left(\frac{\alpha_H}{\underline\Delta_{\rm ax}}\right),
\label{eq:care-rectangle-normalization-app}
\end{equation}
uniformly in \(p\). For \(0<\delta\le1/8\), the positive integer
\begin{equation}
p=\left\lceil
\frac{\log(4\ell(\Gamma_-)/(\pi d\delta))}{2\log2}
\right\rceil
\label{eq:care-rectangle-order-app}
\end{equation}
ensures \(\|S_m-\Pi_-\|\le\delta/2\), with
\begin{equation}
m=O\!\left(\frac{\alpha_H}{\underline\Delta_{\rm ax}}
\log\!\left(e+\frac{\alpha_H}{\underline\Delta_{\rm ax}\delta}\right)\right).
\label{eq:care-rectangle-nodes-app}
\end{equation}

For \(0<\varepsilon\le1\), take
\(\delta=\varepsilon/[8(1+\|X\|^2)]\), using a supplied upper
bound on \(\|X\|\) consistently wherever it occurs, and allocate
at most \(\delta/2\) to implementing \(S_m\).
Proposition~\ref{prop:dre-direct-projector-stability-app} gives decoded
error at most \(\varepsilon\), with normalization
\begin{equation}
\alpha_X=8\alpha_{\Pi_-}\sqrt{1+\|X\|^2}
=O\!\left(\frac{\alpha_H}{\underline\Delta_{\rm ax}}
\sqrt{1+\|X\|^2}\right)
\label{eq:care-rectangle-output-app}
\end{equation}
and matrix-query complexity
\begin{equation}
\begin{aligned}
Q_{\rm CARE}=O\!\Bigg(&
\mathcal R_{\rm CARE}\alpha_X
\log\!\left(e+
\frac{\mathcal R_{\rm CARE}\alpha_{\Pi_-}}{\delta}\right)
\log\!\left(e+\frac{\alpha_X}{\varepsilon}\right)\Bigg).
\end{aligned}
\label{eq:care-rectangle-query-app}
\end{equation}
For this precision choice, the sufficient node count is
\begin{equation}
m=O\!\left(\frac{\alpha_H}{\underline\Delta_{\rm ax}}
\log\!\left(e+
\frac{\alpha_H(1+\|X\|^2)}
{\underline\Delta_{\rm ax}\varepsilon}\right)\right).
\label{eq:care-rectangle-final-nodes-app}
\end{equation}
\end{proposition}
\begin{proof}
Theorem~\ref{thm:factor-axis-app}, with supplied gap
\(\underline\Delta_{\rm ax}\) and inner-edge distance
\(d=\underline\Delta_{\rm ax}/2\), gives the stable spectral
enclosure, inverse bound \(1/d\), and generalized singularity bound
in \eqref{eq:care-rectangle-bounds-app}. Integrating the inverse bound
over the contour gives the stated projector-norm bound.

For the constant weight \(g=1\), the panels may have length \(d\).
At their centers, \(\|(c_\ell\mathsf I-\mathcal H_{\rm CARE})^{-1}\|\le1/d\)
and \(|h_\ell|\le d/2\); the Neumann series truncated after degree
\(2p-1\) therefore has remainder at most \((2/d)4^{-p}\).
Gauss--Legendre exactness and positivity, as in the proof of
Proposition~\ref{prop:dre-weighted-quadrature-app}, give panel error
at most \(4|h_\ell|4^{-p}/(\pi d)\).
Summing \(2|h_\ell|\) over the contour proves the stated error;
the four ceiling bounds give \(N_{\rm pan}\).
Lemma~\ref{lem:care-normalization-factor-app} applies with
\(\tau=d/\Delta_{\rm ax}\le1/2\).

The order in \eqref{eq:care-rectangle-order-app} gives quadrature
error at most \(\delta/2\) and the stated node count, since
\(\ell(\Gamma_-)/d=O(\alpha_H/\underline\Delta_{\rm ax})\).
The remaining implementation budget and
Proposition~\ref{prop:dre-direct-projector-stability-app}, with
the choices in \eqref{eq:algebraic-projector-output-app}
and~\eqref{eq:algebraic-projector-budget-app}, give the output
scale and query bound. Substituting \(\delta\) gives
\eqref{eq:care-rectangle-final-nodes-app}.
\end{proof}

\begin{remark}[Nodewise scales from coefficient intervals]
\label{rem:care-nodewise-scales-app}
Useful nodewise inverse scales require only mild coefficient assumptions.
For example, take the case of \(A=A^*\), put \(a=\|A\|<1\), and use
\eqref{eq:factor-coefficient-data-app}; put
\(\mu=\min\{g_0,q_0\}\), \(m=\max\{\bar g,\bar q\}\), and assume
\(0<\mu\le1\), \(m/\mu=O(1)\). Write
\(f(t)=\sqrt{t^2+((m+\mu)/2)^2}-(m-\mu)/2\).
\[
\beta_{\rm DARE}(z)=\frac4{f(\operatorname{dist}(z,[-a,a]))},\qquad |z|=1,
\quad
\|X\|\le\frac{\bar q+\sqrt{\bar q^2+4\bar q/g_0}}2=O(1).
\]
At \(z_j=e^{2\pi\mathrm i j/N}\), with \(N\ge\mu^{-1}\) also satisfying
\eqref{eq:dare-circle-count-app},
and bounded input scales,
\(\mathcal R_{\rm DARE}=\Theta((\mu+1-a)^{-1})\) and
\(\alpha_{\rm DARE}=N^{-1}\sum_j\beta_j
=O(\log(e+\mathcal R_{\rm DARE}))\);
a common \(\beta_j=4/f(1-a)\) instead gives \(\Theta(\mathcal R_{\rm DARE})\).
The other three problems in the same condition admit
\[
\begin{aligned}
\beta_{\rm CARE}(\pm\mu/2+\mathrm iy)
&=\frac4{\sqrt{y^2+m^2/4}-(m-\mu)/2},\\
\beta_{\rm DRE}(z)&=\alpha_H\beta_{\rm CARE}(\alpha_Hz),\\
\beta_{\rm RR}(re^{\mathrm i\theta})
&=\frac{4(1+m)\max\{1,r^{-1}\}}{f(\kappa_r|\sin\theta|)-\delta_r},
\end{aligned}
\]
where RR still requires \(A\) to be invertible and uses
\(r=1\pm\mu/[2(1+m)]\), \(\kappa_r=(r+r^{-1})/2\),
\(\delta_r=|r-r^{-1}|/2\le\mu/2\).
On remaining boundary pieces with \(|z|>\alpha\), use \(4/(|z|-\alpha)\)
for the corresponding operator scale \(\alpha\).
The logarithm follows from \(f(t)=\Theta(t+\mu)\).
These bounds need not be pointwise tight, but nodewise inversion can reduce
\(\alpha_{\rm CARE/DARE/DRE/RR}\) from \(O(\mathcal R_{\rm CARE/DARE/DRE/RR})\) to
\(O(\log(e+\mathcal R_{\rm CARE/DARE/DRE/RR}))\), as this example shows.
\end{remark}

The panel quadrature applies to matrices with arbitrary Jordan structure
and to contours with corners. Under coherent access, the node count does
not multiply the matrix-query bound. The costs of preparing the nodes and
weights are accounted for separately.

Two types of supplied data give explicit choices of
\(\underline\Delta_{\rm ax}\). Under the coefficient assumptions of
Proposition~\ref{prop:factor-care-coefficients-app}, take
\(\underline\Delta_{\rm ax}=\min\{g_0,q_0\}\).
A supplied diagonalization of \(\mathcal H_{\rm CARE}/\alpha_H\)
instead gives \(\underline\Delta_{\rm ax}=\alpha_H\Delta_H/\kappa_V\),
where \(\Delta_H\) is its half-plane gap and \(\kappa_V\) bounds the
eigenvector condition number. These lower bounds need not approximate
\(\Delta_{\rm ax}\) within a constant factor. The nodewise constant-factor
inverse-norm estimates required by Proposition~\ref{thm:local-contour-sum}
are separate inputs.

For approximate Hamiltonian inputs, use
Proposition~\ref{prop:contour-input-errors-app} with the exact right
input \(R=\mathsf I\).
The quadrature, input, and implementation errors must sum to at most
\(\delta\); assigning each at most \(\delta/3\) suffices.
The nodewise smallness conditions in that proposition preserve the inverse
bounds up to constant factors. The same projector-recovery proof then applies.

Weyl--LCHM constructs the same full projector. Its degree and actual
normalization both enter the recovery cost.

\begin{corollary}[CARE through Weyl--LCHM]
\label{cor:care-weyl-recovery-app}
For the CARE projector of Lemma~\ref{lem:care-graph-recovery}, let
\(\widehat H=\mathcal H_{\rm CARE}/\alpha_H\) have the exact-input
access of Proposition~\ref{cor:lchm-branch-block}. Choose its
left-half-plane selector to approximate \(\Pi_-\) to the tolerance
\(\delta\) in \eqref{eq:algebraic-projector-budget-app}, with degree
bound \(d_S\) and normalization \(\alpha_{S_{-,r}}\).
Pad the projector encoding to normalization \(2\alpha_{S_{-,r}}\),
and set \(\alpha_X=16\alpha_{S_{-,r}}\sqrt{1+\|X\|^2}\), with the supplied
solution-norm bound substituted as in
Proposition~\ref{prop:dre-direct-projector-stability-app}.
Then the output has decoded error at most \(\varepsilon\) and query cost
\begin{equation}
Q_{\rm CARE}^{\rm Weyl}
=O\!\left(d_S\alpha_X
\log\!\left(e+\frac{\alpha_X}{\varepsilon}\right)\right)
\label{eq:care-weyl-query-app}
\end{equation}
to the Hamiltonian encoding and its adjoint and controlled versions.
If \(\widehat H=V\Lambda V^{-1}\) satisfies
Theorem~\ref{thm:lchm-explicit-query}, with supplied
\(\kappa_V\ge\|V\|\|V^{-1}\|\), half-plane gap \(\Delta_H>0\),
and sector parameter \(0\le\eta<1\), then
\begin{equation}
Q_{\rm CARE}^{\rm Weyl}
=\widetilde O\!\left(
\frac{\alpha_X}{\Delta_H^2(1-\eta^2)^2}\right).
\label{eq:care-weyl-sector-app}
\end{equation}
Only the recovery logarithm and the degree logarithms in
Theorem~\ref{thm:lchm-explicit-query}, evaluated at \(\delta\),
are suppressed.
\end{corollary}
\begin{proof}
Apply Proposition~\ref{cor:lchm-branch-block} with \(\chi=-\) and
\(R=\mathsf I\). Positive scaling preserves the spectral projector.
The circuit uses \(O(d_S)\) Hamiltonian queries and has normalization
\(\alpha_{S_{-,r}}\) before padding. Since \(\|\Pi_-\|\ge1\)
and \(\delta\le1/8\), this normalization is at least \(7/8\),
so padding by two also bounds the exact projector norm.
Proposition~\ref{prop:dre-direct-projector-stability-app}
gives \eqref{eq:care-weyl-query-app}; substitution of the sector degree
bound proves \eqref{eq:care-weyl-sector-app}.
\end{proof}

The sufficient normalization \(\alpha_{S_{-,r}}=2^r\) in
\eqref{eq:lchm-explicit-normalization} can grow with the degree.
Thus the degree alone does not bound the full recovery cost
polynomially in the inverse gap or target error. Any reduction of this
normalization includes its implementation overhead in each projector
call. For approximate inputs, the selector, input, and circuit errors
of Proposition~\ref{prop:lchm-approximate-inputs} must together fit
within \(\delta\).

\subsection{DARE finite projector, unit-circle quadrature, and construction bound}
\label{app:dare-construction}

The finite stable graph of the symplectic pencil
\cite{VanDooren1981Generalized,LancasterRodman1995Riccati} gives the
same recovery, including when \(L\) is singular.

\newtheorem*{daregraphrestatement}{Lemma~\ref{lem:dare-graph-recovery}}
\begin{daregraphrestatement}
Let \(X\) be the stabilizing solution in
Definition~\ref{def:problem-dare}.
For \(E_1=[\mathsf I;0]\) and \(E_2=[0;\mathsf I]\), the finite
unit-disk projector of \(M-zL\) satisfies
\begin{equation*}
\Pi_<^2=\Pi_<,
\qquad
\operatorname{ran}\Pi_<=\operatorname{ran}\begin{bmatrix}\mathsf I\\X\end{bmatrix}.
\end{equation*}
Its upper block row has full row rank, and
\begin{equation*}
X=(E_2^*\Pi_<)(E_1^*\Pi_<)^+.
\end{equation*}
\end{daregraphrestatement}

\begin{proof}
Set \(F=(\mathsf I+GX)^{-1}A\) and \(V=[\mathsf I;X]\).
The compact DARE gives \(MV=LVF\).
Consider the invertible matrices
\[
T_\ell=
\begin{bmatrix}
(\mathsf I+GX)^{-1}&0\\
-A^*X(\mathsf I+GX)^{-1}&\mathsf I
\end{bmatrix},\qquad
T_r=\begin{bmatrix}\mathsf I&0\\X&\mathsf I\end{bmatrix}.
\]
Using \(X=X^*\), \(G=G^*\), and the compact DARE, block multiplication gives
\begin{equation}
T_\ell(M-zL)T_r
=\begin{bmatrix}
F-z\mathsf I&-z(\mathsf I+GX)^{-1}G\\
0&\mathsf I-zF^*
\end{bmatrix}.
\label{eq:dare-triangularization-app}
\end{equation}
Hence \(\det(M-zL)\) is a nonzero constant multiple of
\(\det(F-z\mathsf I)\det(\mathsf I-zF^*)\).
The first factor has \(n\) roots inside the unit disk, counted with
algebraic multiplicity, because \(\rho(F)<1\).
The second has no roots on or inside that disk.
This proves regularity and absence of finite unit-circle spectrum.
Proposition~\ref{prop:riesz-calculus-app} therefore gives a finite
unit-disk projector \(\Pi_<\) of rank \(n\), even when the pencil
also has infinite modes.

Choose a positively oriented contour \(\Gamma_<\) enclosing precisely
the finite unit-disk spectrum. The graph relation implies, on this contour,
\[
(zL-M)^{-1}LV=V(z\mathsf I-F)^{-1}.
\]
Integration yields
\[
\Pi_<V
=\frac{1}{2\pi\mathrm i}\oint_{\Gamma_<}
(zL-M)^{-1}LV\,\mathrm dz
=V,
\]
since all eigenvalues of \(F\) lie inside the contour.
Both \(\operatorname{range}(V)\) and \(\operatorname{range}(\Pi_<)\)
have dimension \(n\), so they coincide.
Idempotence and this graph range permit
Lemma~\ref{lem:dre-projector-block-recovery-app} to give the full-row
recovery formula. The argument includes infinite modes in the
complementary branch and does not invert \(L\).
\end{proof}

\newtheorem*{dareproblemrestatement}{Problem~\ref{prob:quantum-dare}}
\begin{dareproblemrestatement}
For the LQR data in Definition~\ref{def:problem-dare}, let
\(X\succeq0\) be the stabilizing solution, with \(\|X\|<\infty\).
Given exact block-encodings of \(M,L\) in \eqref{eq:dare-pencil},
their adjoints and controlled versions, with normalizations
\(\alpha_M,\alpha_L\), and a target \(0<\varepsilon\le1\),
construct a block-encoding with decoded matrix \(\widetilde X\)
satisfying \(\|\widetilde X-X\|\le\varepsilon\).
\end{dareproblemrestatement}

\newtheorem*{daremainrestatement}{Theorem~\ref{thm:dare-construction-main}}
\begin{daremainrestatement}
For Problem~\ref{prob:quantum-dare}, supply an upper bound on \(\|X\|\)
and a positively oriented contour \(\Gamma_<\) enclosing exactly the
finite unit-disk spectrum of \(M-zL\). Define
\begin{equation*}
\mathcal R_{\rm DARE}
=\sup_{z\in\Gamma_<}(|z|\alpha_L+\alpha_M)\|(zL-M)^{-1}\|.
\end{equation*}
Under the implementation and normalization assumptions of
Appendix~\ref{app:algebraic-riccati}, the algorithm returns an
\((\alpha_X,a_X,\varepsilon)\) block-encoding of \(X\), using
\begin{equation*}
Q_{\rm DARE}=\widetilde O\!\left(\mathcal R_{\rm DARE}\alpha_X\right)
\end{equation*}
queries to the supplied pencil encodings and their adjoint and
controlled versions. Here \(\alpha_X\) is the actual output
normalization in \eqref{eq:algebraic-projector-output-app};
\(\widetilde O\) suppresses logarithmic dependence on precision and
the supplied scales.
\end{daremainrestatement}

\begin{proof}
Lemma~\ref{lem:dare-graph-recovery} identifies the complete finite
projector \(\Pi_<\). Choose \(\delta\) by
\eqref{eq:algebraic-projector-budget-app}, and apply
Proposition~\ref{thm:affine-pencil-lcu-main} with \(g=1\) and
\(R=\mathsf I\) to a quadrature with error at most \(\delta/2\).
The implemented sum is \(\sum_j\omega_j(z_jL-M)^{-1}L\), with
raw normalization \(\alpha_L\alpha_{\rm DARE}\).
Allocate the remaining \(\delta/2\) to implementation.
Including any declared normalization change, one projector call costs
\[
Q_{\Pi_<}(\delta)=O\!\left(r_{\Pi_<}\mathcal R_{\rm DARE}
\log\!\left(e+
\frac{\mathcal R_{\rm DARE}\alpha_L\alpha_{\rm DARE}}{\delta}\right)\right).
\]
The extra call implementing the right factor \(L\) is absorbed in
this bound. Proposition~\ref{prop:dre-direct-projector-stability-app}
then yields the output normalization and error, with
\begin{equation}
\begin{aligned}
Q_{\rm DARE}=O\!\Bigg(&r_{\Pi_<}\mathcal R_{\rm DARE}\alpha_X
\log\!\left(e+
\frac{\mathcal R_{\rm DARE}\alpha_L\alpha_{\rm DARE}(1+\|X\|^2)}{\varepsilon}\right)\\
&\times\log\!\left(e+\frac{\alpha_X}{\varepsilon}\right)\Bigg).
\end{aligned}
\label{eq:dare-query-expanded-app}
\end{equation}
With the supplied norm bound and \(r_{\Pi_<}=O(1)\), this is the
claimed query complexity.
\end{proof}

Algorithm~\ref{alg:dare-direct-projector-construction} uses the
projector and inverse accuracies in \eqref{eq:algebraic-projector-budget-app},
with the product error budget specified there.

\begin{algorithm}[H]
\caption{DARE: construct a block-encoding of \(X\)}
\label{alg:dare-direct-projector-construction}
\small
\begin{algorithmic}[1]
\Require Exact BEs of \(M,L\) with scales \(\alpha_M,\alpha_L\);
\(0<\varepsilon\le1\).
\Ensure A BE of \(X\) with decoded error \(\le\varepsilon\)
and scale \(\alpha_X\) defined in
Appendix~\ref{app:direct-graph-recovery}.

\State Construct the complete finite unit-disk projector BE:
\[
U_{\Pi_<}\gets\operatorname{RieszBE}
(U_M,U_L,\Gamma_<,1,U_{\mathsf I};\delta).
\]
\State Encode \((E_1^*\Pi_<)^+\) from the complete upper row of
\(U_{\Pi_<}\).
\State Compose the complete lower-row BE from the same
\(U_{\Pi_<}\) with this pseudoinverse BE to encode
\[
X=(E_2^*\Pi_<)(E_1^*\Pi_<)^+.
\]
\Statex \Return \((U_X,\alpha_X,a_X)\).
\end{algorithmic}
\end{algorithm}

The unit-circle rule controls the entire pencil integrand, including
finite-rule residuals from unselected and infinite modes. Its right
factor \(L\) remains in the normalization and the error bound.

\begin{lemma}[DARE normalization on the unit circle]
\label{lem:dare-normalization-factor-app}
For the parameters in \eqref{eq:dare-construction-parameters}
and~\eqref{eq:algebraic-quadrature-normalizations-app}, let
\(c_\beta\) be a uniform constant such that
\(\beta_j\le c_\beta\|(z_jL-M)^{-1}\|\).
Then
\begin{equation}
\alpha_{\rm DARE}\le c_\beta\mathcal R_{\rm DARE}
\sum_j\frac{|\omega_j|}{|z_j|\alpha_L+\alpha_M}.
\label{eq:dare-normalization-weight-bound-app}
\end{equation}
In particular, on the counterclockwise unit circle use the trapezoidal
nodes and coefficients
\[
z_j=e^{2\pi\mathrm i j/m},\qquad
\omega_j=\frac{z_j}{m},\qquad 0\le j<m.
\]
For the standard pencil in \eqref{eq:dare-pencil},
\begin{equation}
\begin{aligned}
\alpha_{\rm DARE}
&\le\frac{c_\beta}{\alpha_M+\alpha_L}\mathcal R_{\rm DARE}
\le\frac{c_\beta}{2}\mathcal R_{\rm DARE},\\
\alpha_{\Pi_<}=\alpha_L\alpha_{\rm DARE}
&\le\frac{c_\beta\alpha_L}{\alpha_M+\alpha_L}
\mathcal R_{\rm DARE}=O(\mathcal R_{\rm DARE}),
\end{aligned}
\label{eq:dare-normalization-factor-app}
\end{equation}
uniformly in \(m\), using the direct LCU normalization without reduction.
\end{lemma}
\begin{proof}
The nodewise estimate
\(\beta_j\le c_\beta\mathcal R_{\rm DARE}/(|z_j|\alpha_L+\alpha_M)\)
gives \eqref{eq:dare-normalization-weight-bound-app}.
As in Proposition~\ref{prop:rr-weighted-quadrature-app}, the circle
has \(|z_j|=1\) and \(\sum_j|\omega_j|=1\).
The standard pencil satisfies \(ME_2=E_2\), \(LE_1=E_1\), so
\(\alpha_M,\alpha_L\ge1\). Multiplication by the encoded \(L\)
gives the full-projector normalization, whether or not \(L\) is invertible.
\end{proof}

A supplied singular-value bound gives an explicit quadrature order.

\begin{proposition}[DARE on the unit circle]
\label{prop:dare-circle-quadrature-app}
For the DARE projector of Lemma~\ref{lem:dare-graph-recovery}, use the
input-access, node inverse-norm, and recovery-bound assumptions stated
at the start of Appendix~\ref{app:algebraic-riccati}, and suppose
\[
0<\underline\eta_{\rm pen}\le
\min_{|z|=1}\sigma_{\min}(zL-M).
\]
Choose the counterclockwise unit circle and set
\begin{equation}
\begin{aligned}
a_{\rm pen}&=\log\!\left(1+
\frac{\underline\eta_{\rm pen}}{2\alpha_L}\right),\\
z_j&=e^{2\pi\mathrm i j/m},\qquad \omega_j=z_j/m,
\qquad 0\le j<m,\\
S_m&=\frac1m\sum_{j=0}^{m-1}z_j(z_jL-M)^{-1}L.
\end{aligned}
\label{eq:dare-circle-rule-app}
\end{equation}
For every integer \(m\ge1\), the full-projector error satisfies
\begin{equation}
\|S_m-\Pi_<\|
\le\frac{4\alpha_Le^{a_{\rm pen}}}
{\underline\eta_{\rm pen}(e^{a_{\rm pen}m}-1)}.
\label{eq:dare-circle-error-app}
\end{equation}
For any \(\delta>0\), error at most \(\delta/2\) is obtained whenever
\begin{equation}
m\ge\left\lceil
\frac1{a_{\rm pen}}\log\!\left(1+
\frac{8\alpha_Le^{a_{\rm pen}}}{\underline\eta_{\rm pen}\delta}\right)
\right\rceil.
\label{eq:dare-circle-count-app}
\end{equation}
The construction parameters obey
\begin{equation}
\begin{aligned}
\mathcal R_{\rm DARE}&\le
\frac{\alpha_M+\alpha_L}{\underline\eta_{\rm pen}},&
\alpha_{\rm DARE}&\le\frac{c_\beta}{\underline\eta_{\rm pen}},\\
\alpha_{\Pi_<}&\le\frac{c_\beta\alpha_L}{\underline\eta_{\rm pen}},&
\|\Pi_<\|&\le\frac{\alpha_L}{\underline\eta_{\rm pen}},
\end{aligned}
\label{eq:dare-circle-parameters-app}
\end{equation}
where \(c_\beta\) is the uniform constant in
Lemma~\ref{lem:dare-normalization-factor-app}.
For \(0<\varepsilon\le1\), take
\(\delta=\varepsilon/[8(1+\|X\|^2)]\), using a supplied upper
bound on \(\|X\|\) consistently wherever it occurs.
Taking \(m\) equal to the right-hand side of
\eqref{eq:dare-circle-count-app} gives
\begin{equation}
m=O\!\left(\frac{\alpha_L}{\underline\eta_{\rm pen}}
\log\!\left(e+
\frac{\alpha_L(1+\|X\|^2)}
{\underline\eta_{\rm pen}\varepsilon}\right)\right).
\label{eq:dare-circle-count-epsilon-app}
\end{equation}
With projector implementation error at most \(\delta/2\), the output
has decoded error at most \(\varepsilon\), normalization
\begin{equation}
\alpha_X=8\alpha_{\Pi_<}\sqrt{1+\|X\|^2}
=O\!\left(
\frac{\alpha_L\sqrt{1+\|X\|^2}}{\underline\eta_{\rm pen}}\right),
\label{eq:dare-circle-output-app}
\end{equation}
and matrix-query cost
\begin{equation}
\begin{aligned}
Q_{\rm DARE}=O\!\Bigg(&
\mathcal R_{\rm DARE}\alpha_X
\log\!\left(e+
\frac{\mathcal R_{\rm DARE}\alpha_{\Pi_<}}{\delta}\right)
\log\!\left(e+\frac{\alpha_X}{\varepsilon}\right)\Bigg).
\end{aligned}
\label{eq:dare-circle-query-app}
\end{equation}
Supplied norm bounds are used consistently in the precision choices,
output normalization, and logarithms.
\end{proposition}
\begin{proof}
For the complete pencil integrand
\(e^{\mathrm it}(e^{\mathrm it}L-M)^{-1}L\), set
\(t=u+\mathrm iv\) with \(|v|\le a_{\rm pen}\). Then
\[
\|(e^{\mathrm it}-e^{\mathrm iu})L\|
\le(e^{a_{\rm pen}}-1)\alpha_L=\underline\eta_{\rm pen}/2.
\]
Lemma~\ref{lem:inverse-input-perturbation}, with
\(T=e^{\mathrm iu}L-M\) and
\(\Delta T=(e^{\mathrm it}-e^{\mathrm iu})L\), bounds the inverse
norm by \(2/\underline\eta_{\rm pen}\).
The integrand is therefore analytic on a neighborhood of the closed strip,
with norm at most \(2\alpha_Le^{a_{\rm pen}}/\underline\eta_{\rm pen}\).
Lemma~\ref{lem:analytic-quadrature-app} gives
\eqref{eq:dare-circle-error-app} and~\eqref{eq:dare-circle-count-app}
for the entire projector, including infinite modes, without inverting \(L\).
The unit-circle inverse bound, Lemma~\ref{lem:dare-normalization-factor-app},
and the exact integral give \eqref{eq:dare-circle-parameters-app}.

For a finite eigenpair \(Mv=\lambda Lv\), with \(\|v\|=1\)
and \(|\lambda|<1\), its nearest unit-circle point \(\zeta\) gives
\(\underline\eta_{\rm pen}\le\|(\zeta L-M)v\|
\le(1-|\lambda|)\|L\|\le\alpha_L\).
Hence \(a_{\rm pen}=\Theta(\underline\eta_{\rm pen}/\alpha_L)\),
so the displayed choice of \(m\) proves
\eqref{eq:dare-circle-count-epsilon-app}.
Proposition~\ref{prop:dre-direct-projector-stability-app}, with
\eqref{eq:algebraic-projector-output-app}
and~\eqref{eq:algebraic-projector-budget-app}, gives the output
scale and query cost.
\end{proof}

The node count describes the supplied coherent quadrature data; it
does not multiply the matrix-query bound. The lower bound
\(\underline\eta_{\rm pen}\) alone need not supply the required
constant-factor inverse-norm estimate at every node. Under the
coefficient assumptions of
Proposition~\ref{prop:factor-dare-resolvent-app}, one may choose
\(\underline\eta_{\rm pen}\) as the reciprocal of the right-hand
side of \eqref{eq:factor-dare-resolvent-app}; any conservatism in
this choice remains in the displayed bounds.

For approximate pencil inputs, apply
Proposition~\ref{prop:contour-input-errors-app} with exact right input
\(R=\mathsf I\). The implemented sum contains
\((z_j\widetilde L-\widetilde M)^{-1}\widetilde L\), so both the
shifted inverse and the right factor contribute input error.
Under its nodewise smallness conditions, allocate quadrature, input,
and implementation errors a total budget \(\delta\), then use the
same recovery proposition. No spectral assumption on the perturbed
pencil over the entire contour is needed.

\subsection{Cayley construction under additional invertibility}
\label{app:dare-cayley}

When \(L\) is invertible, the Cayley transformation converts the
DARE disk selection to a half-plane selection
\cite{Higham2008Functions}. The following implementation includes
both the formation of \(L^{-1}M\) and the shifted-inverse cost.

\begin{proposition}[DARE through a Cayley transformation]
\label{prop:dare-cayley-construction-app}
For the DARE projector of Lemma~\ref{lem:dare-graph-recovery}, use
the recovery choices \eqref{eq:algebraic-projector-output-app}
and~\eqref{eq:algebraic-projector-budget-app}.
Suppose \(L\) is invertible and set \(S_D=L^{-1}M\).
Then \(S_D+\mathsf I\) is invertible, and
\begin{equation}
C_D=(S_D-\mathsf I)(S_D+\mathsf I)^{-1},\qquad
\Pi_<=\mathfrak R_-[1;C_D,\mathsf I].
\label{eq:dare-cayley-projector-app}
\end{equation}
Suppose an encoding of \(S_D\), with normalization
\(\alpha_{S_D}\), decoded error \(\varepsilon_{S_D}\), and cost
\(Q_{S_D}\) base matrix queries, is supplied. Assume a constant-factor
upper bound on \(\|(S_D+\mathsf I)^{-1}\|\),
\(\varepsilon_{S_D}\|(S_D+\mathsf I)^{-1}\|\le1/2\), and the
inverse access of Proposition~\ref{prop:local-inverse}.
For \(\varepsilon_+>0\) with
\((\alpha_{S_D}+1)\varepsilon_+\le1\), a Cayley encoding has
\begin{equation}
\begin{aligned}
\alpha_C&=1+2\beta_+,
\qquad \beta_+=O(\|(S_D+\mathsf I)^{-1}\|),\quad
\beta_+\ge2\|(S_D+\mathsf I)^{-1}\|,\\
\varepsilon_C&\le4\|(S_D+\mathsf I)^{-1}\|^2\varepsilon_{S_D}
+2\varepsilon_+,\\
Q_C&=O\!\left(Q_{S_D}(\alpha_{S_D}+1)\|(S_D+\mathsf I)^{-1}\|
\log\!\left(e+\frac{\|(S_D+\mathsf I)^{-1}\|}{\varepsilon_+}\right)\right).
\end{aligned}
\label{eq:dare-cayley-input-app}
\end{equation}
For \(\widehat C_D=C_D/\alpha_C\), choose a left-half-plane
selector under Proposition~\ref{prop:lchm-approximate-inputs}, with
degree bound \(d_S\) and normalization \(\alpha_{S_{-,r}}\), such that
\begin{equation}
\varepsilon_{\rm sel}
+\mathcal L_{-,r}\frac{\varepsilon_C}{\alpha_C}
+\alpha_{S_{-,r}}\varepsilon_{\rm circ}\le\delta,
\label{eq:dare-cayley-branch-error-app}
\end{equation}
where \(\delta\) is given in
\eqref{eq:algebraic-projector-budget-app}. Pad its normalization to
\(2\alpha_{S_{-,r}}\). Direct recovery then has decoded error at
most \(\varepsilon\), output scale and query cost
\begin{equation}
\begin{aligned}
\alpha_X&=16\alpha_{S_{-,r}}\sqrt{1+\|X\|^2},\\
Q_X&=O\!\left(Q_Cd_S\alpha_X
\log\!\left(e+\frac{\alpha_X}{\varepsilon}\right)\right),
\end{aligned}
\label{eq:dare-cayley-recovery-app}
\end{equation}
with the supplied solution-norm bound substituted in \(\alpha_X\).
\end{proposition}
\begin{proof}
The spectrum of \(S_D\) is the generalized spectrum of \((M,L)\)
and avoids the unit circle by Lemma~\ref{lem:dare-graph-recovery}.
Since
\(\operatorname{Re}((z-1)/(z+1))=(|z|^2-1)/|z+1|^2\),
analytic functional calculus maps the inside and outside spectral
subspaces to the left and right half-planes, proving
\eqref{eq:dare-cayley-projector-app}.
The input-error assumption gives
\(\sigma_{\min}(\widetilde S_D+\mathsf I)
\ge\tfrac12\sigma_{\min}(S_D+\mathsf I)\).
Apply Proposition~\ref{prop:local-inverse} at half the reciprocal
of the supplied inverse-norm upper bound, with inverse error
\(\varepsilon_+\) and normalization \(\beta_+\).
Lemma~\ref{lem:inverse-input-perturbation} bounds its total decoded
inverse error by
\(2\|(S_D+\mathsf I)^{-1}\|^2\varepsilon_{S_D}+\varepsilon_+\).
The identity \(C_D=\mathsf I-2(S_D+\mathsf I)^{-1}\) now gives
\eqref{eq:dare-cayley-input-app} by exact affine LCU; any finite
precision composition error is added to \(\varepsilon_C\).
The declared \(\alpha_C\) bounds both the ideal and decoded norms.
Proposition~\ref{prop:lchm-approximate-inputs}, with \(R=\mathsf I\)
and normalized input error \(\varepsilon_C/\alpha_C\), constructs
the projector using \(O(d_S)\) Cayley calls. Padding by two also
bounds its exact norm, as in Corollary~\ref{cor:care-weyl-recovery-app}.
Proposition~\ref{prop:dre-direct-projector-stability-app}, with
\(P(t)=X\), \(\mathcal E(t)=\Pi_<\), and the supplied solution
bound, proves \eqref{eq:dare-cayley-recovery-app}.
\end{proof}

Fix \(\beta_+\) and \(\alpha_C\) from the supplied inverse-norm
bound before choosing the selector for the ideal normalized matrix.
Its sensitivity then determines the permitted
\(\varepsilon_C/\alpha_C\), and hence the input and inverse
accuracies. Apply both the Cayley map and normalization when computing
the half-plane gap. If only \(M,L\) are encoded, first construct
\(L^{-1}\) by Proposition~\ref{prop:local-inverse} and multiply by
\(M\); include this cost in \(Q_{S_D}\) and its decoded error in
\(\varepsilon_{S_D}\).

Alternatively, Proposition~\ref{prop:care-rectangle-quadrature-app}
applies with \(\mathcal H_{\rm CARE},\alpha_H,
\underline\Delta_{\rm ax}\) replaced by
\(C_D,\alpha_C,\underline\Delta_{{\rm ax},C}\), provided
\(0<\underline\Delta_{{\rm ax},C}\le
\inf_{\omega\in\mathbb R}\sigma_{\min}(\mathrm i\omega\mathsf I-C_D)\)
and the node inverse access are supplied. Each Cayley call incurs
\(Q_C\); approximate inputs use
Proposition~\ref{prop:contour-input-errors-app} on the fixed ideal
contour. A unit-circle eigenvalue gap alone does not give this
singular-value bound, which is divided by \(\alpha_C\) upon normalization.

\clearpage
\section{Classical outputs and application complexity}
\label{app:control-applications}

This appendix proves the application results of
Section~\ref{sec:control-applications}: classical feedback, the stable
RPA comparison, and thermal-network input and parameter bounds.
The proofs reuse the contour, inverse, and recovery estimates of
Appendices~\ref{app:algebraic-riccati},
\ref{app:dre-recovery}, and~\ref{app:dre-complexity},
with the required instance-specific substitutions.
Appendix~\ref{app:hbc-numerics} gives the HBC matrix reconstruction
and closed-loop validation.

\subsection{Finite-horizon continuous LQR feedback}
\label{app:lqr-classical-output}

\newtheorem*{lqroutputrestatement}{Proposition~\ref{prop:lqr-be-classical-main}}
\begin{lqroutputrestatement}
For \eqref{eq:lqr-continuous-demonstration}, suppose the input encodings
are exact and the reverse-time DRE satisfies the structural, weighted-block
access and general inverse assumptions of Appendix~\ref{app:dre-complexity}.
Use direct sums and products with the raw normalizations, without
normalization reduction. Supply a bound on the solution over the
requested interval and use the output normalization
\(\alpha_{P(s)}\) specified in that appendix.
For \(0<\delta\le1/3\) and
\(0<\varepsilon_u<\alpha_b\alpha_{P(s)}\alpha_x/r\), the algorithm
returns a classical estimate satisfying
\begin{equation*}
|\widehat u(t)+r^{-1}b^*P(s)x|\le\varepsilon_u
\end{equation*}
with probability at least \(1-\delta\), using a total of
\begin{equation*}
Q_u=O\!\left(
\frac{\alpha_b\alpha_x\alpha_{P(s)}}{r\varepsilon_u}\,
\log\frac1\delta
\right)Q_{\rm DRE}
\end{equation*}
queries to the original coefficient and state block-encodings, their
adjoints and controlled versions. Here \(Q_{\rm DRE}\) is the
full query cost of one DRE solution circuit at the accuracy specified
in \eqref{eq:lqr-continuous-outer-accuracy-app}, including its actual
normalizations and inverse thresholds in
Proposition~\ref{prop:dre-general-resources-app}.
The compact bound \eqref{eq:dre-query-main} applies under its additional
normalization assumptions. For a tolerance at least the displayed output bound, or a
known zero factor, returning zero suffices.
\end{lqroutputrestatement}
\begin{proof}
Block-encoding multiplication constructs \(G=bb^*/r\) with scale
\(\alpha_b^2/r\). The time convention in
Proposition~\ref{prop:lqr-conventions-app} uses
\((-A,-G,-Q,P_T)\) as the forward DRE data at \(s=T-t\).
Consequently the Hamiltonian and initial column are exactly those in
Algorithm~\ref{alg:lqr-be-classical}. Direct block sums have, for example,
scales \(\alpha_A+\alpha_b^2/r+\alpha_Q\) and \(1+\alpha_{P_T}\);
each uses a constant number of original input calls.

Apply Proposition~\ref{prop:dre-general-resources-app} at solution
accuracy \(\varepsilon\), and denote the full solution-circuit query
cost by \(Q_{\rm DRE}\). Each lifted-input query uses a constant
number of original coefficient or initial-data calls.
The general bound \eqref{eq:dre-general-query-expanded-app} includes
the initialization pseudoinverse, upper-projector-block pseudoinverse
and their product normalization. Using the raw scales sets its
normalization-reduction overheads to one; no initialization norm
calibration is assumed.

For a unitary signal element \(z\), a Hadamard test has success
probability \((1+\operatorname{Re}z)/2\); changing the control phase
gives the imaginary part. Amplitude estimation and median amplification
estimate each part to error \(\tau\) using
\(O(\tau^{-1}\log(1/\delta))\) controlled unitary and adjoint calls
\cite{BrassardHoyerMoscaTapp2002AmplitudeEstimation}.
Thus reading \eqref{eq:lqr-feedback-amplitude} to physical signal error
\(r\varepsilon_u/4\) requires
\[
O\!\left(\frac{\alpha_b\alpha_{P(s)}\alpha_x}{r\varepsilon_u}
\log\frac1\delta\right)
\]
uses of its signal circuit. For complex data, the real and imaginary
estimates share the error and failure probability, changing only
absolute constants.

The decoded solution error changes the physical control by at most
\(\alpha_b\alpha_x\varepsilon/r\). It suffices to choose
\begin{equation}
0<\varepsilon\le1,\qquad
\alpha_b\alpha_x\varepsilon\le r\varepsilon_u/4.
\label{eq:lqr-continuous-outer-accuracy-app}
\end{equation}
For example,
\(\varepsilon=r\varepsilon_u/(4\alpha_b\alpha_x+r\varepsilon_u)\)
meets both inequalities. Use this accuracy in
\eqref{eq:dre-outer-budget-app}. The solution and readout errors then
sum to at most \(\varepsilon_u/2\), leaving room for input and
classical-arithmetic approximations.

Each signal circuit makes one call to the solution encoding and one
each to \(U_b^*\) and \(U_x\), hence costs \(O(Q_{\rm DRE})\)
original input queries. Multiplying by the amplitude-estimation
count proves \eqref{eq:lqr-continuous-total}.
All three DRE precision logarithms remain inside \(Q_{\rm DRE}\),
evaluated at the sufficient solution accuracy just chosen.
The direct recovery cost already contains one factor \(\alpha_{P(s)}\),
and classical readout contributes another. Under the compact theorem's
calibration assumptions, the resulting bound is
\[
Q_u=\widetilde O\!\left(
\frac{\alpha_b\alpha_x}{r\varepsilon_u}
\mathcal R_{\rm DRE}\kappa(\Pi_+R_0)\alpha_{P(s)}^2
\log\frac1\delta\right).
\]
Without those assumptions, the actual initial-column scale and supplied
inverse threshold remain in \(Q_{\rm DRE}\).
\end{proof}

Approximate inputs use the same outer allocation. Take each displayed
normalization to bound the ideal and decoded factors, enlarging its
implemented scale if necessary. With decoded vector errors
\(\varepsilon_b,\varepsilon_x\), a telescoping product expansion
bounds the signal error by
\begin{equation}
\alpha_b\alpha_x\varepsilon
+\alpha_{P(s)}(\alpha_x\varepsilon_b+\alpha_b\varepsilon_x).
\label{eq:lqr-continuous-product-error-app}
\end{equation}
Here \(\varepsilon\) is the total solution error relative to the ideal
problem, including coefficient errors. Requiring this expression to
be at most \(r\varepsilon_u/4\) preserves the stated guarantee.
For instance, \(\|\widetilde G-G\|\le2\alpha_b\varepsilon_b/r\)
bounds the Gramian input error. The block sum then bounds the
Hamiltonian perturbation, and
Propositions~\ref{prop:contour-input-errors-app}
and~\ref{prop:dre-approximate-inputs-app} feed these errors into the
same common internal accuracy. This directly reuses the existing
quadrature, inverse and recovery analysis.

For multiple inputs, the formula is \(-R_c^{-1}B^*P(s)x\).
Proposition~\ref{prop:local-inverse} and
Lemma~\ref{lem:inverse-input-perturbation} control the additional
control-space inverse and its input error. Coordinate selection gives
the same scalar estimation task for each requested component, with
the actual control-space and repeated solution-circuit costs included.

\subsection{Stable RPA family and comparison proof}
\label{app:rpa-review}

\newtheorem*{rparestated}{Proposition~\ref{prop:rpa-review-comparison}}
\begin{rparestated}
\RPAComparisonStatement
\end{rparestated}
\begin{proof}
Write \(E_1=[\mathsf I;0]\), \(E_2=[0;\mathsf I]\), and let
\(\Pi_-\) denote the stable half-plane projector.
For~\eqref{eq:rpa-review-family}, direct substitution gives
\begin{equation}
\begin{aligned}
T&=\operatorname{diag}(t,0),\qquad
t=-\frac{1-u}{1+\sqrt{2u-u^2}},\\
\Pi_-&=\begin{bmatrix}\mathsf I\\T\end{bmatrix}
 (\mathsf I-T^2)^{-1}\begin{bmatrix}\mathsf I&-T\end{bmatrix}.
\end{aligned}
\label{eq:rpa-review-stable-projector}
\end{equation}
The graph \([\mathsf I;T]\) is invariant with eigenvalues
\(-\sqrt{2u-u^2}\) and \(-1\); the complementary graph
\([T;\mathsf I]\) has their positives. The full upper block row is
\begin{equation}
\begin{aligned}
E_1^*\Pi_-&=(\mathsf I-T^2)^{-1}
 \begin{bmatrix}\mathsf I&-T\end{bmatrix},\\
\|(E_1^*\Pi_-)^+\|&=1.
\end{aligned}
\label{eq:rpa-review-recovery-scales}
\end{equation}
Its two singular values are \(\sqrt{1+t^2}/(1-t^2)\) and \(1\);
the first is at least one because \(|t|<1\).
Both constructions therefore recover
\(T=(E_2^*\Pi_-)(E_1^*\Pi_-)^+\) by
Lemma~\ref{lem:dre-projector-block-recovery-app}.
With \(J=\operatorname{diag}(\mathsf I_2,-\mathsf I_2)\), the identity
\(J(\mathrm iy\mathsf I-\mathcal H)=K+\mathrm iyJ\)
gives the axis gap \(u\), with equality at \(y=0\), and
\(\alpha_H=\|K\|=2-u\). These identities establish the required
graph and separation properties directly; the LQR sign assumptions
are not needed.

Use the rectangle and panel rule of
Proposition~\ref{prop:care-rectangle-quadrature-app} with
\(\alpha_H=2-u\), \(d=u/2\), and right input \(\mathsf I\).
Its quadrature proof depends only on the supplied axis gap and matrix
normalization, so \eqref{eq:care-rectangle-bounds-app}
and~\eqref{eq:care-rectangle-nodes-app} give full-projector error
at most \(\delta/2\) with
\(m=O(u^{-1}\log(e+1/(u\delta)))\) nodes.
Only the sharper local normalization needs verification.
After permuting the coordinates, the coupled resolvent block is
\begin{equation}
\frac{1}{z^2-(2u-u^2)}
\begin{bmatrix}z-1&-(1-u)\\1-u&z+1\end{bmatrix}.
\label{eq:rpa-review-resolvent}
\end{equation}
On the right edge \(z=-u/2+\mathrm iy\), \(|y|\le1\), its numerator
is bounded above and below in norm, and its denominator has magnitude
\(\Theta(u+y^2)\). Hence the full inverse norm is
\(\Theta((u+y^2)^{-1})\); the other contour segments contribute
\(O(1)\) to its integral. The Neumann estimate in the same quadrature
proof bounds inverse-norm ratios within a panel by three. Positive
Gauss--Legendre weights integrate constants exactly, so their weighted
inverse-norm sum is within fixed factors of the contour integral.
Choose actual node scales satisfying
\(4\le\beta_j/\|(z_j\mathsf I-\mathcal H)^{-1}\|\le c_\beta\)
for fixed \(c_\beta\ge4\).
For example, the computable choice
\[
\beta_j=4\max\!\left\{
\frac{\sqrt{2|z_j|^2+2+2(1-u)^2}}{|z_j^2-(2u-u^2)|},
\frac1{|z_j-1|},\frac1{|z_j+1|}
\right\}
\]
satisfies
\(4\|(z_j\mathsf I-\mathcal H)^{-1}\|
\le\beta_j\le4\sqrt2\|(z_j\mathsf I-\mathcal H)^{-1}\|\).
The local construction thus has
\begin{equation}
\mathcal R_{\rm CARE}=\Theta(u^{-1}),\qquad
\alpha_{\Pi_-}=\alpha_{\rm CARE}
=\sum_{j=1}^{m}|\omega_j|\beta_j=\Theta(u^{-1/2}).
\label{eq:rpa-review-local-scales-app}
\end{equation}

For the uniform construction on this same stable contour, use the
node inverse threshold in the proof of
\cite[Theorem~C.4]{RodenasRuizZhaoLee2026NonlinearMatrixEquations}.
The bound \(\kappa=2(1+2\sqrt2)\alpha_H/u\) is valid at every node,
and its inverse normalization is
\begin{equation}
\beta_j=\frac{8\kappa}{3(\alpha_H+|z_j|)}=\Theta(u^{-1}),\qquad
\alpha_{\Pi_-}=\sum_{j=1}^{m}|\omega_j|\beta_j=\Theta(u^{-1}).
\label{eq:rpa-review-uniform-scales}
\end{equation}
Both projector constructions cost \(\widetilde O(u^{-1})\) queries
per call at the required precision. The displayed LCU scales are
used, with no normalization reduction; they are uniform in the
quadrature order.

Apply Proposition~\ref{prop:dre-direct-projector-stability-app}
to both constructions with the supplied solution bound \(M=1\),
using its common precision budget and
\(\delta=\varepsilon_T/16\). Allocate half of \(\delta\) to
quadrature and half to the projector circuit. The normalization and
recovery-query rule in Appendix~\ref{app:direct-graph-recovery} give
\begin{equation}
\alpha_T=8\sqrt2\,\alpha_{\Pi_-},\qquad
Q_T=\widetilde O(\mathcal R_{\rm CARE}\alpha_T).
\label{eq:rpa-review-local-query}
\end{equation}
Substituting the two actual projector scales yields
\(Q_T=\widetilde O(u^{-2})\) for uniform thresholds and
\(Q_T=\widetilde O(u^{-3/2})\) for local inverse scales, with the
same decoded error \(\varepsilon_T\).

Finally, apply the scalar-readout argument of
Appendix~\ref{app:lqr-classical-output} with both states equal to
\(e_1\) and prefactor \((1-u)/(4V)\).
Taking \(\varepsilon_T=2V\varepsilon_c/(1-u)\le1/2\) leaves at most
\(\varepsilon_c/2\) encoding bias. The readout requires
\(O((1-u)\alpha_T/(V\varepsilon_c))\) solution calls for the
remaining half of the error budget. Substituting the two values of
\(\alpha_T\) and \(Q_T\) gives the energy bounds in
Table~\ref{tab:rpa-review-comparison}.
\end{proof}

\subsection{Thermal-network input constructions}
\label{app:thermal-inputs}

The local thermal coefficients admit direct encodings with scales
that remain bounded as nodes are added.

\begin{lemma}[Thermal-network inputs]
\label{lem:thermal-inputs-app}
Let \(n\ge2\), \(\nu,q,r>0\), and \(\kappa\ge0\).
Let \(\mathcal L\) be the path Laplacian defined in
Subsection~\ref{subsec:thermal-model}, and set
\begin{equation}
A_{\rm c}=-\nu\mathsf I-\kappa\mathcal L,\qquad
b_{\rm c}=e_1,\qquad P_T=0.
\label{eq:thermal-inputs-app}
\end{equation}
The coefficient inputs have block-encodings with normalizations
\begin{equation}
\alpha_{A_{\rm c}}=\nu+4\kappa,\qquad
\alpha_{b_{\rm c}}=1.
\label{eq:thermal-continuous-scales-app}
\end{equation}
The state weight \(q\mathsf I\) has normalization \(q\), and the
zero terminal input has a constant-cost encoding.
These constructions use \(\operatorname{polylog}n\) gates,
with polylogarithmic overhead in the inverse decoded input accuracy
for scalar-rotation synthesis.
\end{lemma}
\begin{proof}
The swaps \(S_{0}\) and \(S_{1}\) exchange the pairs
\((1,2),(3,4),\ldots\) and \((2,3),(4,5),\ldots\), respectively.
Both fix every unpaired endpoint. Then
\[
\mathcal L=(\mathsf I-S_{0})+(\mathsf I-S_{1}),\qquad
A_{\rm c}=-(\nu+2\kappa)\mathsf I
                    +\kappa S_{0}+\kappa S_{1}.
\]
The swaps use reversible index arithmetic and endpoint comparisons on
\(\lceil\log_2n\rceil\) qubits. Fixing all padded indices preserves
the physical subspace and extends \(A_{\rm c}\) by
\(-\nu\mathsf I\) on its complement. A three-term LCU gives
\eqref{eq:thermal-continuous-scales-app}; preparing \(e_{1}\),
encoding \(q\mathsf I\), and encoding the zero matrix use constant
additional operations. Approximating the scalar LCU rotations to a
prescribed decoded input error costs a polylogarithmic number of gates
in the inverse error. Thus each call to the base matrix oracle itself
has the stated polylogarithmic dimension cost, including its index
circuit, with scalar synthesis charged at the accuracy required by
the consuming construction. The same decomposition and the constant null vector of the
path Laplacian give
\begin{equation}
\begin{gathered}
0\preceq\mathcal L\preceq4\mathsf I,\qquad
\sigma(A_{\rm c})\subseteq[-(\nu+4\kappa),-\nu],\\
\|(A_{\rm c})^{-1}\|=\nu^{-1}.
\end{gathered}
\label{eq:thermal-dissipation}
\end{equation}
\end{proof}

\subsection{DRE parameter bounds}
\label{app:thermal-dre}

The continuous model supplies the spectral gap and the absolute
initialization bound required by Appendix~\ref{app:dre-complexity}.
The estimate below allows the rank-one input
\(G_{\rm c}=b_{\rm c}(b_{\rm c})^*/r\).

\begin{lemma}[DRE parameters for the thermal network]
\label{lem:thermal-dre-parameters-app}
Let \(A_{\rm c}=(A_{\rm c})^*\preceq-\nu\mathsf I_{n}\),
\(G_{\rm c}\succeq0\), and let the state weight be
\(q\mathsf I_n\), where \(\nu,q>0\).
The physical primal and dual CAREs have stabilizing solutions
\(X_{\rm c}\succ0\) and \(Y_{\rm c}\succeq0\), respectively,
and the finite-horizon solution with \(P_{T}=0\) satisfies
\begin{equation}
\begin{aligned}
\|X_{\rm c}\|&\le\frac{q}{2\nu},&
\|Y_{\rm c}\|&\le\frac{\|G_{\rm c}\|}{2\nu},&
0&\preceq P(s)\preceq X_{\rm c},\qquad s\ge0.
\end{aligned}
\label{eq:thermal-dre-solution-bounds-app}
\end{equation}
For the reverse-time Hamiltonian and initial column
\begin{equation}
\mathcal H_{{\rm DRE}}
=\begin{bmatrix}-A_{\rm c}&G_{\rm c}\\
q\mathsf I_{n}&A_{\rm c}\end{bmatrix},
\qquad R_{0}=E_{1},
\label{eq:thermal-dre-hamiltonian-app}
\end{equation}
the imaginary-axis resolvent satisfies
\begin{equation}
\sup_{\omega\in\mathbb R}
\|(\mathrm i\omega\mathsf I_{2n}
-\mathcal H_{{\rm DRE}})^{-1}\|
\le\left(1+\frac{\|G_{\rm c}\|}{2\nu}\right)^2
\left(\frac1{\nu}+\frac{q}{\nu^2}\right).
\label{eq:thermal-dre-axis-bound-app}
\end{equation}
Let \(\Pi_{+}\) be its right-half-plane projector. Then
\begin{equation}
\begin{aligned}
\Pi_+R_0&=\begin{bmatrix}\mathsf I_{n}\\X_{\rm c}\end{bmatrix}
 (\mathsf I_{n}+Y_{\rm c}X_{\rm c})^{-1},\\
\sigma_{\min}(\Pi_+R_0)&\ge
 \left(1+\frac{q\|G_{\rm c}\|}{4\nu^2}\right)^{-1},\\
\|\Pi_+R_0\|&\le\sqrt{1+\left(\frac{q}{2\nu}\right)^2}
 \left(1+\frac{q\|G_{\rm c}\|}{4\nu^2}\right),\\
\kappa(\Pi_+R_0)&\le\sqrt{1+\left(\frac{q}{2\nu}\right)^2}
 \left(1+\frac{q\|G_{\rm c}\|}{4\nu^2}\right)^2.
\end{aligned}
\label{eq:thermal-dre-initialization-app}
\end{equation}
\end{lemma}
\begin{proof}
Since \(A_{\rm c}\) is Hurwitz, the primal and dual LQR problems
satisfy the stabilizability and detectability conditions for their
stabilizing semidefinite solutions
\cite{LancasterRodman1995Riccati,Mehrmann1991AutonomousLQ}.
The positive state weight makes \(X_{\rm c}\) positive definite.
Zero control and
\(\|e^{sA_{\rm c}}\|\le e^{-\nu{}s}\) give
\[
X_{\rm c}\preceq
\int_0^\infty e^{sA_{\rm c}}q\mathsf I
 e^{sA_{\rm c}}\,\mathrm ds
\preceq\frac{q}{2\nu}\mathsf I_{n}.
\]
Applying the same comparison to the dual problem, whose state weight
is \(G_{\rm c}\), bounds \(Y_{\rm c}\).
The finite-horizon value with zero terminal cost is nonnegative and
does not exceed the infinite-horizon value~\cite[Sections~3.1--3.2]{AndersonMoore2007OptimalControl}. Together with
Proposition~\ref{prop:lqr-conventions-app}, this proves
\eqref{eq:thermal-dre-solution-bounds-app} and global existence of
\(P(s)\) for \(s\ge0\).

The dual equation gives
\(-A_{\rm c}Y_{\rm c}-Y_{\rm c}A_{\rm c}
+q(Y_{\rm c})^2=G_{\rm c}\). Hence
\begin{equation}
\begin{bmatrix}\mathsf I_{n}&Y_{\rm c}\\0&\mathsf I_{n}\end{bmatrix}
\mathcal H_{{\rm DRE}}
\begin{bmatrix}\mathsf I_{n}&-Y_{\rm c}\\0&\mathsf I_{n}\end{bmatrix}
=\begin{bmatrix}
-A_{\rm c}+qY_{\rm c}&0\\
q\mathsf I_{n}&A_{\rm c}-qY_{\rm c}
\end{bmatrix}.
\label{eq:thermal-dre-shear-app}
\end{equation}
The first diagonal block is Hermitian and at least
\(\nu\mathsf I_{n}\); the second is its negative.
Their imaginary-axis inverses have norm at most \(1/\nu\),
and the off-diagonal block of the triangular inverse has norm at
most \(q/\nu^2\). Each shear factor has norm at most
\(1+\|Y_{\rm c}\|\). This proves
\eqref{eq:thermal-dre-axis-bound-app}.

The right spectral branch of \(\mathcal H_{{\rm DRE}}\) is the
stable graph of the physical CARE Hamiltonian, with the sign reversed.
Its complementary graph is
\(\operatorname{ran}[-Y_{\rm c};\mathsf I_{n}]\).
Lemma~\ref{lem:factor-positive-graphs-app} therefore gives the formula
for \(\Pi_+R_0\). Since \([\mathsf I_{n};X_{\rm c}]\) has smallest
singular value at least one,
\[
\sigma_{\min}(\Pi_+R_0)\ge
\|\mathsf I_{n}+Y_{\rm c}X_{\rm c}\|^{-1}
\ge(1+\|Y_{\rm c}\|\|X_{\rm c}\|)^{-1}.
\]
The identity
\[
(\mathsf I_{n}+Y_{\rm c}X_{\rm c})^{-1}
=\mathsf I_{n}-Y_{\rm c}(X_{\rm c})^{1/2}
\bigl(\mathsf I_{n}+(X_{\rm c})^{1/2}Y_{\rm c}
 (X_{\rm c})^{1/2}\bigr)^{-1}(X_{\rm c})^{1/2}
\]
bounds its norm by \(1+\|Y_{\rm c}\|\|X_{\rm c}\|\).
Using \eqref{eq:thermal-dre-solution-bounds-app} proves the norm bound;
dividing it by the supplied smallest-singular-value bound proves the
condition-number bound in \eqref{eq:thermal-dre-initialization-app}.
\end{proof}

\newtheorem*{thermaldrerestatement}{Corollary~\ref{cor:thermal-dre-main}}
\begin{thermaldrerestatement}
For \eqref{eq:thermal-continuous-model}, fix
\(\nu,\kappa,q,r>0\) and a bounded, finitely represented
time range \(0\le s\le T\), independently of \(n\).
Use the input implementation and state access of
Subsection~\ref{subsec:thermal-model}. For \(0<\varepsilon\le1\),
the direct projector-block construction encodes \(P(s)\) with decoded
error at most \(\varepsilon\), output normalization
\(\alpha_P=O(1)\) as specified in \eqref{eq:thermal-dre-output-scale}, and
\(O(\log^3(e+1/\varepsilon))\) basic continuous-matrix queries.
For \(0<\delta\le1/3\) and a nontrivial control tolerance
\(0<\varepsilon_{u}\le1\), it returns
\begin{equation*}
\widehat u(t)\approx
-r^{-1}(b_{\rm c})^*P(T-t)x(t)
\end{equation*}
with absolute error at most \(\varepsilon_{u}\) and failure
probability at most \(\delta\), using
\begin{equation*}
O\!\left(\frac{\log(1/\delta)}{\varepsilon_{u}}
\log^3(e+1/\varepsilon_{u})\right)
\end{equation*}
basic matrix and state-oracle calls. The total gate cost is
\(\varepsilon_{u}^{-1}\log(1/\delta)
\operatorname{polylog}(n)\operatorname{polylog}(e+1/\varepsilon_{u})\).
The constants depend on the fixed physical parameters, time range and
supplied state normalization, but not on \(n\).
\end{thermaldrerestatement}

\begin{proof}
For the thermal data, \(\|G_{\rm c}\|=1/r\), and direct block
composition supplies
\begin{equation}
\begin{aligned}
\alpha_{H}&=\nu+4\kappa+q+r^{-1},&
\alpha_{R_0}&=1,\\
\underline\Delta_{{\rm ax}}
&=\frac{\nu^2}
 {(1+(2r\nu)^{-1})^2(\nu+q)},&
\gamma_{+}&=\left(1+\frac{q}{4r\nu^2}\right)^{-1}.
\end{aligned}
\label{eq:thermal-dre-supplied-data-app}
\end{equation}
Lemma~\ref{lem:thermal-dre-parameters-app} proves both supplied lower
bounds. In particular, the initial column has full rank. The graph
and recovery hypotheses follow from
Lemma~\ref{lem:dre-graph-flow-app} and
Theorem~\ref{thm:dre-seeded-projector-main}.
Supply the solution bound \(M=q/(2\nu)\). With the decaying graph projector
normalization constructed below, choose
\begin{equation}
\alpha_{P}=8\alpha_{\mathcal E}\left(1+\frac{q}{2\nu}\right).
\label{eq:thermal-dre-output-scale}
\end{equation}
This satisfies \eqref{eq:dre-supplied-output-scale-app} at every requested
time and uses the scale of the direct recovery product.

Use the rectangles and positive Gauss--Legendre rule of
Proposition~\ref{prop:dre-weighted-quadrature-app}, with normalized
gap \(\underline\Delta_{{\rm ax}}/\alpha_{H}\).
The construction has
\begin{equation}
\mathcal R_{{\rm DRE}}
\le2(1+2\sqrt2)\frac{\alpha_{H}}
 {\underline\Delta_{{\rm ax}}}.
\label{eq:thermal-dre-contour-bound-app}
\end{equation}
On every node, the shifted matrix has normalization \(1+|z|\)
and smallest singular value at least
\(\underline\Delta_{{\rm ax}}/(2\alpha_{H})\).
Use the smaller threshold
\(\underline\Delta_{{\rm ax}}/(4\alpha_{H})\) in
Proposition~\ref{prop:local-inverse}; its inverse normalization is
\(16\alpha_{H}/\underline\Delta_{{\rm ax}}\).
(Nodewise inverse scales are unnecessary for the analysis here.)
This choice also
allows a normalized Hamiltonian input error at most
\(\underline\Delta_{{\rm ax}}/(4\alpha_{H})\).
The length estimate and positive weights in the proof of
Proposition~\ref{prop:dre-weighted-quadrature-app} then give the
actual raw scales
\begin{equation}
\begin{aligned}
\alpha_{\Pi_+},\quad\alpha_{\Pi_+R_0}
&\le\frac{96\alpha_H}{\pi\underline\Delta_{\rm ax}},\\
\alpha_{e^{s\mathcal H_{\rm DRE}}\Pi_-R_0},\quad
\alpha_{e^{-s\mathcal H_{\rm DRE}}\Pi_+}
&\le\frac{96\alpha_H}{\pi\underline\Delta_{\rm ax}}.
\end{aligned}
\label{eq:thermal-dre-raw-blocks-app}
\end{equation}
The coherent inverse-sum construction in
Appendix~\ref{app:contour-proofs}, before its constant-factor
inverse-norm simplification, gives a decoded block error
\(\eta\) using
\begin{equation}
O\!\left(\frac{\alpha_{H}}{\underline\Delta_{{\rm ax}}}
\log\!\left(e+
\frac{\alpha_{H}}
 {\underline\Delta_{{\rm ax}}\eta}\right)\right)
\label{eq:thermal-dre-block-query-app}
\end{equation}
Hamiltonian queries. This uses the supplied threshold directly;
it does not assume that this bound estimates each inverse norm
within an absolute constant factor.

Use these raw normalizations and perform no normalization
reduction. Equation~\eqref{eq:dre-raw-normalizations-app} gives
\begin{equation}
\begin{aligned}
\alpha_{\mathcal E}
&=\alpha_{\Pi_+}\\
&\quad+\frac{
8\alpha_{e^{s\mathcal H_{\rm DRE}}\Pi_-R_0}
 \alpha_{e^{-s\mathcal H_{\rm DRE}}\Pi_+}}{\gamma_+}.
\end{aligned}
\label{eq:thermal-dre-composition-scales-app}
\end{equation}
Consequently the DGP normalization is
\(O((\alpha_{H}/\underline\Delta_{{\rm ax}})^2
\gamma_{+}^{-1})\). The initial pseudoinverse still incurs
\(\alpha_{\Pi_+R_0}/\gamma_{+}\); this actual ratio is not replaced
by \(\kappa(\Pi_+R_0)\).
For fixed \(\nu,\kappa,q,r>0\), the supplied data in
\eqref{eq:thermal-dre-supplied-data-app} give \(O(1)\) bounds on
the input scales, generalized singularity factor, absolute initialization
inverse, and DGP normalization, independently of \(n\).
Equation~\eqref{eq:thermal-dre-output-scale} gives the same conclusion
for \(\alpha_P\). These bounds use the actual block construction
and do not assume that its conservative inverse scales equal the
resolvent norms.

Choose the common accuracy \(\eta\) by
\eqref{eq:dre-outer-budget-app}, using these scales and target
\(\varepsilon\). The stability proof of
Proposition~\ref{prop:dre-general-resources-app} applies unchanged.
Its query proof, with \eqref{eq:thermal-dre-block-query-app}
as the weighted-block cost and all re-encoding overheads equal to
one, gives the sufficient bound
\begin{equation}
\begin{aligned}
Q_{P}=O\!\Bigg(&\alpha_{P}
\left(\frac{\alpha_{H}}{\underline\Delta_{{\rm ax}}}\right)^2
\left(1+\frac{q}{4r\nu^2}\right)\\
&\times\log^3\!\left(e+
\frac{\alpha_{P}\alpha_{H}}
 {\underline\Delta_{{\rm ax}}\varepsilon}
\left(1+\frac{q}{4r\nu^2}\right)\right)\Bigg).
\end{aligned}
\label{eq:thermal-dre-physical-query-app}
\end{equation}
Only polynomial expressions in the displayed positive scales enter
\(\eta^{-1}\), so the three logarithms from the resolvent inverse,
initial-column pseudoinverse and upper-block pseudoinverse are each
bounded by the common logarithm. Their product is bounded by its
displayed third power. Each Hamiltonian call
uses a constant number of calls to the physical coefficient encodings.
At fixed physical parameters, \(Q_{P}=O(\log^3(e+1/\varepsilon))\);
the number of nodes for each weighted block is
\(O(\log(e+1/\varepsilon))\) by
\eqref{eq:dre-final-node-counts-app}.

Apply the readout argument of Proposition~\ref{prop:lqr-be-classical-main}
to the solution circuit just constructed, with
\(\alpha_{b_{\rm c}}=1\) and
\(\alpha_{x}=O(1)\). At fixed \(r\) and \(\alpha_{P}\),
a sufficiently small constant multiple of \(\varepsilon_{u}\)
is an admissible Riccati accuracy. This gives
\(O(\varepsilon_{u}^{-1}\log(1/\delta)
\log^3(e+1/\varepsilon_{u}))\) base-input calls for the
classical feedback. The path-input circuits, node preparation,
rotation synthesis, and arithmetic for
\(\alpha_{H}s\) have the gate costs specified in the input
construction. With a fixed finite time range and the stipulated
state-column access, they yield the claimed polylogarithmic
dependence on \(n\). Finite-precision coefficients and gate
synthesis use Proposition~\ref{prop:contour-input-errors-app} and
Proposition~\ref{prop:dre-approximate-inputs-app} at the same supplied
scales; the inverse and recovery error analysis is already contained
in Proposition~\ref{prop:dre-general-resources-app}.
\end{proof}

\subsection{Numerical experiments for heated boundary control}
\label{app:hbc-numerics}

We use the baseline physical parameters and contours of
Subsection~\ref{subsec:thermal-model}.
Figure~\ref{fig:hbc-parameters} varies \(\nu\) and \(r\) at
\(n=64\), and then varies \(n\) at the baseline coefficients.
Figure~\ref{fig:hbc-queries} combines local inverse QSP simulations,
a classical matrix reconstruction, and logical oracle counts for
the complete \(P(1)\) block-encoding.

\begin{figure}[H]
\centering
\includegraphics[width=0.95\textwidth]{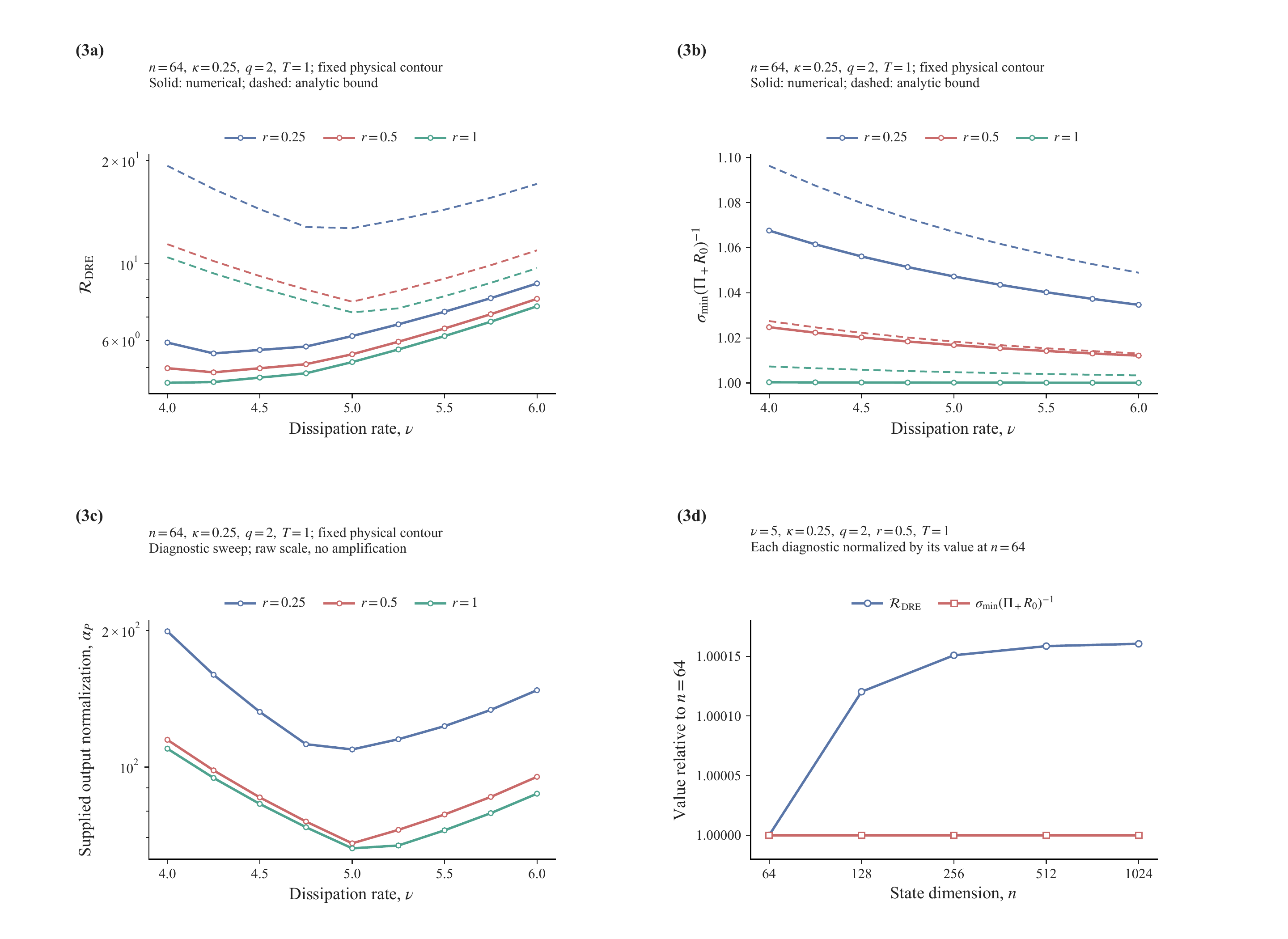}
\small\caption{%
Physical parameters and algorithmic scales, with \(\kappa=1/4\),
\(q=2\), \(T=1\), and a fixed physical contour.
(a) Numerical estimate of the generalized singularity factor
\(\mathcal R_{\rm DRE}\).
(b) Absolute initialization inverse \(\sigma_{\min}(\Pi_+R_0)^{-1}\).
Solid curves are numerical diagnostics; dashed curves are analytic bounds.
(c) Supplied output normalization \(\alpha_P\) from input bounds and
104-node quadrature, without normalization amplification.
Panels (a--c) use \(n=64\), nine equally spaced values of
\(\nu\in[4,6]\), and \(r=0.25,0.5,1\).
These are diagnostic and scale calculations: only three of the 27
parameter variants meet the fixed-node quadrature budget, which alone
does not establish their full-recovery accuracy.
(d) Five dimensions \(n=64,128,256,512,1024\) at \(\nu=5\), \(r=0.5\),
with each diagnostic divided by its value at \(n=64\).
Local coefficients remain fixed as nodes are added; this is not
fixed-interval mesh refinement. All 32 cases are included.%
}
\label{fig:hbc-parameters}
\end{figure}

Figure~\ref{fig:hbc-parameters} is consistent with the \(O(1)\)
parameter bounds at fixed physical coefficients. For fixed time range, accuracy, and failure
probability, Corollary~\ref{cor:thermal-dre-main} then gives \(O(1)\)
oracle queries in \(n\) and \(\operatorname{polylog}n\) gates under
the stated input-access assumptions.

Figure~\ref{fig:hbc-queries}(d) varies the nodes per contour from
104 to 260 while fixing the three QSP degrees at \(757,29,57\).
Each complete \(P(1)\) block-encoding uses \(82\,516\,028\) logical
Hamiltonian-oracle queries and has measured spectral error approximately
\(1.88\times10^{-10}\). The count comes from phase-by-phase success-block
simulation and recursive circuit traces; the expression
\(4\cdot757(3+4\cdot29)(1+4\cdot57)\) provides a cross-check.
The simulation does not apply that many dense oracles in the full
auxiliary space. The construction includes the transient term and
both recovery pseudoinverses, without normalization amplification.
The count excludes readout, node preparation, and other gate costs.

\begin{figure}[H]
\centering
\includegraphics[width=0.95\textwidth]{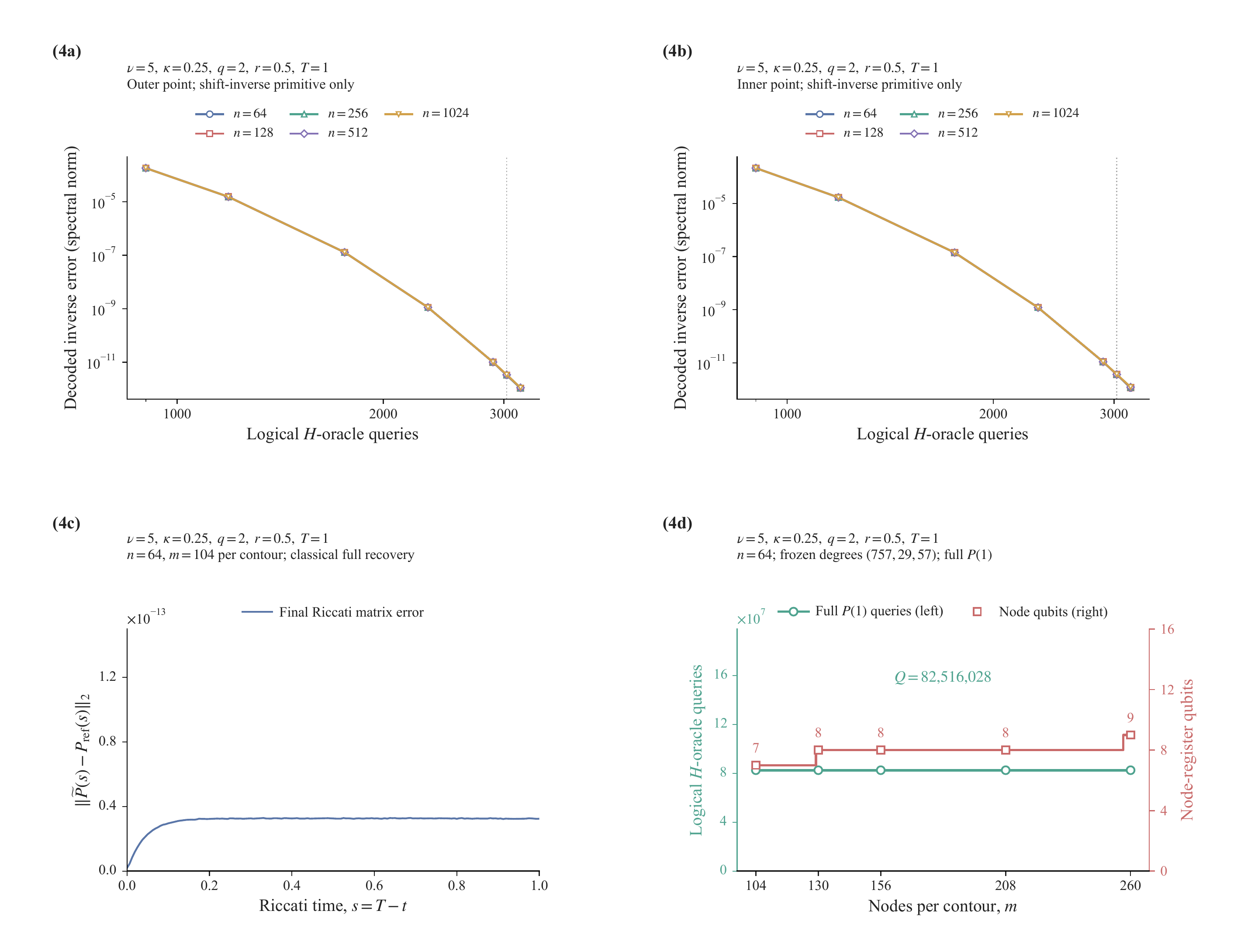}
\small\caption{%
Oracle queries, Riccati error, and node-register resources at
\(\nu=5\), \(\kappa=1/4\), \(q=2\), \(r=0.5\), and \(T=1\).
(a,b) Decoded spectral-norm errors for the inverse of
\(\zeta\mathsf I-\widehat{\mathcal H}_{\rm DRE}\), with
\(\alpha_H=10\) and \(\zeta=z/\alpha_H\), at the outer and inner
real-axis intersections of the right contour.
Five dimensions \(n=64,128,256,512,1024\) and seven frozen phase sets
give 70 local inverse simulations, with \(4d\) logical queries at degree
\(d=225,297,439,581,723,757,793\); all meet their local error budgets.
(c) Classical full-matrix error \(\|\widetilde P(s)-P_{\rm ref}(s)\|_2\)
at 161 times \(s=T-t\), with \(n=64\) and 104 nodes per contour.
It includes quadrature and recovery errors before control interpolation.
(d) Complete \(P(1)\) block-encoding with \(n=64\) and frozen degrees
\((757,29,57)\): \(m=104,130,156,208,260\) nodes per contour give
\(82\,516\,028\) logical queries in each case
(\(41\,258\,014\) each to \(U_H\) and \(U_H^\dagger\)).
The transient term and both recovery inverses are included.
No normalization amplification is used; the count excludes readout.
The right axis counts only the node register,
\(\lceil\log_2m\rceil=7,8,8,8,9\), excluding system and other ancillas.
Squares mark the five cases; the staircase is the register-size formula.%
}
\label{fig:hbc-queries}
\end{figure}

For Figure~\ref{fig:hbc-dimension-cost}, RK4 uses the tridiagonal
structure of \(-A_{\rm c}\) and the rank-one structure of \(G_{\rm c}\).
Each dimension executes five integrations with \(7,14,28,56,112\) steps,
for 217 steps in total; the accepted output uses 112 steps.
Step selection uses the Richardson indicator
\(\sqrt{\|\Delta P\|_1\|\Delta P\|_\infty}/15\), where \(\Delta P\)
is the difference between successive RK4 outputs.
Refinement stops after two consecutive indicators are at most
\(\varepsilon/2\) and the corresponding difference-norm ratios lie
in \([4,64]\). This is an empirical error indicator; the reference
is consulted only after step selection, to measure the final error.
All refinement runs, setup arithmetic, and indicator arithmetic contribute
to the plotted count \(12174n^2-5206n+65\).
Square roots, absolute values, comparisons, and data movement are excluded
from the count of additions, subtractions, multiplications, and divisions.

We set \(\nu=5\), \(\kappa=1/4\), \(q=2\), \(r=0.5\), \(T=1\),
\(P_T=0\), and \(x(0)=(e_1+e_2)/\sqrt2\).
We evaluate the finite Riesz sums and both pseudoinverses in
Algorithm~\ref{alg:dre-direct-projector-construction} classically.
The reference is checked with 70-digit arithmetic at
\(n=8,s=0.25,1\) and DRE integration at all 161 times for \(n=64\).

In Figure~\ref{fig:hbc-convergence}, the Riccati error falls from
about \(3.05\times10^{-7}\) at \(p=4\) to
\(3.27\times10^{-12}\) at \(p=8\) for all five dimensions;
higher orders reach a floating-point plateau of about
\(3.2\times10^{-14}\) to \(1.2\times10^{-13}\).
The initial scalar-feedback error follows the same trend.

\begin{figure}[H]
\centering
\includegraphics[width=\textwidth]{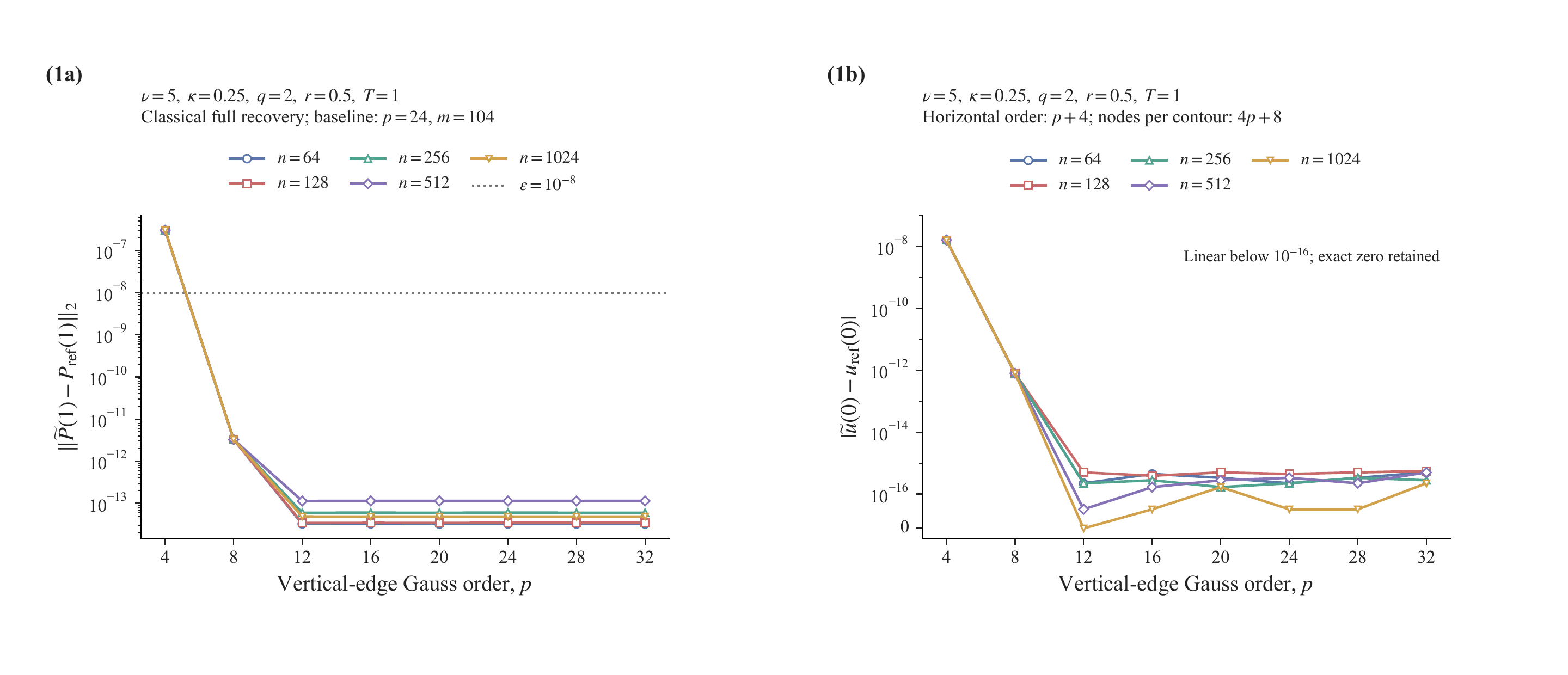}
\small\caption{%
Finite-contour convergence in a classical matrix simulation.
(a) Final Riccati error \(\|\widetilde P(T)-P_{\rm ref}(T)\|_2\).
(b) Initial feedback error \(|\widetilde u(0)-u_{\rm ref}(0)|\).
All 40 reconstructions use \(n=64,128,256,512,1024\) and
vertical-edge Gauss--Legendre order \(p=4,8,12,16,20,24,28,32\).
The horizontal-edge order is \(p+4\), giving \(4p+8\) nodes per
rectangle and \(8p+16\) across both contours; \(p=24\) gives the
104-node baseline. Inverses and both recovery pseudoinverses are
computed classically. The dotted line in (a) marks \(10^{-8}\).
In (b), the ordinate is linear below \(10^{-16}\) and logarithmic above,
so the floating-point zero at \(n=1024,p=12\) remains visible.%
}
\label{fig:hbc-convergence}
\end{figure}

\clearpage
\section{Lower bounds and complexity proofs}
\label{app:lower-bounds}

We prove the common oracle and output estimates, then the four product
lower bounds and the two circuit reductions.

\subsection{Input oracles and output separation}
\label{app:lower-bound-tools}

For a strict contraction \(T\), use the Julia input unitary
\cite{Robinson2018Julia}
\begin{equation}
U_T=\begin{bmatrix}
T&(\mathsf I-TT^*)^{1/2}\\
(\mathsf I-T^*T)^{1/2}&-T^*
\end{bmatrix}.
\label{eq:julia-input-unitary-app}
\end{equation}
For a matrix encoded with public normalization \(\alpha\), substitute
that matrix divided by \(\alpha\) for \(T\). The complete unitary is
specified because its other blocks could otherwise reveal the hidden input.

\begin{lemma}[Continuity of the input unitaries]
\label{lem:julia-input-continuity-app}
The matrix in \eqref{eq:julia-input-unitary-app} is unitary. If
\(\|T_0\|,\|T_1\|\le\rho<1\), then
\begin{equation}
\|U_{T_1}-U_{T_0}\|
\le\left(1+\frac{\rho}{\sqrt{1-\rho^2}}\right)
\|T_1-T_0\|.
\label{eq:julia-input-continuity-app}
\end{equation}
\end{lemma}
\begin{proof}
The singular-value decomposition gives
\(T(\mathsf I-T^*T)^{1/2}=(\mathsf I-TT^*)^{1/2}T\),
so block multiplication proves unitarity.
For \(D_c=\mathsf I-T_cT_c^*\) or
\(D_c=\mathsf I-T_c^*T_c\), we have
\(D_c\succeq(1-\rho^2)\mathsf I\). The difference
\(E=D_1^{1/2}-D_0^{1/2}\) solves
\[
D_1^{1/2}E+ED_0^{1/2}=D_1-D_0.
\]
Its integral solution~\cite[Eq.~(22)]{Simoncini2016LinearMatrixEquations} bounds
\[
\|E\|\le\frac{\|D_1-D_0\|}{2\sqrt{1-\rho^2}}
\le\frac{\rho\|T_1-T_0\|}{\sqrt{1-\rho^2}}.
\]
Adding the norms of the diagonal and off-diagonal parts of
\(U_{T_1}-U_{T_0}\) proves the claim.
\end{proof}

Queries below may select coherently from a fixed finite collection of
input unitaries, their adjoints, and controlled versions at unit cost.
All other gates, auxiliary registers, numerical metadata, and signal
spaces are independent of the hidden bit.

\begin{lemma}[Query hybrid and solution-block separation]
\label{lem:query-hybrid-app}
For \(c\in\{0,1\}\), let \(V_c\) be the output unitary of one common
coherent circuit using at most \(q\) calls to the supplied unitaries
\(U_c^{(j)}\). Suppose \(V_c\) encodes the target \(Y_c\) with a
common public normalization \(\alpha_Y\), a common signal space,
and decoded error at most \(\varepsilon\). Then
\begin{equation}
q\max_j\|U_1^{(j)}-U_0^{(j)}\|
\ge\|V_1-V_0\|
\ge\frac{\|Y_1-Y_0\|-2\varepsilon}{\alpha_Y}.
\label{eq:query-hybrid-output-app}
\end{equation}
\end{lemma}
\begin{proof}
Replace the query gates one at a time. Unitarity bounds each change
by the corresponding oracle difference; adjoints, controls, and
coherent selection preserve the maximum difference norm. Summing
gives the first bound. Compression to the common signal space and
the reverse triangle inequality give the second.
\end{proof}

For the BQP reductions, use the finite-input and controlled-output
conventions of Section~\ref{subsec:bqp-hardness}. The input includes
the source gate list, register layout, and coefficient-access circuit
descriptions, with polynomial cost in their length and the requested
number of precision bits. Amplification and realification give the
required acceptance thresholds using real local gates with polynomial
overhead \cite{BernsteinVazirani1997QuantumComplexity,RudolphGrover2002Rebit}.

\begin{lemma}[From a selected value to a solution oracle]
\label{lem:selected-value-hardness-app}
Let a promised family have a Hermitian target \(X\), a known basis
vector \(x_0\), and a BQP-hard decision problem for \(x_0^*Xx_0\)
with a fixed positive gap, under a polynomial-time reduction supplying
the finite descriptions above. A controlled
\((\alpha_X,a_X,\varepsilon)\) block-encoding of \(X\) decides this
problem with a constant number of calls if the public
\(\alpha_X=O(1)\) and the fixed decoded error \(\varepsilon\) is
smaller than one quarter of the gap. An independent polynomial-time
quantum decision procedure on the same family establishes BQP
completeness there.
\end{lemma}
\begin{proof}
A controlled Hadamard test on the known signal state with system
register \(x_0\) estimates the real diagonal element of the output
unitary. After multiplication by \(\alpha_X\), the encoding error
changes the selected value by at most \(\varepsilon\). A constant
number of tests bounds the statistical error by \(\varepsilon\) with
probability at least \(5/6\). The total error is less than half the
gap, so comparison with its midpoint decides the problem.

In a polynomial-time implementation, sufficiently small
inverse-polynomial operator errors per input call preserve bounded
error by the same telescoping argument as in
Lemma~\ref{lem:query-hybrid-app}. The precision-adjustable input
circuits provide this accuracy with polynomial overhead. Membership
is supplied by the stated decision procedure on the same family.
\end{proof}

\subsection{Product lower bounds for the four Riccati problems}
\label{app:product-lower-bounds}

We use the conventions of Section~\ref{subsec:product-lower-bounds}
and the common argument of Appendix~\ref{app:lower-bound-tools}.
Only the specified input oracles depend on \(c\); all other data,
operations, and output signal spaces are public and common.
The seed normalizations are \(\alpha_{R_0}=\sqrt5/2\) for DRE and
\(\alpha_{R_0}=1\) for RR. All estimates hold for each input.

\newtheorem*{dareproductrestatement}{Theorem~\ref{thm:dare-product-lower-main}}
\begin{dareproductrestatement}
Under the above access and accuracy conventions, there exists a family
of scalar instances of Problem~\ref{prob:quantum-dare} requiring
\[
q_{\rm DARE}=\Omega\!\left(\mathcal R_{\rm DARE}\alpha_X\right)
\]
quantum queries in the worst case. The two factors can be made
independently arbitrarily large.
\end{dareproductrestatement}
\begin{proof}
For \(c\in\{0,1\}\), take the positive scalar coefficients
\begin{equation}
\begin{aligned}
A_c&=\frac{(1-\mu_1)(1+c+\mu_2)}
 {(1+c)(1-\mu_1)^2+\mu_2},\\
G_c&=\frac{\mu_1(2-\mu_1)\mu_2}
 {(1+c)(1-\mu_1)^2+\mu_2},&
Q_c&=\frac{\mu_1(2-\mu_1)(1+c)}
 {(1+c)(1-\mu_1)^2+\mu_2}.
\end{aligned}
\label{eq:dare-product-data-app}
\end{equation}
Choosing \(B_c=\sqrt{G_c}\), \(R_c=1\) gives stabilizability and
detectability. The pencil in \eqref{eq:dare-pencil} has invertible
\(L_c\), and
\begin{equation}
\begin{aligned}
\Pi_<^{(c)}
&=\frac1{1+c+\mu_2}
 \begin{bmatrix}\mu_2&\mu_2\\1+c&1+c\end{bmatrix},\\
S_D^{(c)}
&=L_c^{-1}M_c
 =\frac{\mathsf I-\mu_1(2-\mu_1)\Pi_<^{(c)}}{1-\mu_1}.
\end{aligned}
\label{eq:dare-product-projector-app}
\end{equation}
The projector is idempotent and selects the eigenvalue \(1-\mu_1\);
the other eigenvalue is \((1-\mu_1)^{-1}\).
Its range is the graph of \(X_c=(1+c)/\mu_2\).
Substitution gives
\(Q_c+A_c^2X_c/(1+G_cX_c)=X_c\) and
\(A_c/(1+G_cX_c)=1-\mu_1\), proving that \(X_c\) is stabilizing.

The coefficients give
\(\|M_c\|,\|L_c\|<8/5\), \(\|L_c^{-1}\|<3/2\), and
\(\|\Pi_<^{(c)}\|,\|\mathsf I-\Pi_<^{(c)}\|\le\sqrt2\).
For the pencil bounds, use
\(\max\{A_c,A_c^{-1}\}\le(1-\mu_1)^{-1}\) and
\(G_c,Q_c\le\mu_1(2-\mu_1)/(1-\mu_1)^2\);
the inverse additionally uses
\(G_c/A_c\le\mu_1(2-\mu_1)/(1-\mu_1)\).
The resolvent is
\[
(zL_c-M_c)^{-1}
=\left(\frac{\Pi_<^{(c)}}{z-(1-\mu_1)}
+\frac{\mathsf I-\Pi_<^{(c)}}{z-(1-\mu_1)^{-1}}\right)L_c^{-1}.
\]
On \(|z|=1\), both pole distances are at least \(\mu_1\), giving
an \(O(\mu_1^{-1})\) bound. At \(z=1\), a unit
\((1-\mu_1)\)-eigenvector \(v\) satisfies
\((L_c-M_c)^{-1}L_cv=v/\mu_1\), giving the matching lower bound.
Thus, with the common \(\alpha_M=\alpha_L=2\),
\begin{equation}
\mathcal R_{\rm DARE}
=\sup_{|z|=1}4\|(zL_c-M_c)^{-1}\|
=\Theta(\mu_1^{-1}).
\label{eq:dare-product-resolvent-app}
\end{equation}

Direct subtraction of \eqref{eq:dare-product-data-app} gives
\[
\begin{bmatrix}A_1-A_0\\G_1-G_0\\Q_1-Q_0\end{bmatrix}
=
\frac{\mu_1(2-\mu_1)\mu_2}
 {((1-\mu_1)^2+\mu_2)(2(1-\mu_1)^2+\mu_2)}
\begin{bmatrix}1-\mu_1\\-(1-\mu_1)^2\\1\end{bmatrix}.
\]
Consequently \(\|M_1-M_0\|,\|L_1-L_0\|<3\mu_1\mu_2\).
Their signal matrices have norms below \(4/5\), so the Julia
estimate in Appendix~\ref{app:lower-bound-tools} yields
\[
\max\{\|U_{M_1/2}-U_{M_0/2}\|,
       \|U_{L_1/2}-U_{L_0/2}\|\}<4\mu_1\mu_2.
\]
This also covers their adjoints, controlled versions, and public
selection between the two oracles.
Use the common actual output normalization
\(\alpha_X=\Theta(\|X_c\|)=\Theta(1/\mu_2)\)
(see Proposition~\ref{prop:algebraic-product-normalization-app} for the proof).
Since \(\|X_1-X_0\|=1/\mu_2\), the decoded-error convention gives
\[
\frac{\|X_1-X_0\|-2\varepsilon}{\alpha_X}
\ge\frac{3}{4\mu_2\alpha_X}=\Omega(1).
\]
Lemma~\ref{lem:query-hybrid-app} therefore gives
\(q_{\rm DARE}=\Omega((\mu_1\mu_2)^{-1})\).
Together with \(\mathcal R_{\rm DARE}=\Theta(\mu_1^{-1})\),
this establishes the claimed product for each input.
\end{proof}

\newtheorem*{careproductrestatement}{Theorem~\ref{thm:care-product-lower-main}}
\begin{careproductrestatement}
Under the above access and accuracy conventions, there exists a family
of two-state instances of Problem~\ref{prob:quantum-care} requiring
\[
q_{\rm CARE}=\Omega\!\left(\mathcal R_{\rm CARE}\alpha_X\right)
\]
quantum queries in the worst case. The two factors can be made
independently arbitrarily large.
\end{careproductrestatement}
\begin{proof}
Set
\[
O=\frac1{\sqrt2}\begin{bmatrix}1&1\\1&-1\end{bmatrix},\qquad
\begin{aligned}
A_c&=O\operatorname{diag}(\mu_1(1-2(1+c)\mu_2),0)O^*,\\
G_c&=O\operatorname{diag}
 (4\mu_1(1+c)\mu_2(1-(1+c)\mu_2),1)O^*,\\
Q&=O\operatorname{diag}(\mu_1,1)O^*.
\end{aligned}
\]
Here \(G_c,Q\succ0\); \(B_c=G_c^{1/2}\), \(R_c=\mathsf I\)
give the required LQR conditions. The queried input is the assembled
Hamiltonian, not separate encodings of these control data.

In the coordinates defined by \(O\), the slow Hamiltonian block
squares to \(\mu_1^2\mathsf I\) and has norm at most \(2\mu_1<1\);
the fast block is \(\begin{bmatrix}0&-1\\-1&0\end{bmatrix}\).
Hence
\[
\|\mathcal H_{\rm CARE}^{(c)}\|=1,\qquad
\operatorname{spec}(\mathcal H_{\rm CARE}^{(c)})
=\{-1,-\mu_1,\mu_1,1\}.
\]
On the slow block, subtracting the Hamiltonian divided by \(2\mu_1\)
from \(\mathsf I/2\) gives the stable projector. Both slow
projectors therefore have norm at most \(3/2\), and the fast
projectors are orthogonal. The solution is
\begin{equation}
X_c=O\operatorname{diag}
 \left(\frac1{2(1+c)\mu_2},1\right)O^*,\qquad
\|X_c\|=\frac1{2(1+c)\mu_2}.
\label{eq:care-product-graph-app}
\end{equation}
Substitution verifies the CARE and
\(A_c-G_cX_c=O\operatorname{diag}(-\mu_1,-1)O^*\),
so \(X_c\) is stabilizing.

For \(\alpha_H=2\), use the common positively oriented contour
\[
\Gamma_-=
\partial\{x+\mathrm iy:-3/2\le x\le-\mu_1/2,\ |y|\le1/2\}.
\]
It selects the stable spectrum of the unscaled Hamiltonian.
The two slow projectors have uniformly bounded norms, and their
pole distances are at least \(\mu_1/2\) and \(3\mu_1/2\);
the fast resolvent is uniformly bounded. At \(z=-\mu_1/2\),
a \(-\mu_1\)-eigenvector supplies the matching lower bound.
Since \(|z|+2\) stays between positive constants,
\[
\mathcal R_{\rm CARE}
=\sup_{z\in\Gamma_-}(|z|+2)
 \|(z\mathsf I-\mathcal H_{\rm CARE}^{(c)})^{-1}\|
=\Theta(\mu_1^{-1}).
\]

The explicit coefficients give
\(\|\mathcal H_{\rm CARE}^{(1)}
-\mathcal H_{\rm CARE}^{(0)}\|\le\sqrt{24}\mu_1\mu_2\).
As \(\|\mathcal H_{\rm CARE}^{(c)}/2\|=1/2\), the Julia estimate
in Appendix~\ref{app:lower-bound-tools} gives
\[
\|U_{\mathcal H_{\rm CARE}^{(1)}/2}
   -U_{\mathcal H_{\rm CARE}^{(0)}/2}\|<4\mu_1\mu_2.
\]
Use the common actual output normalization
\(\alpha_X=\Theta(\|X_c\|)=\Theta(1/\mu_2)\)
(see Proposition~\ref{prop:algebraic-product-normalization-app} for the proof).
Here \(\|X_1-X_0\|=1/(4\mu_2)\), so the same decoded-error
convention gives
\[
\frac{\|X_1-X_0\|-2\varepsilon}{\alpha_X}
\ge\frac{3}{16\mu_2\alpha_X}=\Omega(1).
\]
The same hybrid argument yields
\(q_{\rm CARE}=\Omega((\mu_1\mu_2)^{-1})\).
Since \(\mathcal R_{\rm CARE}=\Theta(\mu_1^{-1})\), this proves
the claimed product for either input.
\end{proof}

\begin{proposition}[Actual normalization on the algebraic lower-bound families]
\label{prop:algebraic-product-normalization-app}
For the CARE and DARE families above, public contours and quadrature
rules independent of \(c\) implement the full-projector recovery
algorithm with
\begin{equation}
\alpha_\Pi=\Theta(1),\qquad r_\Pi=1,\qquad
\alpha_X=\Theta(\|X_c\|)=\Theta(\mu_2^{-1}),\qquad
\mathcal R=\Theta(\mu_1^{-1}).
\label{eq:algebraic-product-actual-normalization-app}
\end{equation}
The supplied scales, including \(\alpha_X\), are common to both
inputs and independent of quadrature accuracy. Under the coherent node
access of the construction theorems, projector error
\(0<\delta\le1/8\) requires \(O(\log(1/\delta))\) nodes and
\[
Q_\Pi(\delta)
=O\!\left(\mu_1^{-1}
\log\!\left(e+\frac1{\mu_1\delta}\right)\right)
\]
input queries. Thus the solution construction uses
\(\widetilde O(\mathcal R\alpha_X)\) queries at an actual
normalization comparable to the solution norm, uniformly as
\(\mu_2\to0\).
\end{proposition}
\begin{proof}
For CARE, use the two counterclockwise circles
\begin{equation}
|z+\mu_1|=\frac{\mu_1}{2},\qquad |z+1|=\frac14,
\label{eq:care-product-normalization-contours-app}
\end{equation}
and sum their integrals. Each circle encloses one stable eigenvalue
and no other eigenvalue. On the first circle, the slow-block
resolvent norm is at most \(4/\mu_1\), because its projectors have
norm at most \(3/2\) and the two pole distances are at least
\(\mu_1/2\) and \(3\mu_1/2\). The fast block is uniformly bounded.
On the second circle, the fast-block resolvent has norm \(4\),
while the slow-block norm is at most
\((3/2)(8/5+4/3)<5\).
A unit eigenvector at the enclosed eigenvalue gives the matching
lower bound on each circle. Consequently the public inverse scales
\(\beta_j=32/\mu_1\) and \(\beta_j=32\), respectively, are
admissible constant-factor inverse normalizations for both inputs.
The generalized singularity factor on their union is
\(\Theta(\mu_1^{-1})\).

For DARE, use the counterclockwise circle
\begin{equation}
|z-(1-\mu_1)|=\frac{\mu_1}{2}.
\label{eq:dare-product-normalization-contour-app}
\end{equation}
It encloses the stable eigenvalue and excludes its reciprocal.
Their separation is at least \(2\mu_1\). The projector and
\(L_c^{-1}\) bounds in the DARE proof give
\[
\|(zL_c-M_c)^{-1}\|
\le\frac32\sqrt2\left(\frac2{\mu_1}
+\frac2{3\mu_1}\right)
=\frac{4\sqrt2}{\mu_1}.
\]
For a unit stable eigenvector \(v\),
\((zL_c-M_c)^{-1}L_cv=v/(z-(1-\mu_1))\), giving a matching
\(\Omega(\mu_1^{-1})\) bound. Thus \(\beta_j=32/\mu_1\)
is an admissible constant-factor inverse normalization, and
\(\mathcal R_{\rm DARE}=\Theta(\mu_1^{-1})\) also on this circle.
The construction uses the pencil inverses and the right factor
\(L_c\); \(L_c^{-1}\) is used only in this estimate.
(Nodewise inverse scales are unnecessary for this analysis;
a fixed scale on each circle suffices.)

On a circle with center \(a\) and radius \(\rho\), take
\(z_j=a+\rho e^{2\pi\mathrm i j/m}\) and
\(\omega_j=\rho e^{2\pi\mathrm i j/m}/m\).
The selected eigenvalue equals \(a\), so its scalar quadrature is
exact. For any excluded eigenvalue \(\lambda\), the geometric-series
identity gives
\[
\frac1m\sum_{j=0}^{m-1}
\frac{\rho e^{2\pi\mathrm i j/m}}
 {a+\rho e^{2\pi\mathrm i j/m}-\lambda}
=\frac1{1-((\lambda-a)/\rho)^m}.
\]
For the CARE circles, every excluded pole satisfies
\(\rho/|\lambda-a|\le2/7<1/3\); for DARE the ratio is at
most \(1/4\). The uniformly bounded spectral projectors therefore
give quadrature error \(O(3^{-m})\) and \(O(4^{-m})\),
respectively. Taking \(m=O(\log(1/\delta))\) allocates at most
\(\delta/2\) to quadrature.

The absolute weights sum to the circle radius. Hence the actual
raw normalizations, including \(\alpha_L=2\) for DARE, are
\begin{equation}
\begin{aligned}
\alpha_{\Pi_-}^{\rm raw}
&=\frac{\mu_1}{2}\frac{32}{\mu_1}+\frac14\,32=24,\\
\alpha_{\Pi_<}^{\rm raw}
&=2\,\frac{\mu_1}{2}\frac{32}{\mu_1}=32.
\end{aligned}
\label{eq:algebraic-product-raw-scales-app}
\end{equation}
Both dominate the exact projector norms and can be used without
normalization reduction, so \(r_\Pi=1\). The contour primitives
implement the remaining error \(\delta/2\) at the stated
\(Q_\Pi(\delta)\), using public node scales and weights.
Finally, substitute the public solution bounds \(1/(2\mu_2)\)
for CARE and \(2/\mu_2\) for DARE into
\eqref{eq:algebraic-projector-output-app} and
\eqref{eq:algebraic-projector-budget-app}. Each bound lies within a
factor of two of \(\|X_c\|\), so the resulting actual
\(\alpha_X\) has the order claimed. Equation~\eqref{eq:algebraic-projector-query-app}
then gives the total query bound.

This construction supplies the common actual output normalizations
used in Section~\ref{subsec:product-lower-bounds}, without
normalization reduction.
\end{proof}

\newtheorem*{dreproductrestatement}{Theorem~\ref{thm:dre-product-lower-main}}
\begin{dreproductrestatement}
Under the above access and accuracy conventions, there exists a family
of two-state instances of Problem~\ref{prob:quantum-dre} requiring
\[
q_{\rm DRE}=\Omega\!\left(\mathcal R_{\rm DRE}\kappa(\Pi_+R_0)\right)
\]
quantum queries in the worst case. The two factors can be made
independently arbitrarily large.
\end{dreproductrestatement}
\begin{proof}
Take
\begin{equation}
\begin{gathered}
A=\operatorname{diag}(-\mu_1,1),\qquad
G=0,\qquad P_0=\operatorname{diag}(1/2,0),\\
Q_c=\begin{bmatrix}
\mu_1(1-2(1+c)\mu_2)&\mu_1/4\\
\mu_1/4&1
\end{bmatrix}.
\end{gathered}
\label{eq:dre-product-data-app}
\end{equation}
The first diagonal entry of \(Q_c\) is at least \(\mu_1/2\), and
\(\det Q_c\ge\mu_1/2-\mu_1^2/16>0\).
For the Hamiltonian and common seed,
\[
\mathcal H_{\rm DRE}^{(c)}
=\begin{bmatrix}A&0\\-Q_c&-A\end{bmatrix},\qquad
R_0=\begin{bmatrix}\mathsf I\\P_0\end{bmatrix},
\]
we have \(1\le\|\mathcal H_{\rm DRE}^{(c)}\|\le65/32\) and
\(\|R_0\|=\sqrt5/2\). Thus \(\alpha_H=3\) is within a constant
factor of the actual norm. The nonzero off-diagonal entries of \(Q_c\)
connect the two states.

To compute both the projector and resolvent, use
\[
Y_c=\begin{bmatrix}
1/2-(1+c)\mu_2&-\mu_1/(4(1-\mu_1))\\
-\mu_1/(4(1-\mu_1))&-1/2
\end{bmatrix}.
\]
Since \(AY_c+Y_cA=-Q_c\),
\[
\mathcal H_{\rm DRE}^{(c)}
=\begin{bmatrix}\mathsf I&0\\Y_c&\mathsf I\end{bmatrix}
 \begin{bmatrix}A&0\\0&-A\end{bmatrix}
 \begin{bmatrix}\mathsf I&0\\-Y_c&\mathsf I\end{bmatrix}.
\]
The similarity has condition number at most
\((1+\|Y_c\|)^2<4\), because
\(\|Y_c\|\le1/2+\mu_1/(4(1-\mu_1))<1\).
Conjugating the positive-eigenvalue projector gives
\[
\Pi_+^{(c)}
=\begin{bmatrix}
0&0&0&0\\0&1&0&0\\
-1/2+(1+c)\mu_2&0&1&0\\0&-1/2&0&0
\end{bmatrix}.
\]
Its seeded columns are orthogonal, so
\begin{equation}
\begin{aligned}
\Pi_+^{(c)}R_0&=
 \begin{bmatrix}0&0\\0&1\\(1+c)\mu_2&0\\0&-1/2\end{bmatrix},&
\|\Pi_+^{(c)}R_0\|&=\sqrt5/2,\\
\sigma_{\min}(\Pi_+^{(c)}R_0)&=(1+c)\mu_2,&
\kappa(\Pi_+^{(c)}R_0)&=\frac{\sqrt5}{2(1+c)\mu_2}.
\end{aligned}
\label{eq:dre-product-seeded-columns-app}
\end{equation}

For \(\widehat{\mathcal H}_{\rm DRE}^{(c)}
=\mathcal H_{\rm DRE}^{(c)}/3\), choose the positively oriented rectangles
\[
\Gamma_+
=\partial\{x+\mathrm iy:\mu_1/6\le x\le1/2,\ |y|\le1/12\},
\qquad \Gamma_-=-\Gamma_+.
\]
Their pole distances are at least \(\mu_1/6\).
The bounded similarity gives an \(O(\mu_1^{-1})\) resolvent bound,
while the \(\mu_1/3\)-eigenvector \(e_3\) gives \(6/\mu_1\) at
\(z=\mu_1/6\). Therefore
\begin{equation}
\begin{aligned}
\mathcal R_{\rm DRE}
&=\max_{\sigma\in\{-,+\}}\sup_{z\in\Gamma_\sigma}
(1+|z|)\|(z\mathsf I-\widehat{\mathcal H}_{\rm DRE}^{(c)})^{-1}\|\\
&=\Theta(\mu_1^{-1}).
\end{aligned}
\label{eq:dre-product-resolvent-app}
\end{equation}

Only one entry of \(Q_c\) varies, giving
\(\|\widehat{\mathcal H}_{\rm DRE}^{(1)}
-\widehat{\mathcal H}_{\rm DRE}^{(0)}\|=2\mu_1\mu_2/3\).
Since both signal norms are at most \(65/96<1\), the Julia estimate
of Appendix~\ref{app:lower-bound-tools} implies
\begin{equation}
\|U_{\mathcal H_{\rm DRE}^{(1)}/3}
  -U_{\mathcal H_{\rm DRE}^{(0)}/3}\|
<\frac43\mu_1\mu_2.
\label{eq:dre-product-oracle-separation-app}
\end{equation}
The common seed encoding contributes no difference.

With \(G=0\), the DRE is linear and exists at every finite time.
Its solution is
\begin{equation}
P^{(c)}(s)=
\begin{bmatrix}
\displaystyle\frac12+(1+c)\mu_2(e^{2\mu_1s}-1)&
\displaystyle-\frac{\mu_1(1-e^{-(1-\mu_1)s})}{4(1-\mu_1)}\\[2mm]
\displaystyle-\frac{\mu_1(1-e^{-(1-\mu_1)s})}{4(1-\mu_1)}&
\displaystyle-\frac12(1-e^{-2s})
\end{bmatrix}.
\label{eq:dre-product-solution-app}
\end{equation}
Direct differentiation verifies the equation and initial value.
For \(0\le s\le t=\log(1/\mu_2)/(2\mu_1)\), its row sums give
\(\|P^{(c)}(s)\|\le5/2+\mu_1/(4(1-\mu_1))<3\), and
\begin{equation}
\begin{aligned}
P^{(1)}(s)-P^{(0)}(s)&=\mu_2(e^{2\mu_1s}-1)e_1e_1^*,\\
\|P^{(1)}(t)-P^{(0)}(t)\|&=1-\mu_2\ge7/8.
\end{aligned}
\label{eq:dre-product-output-difference-app}
\end{equation}
With \(\alpha_{P(t)}=3\) and decoded error \(1/16\),
Lemma~\ref{lem:query-hybrid-app} and
\eqref{eq:dre-product-oracle-separation-app} give
\begin{equation}
\|V_1-V_0\|\ge\frac{7/8-2/16}{3}=\frac14,\qquad
q_{\rm DRE}\ge\frac{3}{16\mu_1\mu_2}.
\label{eq:dynamic-product-hybrid-app}
\end{equation}
Equations~\eqref{eq:dre-product-seeded-columns-app} and
\eqref{eq:dre-product-resolvent-app} identify the two actual factors.
\end{proof}

\newtheorem*{rrproductrestatement}{Theorem~\ref{thm:rr-product-lower-main}}
\begin{rrproductrestatement}
Under the above access and accuracy conventions, there exists a family
of two-state instances of Problem~\ref{prob:quantum-rr} requiring
\[
q_{\rm RR}=\Omega\!\left(\mathcal R_{\rm RR}\kappa(\Pi_>R_0)\right)
\]
quantum queries in the worst case. The two factors can be made
independently arbitrarily large.
\end{rrproductrestatement}
\begin{proof}
Set
\begin{equation}
\begin{gathered}
A=\operatorname{diag}(1+\mu_1,1/2),\qquad G=0,\qquad P_0=0,\\
Q_c=\operatorname{diag}(\mu_1(2+\mu_1)(1+c)\mu_2,0).
\end{gathered}
\label{eq:rr-product-data-app}
\end{equation}
The data satisfy \(A\) invertible and \(G,Q_c,P_0\succeq0\).
All iterates of \(P_{j+1}=Q_c+A^*P_jA\) exist and are positive
semidefinite. The unscaled lift is
\[
\mathcal S_F^{(c)}
=\begin{bmatrix}
(1+\mu_1)^{-1}&0&0&0\\
0&2&0&0\\
\frac{\mu_1(2+\mu_1)(1+c)\mu_2}{1+\mu_1}&0&1+\mu_1&0\\
0&0&0&1/2
\end{bmatrix},\qquad R_0=\begin{bmatrix}\mathsf I\\0\end{bmatrix}.
\]
Its eigenvalues are \((1+\mu_1)^{-1},1/2,1+\mu_1,2\).
Direct multiplication verifies the exterior projector
\[
\Pi_>^{(c)}
=\begin{bmatrix}
0&0&0&0\\0&1&0&0\\(1+c)\mu_2&0&1&0\\0&0&0&0
\end{bmatrix}.
\]
Thus \(\Pi_>^{(c)}R_0\) is the matrix in
\eqref{eq:dre-product-seeded-columns-app} with its lower-right entry
\(-1/2\) replaced by \(0\). The same column calculation gives
\[
\sigma_{\min}(\Pi_>^{(c)}R_0)=(1+c)\mu_2,\qquad
\|\Pi_>^{(c)}R_0\|=1,\qquad
\kappa(\Pi_>^{(c)}R_0)=\frac1{(1+c)\mu_2}.
\]

For the complete contours, take
\begin{equation}
\Gamma_<:\ |z|=1-\mu_1/4,\qquad
\Gamma_>=\partial\{1+\mu_1/4<|z|<9\}.
\label{eq:rr-product-contours-app}
\end{equation}
The interior circle and annular outer boundary are counterclockwise;
the annular inner boundary is clockwise, so \(\Gamma_>\) excludes
the closed unit disk.
The slow block is similar to
\(\operatorname{diag}((1+\mu_1)^{-1},1+\mu_1)\) through
\(\begin{bmatrix}1&0\\-(1+c)\mu_2&1\end{bmatrix}\),
whose condition number is at most \((1+(1+c)\mu_2)^2\).
All slow-pole distances on the near-unit circles are at least
\(\mu_1/2\), since
\[
1-\mu_1/4-(1+\mu_1)^{-1}
=\frac{\mu_1(3-\mu_1)}{4(1+\mu_1)}\ge\frac{\mu_1}{2}.
\]
The fast poles stay a constant distance away.
The slow block has norm below \(2\), whereas the fast block has norm
\(2\). Thus \(\|\mathcal S_F^{(c)}\|=2\), the common input
normalization is \(\alpha_{\mathcal S_F}=3\), and the outer boundary
\(|z|=9=3\alpha_{\mathcal S_F}\) contributes at most
\(12/(9-2)\) to the weighted factor.
At \(z=1+\mu_1/4\), a \((1+\mu_1)\)-eigenvector gives a resolvent
norm \(4/(3\mu_1)\). Together these estimates prove
\[
\mathcal R_{\rm RR}
=\max_{\chi\in\{<,>\}}\sup_{z\in\Gamma_\chi}
(|z|+3)\|(z\mathsf I-\mathcal S_F^{(c)})^{-1}\|
=\Theta(\mu_1^{-1}).
\]

The signal norms are \(2/3\), and
\[
\|(\mathcal S_F^{(1)}-\mathcal S_F^{(0)})/3\|
=\frac{\mu_1(2+\mu_1)\mu_2}{3(1+\mu_1)}
<\frac23\mu_1\mu_2.
\]
The Julia bound \eqref{eq:dre-product-oracle-separation-app}
therefore applies with \(\mathcal H_{\rm DRE}/3\) replaced by
\(\mathcal S_F/3\), giving
\(\|U_{\mathcal S_F^{(1)}/3}-U_{\mathcal S_F^{(0)}/3}\|
<4\mu_1\mu_2/3\). The seed is again public.

Summing the recursion gives
\(P_j^{(c)}=(1+c)\mu_2((1+\mu_1)^{2j}-1)e_1e_1^*\).
Its input difference follows from the first identity in
\eqref{eq:dre-product-output-difference-app} by replacing
\(e^{2\mu_1s}\) with \((1+\mu_1)^{2j}\).
For the public integer
\[
k=\left\lceil\frac{\log(1/\mu_2)}
 {2\log(1+\mu_1)}\right\rceil,\qquad
1\le\mu_2(1+\mu_1)^{2k}<(1+\mu_1)^2,
\]
we have \(\|P_j^{(c)}\|\le\|P_k^{(c)}\|<2(1+\mu_1)^2<3\)
throughout \(0\le j\le k\), and
\(\|P_k^{(1)}-P_k^{(0)}\|\ge1-\mu_2\ge7/8\).
With \(\alpha_{P_k}=3\) and decoded error \(1/16\), the same
hybrid calculation \eqref{eq:dynamic-product-hybrid-app} gives
\(q_{\rm RR}\ge3/(16\mu_1\mu_2)\), proving the stated product.
\end{proof}

The algebraic lower bounds use the actual output normalizations
constructed in Proposition~\ref{prop:algebraic-product-normalization-app},
which are comparable to the unbounded solution norms.
The dynamic outputs are bounded at the stated times or steps.
For general inputs, these examples leave open the optimal dependence
on additional construction normalizations and re-encoding costs.

\subsection{A single-control DARE circuit reduction}
\label{app:dare-bqp}

\newtheorem*{darebqprestatement}{Theorem~\ref{thm:dare-bqp-main}}
\begin{darebqprestatement}
There exists a uniformly circuit-generated family of single-control
instances of Problem~\ref{prob:quantum-dare} for which constructing a
solution block-encoding with \(\alpha_X=2\) and
\(\varepsilon\le1/100\) is \(\mathsf{BQP}\)-hard.
On the construction family in Appendix~\ref{app:dare-bqp},
distinguishing \(x_0^*Xx_0\ge2/9\) from \(x_0^*Xx_0\le1/18\),
for a specified basis vector \(x_0\), is \(\mathsf{BQP}\)-complete.
\end{darebqprestatement}

\begin{proof}
Use the common circuit convention in
Appendix~\ref{app:lower-bound-tools}.
Let \(V\) be the amplified real source circuit, starting on the work
basis state \(|0\rangle_{\rm work}\). Apply \(V\), copy its accepting
output bit to a fresh flag initialized to zero, and then undo \(V\).
The amplitude of \(|0\rangle_{\rm work}|1\rangle_{\rm flag}\) in the
result is
\[
{}_{\rm work}\langle0|V^*
\bigl(|1\rangle\langle1|_{\rm out}\otimes\mathsf I\bigr)
V|0\rangle_{\rm work},
\]
which is the source acceptance probability.
Write the gates of this augmented circuit as \(U_1,\ldots,U_N\),
where \(N\ge1\) is polynomial in the source size. Each gate acts on
a fixed number of registers and is real orthogonal. The symbols
\(V,U_j,N\) are local to this proof.

On clock levels \(0,\ldots,N\), set
\begin{equation}
\begin{aligned}
A&=\sum_{j=0}^{N-1}|j+1\rangle\langle j|\otimes U_{j+1},\\
x_0&=|0\rangle_{\rm clk}|0\rangle_{\rm work}|0\rangle_{\rm flag},
&
b&=|N\rangle_{\rm clk}|0\rangle_{\rm work}|1\rangle_{\rm flag},
\end{aligned}
\label{eq:dare-bqp-clock-app}
\end{equation}
and take \(B=b\), \(R_c=1\), and \(Q=G=bb^*\).
Unused binary-clock labels have zero drift and zero cost.
Different clock columns of \(A\) have orthogonal images, so
\(\|A\|=1\) and \(A^{N+1}=0\). Thus \(A\) is Schur stable:
the stabilizability and detectability assumptions of
Definition~\ref{def:problem-dare} hold. The single nonzero entry of
\(b\) gives exactly one control channel.

The candidate solution is
\begin{equation}
X=Q+\frac12\sum_{j=1}^{N}(A^*)^jQA^j.
\label{eq:dare-bqp-solution-app}
\end{equation}
It is positive semidefinite. Its terminal clock block is \(Q\), and
the other summands are supported on the distinct clock levels
\(N-j\); consequently \(QX=XQ=Q\).
Each \((A^*)^jQA^j\) is an orthogonal rank-one projector on its clock
level, because the corresponding gate product is orthogonal.
The terminal block therefore has norm one, every earlier block has
norm \(1/2\), and \(\|X\|=1\).

Since \(Q\) is a projector,
\[
(\mathsf I+GX)^{-1}=\mathsf I-\frac12Q,
\qquad
X(\mathsf I+GX)^{-1}=X-\frac12Q.
\]
Using \(A^{N+1}=0\), substitution into the DARE gives
\[
Q+A^*X(\mathsf I+GX)^{-1}A
=Q+\frac12A^*QA
+\frac12\sum_{j=2}^{N}(A^*)^jQA^j
=X.
\]
The closed loop is
\begin{equation}
F=(\mathsf I+GX)^{-1}A
=\left(\mathsf I-\frac12Q\right)A,
\qquad F^{N+1}=0.
\label{eq:dare-bqp-stability-app}
\end{equation}
Indeed, the left factor changes only amplitudes within the terminal
clock block and never reverses a clock step.
Thus \(X\) is stabilizing, and the uniqueness statement in
Definition~\ref{def:problem-dare} identifies it with the required
solution. The allowed singularity of \(A\) and of the pencil matrix
\(L\) is consistent with the finite stable-branch convention in
Section~\ref{subsec:dare-construction}.
Since \(FQ=QF^*=0\), the pencil triangularization in
Appendix~\ref{app:dare-construction} gives
\[
\|(zL-M)^{-1}\|=\Theta(N+1),\qquad |z|=1.
\]
The diagonal inverse blocks are finite contraction sums, the
off-diagonal block is \(Q/2\), and the transformations have bounded
norms; an uncontrolled history direction gives the lower bound.
A sufficiently large fixed multiple of \(N+1\) therefore supplies
the required nodewise \(\beta_j\).
With the constant input normalizations constructed below,
Proposition~\ref{prop:dare-circle-quadrature-app} supplies polynomially
many computable nodes and \(\mathcal R_{\rm DARE}=O(N+1)\).
Including the right factor \(L\), direct LCU gives
\(\alpha_{\Pi_<}=O(N+1)\).
Lemma~\ref{lem:dare-graph-recovery} supplies full-row-rank recovery;
\(\|X\|=1\) gives actual output scale \(O(N+1)\) and polynomial
resources for the DARE algorithm.

For the initial vector in \eqref{eq:dare-bqp-clock-app}, all terms in
\eqref{eq:dare-bqp-solution-app} except \(j=N\) vanish. Hence
\begin{equation}
x_0^*Xx_0
=\frac12|b^*A^Nx_0|^2
=\frac12\bigl(\Pr[\text{source circuit accepts}]\bigr)^2.
\label{eq:dare-bqp-selected-value-app}
\end{equation}
The source promises \(2/3\) and \(1/3\) give the two thresholds
\(2/9\) and \(1/18\).
The gate transformation, basis labels, and coefficient descriptions
are generated uniformly in classical polynomial time.
Each row or column of \(A\) uses one fixed-size gate, so its locators
and precision-adjustable entry evaluation also have polynomial cost.
To encode \(A\) with normalization one, encode the projector onto
clock labels \(0,\ldots,N-1\) by a reversible comparison, then apply
a unitary cyclic clock increment carrying \(U_{j+1}\) on \(j<N\)
and the identity elsewhere. The projector removes the unwanted
transitions. Together with normalization-one encodings of \(Q\),
the identity, and the fixed block selectors, the sum rule applied to
\eqref{eq:dare-pencil} gives the actual
\(\alpha_M=\alpha_L=3\). Adjoints and controls use the same finite
gate descriptions. This completes the polynomial-time reduction.

For contraction approximations with
\(\max_j\|\widetilde U_j-U_j\|\le\delta\), define
\(\widetilde A,\widetilde X\) by
\eqref{eq:dare-bqp-clock-app}--\eqref{eq:dare-bqp-solution-app}.
A product of \(j\) gates changes by at most \(j\delta\), so its
quadratic product changes by at most \(2j\delta\).
The disjoint clock supports and the factor \(1/2\) give
\[
\|\widetilde X-X\|\le N\delta.
\]
The same formula solves the perturbed DARE with a nilpotent closed
loop. Taking \(N\delta\) smaller than one quarter of the value gap
preserves the decision. Unitary gate approximations satisfy the
contraction condition, so the common finite-precision convention
suffices. The promised family is the exact circuit construction.

On this family, the supplied gate list allows execution of the
compute--copy--uncompute circuit and measurement of the known target
\(|0\rangle_{\rm work}|1\rangle_{\rm flag}\).
Its probability is \(2x_0^*Xx_0\); the constant gap gives membership
by a constant number of repetitions, hence
\(\mathsf{BQP}\)-completeness on this finite construction image.
Lemma~\ref{lem:selected-value-hardness-app}, with
\(\alpha_X=2\), \(\varepsilon=1/100\), and gap \(1/6\), then gives
the solution-encoding hardness.
\end{proof}

\subsection{Single-control DRE circuit reduction}
\label{app:dre-bqp}

\newtheorem*{drebqprestatement}{Theorem~\ref{thm:dre-bqp-main}}
\begin{drebqprestatement}
There exists a uniformly circuit-generated family of single-control,
zero-terminal LQR instances with polynomial horizon \(T\) whose
time-reversed DRE solution \(P(T)\) in
Problem~\ref{prob:quantum-dre} is \(\mathsf{BQP}\)-hard to
block-encode with \(\alpha_{P(T)}=4\) and \(\varepsilon\le1/100\).
On the construction family in Appendix~\ref{app:dre-bqp},
distinguishing \(x_0^*P(T)x_0>1/20\) from \(x_0^*P(T)x_0<1/200\),
for a specified basis vector \(x_0\), is \(\mathsf{BQP}\)-complete.
\end{drebqprestatement}
\begin{proof}
Amplification and realification give a polynomial-size real orthogonal
circuit for any \(\mathsf{BQP}\) language, with acceptance probability
\(p\ge15/16\) on yes instances and \(p\le1/1024\) on no instances
\cite{BernsteinVazirani1997QuantumComplexity,RudolphGrover2002Rebit}.
Introduce a fresh flag initialized to zero and copy the accepting bit
to it only in the final gate. Denote the resulting gates by
\(U_1,\ldots,U_\ell\), with \(\ell\ge1\), and append \(4\ell\)
identity gates. Set \(N=5\ell\) and interpret
\(U_j\cdots U_1\) as \(\mathsf I\) when \(j=0\). If
\(\lvert\psi_0\rangle\) is the initial work-and-flag basis state, then
\begin{equation}
\langle\psi_0\rvert (U_j\cdots U_1)^*
(\mathsf I_{\rm work}\otimes\lvert1\rangle\langle1\rvert_{\rm flag})
(U_j\cdots U_1)
\lvert\psi_0\rangle
=\begin{cases}0,&j<\ell,\\p,&j\ge\ell.\end{cases}
\label{eq:dre-bqp-flag-app}
\end{equation}
Realification adds one rebit, and each gate still acts on a fixed
number of rebits.

On clock states \(\lvert j\rangle\), \(0\le j\le N\), define
\begin{equation}
\begin{aligned}
A={}&-\frac{\mathsf I}{100N}
+\frac1{2N}\sum_{j=0}^{N-1}\sqrt{(j+1)(N-j)}
\bigl(\lvert j+1\rangle\langle j\rvert\otimes U_{j+1}
      -\lvert j\rangle\langle j+1\rvert\otimes U_{j+1}^*\bigr),\\
x_0={}&\lvert0\rangle\lvert\psi_0\rangle,\qquad
B=x_0,\qquad R_c=N,\qquad
Q=\frac1N\mathsf I_{\rm clk}\otimes\mathsf I_{\rm work}
\otimes\lvert1\rangle\langle1\rvert_{\rm flag},\qquad
T=\pi N.
\end{aligned}
\label{eq:dre-bqp-data-app}
\end{equation}
The matrix \(A+\mathsf I/(100N)\) is real and skew-symmetric.
The damping term acts on the entire encoded space, including unused
binary-clock labels. Changing basis by the block-diagonal orthogonal matrix
\(\sum_{j=0}^N\lvert j\rangle\langle j\rvert\otimes(U_j\cdots U_1)\)
transforms \(NA+\mathsf I/100\) on the valid clock levels into the clock matrix
\[
\frac12\sum_{j=0}^{N-1}\sqrt{(j+1)(N-j)}
\bigl(\lvert j+1\rangle\langle j\rvert
      -\lvert j\rangle\langle j+1\rvert\bigr)
\]
tensored with the work identity. Identify \(\lvert j\rangle\)
with the normalized Hamming-weight-\(j\) vector in the symmetric
subspace of \(N\) rebits. The clock matrix is the restriction of
\[
\frac12\sum_{a=1}^N
\bigl(\lvert1\rangle\langle0\rvert
      -\lvert0\rangle\langle1\rvert\bigr)_a.
\]
Its summands commute and each has eigenvalues \(\pm\mathrm i/2\).
The tensor powers of either one-rebit eigenvector are symmetric and
have eigenvalues \(\pm\mathrm iN/2\). Thus the restriction has norm
\(N/2\), proving \(\|A+\mathsf I/(100N)\|=1/2\) and
\(\|A\|<1\). Exponentiating the commuting one-rebit rotations,
including the scalar damping, gives
\begin{equation}
e^{NtA}x_0
=e^{-t/100}\sum_{j=0}^N\sqrt{\binom Nj}
\cos^{N-j}(t/2)\sin^j(t/2)
\lvert j\rangle(U_j\cdots U_1)\lvert\psi_0\rangle,
\qquad 0\le t\le\pi.
\label{eq:dre-bqp-transfer-app}
\end{equation}
In particular, \(e^{TA}x_0=e^{-\pi/100}\lvert N\rangle
(U_N\cdots U_1)\lvert\psi_0\rangle\).

For the dynamics \(\dot z=Az+Bu\), \(z(0)=x_0\), write
\(J_T(u)=\int_0^T(z^*Qz+Nu^2)\,\mathrm dt\).
The zero-control cost follows from
\eqref{eq:dre-bqp-flag-app}--\eqref{eq:dre-bqp-transfer-app}:
\begin{equation}
J_T(0)=p\int_0^\pi e^{-t/50}\sum_{j=\ell}^N
\binom Nj\cos^{2(N-j)}(t/2)\sin^{2j}(t/2)\,\mathrm dt.
\label{eq:dre-bqp-free-cost-app}
\end{equation}
For each \(t\), the summands give the probabilities of a binomial
variable with \(N\) trials and success probability \(\sin^2(t/2)\).
If \(2\pi/3\le t\le\pi\), the expected number of failures is at
most \(N/4\). Since \(\ell=N/5\), Markov's inequality bounds the
probability of fewer than \(\ell\) successes by
\((N/4)/(4N/5)=5/16\). The sum in
\eqref{eq:dre-bqp-free-cost-app} is therefore at least \(11/16\)
on this interval and at most one everywhere. Consequently,
\begin{equation}
e^{-\pi/50}\frac{11\pi}{48}p\le J_T(0)\le\pi p.
\label{eq:dre-bqp-free-gap-app}
\end{equation}

The single control can reduce this cost only by a constant factor.
Since \(\|Q^{1/2}\|=N^{-1/2}\), \(\|x_0\|=1\), and
\(\|e^{tA}\|=e^{-t/(100N)}\le1\), Young's convolution inequality gives
\begin{equation}
\left\|Q^{1/2}\int_0^t e^{(t-r)A}x_0u(r)\,\mathrm dr
\right\|_{L^2([0,T])}
\le\frac{T}{\sqrt N}\|u\|_{L^2([0,T])}
=\pi\sqrt N\|u\|_{L^2([0,T])}.
\label{eq:dre-bqp-control-scaling-app}
\end{equation}
Here the convolution kernel has \(L^1\) norm at most \(T\).
Variation of constants~\cite[Chap.~IX, Sec.~1.5]{Kato1995Perturbation} and Cauchy--Schwarz then imply
\[
\begin{aligned}
\sqrt{J_T(0)}
&\le\|Q^{1/2}z\|_{L^2([0,T])}
       +\pi\sqrt N\|u\|_{L^2([0,T])}\\
&\le\sqrt{1+\pi^2}\sqrt{J_T(u)}.
\end{aligned}
\]
The admissible choice \(u=0\) supplies the opposite comparison, so
\begin{equation}
\frac{J_T(0)}{1+\pi^2}
\le\inf_u J_T(u)\le J_T(0).
\label{eq:dre-bqp-one-control-app}
\end{equation}

To identify this value with the requested DRE solution, take the
physical data in \eqref{eq:dre-bqp-data-app} and set
\(G=x_0x_0^*/N\). The reverse-time equation is
\begin{equation}
P'(s)=Q+A^*P(s)+P(s)A-P(s)GP(s),\qquad P(0)=0,
\label{eq:dre-bqp-reversed-dre-app}
\end{equation}
which is Definition~\ref{def:problem-dre} with input
\((-A,-G,-Q,P_0=0)\). The polynomial right-hand side gives a unique
local symmetric solution. For any \(s\) in its existence interval,
Proposition~\ref{prop:lqr-conventions-app}, applied to \(P(s-t)\),
identifies \(P(s)\) as the value matrix for horizon \(s\).
Nonnegative cost and zero control give
\begin{equation}
0\preceq P(s)\preceq
\int_0^s e^{tA^*}Qe^{tA}\,\mathrm dt
\preceq\frac sN\mathsf I.
\label{eq:dre-bqp-solution-bound-app}
\end{equation}
This bound excludes finite-time blow-up on \([0,T]\), so local ODE
continuation proves existence throughout that interval and
\(\|P(s)\|\le\pi\). With
\(P_{\rm LQR}(t)=P(T-t)\), the terminal cost is zero and
\(x_0^*P(T)x_0=\inf_uJ_T(u)\). On yes instances,
\eqref{eq:dre-bqp-free-gap-app} and
\eqref{eq:dre-bqp-one-control-app} give
\[
x_0^*P(T)x_0
\ge e^{-\pi/50}\frac{15}{16}\frac{11\pi}{48(1+\pi^2)}
>\frac1{20}.
\]
On no instances,
\[
x_0^*P(T)x_0\le\frac\pi{1024}
<\frac1{200}.
\]

For the reversed lift
\(\mathcal H_{\rm DRE}=[-A,G;Q,A^*]\), take a unit vector
\(w=[x;y]\) and put
\(r=(\mathrm i\omega\mathsf I-\mathcal H_{\rm DRE})w\).
The block equations and \(0\preceq Q,G\preceq\mathsf I/N\) give
\[
x^*Qx+y^*Gy\le\|r\|,\qquad
\frac1{100N}\le\|r\|+\sqrt{\|r\|/N}.
\]
Thus \(\Delta_{\rm ax}=\Omega(N^{-1})\). Since \(A\) is Hurwitz,
the dual CARE in \eqref{eq:factor-care-dual-app} has a stabilizing
solution \(0\preceq Y\preceq50\mathsf I\), by the zero-control bound.
Its graph \(\operatorname{ran}[-Y;\mathsf I]\) is \(\ker\Pi_+\), so
Lemma~\ref{lem:factor-graph-distance-app}, with
\(R_0=[\mathsf I;0]\), gives
\begin{equation}
\sigma_{\min}(\Pi_+R_0)
\ge(1+\|Y\|^2)^{-1/2}\ge1/\sqrt{2501}.
\label{eq:dre-bqp-initialization-app}
\end{equation}

Choose the paired rectangles of
Proposition~\ref{prop:dre-weighted-quadrature-app} at a fixed
constant encoding scale \(\alpha_H\), and let \(d(z)\) be the
distance to the spectrum of
\(\operatorname{diag}(-A,A^*)/\alpha_H\).
This normal matrix has the known clock spectrum, including the padding
modes. Singular-value perturbation away from that spectrum and the
axis bound near it give, at every quadrature node,
\[
\|(z\mathsf I-\widehat{\mathcal H}_{\rm DRE})^{-1}\|
=\Theta\!\left(\frac1{d(z)+1/(\alpha_HN)}\right).
\]
Thus the required node estimates are computable in polynomial time.
Propositions~\ref{prop:dre-weighted-quadrature-app}
and~\ref{prop:dre-general-resources-app} apply with \(M=\pi\),
\(\mathcal R_{\rm DRE}=O(N)\), raw weighted-block scales \(O(N)\),
and actual output scale \(O(N^2)\), giving polynomial resources.

The gate list gives finite uniform matrix-access circuits: each row
or column of \(A\) meets at most two neighboring clock blocks through
local gates and has one diagonal damping entry; \(Q\) is diagonal,
and \(B,G\) are one-sparse.
The weights \(\sqrt{(j+1)(N-j)}/(2N)\), gate entries, and horizon
\(T=\pi N\) are computable to any requested precision in polynomial
time. With a binary clock, invalid labels have drift
\(-\mathsf I/(100N)\), zero cost, and no coupling to \(B\). Sparse-matrix block encodings
then supply the coefficient and lifted-input circuits at constant
normalization with precision-adjustable controlled adjoints
\cite{GilyenSuLowWiebe2019QSVT}, as required in
Appendix~\ref{app:lower-bound-tools}.

An independent \(\mathsf{BQP}\) procedure on this amplified image
simulates the constant-sparse Hermitian matrix
\(\mathrm i(A+\mathsf I/(100N))\) for time \(T=\pi N\) and
measures the flag. Removing the known scalar damping from
\eqref{eq:dre-bqp-transfer-app} gives acceptance probability \(p\),
so the two amplified ranges distinguish exactly the selected-value cases. Sparse Hamiltonian simulation has polynomial
cost at fixed error \cite{BerryChildsKothari2015HamiltonianSimulation}.
Paired skew-symmetric entry rounding of \(A+\mathsf I/(100N)\)
changes this unitary evolution by at most \(T\) times the matrix
error, by Duhamel's formula~\cite[Chap.~IX, Sec.~2.1]{Kato1995Perturbation}; horizon error is multiplied by \(1/2\). Constant sparsity therefore makes
\(O(\log N)\) additional precision bits sufficient. This establishes
membership on the promised finite construction image.

Finally, \(\|P(T)\|\le\pi<4\) permits the common normalization
\(\alpha_{P(T)}=4\). Lemma~\ref{lem:selected-value-hardness-app},
with gap \(9/200\) and \(\varepsilon\le1/100\), gives the stated
block-encoding hardness.
\end{proof}

The spectral, initialization, and resource estimates proved above show
that our algorithms first construct solution block-encodings for these
two families with polynomial resources and polynomial normalization
scales. To reach the fixed scales in the hardness statements, standard
uniform singular-value amplification
\cite[Theorem~30]{GilyenSuLowWiebe2019QSVT} is then applied to these
already constructed encodings. The bounds \(\|X\|/2=1/2\) and
\(\|P(T)\|/4\le\pi/4<1\) provide its fixed margins, allowing
normalization \(2\) for \(X\) and \(4\) for \(P(T)\).
Allocating decoded error \(\varepsilon/2\) to construction and
\(\varepsilon/2\) to this final rescaling gives the prescribed
accuracy with polynomial additional cost.


\end{document}